\documentclass[12pt,a4paper]{article}

\usepackage[UKenglish]{babel}
\usepackage{geometry}
\usepackage[toc, page]{appendix}
\usepackage{cite}
\usepackage[bf]{caption}
\usepackage{subcaption}
\usepackage[bookmarksopen, bookmarksnumbered, bookmarksopenlevel=2, hidelinks]{hyperref}

\hypersetup{
	pdftitle={Non-minimal Elliptic Threefolds at Infinite Distance I: Log Calabi-Yau Resolutions},
	pdfauthor={Rafael \'Alvarez-Garc\'ia, Seung-Joo Lee and Timo Weigand}
}

\def\fnote#1#2{\begingroup\def\thefootnote{#1}\footnote{#2}\addtocounter{footnote}{-1}\endgroup}

\usepackage{amsmath}
\usepackage{amssymb}
\usepackage{amsthm}
\usepackage{mathtools}
\usepackage{thmtools}
\usepackage{mathrsfs}

\numberwithin{equation}{section}
\numberwithin{table}{section}

\newtheorem{theorem}{Theorem}[section]
\newtheorem{lemma}[theorem]{Lemma}
\newtheorem{proposition}[theorem]{Proposition}
\newtheorem{corollary}[theorem]{Corollary}
\theoremstyle{definition}
\newtheorem{definition}[theorem]{Definition}
\theoremstyle{definition}
\newtheorem{example}[theorem]{Example}
\theoremstyle{remark}
\newtheorem{remark}[theorem]{Remark}

\makeatletter
\def\th@plain{%
  \thm@notefont{}% same as heading font
  \itshape % body font
}
\def\th@definition{%
  \thm@notefont{}% same as heading font
  \normalfont % body font
}
\makeatother

\usepackage[capitalise,noabbrev]{cleveref}
\usepackage{graphicx}
\usepackage{xcolor}
\usepackage{tabularray}
\UseTblrLibrary{diagbox}
\usepackage{tcolorbox}
\tcbset{title filled, arc=0mm}
\usepackage{enumitem}
\usepackage{accents}
\usepackage{trimclip}
\usepackage{rotating}
\usepackage[all]{nowidow}
\usepackage{microtype}

\definecolor{diagDarkPurple}{RGB}{159,64,116}
\definecolor{diagLightPurple}{RGB}{201,127,232}
\definecolor{diagDarkGreen}{RGB}{80,82,34}
\definecolor{diagLightGreen}{RGB}{110,139,32}
\definecolor{diagPurpleBlue}{RGB}{69,26,121}
\definecolor{diagDarkBlue}{RGB}{45,35,133}
\definecolor{diagLightBlue}{RGB}{32,107,190}
\definecolor{diagCyan}{RGB}{64,215,169}
\definecolor{diagDarkBrown}{RGB}{91,71,56}
\definecolor{diagLightBrown}{RGB}{119,107,84}
\definecolor{diagDarkRed}{RGB}{138,52,31}
\definecolor{diagMediumRed}{RGB}{214,69,32}
\definecolor{diagLightRed}{RGB}{214,127,127}
\definecolor{diagDarkYellow}{RGB}{116,75,26}
\definecolor{diagMediumYellow}{RGB}{181,124,16}
\definecolor{diagLightYellow}{RGB}{181,153,63}
\definecolor{mixPurpleGreen}{RGB}{156,133,132}
\definecolor{mixYellowGreen}{RGB}{146,146,48}

\usepackage{tikz}
\usetikzlibrary{cd, decorations.pathmorphing, shapes.misc}

\tikzstyle{black line}=[-, draw=black, line width=1.5pt, line cap=rect]
\tikzstyle{black fine-line}=[-, draw=black, line width=0.75pt, line cap=rect]
\tikzstyle{dashed black line}=[dashed, draw=black, line width=1.5pt, line cap=rect]
\tikzstyle{dotted black line}=[loosely dotted, draw=black, line width=1.5pt, line cap=rect]
\tikzstyle{light-yellow thick-line}=[-, draw=diagLightYellow, line width=1.5pt, line cap=rect]
\tikzstyle{medium-red thick-line}=[-, draw=diagMediumRed, line width=1.5pt, line cap=rect]

\tikzstyle{dark-purple line}=[-, draw=diagDarkPurple, line width=1.25pt]
\tikzstyle{light-purple line}=[-, draw=diagLightPurple, line width=1.25pt]
\tikzstyle{dark-green line}=[-, draw=diagDarkGreen, line width=1.25pt]
\tikzstyle{light-green line}=[-, draw=diagLightGreen, line width=1.25pt]
\tikzstyle{purple-blue line}=[-, draw=diagPurpleBlue, line width=1.25pt]
\tikzstyle{dark-blue line}=[-, draw=diagDarkBlue, line width=1.25pt]
\tikzstyle{light-blue line}=[-, draw=diagLightBlue, line width=1.25pt]
\tikzstyle{cyan line}=[-, draw=diagCyan, line width=1.25pt]
\tikzstyle{dark-brown line}=[-, draw=diagDarkBrown, line width=1.25pt]
\tikzstyle{light-brown line}=[-, draw=diagLightBrown, line width=1.25pt]
\tikzstyle{dark-red line}=[-, draw=diagDarkRed, line width=1.25pt]
\tikzstyle{medium-red line}=[-, draw=diagMediumRed, line width=1.25pt]
\tikzstyle{light-red line}=[-, draw=diagLightRed, line width=1.25pt]
\tikzstyle{dark-yellow line}=[-, draw=diagDarkYellow, line width=1.25pt]
\tikzstyle{medium-yellow line}=[-, draw=diagMediumYellow, line width=1.25pt, line cap=round]
\tikzstyle{light-yellow line}=[-, draw=diagLightYellow, line width=1.25pt]
\tikzstyle{medium-yellow snake}=[-, draw=diagMediumYellow, line width=1.25pt, decorate, decoration=snake]
\tikzstyle{light-yellow snake}=[-, draw=diagLightYellow, line width=1.25pt, decorate, decoration=snake]
\tikzstyle{light-yellow snake dotted}=[densely dotted, draw=diagLightYellow, line width=1.25pt, decorate, decoration=snake]
\tikzstyle{light-yellow dotted}=[densely dotted, draw=diagLightYellow, line width=1.25pt]
\tikzstyle{medium-red zigzag}=[-, draw=diagMediumRed, line width=1.25pt, decorate, decoration=zigzag]
\tikzstyle{medium-red zigzag thick line}=[-, draw=diagMediumRed, line width=1.5pt, decorate, decoration=zigzag]

\tikzset{cross-purple/.style={cross out, draw=diagLightPurple, line width=1.5pt, minimum size=2*(#1-\pgflinewidth), inner sep=0pt, outer sep=0pt},cross/.default={1pt}}
\tikzset{cross-green/.style={cross out, draw=diagLightGreen, line width=1.5pt, minimum size=2*(#1-\pgflinewidth), inner sep=0pt, outer sep=0pt},cross/.default={1pt}}
\tikzset{cross-yellow/.style={cross out, draw=diagLightYellow, line width=1.5pt, minimum size=2*(#1-\pgflinewidth), inner sep=0pt, outer sep=0pt},cross/.default={1pt}}
\tikzset{cross-blue/.style={cross out, draw=diagLightBlue, line width=1.5pt, minimum size=2*(#1-\pgflinewidth), inner sep=0pt, outer sep=0pt},cross/.default={1pt}}
\tikzset{cross-red/.style={cross out, draw=diagMediumRed, line width=1.5pt, minimum size=2*(#1-\pgflinewidth), inner sep=0pt, outer sep=0pt},cross/.default={1pt}}
\tikzset{cross-cyan/.style={cross out, draw=diagCyan, line width=1.5pt, minimum size=2*(#1-\pgflinewidth), inner sep=0pt, outer sep=0pt},cross/.default={1pt}}
\tikzset{cross-purple-green/.style={cross out, draw=mixPurpleGreen, line width=1.5pt, minimum size=2*(#1-\pgflinewidth), inner sep=0pt, outer sep=0pt},cross/.default={1pt}}
\tikzset{cross-yellow-green/.style={cross out, draw=mixYellowGreen, line width=1.5pt, minimum size=2*(#1-\pgflinewidth), inner sep=0pt, outer sep=0pt},cross/.default={1pt}}
\tikzset{cross-node/.style={path picture={\draw[diagLightYellow] (path picture bounding box.south east) -- (path picture bounding box.north west) (path picture bounding box.south west) -- (path picture bounding box.north east);}}}

\newcommand{\ord}[1]{\mathrm{ord}_{#1}}
\newcommand{\fphys}{f_{\mathrm{phys}}}
\newcommand{\gphys}{g_{\mathrm{phys}}}

\newcommand{\Dphys}{\Delta_{\mathrm{phys}}}
\newcommand{\ninftot}{n_{\infty}^{\mathrm{tot}}}
\newcommand{\Ghres}{\hat{G}_{\mathrm{res}}}
\newcommand{\bigslant}[2]{{\left.\raisebox{.2em}{$#1$}\middle/\raisebox{-.2em}{$#2$}\right.}}
\DeclareRobustCommand*{\phat}[1]{{\accentset{\trimbox{0pt 0.15ex}{\scalebox{0.3}{\ensuremath{\boldsymbol{(}}}}\!\trimbox{0pt 1.25ex}{\ensuremath{\string^}}\!\trimbox{0pt 0.15ex}{\scalebox{0.3}{\ensuremath{\boldsymbol{)}}}}}{#1}}}
\makeatletter
\newcommand{\smallbullet}{} % for safety
\DeclareRobustCommand\smallbullet{%
  \mathord{\mathpalette\smallbullet@{0.5}}%
}
\newcommand{\smallbullet@}[2]{%
  \vcenter{\hbox{\scalebox{#2}{$\m@th#1\bullet$}}}%
}
\makeatother

\begin{document}

% Preprint numbers
\begin{flushright}
{\tt\normalsize CTPU-PTC-23-44}\\
{\tt\normalsize ZMP-HH/23-14}
\end{flushright}

% Title, authors, addresses and emails
\vskip 40 pt
\begin{center}
{\large \bf%
Non-minimal Elliptic Threefolds at Infinite Distance I:\\ \vspace{2mm} Log Calabi-Yau Resolutions
} 

\vskip 11 mm

Rafael \'Alvarez-Garc\'ia,${}^{1}$
Seung-Joo Lee,${}^{2}$
and Timo Weigand${}^{1,3}$

\vskip 11 mm
\small ${}^{1}${\textit{II. Institut f\"ur Theoretische Physik, Universit\"at Hamburg,\\  Luruper Chaussee 149, 22607 Hamburg, Germany}}\\[3 mm]
\small ${}^{2}${\textit{%
Particle Theory and Cosmology Group, Center for Theoretical Physics of the Universe,\\ 
Institute for Basic Science (IBS), Daejeon 34126, Korea
}}\\[3 mm]
\small ${}^{3}${\textit{Zentrum f\"ur Mathematische Physik, Universit\"at Hamburg,\\ Bundesstrasse 55, 20146 Hamburg, Germany}}\\[3 mm]

\fnote{}{Email: \href{mailto:rafael.alvarez.garcia@desy.de}{rafael.alvarez.garcia{\fontfamily{ptm}\selectfont @}desy.de}, \href{mailto:seungjoolee@ibs.re.kr}{seungjoolee{\fontfamily{ptm}\selectfont @}ibs.re.kr}, \href{mailto:timo.weigand@desy.de}{timo.weigand{\fontfamily{ptm}\selectfont @}desy.de}}
\end{center}

% Abstract
\vskip 7mm
\begin{abstract}
We study infinite-distance limits in the complex structure moduli space of elliptic Calabi-Yau threefolds. In F-theory compactifications to six dimensions, such limits include infinite-distance trajectories in the non-perturbative open string moduli space. The limits are described as degenerations of elliptic threefolds whose central elements exhibit non-minimal elliptic fibers, in the Kodaira sense, over curves on the base. We show how these non-crepant singularities can be removed by a systematic sequence of blow-ups of the base, leading to a union of log Calabi-Yau spaces glued together along their boundaries. We identify criteria for the blow-ups to give rise to open chains or more complicated trees of components and analyse the blow-up geometry. While our results are general and applicable to all non-minimal degenerations of Calabi-Yau threefolds in codimension one, we exemplify them in particular for elliptic threefolds over Hirzebruch surface base spaces. We also explain how to extract the gauge algebra for F-theory probing such reducible asymptotic geometries. This analysis is the basis for a detailed F-theory interpretation of the associated infinite-distance limits that will be provided in a companion paper \cite{ALWPart2}.
\end{abstract}

\vfill
\thispagestyle{empty}
\pagenumbering{roman}
\setcounter{page}{0}
\newpage

% Table of contents
\tableofcontents
\newpage

\pagenumbering{arabic}
\setcounter{page}{1}

%auto-ignore

\section{Introduction and summary}

Infinite-distance trajectories in the moduli space of string theory have received considerable attention through the lens of the Swampland Program \cite{Vafa:2005ui}, the ongoing effort to distinguish those effective field theories that can be consistently coupled to gravity from those that cannot, see \cite{Brennan:2017rbf,Palti:2019pca,vanBeest:2021lhn,Grana:2021zvf,Agmon:2022thq} for reviews on the topic. The Swampland Distance Conjecture (SDC) \cite{Ooguri:2006in} states that as we traverse an infinite distance in the moduli space of a consistent theory of quantum gravity, an infinite tower of states becomes asymptotically massless at an exponential rate, eventually breaking the effective description of the theory. Accepting this premise, it is only natural to cogitate on the nature of the states becoming light and the kind of theories that we encounter at infinite distance.

The Emergent String Conjecture (ESC) \cite{Lee:2019wij} was proposed as a refinement of the SDC to address these questions. It claims that infinite-distance limits in moduli space are either pure decompactification limits, in which the infinite tower of states is furnished by Kaluza-Klein replicas, or emergent string limits, in which we transition to a duality frame determined by a unique, emergent, critical and weakly coupled string. It has survived non-trivial tests in various corners of the moduli space, like in the K\"ahler moduli space of F/M/IIA-theory in 6D/5D/4D in \cite{Lee:2019wij,Lee:2018urn,Lee:2019xtm,Rudelius:2023odg}, in the complex structure moduli space of F-theory in 8D in \cite{Lee:2021qkx,Lee:2021usk}, in the 4D $\mathcal{N} = 2$ hypermultiplet moduli space of Type IIB in \cite{Marchesano:2019ifh,Baume:2019sry}, in M-theory on $\mathrm{G}_{2}$ manifolds in \cite{Xu:2020nlh}, and in 4D $\mathcal{N} = 1$ F-theory in \cite{Lee:2019tst,Klaewer:2020lfg}. The possibility of membrane limits, in which a fundamental membrane becomes light at leading parametric scale, was excluded in \cite{Alvarez-Garcia:2021pxo}. The nature of the towers has consequences for the asymptotic vanishing rates appearing in the SDC \cite{Rudelius:2023mjy,Etheredge:2023odp}, the species scale \cite{vandeHeisteeg:2023ubh}, and for the realisation \cite{Lee:2021usk,Marchesano:2022axe,Blumenhagen:2023yws,Blumenhagen:2023tev} of the emergence proposal \cite{Harlow:2015lma, Heidenreich:2018kpg,Grimm:2018ohb}.

Closed-moduli infinite-distance limits have been thoroughly explored across various contexts, providing substantial supporting evidence for the SDC\,---\,see \cite{Grimm:2018ohb,Grimm:2018cpv,Grimm:2019ixq} for works regarding the Type II/M-theory complex structure moduli\,---\,and the ESC; their open-moduli counterparts, by contrast, have remained relatively neglected in the existing literature. F-theory constitutes a natural setting in which to study this class of infinite-distance limits. In it, what one would naively call the open moduli is in fact part of the complex structure moduli space of an elliptic fibration; finite-distance complex structure deformations rendering the internal elliptic fibration singular over codimension-one loci can be physically interpreted as moving a collection of \mbox{7-branes} on top of each other. The stack that they form supports a non-abelian gauge algebra given by a finite number of states becoming light along the trajectory in the moduli space. Since we are concerned with the study of quantum gravity, the internal space in which the 7-branes move is compact. Interestingly, varying the location of 7-branes within this space can still correspond to an infinite-distance limit in the complex structure moduli space of F-theory, if the brane stack includes a suitable group of mutually non-local $[p,q]$ 7-branes \cite{DeWolfe:1998zf,DeWolfe:1998pr,DeWolfe:1998eu}.

These are the open-moduli infinite-distance limits that we analyse, extending previous work in eight-dimensional F-theory \cite{Lee:2021qkx,Lee:2021usk} to the geometrically far richer framework of six-dimensional F-theory. Our first objective is to find a useful geometric description of the limiting points, which is the content of the present work. We exploit these results in \cite{ALWPart2} to provide a physical interpretation of the infinite-distance limits under scrutiny, with the assessment of the ESC threading the discussion. Our geometric approach is complementary to the analysis of complex structure degenerations using the formalism of asymptotic Hodge theory \cite{Grimm:2018ohb,Grimm:2018cpv,Grimm:2019ixq}. It would be extremely beneficial to connect these two approaches, which respectively take a more geometric or algebraic viewpoint, in the future.

While we have framed the discussion thus far from the perspective of the Swampland Program, the properties of these limits are also intriguing from an intrinsically F-theoretic standpoint. Above we alluded to the existing relation between stacks of 7-branes carrying a non-abelian gauge algebra and the singularities of the internal elliptic fibration in codimension-one. Increasing the number of $[p,q]$ 7-branes in the stack suitably until we reach an infinite-distance point in the~complex structure moduli space also has a manifestation in the internal geometry, namely, the singular elliptic fibers over the locus become non-minimal. By this, we mean that the vanishing orders of the sections $f$ and $g$ which enter the Weierstrass equation 
\begin{equation}
    y^2 = x^3 + f x z^4 + g z^6
\label{eq:Weierstrass}
\end{equation}
of the elliptic fiber exceed the Kodaira bound that either $\mathrm{ord}(f) < 4$ or $\mathrm{ord}(g) < 6$. Such singularities behave radically different to their minimal analogues, since they do not admit a crepant resolution in the fiber. For this reason, F-theory models with codimension-one non-minimal fibers are usually discarded.

However, non-minimal singularities can still be resolved while preserving the flatness of the elliptic fibration through a sequence of base blow-ups. This turns the internal space at the endpoint of the infinite-distance limit into an arrangement of log Calabi-Yau components appropriately glued together. The study of log Calabi-Yau degenerations is a fascinating endeavour in its own right, constituting an active discipline in mathematics that has also percolated to the physics literature \cite{Donagi:2012ts}. The stable degeneration limit sometimes taken in the F-theory literature \cite{Morrison:1996na,Morrison:1996pp,Aspinwall:1997ye,Clingher:2003ui} is an example of this phenomenon.

The equivalent problem for F-theory compactifications to eight dimensions leads to the theory of degenerations of elliptic K3 surfaces. Their resolutions give rise to so-called Kulikov models, which for general K3 surfaces are classified into Type I, Type II and Type III \cite{Kulikov1977,Persson1977,FriedmanMorrison1983}. Models of Type I correspond to finite-distance degenerations, while Type II and III models lie at infinite distance. For infinite-distance degenerations of elliptic K3 surfaces, a finer subdivision \cite{alexeev2021compactifications,Brunyantethesis, ascher2021compact, odaka2021collapsing, odaka2020pl,Lee:2021qkx} into models of Type II.a, II.b, III.a and III.b parallels the physics interpretation of the associated infinite-distance limits \cite{Lee:2021usk} as (possibly partial) decompactification or weak coupling limits.

Open-moduli infinite-distance limits in six-dimensional F-theory can not only be produced by stacking suitable branes together as explained above, a codimension-one effect, but also by forcing them to intersect with high enough multiplicity, a codimension-two effect. This leads to non-minimal elliptic fibers supported over points in the base of the internal elliptic fibration. When this occurs at finite distance in the moduli space, the degenerations encode strongly coupled SCFTs in six dimensions. These have been understood in a major classification effort, see the review \cite{Heckman:2018jxk} and references therein. Their counterparts at infinite distance, by contrast, remain mysterious. Although we include occasional remarks in the text, our primary focus lies in comprehending non-minimal F-theory models in codimension one, deferring the exploration of codimension-two cases for future research.

\subsection*{Summary of Part I}

We now guide the reader through the results of the article.

The infinite-distance limits we are considering are most conveniently formulated in the language of semi-stable degenerations of elliptic Calabi-Yau threefolds. As we explain in \cref{sec:definition-of-degenerations}, the idea is to consider a family of Weierstrass models $Y_u$ parametrised by a complex \mbox{parameter $u$}; then, for $u=0$, non-minimal singularities in the elliptic fiber arise over curves in the base of the fibration. As noted above, these singularities are non-crepant and hence do not admit a Calabi-Yau resolution in the fiber. Our strategy is, instead, to blow up the base of the fibration in a suitable fashion. This can be achieved without compromising the~Calabi-Yau condition, but at the cost of making the compactification space reducible: The original Calabi-Yau is replaced by a union of elliptic threefolds glued together over divisors. The general structure of these base blow-ups is described in \cref{sec:modifications-of-degenerations}, and made concrete in many examples thereafter. Upon performing a suitable number of blow-ups, the central element of the family is free of non-minimal singularities in codimension-one. The criterion for this is formulated in terms of the so-called family vanishing orders, introduced in \cref{sec:orders-of-vanishing}, which compute the vanishing orders of defining polynomials of the Weierstrass model of the family of threefolds; these are to be distinguished from the component vanishing orders, which are defined by restriction to the components of the central element. The latter contain information on the 7-brane content of the Weierstrass model. If the family vanishing orders are minimal, but the component vanishing orders are not, we cannot read off this brane content, nor can we perform a base blow-up. However, as we explain in \cref{sec:obscured-infinite-distance-limits}, such obscured non-minimalities can be removed by first performing a base change $u \mapsto u^k$ for an appropriate integral $k > 1$ making both vanishing orders agree.

The key point of our analysis, and that of \cite{Lee:2021qkx,Lee:2021usk} for elliptic K3 surfaces, is that the non-minimalities can be grouped into five classes according to their degree of non-minimality, which we display in \cref{def:class-1-5}. If both $f$ and $g$ in \eqref{eq:Weierstrass} vanish to order higher than $4$ and $6$, respectively, we speak of a Class~5 non-minimality. Importantly, it turns out that such degenerations  can always be transformed either into degenerations with only minimal Kodaira type vanishing orders, which arise at finite distance in the moduli space, or into the less extreme degenerations of Class~1--4, where either $\mathrm{ord}(f) = 4$ or $\mathrm{ord}(g) = 6$. The transformations we have in mind are combinations of base changes and the set of birational transformations given by blow-ups and blow-downs of the base. That such transformations must exist is, generally, a consequence of Mumford's Semi-stable Reduction Theorem \cite{Mumford1973} because a resolution of Class~5 singularities would lead to degenerations which are not semi-stable. However, this proof is not constructive. We therefore go a substantial step beyond this general statement for at least a subclass of configurations, namely, codimension-one non-minimal fibers on Hirzebruch surfaces. For these, we are able to describe an explicit algorithm to transform the Class~5 singularities into milder ones. This technical discussion is the content of \cite{ALWClass5}. 

The importance of these considerations comes from the fact that the different components after the blow-up, as required for Class~1--4 non-minimalities, have a rather constrained form: The fiber over general points can only be of Kodaira type $\mathrm{I}_{m}$, with $m > 0$ for Class 4 and $m=0$ otherwise. Regions with generic $\mathrm{I}_{m>0}$ fibers are interpreted as weakly coupled regions, where the string coupling $g_s \rightarrow 0$, while non-perturbative effects can only occur in those with generic $\mathrm{I}_{0}$~fibers. This will be very relevant for the interpretation of the degenerations from a physics point of view in \cite{ALWPart2}.

In \cref{sec:curves-of-non-minimal-fibers} we establish two important results that further constrain the possible degenerations and the structure of the blow-ups. First, we show that non-minimal fibers can be tuned only over curves of genus zero or one, and in the second case only over an anti-canonical divisor of the base. This greatly restricts the possibilities. We explicitly show this for all possible types of base spaces, an analysis which we delegate to \cref{sec:genus-restriction-proof}. 

Next, we study the resolutions in more detail. First, we define what one could call the simplest type of infinite-distance limit, namely one where essentially only a single curve supports codimension-one non-minimal fibers. As it turns out, this statement is well-defined only up to base change and birational transformations, but modulo these complications leads to a clear restriction on the types of curves over which non-minimal singularities can occur. Such so-called single infinite-distance limit degenerations, defined precisely in \cref{def:single-infinite-distance-limit}, admit resolutions whose components form an open chain, as in \cref{fig:semi-stable-degeneration}. The proof of this statement, which is our \cref{prop:single-infinite-distance-limits-and-open-resolutions}, requires again explicit checks for the possible base spaces and is performed in \cref{sec:single-infinite-distance-limits-and-open-chain-resolutions}. If we depart from these simple single infinite-distance limits, one instead finds trees of intersecting resolution components, see \cref{sec:analysis-general-case}.

In \cref{sec:geometry-components-single-limit} we analyse the components of the resolution, focusing only on degenerations over genus-zero curves. For single infinite-distance limits, the exceptional components of the base of the Weierstrass model are all Hirzebruch surfaces of type $\mathbb{F}_{n}$, where $n$ is the self-intersection of the blow-up curve. The Weierstrass models over the individual base components give rise to elliptic fibrations which by themselves are not Calabi-Yau; they are, however, glued together to form a reducible Calabi-Yau space, which we should view as the resolved compactification space. The individual components are, in fact, log Calabi-Yau spaces, as we explain in detail in \cref{sec:line-bundles-components-single-limit}. In degenerations which are not of the single infinite-distance limit type, the individual base components are Hirzebruch surfaces and blow-ups thereof, which intersect in a more complicated tree structure. We highlight this phenomenon in \cref{sec:analysis-general-case}, with derivations and examples contained in \cref{sec:res-trees}. To conclude the general discussion, we also briefly comment in \cref{sec:comments-genus-one-degenerations} on degenerations over curves of genus-one, whose in-depth study, however, is beyond the scope of this paper. 

All results so far have been obtained for elliptic fibrations over any of the allowed base spaces $\mathbb{P}^{2}$, $\mathbb{F}_{n}$, and their blow-ups. Of particular interest for us are degenerations over Hirzebruch surfaces, due to the duality with the heterotic string. Indeed, in \cite{ALWPart2} we will focus on this class of models in proposing an interpretation of the infinite-distance degenerations from a physics point of view. For this reason, we devote \cref{sec:degenerations-Hirzebruch-models} to an explicit investigation of the single infinite-distance limits over genus-zero curves on Hirzebruch surfaces. Apart from the relation to the heterotic string, this is also motivated by the fact that some of these degenerations admit toric resolutions, which serves as an independent check and illustration of our general results. First, we classify the possible genus-zero and genus-one curves which can support non-minimal degenerations of single infinite-distance limit type in \cref{sec:single-Hirzebruch}. Focusing in the sequel on the much richer class of genus-zero degenerations, we can naturally group them into models with non-minimal fibers over the $(\pm n)$-curves of the Hirzebruch surface (horizontal models), over a fiber (vertical models), or over a mixed curve (mixed models), which we analyse in turn in the following subsections. The types of resolutions one obtains for the base space and the types of Weierstrass models over the blown up base are summarised concisely in \cref{tab:genus-zero-Hirzebruch-summary}.

As an immediate consequence of these geometric results, we conclude that the infinite-distance behaviour of F-theory is encoded in the way in which the theory probes the chain of log~\mbox{Calabi-Yau} spaces, glued together along their boundaries. While this question is reserved for \cite{ALWPart2}, there is one aspect which we analyse already in this article, and which concerns the structure of the discriminant of the Weierstrass model. In F-theory on elliptic Calabi-Yau spaces, the Kodaira types of singular fibers over the irreducible components of the discriminant determine the non-abelian gauge algebra. When the elliptic fibration factors into log Calabi-Yau spaces, there arise a number of subtleties that demand our attention. First, special care must be taken in determining the components of the discriminant which correspond to global divisors from the point of view of the resolved space. An analysis of the discriminant in each component of the base can be misleading, for two reasons: Two divisors may coincide in one component of the base while being clearly distinct in others. Or a single such global divisor may factor into various irreducible components in some of the base components. We illustrate these two phenomena, which are special to the factorisation of the resolution space, in \cref{sec:illustrative-example} and \cref{sec:overcounting-example}, and explain how to read off the correct vanishing orders of the discriminant in both types of configurations. Armed with this intuition, we give the general prescription for reading off the components of the discriminant in \cref{sec:physical-discriminant}. Technically, we must perform a suitable factorisation of the discriminant polynomial modulo the product of all exceptional coordinates. We also comment on the proper formulation of this intuitively clear concept in terms of ideal theory. When the dust settles, we can assign, as summarised in \cref{sec:algorithm-gauge-algebra}, to each discrimi\-nant component a non-abelian gauge algebra, for which one also has to take into account that monodromies in the fiber may act only locally over some of the base components, as stressed in \cref{sec:monodromy-cover}. This gauge algebra, however, may enhance to a higher algebra as a consequence of a second effect not yet taken into account, namely the fact that the infinite-distance limits may describe partial decompactifications of the originally six-dimensional effective theory. The analogous effect was discussed in eight dimensions in \cite{Lee:2021qkx,Lee:2021usk} for F-theory on non-minimal K3 surfaces. There, the enhanced gauge algebra after (possibly partial) decompactification combines with the Kaluza-Klein $\mathrm{U}(1)$ factors into a loop algebra (see \cite{Collazuol:2022jiy,Collazuol:2022oey} for the dual heterotic effect). We leave a discussion of the analogous effect for F-theory on Calabi-Yau threefolds to the physics analysis in \cite{ALWPart2}.
%auto-ignore

\section{Geometric description of 6D F-theory limits}
\label{sec:geometric-description-6D-F-theory-limits}

Infinite-distance limits in the complex structure moduli space of six-dimensional F-theory can be  described geometrically in the language of degenerations of elliptically fibered threefolds, a notion that we review in \cref{sec:definition-of-degenerations}. The elliptically fibered threefolds undergoing the degeneration must have as their base one of the allowed six-dimensional F-theory bases, whose geometry we recall in \cref{sec:six-dimensional-F-theory-bases}.

The same infinite-distance limit may be represented by various degenerations, which differ in their geometrical representative of the endpoint of the limit. In \cref{sec:modifications-of-degenerations} we discuss how to obtain the degenerations in which this geometrical representative has the most convenient form for our purposes, and which will be the one used throughout the rest of the document and in \cite{ALWPart2} to extract the physics.

The infinite-distance nature of the limits we study is associated with the presence of non-minimal fibers over certain curves in the central fiber of the degeneration. These curves can only be of genus zero or genus one, as we discuss in \cref{sec:curves-of-non-minimal-fibers} and proof in \cref{sec:genus-restriction-proof}. We will furthermore introduce the notion of so-called single infinite-distance limits, and show that their resolution always takes the form of an open chain of log Calabi-Yau spaces, relegating most of the technicalities of the discussion to \cref{sec:single-infinite-distance-limits-and-open-chain-resolutions}.

Genus-one curves supporting non-minimal singularities only occur in highly tuned models, and we hence focus on the much more prevalent genus-zero degenerations for the rest of the article. Their geometrical properties are studied in great detail in \cref{sec:geometry-components-single-limit,sec:line-bundles-components-single-limit,sec:analysis-general-case}, with additional technical remarks contained in \cref{sec:obscured-infinite-distance-limits,sec:res-trees,sec:single-infinite-distance-limits-and-open-chain-resolutions,sec:restricting-star-degenerations,sec:blowing-down-vertical-components}. We conclude the section with some comments on genus-one degenerations in \cref{sec:comments-genus-one-degenerations}, leaving their systematic study for future work. While the discussion is kept general throughout, we include numerous examples illustrating each of the features analysed.

\subsection{Semi-stable degenerations of Calabi-Yau threefolds}
\label{sec:definition-of-degenerations}

Our goal is to study infinite-distance limits in the complex structure moduli space of F-theory compactified on elliptically fibered Calabi-Yau threefolds and the associated physics. Throughout the text, we will focus on those limits that can be described by one parameter, which we will denote by $u$. In order to formulate this mathematically, we will use the algebro-geometric language of degenerations.

Let $D := \left\{u \in \mathbb{C} : |u| < 1 \right\}$ be the unit disk.\footnote{More generally, in the context of the semi-stable reduction theorem \cite{Mumford1973} one can substitute $D$ for a non-singular curve $C$ with an origin, i.e.\ with a marked point $0 \in C$. In those concrete examples that we treat by toric methods, we will substitute $D$ for $\mathbb{C}$; this does not affect any of the properties relevant to us.} A one-parameter family of varieties is a variety\footnote{By variety we will always mean an algebraic variety over the field $\mathbb{C}$.} $\hat{\mathcal{Y}}$ together with a morphism
\begin{equation}
	\hat{\rho}: \hat{\mathcal{Y}} \longrightarrow D\,.
\end{equation}
The members of the family are given by the fibers $\hat{Y}_{u} := \rho^{-1}(u)$, with $u \in D$. We will often denote the family simply as $\hat{\mathcal{Y}}$. By distinguishing the central fiber $\hat{Y}_{0}$, we can see this as a degeneration, in which the elements $\hat{Y}_{u \neq 0}$ of the family degenerate to $\hat{Y}_{0}$. A degeneration is called semi-stable if $\hat{\mathcal{Y}}$ is smooth and the central fiber $\hat{Y}_{0}$ is reduced with local normal crossings.

Denote by $D^{*}$ the punctured unit disk. A modification of $\rho: \hat{\mathcal{Y}} \rightarrow D$ is another family of varieties $\rho': \hat{\mathcal{Y}}' \rightarrow D$ such that there exists a birational morphism $f: \hat{\mathcal{Y}} \rightarrow \hat{\mathcal{Y}}'$ that is compatible with the projections to $D$ and an isomorphism over $D^{*}$.

Intuitively speaking, $D$ is a small patch in the moduli space of the Calabi-Yau under consideration. A degeneration $\rho: \hat{\mathcal{Y}} \rightarrow D$ and its modifications and base changes all describe the same limit\footnote{For a precise mathematical definition of equivalent degenerations (particularized for degenerations of elliptic K3 surfaces) see \cite{Clingher:2003ui}.} in $D^{*}$, approaching the boundary point at $u = 0$. They will present, however, different central fibers, which constitute equally valid geometrical representatives of the endpoint of the limit. To extract the physics of the infinite-distance limits, we will birationally transform the family, with the aim of finding a modification from which the physical information can be most directly read off.

For degenerations of K3 surfaces, there, in fact, exists a canonical model of the central fiber. The associated degenerations are the Kulikov degenerations, which we review in \cite{ALWPart2} and were studied in the context of F-theory in \cite{Lee:2021qkx,Lee:2021usk}. For degenerations of Calabi-Yau threefolds, such a canonical geometrical representative for the endpoint of the limit is no longer available.

In the context of six-dimensional compactifications of F-theory, we are interested in (degenerations of) genus-one fibered Calabi-Yau threefolds. Throughout the text, we will assume the existence of a section, i.e.\ focus on elliptically fibered Calabi-Yau threefolds. A Calabi-Yau $\hat{Y}$ within this class can be described as a Weierstrass model over a twofold base $\hat{B}$ given by the hypersurface
\begin{equation}
    \hat{Y}: \quad y^{2} = x^{3} + f x z^{4} + g z^{6}\,,
\end{equation}
with discriminant
\begin{equation}
    \Delta := 4 f^{3} + 27 g^{2}\,,
\end{equation}
in the ambient $\mathbb{P}_{231}$-bundle over $\hat{B}$
\begin{equation}
	\mathbb{P}_{231}\left( \mathcal{E} \right) := \mathbb{P}_{231}\left( \mathcal{L}^{2} \oplus \mathcal{L}^{3} \oplus \mathcal{O} \right)\,.
\end{equation}
Here, the holomorphic line bundle\footnote{We will often not distinguish between divisors and their associated line bundles, denoting both objects by the same symbol. It will also be implicit, but clear from context, if we are referring to a divisor class or to a concrete representative of it.} is $\mathcal{L} = \overline{K}_{\hat{B}}$, where $\overline{K}_{\hat{B}}$ denotes the anticanonical class of $\hat{B}$, such that the Calabi-Yau condition $c_{1}\left( \mathcal{L} \right) = c_{1}\left( \overline{K}_{\hat{B}} \right)$ is fulfilled. The defining polynomials $f$ and $g$ of the Weierstrass model and the discriminant $\Delta$ are global holomorphic sections
\begin{equation}
	f \in \Gamma\left( \hat{B}, \overline{K}_{\hat{B}}^{\otimes 4} \right)\,,\quad g \in \Gamma\left( \hat{B}, \overline{K}_{\hat{B}}^{\otimes 6} \right)\,,\quad \Delta \in \Gamma\left( \hat{B}, \overline{K}_{\hat{B}}^{\otimes 12} \right)\,.
\end{equation}
For these  to exist, $\overline{K}_{\hat{B}}$ must be effective.

To construct a family of elliptically fibered Calabi-Yau threefolds, we consider a relative version of this Weierstrass model by taking the family base 
\begin{equation}
    \hat{\mathcal{B}} = \hat{B} \times D
\end{equation}
and promoting $f$, $g$ and $\Delta$ to global holomorphic sections of $\overline{K}_{\hat{\mathcal{B}}}^{\otimes m}$. We therefore consider the degeneration $\rho: \hat{\mathcal{Y}} \rightarrow D$ with fibers
\begin{equation}
    \hat{Y}_{u}: \quad y^{2} = x^{3} + f_{u} x z^{4} + g_{u} z^{6}\,,
\end{equation}
and hence the elliptic fibration naturally extends to the family of Calabi-Yau threefolds. Note that $\hat{\mathcal{B}}$ is itself a (trivial) family of two-dimensional complex varieties with $\hat{B}_{u} = \hat{B}$ for all $u \in D$.

In practical terms, given a set of homogeneous coordinates $\left\{ x_{i} \right\}_{i \in \mathcal{I}}$ describing the surface $\hat{B}$, a Weierstrass model over $\hat{B}$ is fixed by choosing two defining polynomials $f = f(x_{i})$ and $g = g(x_{i})$, homogeneous under the $\mathbb{C}^{*}$-actions with degrees such that they are global holomorphic sections of the line bundles $F = 4 \overline{K}_{\hat{B}}$ and $G = 6\overline{K}_{\hat{B}}$, respectively. In case a global description of $\hat{B}$ in terms of homogeneous coordinates is not available, the same explicit construction can be done in patches using local coordinates. Then, in order to pass from a fixed Weierstrass model $\pi_{\mathrm{ell}}: \hat{Y} \rightarrow \hat{B}$ to the elliptically fibered fourfold $\Pi_{\mathrm{ell}}: \hat{\mathcal{Y}} \rightarrow \hat{\mathcal{B}}$, which represents the one-parameter family of threefolds $\hat{Y}_{u}$, we simply introduce a $u$ dependence into the defining polynomials $f_{u} = f_{u}(x_{i},u)$ and $g_{u} = g_{u}(x_{i},u)$. As mentioned already, the base of the family variety is taken to be $\hat{\mathcal{B}} = \hat{B} \times D$. Since the class of the divisor $\mathcal{U} := \{ u=0 \}_{\hat{\mathcal{B}}} = \pi^{*}_{D}(0)$ is trivial, we obtain $\overline{K}_{\hat{\mathcal{B}}} = \pi^{*}_{\hat{B}} \left( \overline{K}_{\hat{B}} \right )$, where the maps are the projections $\pi_{\hat{B}}: \hat{B} \times D \rightarrow \hat{B}$ and $\pi_{D}: \hat{B} \times D \rightarrow D$. This means that $u$ can appear with arbitrary degrees in $f_{u}$ and $g_{u}$. The effect of varying $u$ is to vary the monomial coefficients in $f_{u}$ and $g_{u}$, and therefore $u$-trajectories correspond to complex structure deformations in the moduli space of six-dimensional F-theory. In what follows, we will denote by $f$, $g$ and $\Delta$ the defining polynomials of the family of Weierstrass models, and only use the subscript when we want to highlight that we are working with a concrete fiber of the family associated to a set value of $u$, most commonly the central fiber $\hat{Y}_{0}$.

Consider a fixed  element $\hat{Y}_{u}$ of the family $\hat{\mathcal{Y}}$. The elliptic fiber of such a threefold will be singular over the divisor of the base $\hat{B}_{u}$ defined by the vanishing locus of the discriminant $\Delta_{u}$. The type of fibral singularity over a given point in the base $\hat{B}_{u}$ can be read off from the vanishing orders\footnote{Various notions of vanishing order will be at play when analysing the geometry of infinite-distance complex structure limits. We will define what we exactly mean by each of them in \cref{sec:orders-of-vanishing}.} of the defining polynomials of the Weierstrass model at that locus. The non-minimal singularities that may lie at infinite distance in the moduli space are those with
\begin{subequations}
\begin{align}
    \textrm{codimension-one:}&\qquad \ord{\hat{Y}_{u}}(f_{u},g_{u},\Delta_{u})_{\mathcal{D}} \geq (4,6,12)\,,\\
    \textrm{codimension-two:}&\qquad \ord{\hat{Y}_{u}}(f_{u},g_{u},\Delta_{u})_{p} \geq (8,12,24)\,.
\end{align}
\label{eq:infinite-distance-non-minimal-vanishing-orders}%
\end{subequations}
Here $\mathcal{D}$ and $p$ denote an irreducible divisor and a point of $\hat{B}_{u}$, respectively. We will refer to vanishing orders \eqref{eq:infinite-distance-non-minimal-vanishing-orders} as infinite-distance non-minimal vanishing orders.\footnote{While infinite-distance complex structure degenerations are associated to infinite-distance non-minimal vanishing orders, the converse is not necessarily true. Some degenerations presenting loci with infinite-distance non-minimal vanishing orders can be seen to lie at finite distance, a phenomenon on which we comment in \cref{sec:class-1-5-models}.} Note that over points, vanishing orders satisfying
\begin{equation}
	(4,6,12) \leq \ord{\hat{Y}_{u}}(f_{u},g_{u},\Delta_{u})_{p} < (8,12,24)\,,
\end{equation}
are not associated to infinite-distance points in the moduli space, even though they are non-minimal. Such vanishing orders will be called finite-distance non-minimal vanishing orders. We will assume that the generic elements $\hat{Y}_{u \neq 0}$ of the family do not present any infinite-distance non-minimal fibral singularities, while the central fiber $\hat{Y}_{0}$ does, corresponding to the fact that the degeneration $\hat{\mathcal{Y}}$ potentially represents an infinite-distance limit in complex structure moduli space.

The family $\hat{\mathcal{Y}}$ can be birationally transformed in such a way that the central fiber decomposes into a union of threefolds with normal crossings, as we discuss in detail in \cref{sec:modifications-of-degenerations}. This is achieved by performing a series of blow-ups along the base $\hat{\mathcal{B}}$ of the degeneration $\hat{\mathcal{Y}}$. The resulting modification $\mathcal{Y}$ is free of infinite-distance non-minimal fibral singularities, with the geometrical representative of its central fiber taking the form 
\begin{equation}
    \pi_{0}: Y_{0} = \bigcup_{p=0}^{P} Y^{p} \longrightarrow B_{0} = \bigcup_{p=0}^{P} B^{p}\,,
\end{equation}
with components
\begin{equation}
	\pi^{p}: Y^{p} \longrightarrow B^{p}\,,\qquad p = 0, \dotsc, P\,.
\end{equation}
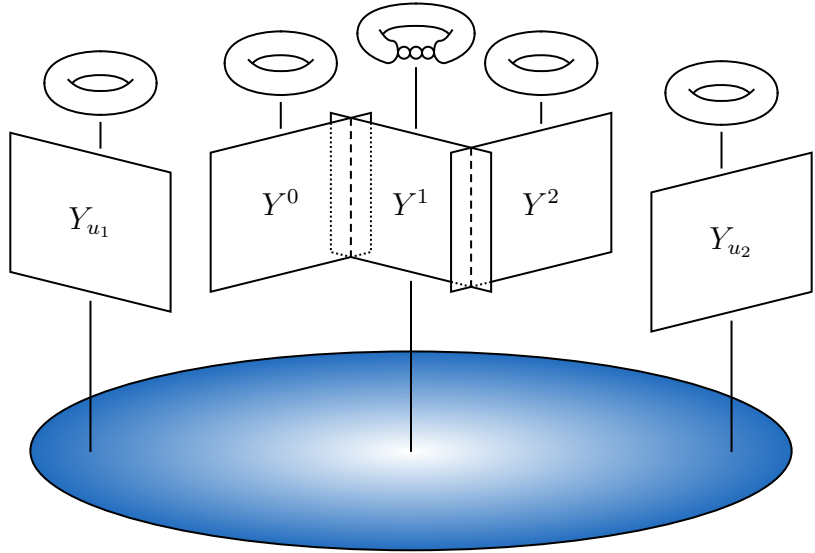
\begin{figure}[t!]
    \centering
	\begin{tikzpicture}[x=0.75pt,y=0.75pt]
		% Labels
		\node [] at (-65,125) {$Y^{0}$};
		\node [] at (0,125) {$Y^{1}$};		
		\node [] at (65,125) {$Y^{2}$};
		\node [] at (-160,115) {$Y_{u_{1}}$};
		\node [] at (160,105) {$Y_{u_{2}}$};
  
		% Ellipse
		\usepgflibrary{shadings}
		\draw[style=black fine-line, outer color=diagLightBlue, inner color=white] (0,0) ellipse (190 and 50);
		
		% Component 1
		\draw [style=densely dotted, line width=0.75pt] (-20,165) -- (-20,100) -- (-30,195/2);
		\draw [style=black fine-line] (-30,195/2) -- (-100,80) -- (-100,150) -- (-20,170) -- (-20,165);
		
		% Component 2
		\draw [style=black fine-line] (30,165/2) -- (40,80) -- (40,150) -- (-40,170) -- (-40,165);
		\draw [style=densely dotted, line width=0.75pt] (-40,165) -- (-40,100) -- (-30,195/2);
		\draw [style=black fine-line] (-30,195/2) -- (20,85);
		\draw [style=densely dotted, line width=0.75pt] (20,85) -- (30,165/2);
		
		% Component 3
		\draw [style=black fine-line] (30,165/2) -- (20,80) -- (20,150) -- (100,170) -- (100,100) -- (40,85);
		\draw [style=densely dotted, line width=0.75pt] (40,85) -- (30,165/2);
		
		% Component 4
		\draw [style=black fine-line] (120,130) -- (200,150) -- (200,80) -- (120,60) -- cycle;
		
		% Component 5
		\draw [style=black fine-line] (-120,140) -- (-200,160) -- (-200,90) -- (-120,70) -- cycle;
		
		% Fibration lines
		\draw [style=black fine-line] (0,85) -- (0,0);
		\draw [style=black fine-line] (-160,75) -- (-160,0);
		\draw [style=black fine-line] (160,65) -- (160,0);
		
		\draw [style=black fine-line] (-65,175) -- (-65,162.5);
		\draw [style=black fine-line] (2.5,192.5) -- (2.5,162.5);
		\draw [style=black fine-line] (65,175) -- (65,165);
		\draw [style=black fine-line] (155,160) -- (155,142.5);
		\draw [style=black fine-line] (-155,165) -- (-155,152.5);
		
		% Intersection lines
		\draw [style=densely dashed, line width=0.75pt] (30,70+165/2) -- (30,165/2);
		\draw [style=densely dashed, line width=0.75pt] (-30,70+195/2) -- (-30,195/2);
		
		% Fiber 1
		\def\torusoneposx{-65}
		\def\torusoneposy{195}
		\def\torusonescale{16}
		
		\draw [style=black fine-line] (-1.75*\torusonescale+\torusoneposx,0*\torusonescale+\torusoneposy) .. controls (-1.75*\torusonescale+\torusoneposx,0.75*\torusonescale+\torusoneposy) and (-0.75*\torusonescale+\torusoneposx,1*\torusonescale+\torusoneposy) .. (0*\torusonescale+\torusoneposx,1*\torusonescale+\torusoneposy);
		\draw[style=black fine-line, xscale=-1] (-1.75*\torusonescale-\torusoneposx,0*\torusonescale+\torusoneposy) .. controls (-1.75*\torusonescale-\torusoneposx,0.75*\torusonescale+\torusoneposy) and (-0.75*\torusonescale-\torusoneposx,1*\torusonescale+\torusoneposy) .. (0*\torusonescale-\torusoneposx,1*\torusonescale+\torusoneposy);
		\draw[style=black fine-line, rotate=180] (-1.75*\torusonescale-\torusoneposx,0*\torusonescale-\torusoneposy) .. controls (-1.75*\torusonescale-\torusoneposx,0.75*\torusonescale-\torusoneposy) and (-0.75*\torusonescale-\torusoneposx,1*\torusonescale-\torusoneposy) .. (0*\torusonescale-\torusoneposx,1*\torusonescale-\torusoneposy);
		\draw[style=black fine-line, yscale=-1] (-1.75*\torusonescale+\torusoneposx,0*\torusonescale-\torusoneposy) .. controls (-1.75*\torusonescale+\torusoneposx,0.75*\torusonescale-\torusoneposy) and (-0.75*\torusonescale+\torusoneposx,1*\torusonescale-\torusoneposy) .. (0*\torusonescale+\torusoneposx,1*\torusonescale-\torusoneposy);

		\draw [style=black fine-line] (-1*\torusonescale+\torusoneposx,0.1*\torusonescale+\torusoneposy) .. controls (-0.75*\torusonescale+\torusoneposx,-0.15*\torusonescale+\torusoneposy) and (-0.5*\torusonescale+\torusoneposx,-0.25*\torusonescale+\torusoneposy) .. (0*\torusonescale+\torusoneposx,-0.25*\torusonescale+\torusoneposy) .. controls (0.5*\torusonescale+\torusoneposx,-0.25*\torusonescale+\torusoneposy) and (0.75*\torusonescale+\torusoneposx,-0.15*\torusonescale+\torusoneposy) .. (1*\torusonescale+\torusoneposx,0.1*\torusonescale+\torusoneposy);

		\draw [style=black fine-line] (-0.875*\torusonescale+\torusoneposx,0*\torusonescale+\torusoneposy) .. controls (-0.75*\torusonescale+\torusoneposx,0.15*\torusonescale+\torusoneposy) and (-0.5*\torusonescale+\torusoneposx,0.25*\torusonescale+\torusoneposy) .. (0*\torusonescale+\torusoneposx,0.25*\torusonescale+\torusoneposy) .. controls (0.5*\torusonescale+\torusoneposx,0.25*\torusonescale+\torusoneposy) and (0.75*\torusonescale+\torusoneposx,0.12*\torusonescale+\torusoneposy) .. (0.875*\torusonescale+\torusoneposx,0*\torusonescale+\torusoneposy);
		
		% Fiber 2
		\def\torustwoposx{2.5}
		\def\torustwoposy{210}
		\def\torustwoscale{15}
		
		\draw [style=black fine-line] (-1.95*\torustwoscale+\torustwoposx,0*\torustwoscale+\torustwoposy) .. controls (-1.95*\torustwoscale+\torustwoposx,0.75*\torustwoscale+\torustwoposy) and (-0.95*\torustwoscale+\torustwoposx,1*\torustwoscale+\torustwoposy) .. (0*\torustwoscale+\torustwoposx,1*\torustwoscale+\torustwoposy);
		\draw[style=black fine-line, xscale=-1] (-1.95*\torustwoscale-\torustwoposx,0*\torustwoscale+\torustwoposy) .. controls (-1.95*\torustwoscale-\torustwoposx,0.75*\torustwoscale+\torustwoposy) and (-0.95*\torustwoscale-\torustwoposx,1*\torustwoscale+\torustwoposy) .. (0*\torustwoscale-\torustwoposx,1*\torustwoscale+\torustwoposy);
		\draw [style=black fine-line] (-1.075*\torustwoscale+\torustwoposx,0*\torustwoscale+\torustwoposy) .. controls (-0.95*\torustwoscale+\torustwoposx,0.15*\torustwoscale+\torustwoposy) and (-0.7*\torustwoscale+\torustwoposx,0.25*\torustwoscale+\torustwoposy) .. (0*\torustwoscale+\torustwoposx,0.25*\torustwoscale+\torustwoposy) .. controls (0.7*\torustwoscale+\torustwoposx,0.25*\torustwoscale+\torustwoposy) and (0.95*\torustwoscale+\torustwoposx,0.12*\torustwoscale+\torustwoposy) .. (1.075*\torustwoscale+\torustwoposx,0*\torustwoscale+\torustwoposy);
		
		\draw [style=black fine-line] (-1.2*\torustwoscale+\torustwoposx,0.1*\torustwoscale+\torustwoposy) .. controls (-1*\torustwoscale+\torustwoposx,-0.2*\torustwoscale+\torustwoposy) and (-0.9*\torustwoscale+\torustwoposx,-0.25*\torustwoscale+\torustwoposy) .. (-0.8*\torustwoscale+\torustwoposx,-0.25*\torustwoscale+\torustwoposy) .. controls (-0.61*\torustwoscale+\torustwoposx,-0.25*\torustwoscale+\torustwoposy) and (-0.6*\torustwoscale+\torustwoposx,-0.55*\torustwoscale+\torustwoposy) .. (-0.6*\torustwoscale+\torustwoposx,-0.65*\torustwoscale+\torustwoposy);
		\draw [style=black fine-line, xscale=-1] (-1.2*\torustwoscale-\torustwoposx,0.1*\torustwoscale+\torustwoposy) .. controls (-1*\torustwoscale-\torustwoposx,-0.2*\torustwoscale+\torustwoposy) and (-0.9*\torustwoscale-\torustwoposx,-0.25*\torustwoscale+\torustwoposy) .. (-0.8*\torustwoscale-\torustwoposx,-0.25*\torustwoscale+\torustwoposy) .. controls (-0.61*\torustwoscale-\torustwoposx,-0.25*\torustwoscale+\torustwoposy) and (-0.6*\torustwoscale-\torustwoposx,-0.55*\torustwoscale+\torustwoposy) .. (-0.6*\torustwoscale-\torustwoposx,-0.65*\torustwoscale+\torustwoposy);
		
		\draw[style=black fine-line, yscale=-1] (-1.95*\torustwoscale+\torustwoposx,0*\torustwoscale-\torustwoposy) .. controls (-1.95*\torustwoscale+\torustwoposx,0.75*\torustwoscale-\torustwoposy) and (-0.95*\torustwoscale+\torustwoposx,1*\torustwoscale-\torustwoposy) .. (-0.95*\torustwoscale+\torustwoposx,1*\torustwoscale-\torustwoposy) .. controls (-0.7*\torustwoscale+\torustwoposx,1*\torustwoscale-\torustwoposy) and (-0.6*\torustwoscale+\torustwoposx,0.75*\torustwoscale-\torustwoposy) .. (-0.6*\torustwoscale+\torustwoposx,0.65*\torustwoscale-\torustwoposy);
		\draw[style=black fine-line, xscale=-1, yscale=-1] (-1.95*\torustwoscale-\torustwoposx,0*\torustwoscale-\torustwoposy) .. controls (-1.95*\torustwoscale-\torustwoposx,0.75*\torustwoscale-\torustwoposy) and (-0.95*\torustwoscale-\torustwoposx,1*\torustwoscale-\torustwoposy) .. (-0.95*\torustwoscale-\torustwoposx,1*\torustwoscale-\torustwoposy) .. controls (-0.7*\torustwoscale-\torustwoposx,1*\torustwoscale-\torustwoposy) and (-0.6*\torustwoscale-\torustwoposx,0.75*\torustwoscale-\torustwoposy) .. (-0.6*\torustwoscale-\torustwoposx,0.65*\torustwoscale-\torustwoposy);
		
		\draw[style=black fine-line] (-0.4*\torustwoscale+\torustwoposx,-0.65*\torustwoscale+\torustwoposy) circle (0.2*\torustwoscale);
		\draw[style=black fine-line] (0.4*\torustwoscale+\torustwoposx,-0.65*\torustwoscale+\torustwoposy) circle (0.2*\torustwoscale);
		\draw[style=black fine-line] (0.0*\torustwoscale+\torustwoposx,-0.65*\torustwoscale+\torustwoposy) circle (0.2*\torustwoscale);
		
		% Fiber 3
		\def\torusoneposx{65}
		\def\torusoneposy{195}
		\def\torusonescale{16}
		
		\draw [style=black fine-line] (-1.75*\torusonescale+\torusoneposx,0*\torusonescale+\torusoneposy) .. controls (-1.75*\torusonescale+\torusoneposx,0.75*\torusonescale+\torusoneposy) and (-0.75*\torusonescale+\torusoneposx,1*\torusonescale+\torusoneposy) .. (0*\torusonescale+\torusoneposx,1*\torusonescale+\torusoneposy);
		\draw[style=black fine-line, xscale=-1] (-1.75*\torusonescale-\torusoneposx,0*\torusonescale+\torusoneposy) .. controls (-1.75*\torusonescale-\torusoneposx,0.75*\torusonescale+\torusoneposy) and (-0.75*\torusonescale-\torusoneposx,1*\torusonescale+\torusoneposy) .. (0*\torusonescale-\torusoneposx,1*\torusonescale+\torusoneposy);
		\draw[style=black fine-line, rotate=180] (-1.75*\torusonescale-\torusoneposx,0*\torusonescale-\torusoneposy) .. controls (-1.75*\torusonescale-\torusoneposx,0.75*\torusonescale-\torusoneposy) and (-0.75*\torusonescale-\torusoneposx,1*\torusonescale-\torusoneposy) .. (0*\torusonescale-\torusoneposx,1*\torusonescale-\torusoneposy);
		\draw[style=black fine-line, yscale=-1] (-1.75*\torusonescale+\torusoneposx,0*\torusonescale-\torusoneposy) .. controls (-1.75*\torusonescale+\torusoneposx,0.75*\torusonescale-\torusoneposy) and (-0.75*\torusonescale+\torusoneposx,1*\torusonescale-\torusoneposy) .. (0*\torusonescale+\torusoneposx,1*\torusonescale-\torusoneposy);

		\draw [style=black fine-line] (-1*\torusonescale+\torusoneposx,0.1*\torusonescale+\torusoneposy) .. controls (-0.75*\torusonescale+\torusoneposx,-0.15*\torusonescale+\torusoneposy) and (-0.5*\torusonescale+\torusoneposx,-0.25*\torusonescale+\torusoneposy) .. (0*\torusonescale+\torusoneposx,-0.25*\torusonescale+\torusoneposy) .. controls (0.5*\torusonescale+\torusoneposx,-0.25*\torusonescale+\torusoneposy) and (0.75*\torusonescale+\torusoneposx,-0.15*\torusonescale+\torusoneposy) .. (1*\torusonescale+\torusoneposx,0.1*\torusonescale+\torusoneposy);

		\draw [style=black fine-line] (-0.875*\torusonescale+\torusoneposx,0*\torusonescale+\torusoneposy) .. controls (-0.75*\torusonescale+\torusoneposx,0.15*\torusonescale+\torusoneposy) and (-0.5*\torusonescale+\torusoneposx,0.25*\torusonescale+\torusoneposy) .. (0*\torusonescale+\torusoneposx,0.25*\torusonescale+\torusoneposy) .. controls (0.5*\torusonescale+\torusoneposx,0.25*\torusonescale+\torusoneposy) and (0.75*\torusonescale+\torusoneposx,0.12*\torusonescale+\torusoneposy) .. (0.875*\torusonescale+\torusoneposx,0*\torusonescale+\torusoneposy);
		
		% Fiber 4
		\def\torusoneposx{155}
		\def\torusoneposy{180}
		\def\torusonescale{16}
		
		\draw [style=black fine-line] (-1.75*\torusonescale+\torusoneposx,0*\torusonescale+\torusoneposy) .. controls (-1.75*\torusonescale+\torusoneposx,0.75*\torusonescale+\torusoneposy) and (-0.75*\torusonescale+\torusoneposx,1*\torusonescale+\torusoneposy) .. (0*\torusonescale+\torusoneposx,1*\torusonescale+\torusoneposy);
		\draw[style=black fine-line, xscale=-1] (-1.75*\torusonescale-\torusoneposx,0*\torusonescale+\torusoneposy) .. controls (-1.75*\torusonescale-\torusoneposx,0.75*\torusonescale+\torusoneposy) and (-0.75*\torusonescale-\torusoneposx,1*\torusonescale+\torusoneposy) .. (0*\torusonescale-\torusoneposx,1*\torusonescale+\torusoneposy);
		\draw[style=black fine-line, rotate=180] (-1.75*\torusonescale-\torusoneposx,0*\torusonescale-\torusoneposy) .. controls (-1.75*\torusonescale-\torusoneposx,0.75*\torusonescale-\torusoneposy) and (-0.75*\torusonescale-\torusoneposx,1*\torusonescale-\torusoneposy) .. (0*\torusonescale-\torusoneposx,1*\torusonescale-\torusoneposy);
		\draw[style=black fine-line, yscale=-1] (-1.75*\torusonescale+\torusoneposx,0*\torusonescale-\torusoneposy) .. controls (-1.75*\torusonescale+\torusoneposx,0.75*\torusonescale-\torusoneposy) and (-0.75*\torusonescale+\torusoneposx,1*\torusonescale-\torusoneposy) .. (0*\torusonescale+\torusoneposx,1*\torusonescale-\torusoneposy);

		\draw [style=black fine-line] (-1*\torusonescale+\torusoneposx,0.1*\torusonescale+\torusoneposy) .. controls (-0.75*\torusonescale+\torusoneposx,-0.15*\torusonescale+\torusoneposy) and (-0.5*\torusonescale+\torusoneposx,-0.25*\torusonescale+\torusoneposy) .. (0*\torusonescale+\torusoneposx,-0.25*\torusonescale+\torusoneposy) .. controls (0.5*\torusonescale+\torusoneposx,-0.25*\torusonescale+\torusoneposy) and (0.75*\torusonescale+\torusoneposx,-0.15*\torusonescale+\torusoneposy) .. (1*\torusonescale+\torusoneposx,0.1*\torusonescale+\torusoneposy);

		\draw [style=black fine-line] (-0.875*\torusonescale+\torusoneposx,0*\torusonescale+\torusoneposy) .. controls (-0.75*\torusonescale+\torusoneposx,0.15*\torusonescale+\torusoneposy) and (-0.5*\torusonescale+\torusoneposx,0.25*\torusonescale+\torusoneposy) .. (0*\torusonescale+\torusoneposx,0.25*\torusonescale+\torusoneposy) .. controls (0.5*\torusonescale+\torusoneposx,0.25*\torusonescale+\torusoneposy) and (0.75*\torusonescale+\torusoneposx,0.12*\torusonescale+\torusoneposy) .. (0.875*\torusonescale+\torusoneposx,0*\torusonescale+\torusoneposy);
		
		% Fiber 5
		\def\torusoneposx{-155}
		\def\torusoneposy{185}
		\def\torusonescale{16}
		
		\draw [style=black fine-line] (-1.75*\torusonescale+\torusoneposx,0*\torusonescale+\torusoneposy) .. controls (-1.75*\torusonescale+\torusoneposx,0.75*\torusonescale+\torusoneposy) and (-0.75*\torusonescale+\torusoneposx,1*\torusonescale+\torusoneposy) .. (0*\torusonescale+\torusoneposx,1*\torusonescale+\torusoneposy);
		\draw[style=black fine-line, xscale=-1] (-1.75*\torusonescale-\torusoneposx,0*\torusonescale+\torusoneposy) .. controls (-1.75*\torusonescale-\torusoneposx,0.75*\torusonescale+\torusoneposy) and (-0.75*\torusonescale-\torusoneposx,1*\torusonescale+\torusoneposy) .. (0*\torusonescale-\torusoneposx,1*\torusonescale+\torusoneposy);
		\draw[style=black fine-line, rotate=180] (-1.75*\torusonescale-\torusoneposx,0*\torusonescale-\torusoneposy) .. controls (-1.75*\torusonescale-\torusoneposx,0.75*\torusonescale-\torusoneposy) and (-0.75*\torusonescale-\torusoneposx,1*\torusonescale-\torusoneposy) .. (0*\torusonescale-\torusoneposx,1*\torusonescale-\torusoneposy);
		\draw[style=black fine-line, yscale=-1] (-1.75*\torusonescale+\torusoneposx,0*\torusonescale-\torusoneposy) .. controls (-1.75*\torusonescale+\torusoneposx,0.75*\torusonescale-\torusoneposy) and (-0.75*\torusonescale+\torusoneposx,1*\torusonescale-\torusoneposy) .. (0*\torusonescale+\torusoneposx,1*\torusonescale-\torusoneposy);

		\draw [style=black fine-line] (-1*\torusonescale+\torusoneposx,0.1*\torusonescale+\torusoneposy) .. controls (-0.75*\torusonescale+\torusoneposx,-0.15*\torusonescale+\torusoneposy) and (-0.5*\torusonescale+\torusoneposx,-0.25*\torusonescale+\torusoneposy) .. (0*\torusonescale+\torusoneposx,-0.25*\torusonescale+\torusoneposy) .. controls (0.5*\torusonescale+\torusoneposx,-0.25*\torusonescale+\torusoneposy) and (0.75*\torusonescale+\torusoneposx,-0.15*\torusonescale+\torusoneposy) .. (1*\torusonescale+\torusoneposx,0.1*\torusonescale+\torusoneposy);

		\draw [style=black fine-line] (-0.875*\torusonescale+\torusoneposx,0*\torusonescale+\torusoneposy) .. controls (-0.75*\torusonescale+\torusoneposx,0.15*\torusonescale+\torusoneposy) and (-0.5*\torusonescale+\torusoneposx,0.25*\torusonescale+\torusoneposy) .. (0*\torusonescale+\torusoneposx,0.25*\torusonescale+\torusoneposy) .. controls (0.5*\torusonescale+\torusoneposx,0.25*\torusonescale+\torusoneposy) and (0.75*\torusonescale+\torusoneposx,0.12*\torusonescale+\torusoneposy) .. (0.875*\torusonescale+\torusoneposx,0*\torusonescale+\torusoneposy);
	\end{tikzpicture}
    \caption{A representation of a semi-stable degeneration of elliptically fibered threefolds. The blue disk represents $D$, over which we have two generic fibers $Y_{u_{1}}$ and $Y_{u_{2}}$, that degenerate to the multi-component central fiber $Y_{0}$.}
    \label{fig:semi-stable-degeneration}
\end{figure}%
Here, each $Y^{p}$ is an elliptic threefold with base $B^{p}$. Performing then fibral blow-ups\footnote{\label{foot:terminal-singularities}Working with elliptic Calabi-Yau threefolds, we may encounter codimension-two terminal singularities that do not admit a crepant resolution, signalling the presence of localised matter uncharged under any continuous gauge group \cite{Arras:2016evy,Grassi:2018rva}. These are kept unresolved in $\mathcal{X}$ in order to preserve the Calabi-Yau condition.} at every value of $u$ leads to yet another modification $\mathcal{X}$, which corresponds to the Coulomb branch in the dual M-theory description. Altogether, we can schematically summarize the modifications of the degeneration of interest as
\begin{equation}
    \begin{tikzcd}[every arrow/.append style={shift left}, row sep=large, column sep=large]
    \mathcal{X} \arrow{rr}{\textrm{fibral blow-down}} \arrow[d, shift right] & & \mathcal{Y} \arrow{ll}{\textrm{fibral blow-up}} \arrow{rr}{\textrm{base blow-down}} \arrow[d, shift right] & & \hat{\mathcal{Y}} \arrow{ll}{\textrm{base blow-up}} \arrow[d, shift right]\\
    \mathcal{B} & \simeq & \mathcal{B} \arrow{rr} & & \hat{\mathcal{B}}\mathrlap{\,.} \arrow{ll}
\end{tikzcd}
\end{equation}
We will often refer to the modification $\rho: \mathcal{Y} \rightarrow D$ as the resolved degeneration, even if the fibral singularities are kept unresolved in it. Along with a general discussion, we also provide some explicit examples of the degenerations just described and their modifications starting in \cref{sec:modifications-of-degenerations}.

Finally, recall that the set of allowed six-dimensional F-theory bases that can play the role of $\hat{B}$ consists of the Enriques surface, the complex projective plane $\mathbb{P}^{2}$, the Hirzebruch surfaces $\mathbb{F}_{n}$ (with $0 \leq n \leq 12$) and their blow-ups $\mathrm{Bl}\left( \mathbb{F}_{n} \right)$ \cite{Grassi1991}. The Enriques surface has trivial $\overline{K}_{\hat{B}}$ (up to torsion), and we therefore discard it; the models constructed over it support no gauge group or matter content, and in particular cannot exhibit
the fibral non-minimal singularities that we seek to study. In \cref{sec:six-dimensional-F-theory-bases} we briefly review the geometry of the other complex surfaces listed in order to set the notation used throughout the text.

\subsection{Modifications of an infinite-distance degeneration}
\label{sec:modifications-of-degenerations}

Given a starting degeneration $\hat{\rho}: \hat{\mathcal{Y}} \rightarrow D$, our task is now to find the equivalent resolved degeneration $\rho: \mathcal{Y} \rightarrow D$  and, eventually, to extract the physics in \cite{ALWPart2}.

We  describe, in \cref{sec:base-blow-ups}, the resolution process that eliminates the infinite-distance non-minimal singularities of the family variety $\hat{\mathcal{Y}}$, 
both in general terms and in an explicit example. In \cref{sec:orders-of-vanishing} we clarify the different notions of vanishing orders that we employ in the description of the degeneration. A subtlety is explained in \cref{sec:obscured-infinite-distance-limits}: In the context of the Semi-stable Reduction Theorem \cite{Mumford1973}, it is well known that performing a base change may be necessary in order to make a modification to a semi-stable degeneration possible; we observe in \cref{sec:obscured-infinite-distance-limits} how the need for base change is explicitly realised in our context through the appearance of non-minimal elliptic fibers of the central fiber of the degeneration that may be minimal for the family variety. We conclude by classifying the codimension-one infinite-distance degenerations of elliptically fibered Calabi-Yau threefolds into five classes, mirroring the classification performed in \cite{Lee:2021qkx,Lee:2021usk} for the infinite-distance degenerations of elliptic K3 surfaces. As an important result of this discussion, we will explain that it suffices to restrict our attention to geometrical representatives of the central fiber in which the components only support Kodaira type $\mathrm{I}_{m}$ singularities in codimension-zero.

\subsubsection{Base blow-ups}
\label{sec:base-blow-ups}

Let $\hat{\mathcal{Y}}$ be a degeneration of an elliptically fibered Calabi-Yau threefold. As argued in \cref{sec:definition-of-degenerations}, we only allow infinite-distance non-minimal fibers to appear over the central fiber $\hat{Y}_{0}$ of the degeneration. Furthermore, we focus on infinite-distance limits related to codimension-one non-minimal fibers in $\hat{Y}_{0}$, leaving codimension-two degenerations for future work. Let us denote the divisor of $\hat{B}_{0}$ over which we find the non-minimal fibers by $C$. We will assume that $C$ is a smooth, irreducible curve. Then, the vanishing orders of the defining polynomials of the Weierstrass model will be
\begin{equation}
    \ord{\hat{Y}_{0}}(f_{0},g_{0},\Delta_{0})_{C} \geq (4,6,12)\,.
\end{equation}
Turning our attention to the family fourfold $\hat{\mathcal{Y}}$, and defining 
\begin{equation}
    \mathcal{U} := \left\{ u = 0 \right\}_{\hat{\mathcal B}} \,,
\end{equation}
the curve $C \subset \hat{B}_{0} \subset \hat{\mathcal{B}}$ can be written as $C = \mathcal{C} \cap \mathcal{U}$, where $\mathcal{C}$ is the divisor of $\hat{\mathcal{B}}$ given by $\mathcal{C} = C \times D$. We will sometimes denote the curve by $C$ when we want to regard it as a divisor in $\hat{B}_{0}$, and by $\mathcal{C} \cap \mathcal{U}$ when we want to see it as a curve in $\hat{\mathcal{B}}$. The non-minimal nature of the locus will then manifest itself\footnote{This is not entirely precise; we discuss the interplay between family and component vanishing orders in \cref{sec:orders-of-vanishing}, as well as the appearance of ``obscured" infinite-distance limits (in which the former are minimal while the latter are not) in \cref{sec:obscured-infinite-distance-limits}. We can, however, always find an equivalent degeneration in which they mutually agree, and we therefore assume that this is the case during the remainder of this section.} in the family variety through
\begin{equation}
    \ord{\hat{\mathcal{Y}}}(f,g,\Delta)_{\mathcal{C} \cap\:\! \mathcal{U}} = (4 + \alpha, 6 + \beta, 12 + \gamma)\,,\quad \alpha, \beta, \gamma \geq 0\,.
\label{eq:family-non-minimal-vanishing-orders}
\end{equation}

To arrive at the modification of the degeneration that we will use to extract the physics in \cite{ALWPart2}, we blow  up the fourfold $\hat{\mathcal{Y}}$ until we have removed all of its codimension-one\footnote{When we refer to the codimension of a locus supporting non-minimal fibers we always compute it in the central fiber $\hat{Y}_{0}$, rather than in the family variety $\hat{\mathcal{Y}}$, unless explicitly stated. Hence, we say that $\mathcal{C} \cap \mathcal{U}$ is a codimension-one degeneration instead of a codimension-two one.} non-minimal elliptic fibers, obtaining the equivalent degeneration $\mathcal{Y}$. 

We first focus our attention on limits presenting a single non-minimal codimension-one locus. This corresponds to the naive notion of taking a single infinite-distance limit in complex structure moduli space instead of a superposition of several limits. We make this concept more precise in \cref{sec:curves-of-non-minimal-fibers} and return to the general case in \cref{sec:analysis-general-case}. The resolution process that we are about to discuss is no different in the presence of multiple codimension-one loci of non-minimal fibers; we simply apply the same procedure iteratively to each irreducible locus until all the non-minimal elliptic fibers of the family variety have been removed. 

The resolution of a codimension-one non-minimal locus amounts to blowing up the intersection curve of the divisors $\mathcal{C}$ and $\mathcal{U}$ in $\hat{\mathcal{B}}$ and performing a line bundle shift in order to ensure that the Calabi-Yau condition still holds for the blown up Weierstrass model. Even if only a single codimension-one locus of non-minimal elliptic fibers is present in $\hat{\mathcal{Y}}$, multiple blow-ups may be needed before the non-minimal fibers are fully removed from the family fourfold, since we may encounter new curves of non-minimal fibers in the exceptional components of the blow-ups. Let us analyse how this resolution process affects the geometry of the central fiber.

The blow-up of $\hat{\mathcal{B}}$ at $\mathcal{C} \cap \mathcal{U}$ is the pair given by a threefold $\mathcal{B}$ and a birational map
\begin{equation}
    \pi : \mathcal{B} \longrightarrow \hat{\mathcal{B}}\,.
\end{equation}
The exceptional set of the blow-up is the irreducible variety 
\begin{equation}
    E := \pi^{-1}(\mathcal{C \cap \mathcal{U}}) \,.
\end{equation}
Given an irreducible divisor $\mathcal{D}$ of $\hat{\mathcal{B}}$ that intersects the blow-up locus, the strict transform of $\mathcal{D}$ is the irreducible divisor of $\mathcal{B}$ given by the closure of $\pi^{-1}(\mathcal{D} \setminus \mathcal{C} \cap \mathcal{U})$. It is equal to the reducible divisor $\pi^{*}(\mathcal{D})$, known as the proper or total transform of $\mathcal{D}$, up to copies of $E$. More concretely, for $\mathcal{C}$ and $\mathcal{U}$ we have
\begin{align}
    \tilde{\mathcal{C}} &:= \pi^{*}\left( \mathcal{C} \right) = \mathcal{C}' + E\,,\\
    \tilde{\mathcal{U}} &:= \pi^{*}\left( \mathcal{U} \right) = \mathcal{U}' + E\,,
\end{align}
where we denote the strict transforms by the primes and the proper transforms by the tildes, for brevity. It is common to omit the primes and denote the original divisor and its strict transform by the same symbol, something that we will also do when the context makes clear what is meant.

The anticanonical classes of $\mathcal{B}$ and $\hat{\mathcal{B}}$ are related by
\begin{equation}
    \overline{K}_{\mathcal{B}} = \pi^{*} \left( \overline{K}_{\hat{\mathcal{B}}} \right) - E\,,
\label{eq:anti-canonical-class-blow-up}
\end{equation}
and therefore, after the blow-up, the Calabi-Yau condition is no longer satisfied in the resulting Weierstrass model unless we perform a line bundle shift, as we now explain. In view of \eqref{eq:family-non-minimal-vanishing-orders}, the strict transforms of the divisors associated to the global holomorphic sections given by the defining polynomials are related to their proper transforms by
\begin{subequations}
\begin{align}
    \tilde{F} &= F' + (4 + \alpha) E\,,\\
    \tilde{G} &= G' + (6 + \beta) E\,,\\
    \tilde{\Delta} &= \Delta' + (12 + \gamma) E\,.
\end{align}
\label{eq:FGDelta-total-transforms}%
\end{subequations}
Let us denote the holomorphic line bundle defining the Weierstrass model over $\hat{\mathcal{B}}$ by $\mathcal{L}_{\hat{\mathcal{B}}}$, and the one defining it over $\mathcal{B}$ by $\breve{\mathcal{L}}_{\mathcal{B}}$. The total transforms \eqref{eq:FGDelta-total-transforms} are in the class
\begin{equation}
	m \breve{\mathcal{L}}_{\mathcal{B}} = \pi^{*} \left( m \mathcal{L}_{\hat{\mathcal{B}}} \right) = \pi^{*} \left( m \overline{K}_{\hat{\mathcal{B}}} \right) = m \overline{K}_{\mathcal{B}} + mE\,,
\label{eq:line-bundle-discrepancy}
\end{equation}
where $m = 4$, $6$ and $12$ for $\tilde{F}$, $\tilde{G}$ and $\tilde{\Delta}$, respectively. We can restore the Calabi-Yau condition by shifting the line bundles $\tilde{F}$, $\tilde{G}$ and $\tilde{\Delta}$ such that the holomorphic line bundle $\mathcal{L}_{\mathcal{B}}$ defining the new Weierstrass model over $\mathcal{B}$ fulfils
\begin{equation}
	\mathcal{L}_{\mathcal{B}} = \breve{\mathcal{L}}_{\mathcal{B}} - E = \overline{K}_{\mathcal{B}}\,.
\label{eq:line-bundle-shift}
\end{equation}
In view of \eqref{eq:line-bundle-discrepancy}, the necessary shift leads to the divisors\footnote{We denote the (shifted) line bundles defining the elliptic fibration in $\mathcal{Y}$ in the same way as those defining it in $\hat{\mathcal{Y}}$. For additional clarity, we will denote the defining polynomials of the blown up family Weierstrass model by $f_{b}$, $g_{b}$ and $\Delta_{b}$.}
\begin{equation}
    \begin{aligned}
        F &= \tilde{F} - 4E = F' + \alpha E\,,\\
        G &= \tilde{G} - 6E = G' + \beta E\,,\\
        \Delta &= \tilde{\Delta} - 12E = \Delta' + \gamma E\,, 
    \end{aligned}
\label{eq:line-bundles-blow-up-shifted}
\end{equation}
where we observe that, depending on the original vanishing orders, copies of $E$ still appear factored out. Note that the line bundle shift restoring the Calabi-Yau condition yields effective divisors thanks to the starting non-minimal vanishing orders \eqref{eq:family-non-minimal-vanishing-orders}, without which the operation would not result in a consistent F-theory model. Non-minimal vanishing orders are therefore needed so that we can restore the Calabi-Yau condition by dividing $f$, $g$ and $\Delta$ by the required powers of the blow-up coordinate without turning them into rational functions.

The blow-up operation is a local one, meaning that away from $\mathcal{C} \cap \mathcal{U \subset \mathcal{U}} \subset \mathcal{B}$ the geometry remains unaffected, as we expect for a modification of a degeneration. The geometrical representative of the central fiber has been, however, changed. Let us rename the divisors
\begin{equation}
    E_{0} := \mathcal{U}'\,,\qquad E_{1} := E\,.
\end{equation}
The total transform $\tilde{\mathcal{U}}$ of $\mathcal{U}$ is the locus of the central fiber $B_{0}$ of the modified degeneration $\mathcal{B}$. It is now reducible and, in this case, given by
\begin{equation}
	\tilde{\mathcal{U}} = E_{0} + E_{1}\,.
\end{equation}
In other words, blowing up the base once and shifting the line bundles to restore the Calabi-Yau condition has left us with a two-component model for the central fiber of the degeneration. The components of the base are given by
\begin{equation}
    B_{0} := \mathcal{B}|_{\tilde{\mathcal{U}}}\,,\qquad\qquad B^{p} := \mathcal{B}|_{E_{p}}\,,\quad p = 0,1\,,
\end{equation}
and the holomorphic line bundles defined over them and pertaining to the elliptic fibrations $Y_{0}$, $Y^{0}$ and $Y^{1}$ are
\begin{equation}
	\begin{aligned}
    		F_{B_{0}} &:= F|_{\tilde{\mathcal{U}}}\,,\\
    		G_{B_{0}} &:= G|_{\tilde{\mathcal{U}}}\,,\\
    		\Delta_{B_{0}} &:= \Delta|_{\tilde{\mathcal{U}}}\,,
    \end{aligned}
    \qquad\qquad
    \begin{aligned}
    		F_{p} &:= F|_{E_{p}}\,,\\
    		G_{p} &:= G|_{E_{p}}\,,\\
    		\Delta_{p} &:= \Delta|_{E_{p}}\,,
	\end{aligned}
    \quad p = 0,1\,.
\end{equation}
From \eqref{eq:line-bundles-blow-up-shifted} we read off
\begin{equation}
    \ord{\mathcal{Y}}(f,g,\Delta)_{E_{1}} = (\alpha,\beta,\gamma)\,,
\label{eq:E1-vanishing-orders}
\end{equation}
so that the elliptic fibers over $B^{1}$ could be singular in codimension-zero depending on the values of $\alpha$, $\beta$ and $\gamma$. The fibers over $B^{0}$ could also be singular in codimension-zero, depending on $\ord{\hat{\mathcal{Y}}}(f,g,\Delta)_{\mathcal{U}} = \ord{\mathcal{Y}}(f,g,\Delta)_{E_{0}}$.

\begin{example}
\label{example:modification-of-degeneration}
	In order to make the above discussion more concrete, let us see how the blow-up process works in a particular example. We choose as the base $\hat{B}$ of the degenerating six-dimensional F-theory models the Hirzebruch surface $\hat{B} = \mathbb{F}_{7}$. For details on the notation that we employ, we refer to \cref{sec:six-dimensional-F-theory-bases}. In particular, we denote by $[s:t]$ and $[v:w]$ the homogenous coordinates on the fiber $\mathbb{P}_{f}^{1}$ and the base $\mathbb{P}_{b}^{1}$ of the Hirzebruch surface, respectively, see \eqref{eq:Cstar-action-weights}. The (toric) divisors associated with their vanishing loci are referred to by a corresponding capital letter, as in \eqref{eq:Hirzebruch-toric-divisors}.
    
    An example of a Weierstrass model giving an elliptically fibered family variety $\hat{\mathcal{Y}}$ over the base $\hat{\mathcal{B}} = \mathbb{F}_{7} \times D$ is
	\begin{subequations}
	\begin{align}
		f &= s^4 t^4 v^2 \left(u v^6-3 v^4 w^2+6 v^2 w^4-3 w^6\right)\,,\\
		g &= s^5 t^5 v^3 \left(s^2 v^{16}-2 s t v^6 w^3+6 s t v^4 w^5-6 s t v^2 w^7+2 s t w^9+t^2 u^2 w^2\right)\,,\\
		\Delta &= s^{10} t^{10} v^6 p_{4,32}\left( [s:t], [v:w:t], u \right)\,,
	\end{align}
	\end{subequations}
	where the subscripts in the residual polynomial of $\Delta$ refer to its homogeneous degrees under the two $\mathbb{C}^{*}$-actions of $\mathbb{F}_{7}$ given in \eqref{eq:Cstar-action-weights}. One can see that the generic fibers $\hat{Y}_{u \neq 0}$ of the family only present minimal singular elliptic fibers. The central fiber $\hat{Y}_{0}$, however, supports non-minimal singular elliptic fibers over the curve $\mathcal{S} \cap \mathcal{U} = \{s=u=0\} \subset \hat{B}_{0} \subset \hat{\mathcal{B}}$, as can be seen from the vanishing orders
	\begin{equation}
	    \ord{\hat{\mathcal{Y}}}(f,g,\Delta)_{s=u=0} = (4, 6, 13)\,.
	\end{equation}
	To perform a (toric) blow-up of $\hat{\mathcal{B}}$ along the curve $\mathcal{S} \cap \mathcal{U}$ and obtain $\mathcal{B}$, we introduce a new (exceptional) coordinate $e_{1}$ in the (total) coordinate ring of $\hat{\mathcal{B}}$, accompanied by a new $\mathbb{C}^{*}$-action
	\begin{equation}
	\begin{aligned}
		\mathbb{C}^{*}_{\mu_{1}}: \mathbb{C}^{*}_{(s',t,v,w,e'_{0},e_{1})} &\longrightarrow \mathbb{C}^{*}_{(s',t,v,w,e'_{0},e_{1})}\\
		(s',t,v,w,e'_{0},e_{1}) &\longmapsto (\mu_{1} s',t,v,w,\mu_{1} e'_{0},\mu^{-1}_{1} e_{1})\,.
	\end{aligned}
	\label{eq:example1-Cstar-mu-1}
	\end{equation}
	In the above expression, we have employed the coordinates $s'$ and $e'_{0}$, whose vanishing locus corresponds in $\mathcal{B}$ to the strict transform of the vanishing locus of the coordinates $s$ and $e_{0}$ in $\hat{\mathcal{B}}$, i.e.\ we have the relations
	\begin{subequations}
	\begin{align}
		\pi^{*}\left( \mathcal{S} \right) &= \pi^{*}\big( \{ s = 0\}_{\hat{\mathcal{B}}} \big) = \{ s' = 0\}_{\mathcal{B}} \cup \{e_{1} = 0\}_{\mathcal{B}} = \mathcal{S}' + E_{1}\,,\\
		\pi^{*}\left( \mathcal{U} \right) &= \pi^{*}\big( \{ e_{0} = 0\}_{\hat{\mathcal{B}}} \big) = \{ e'_{0} = 0\}_{\mathcal{B}} \cup \{e_{1} = 0\}_{\mathcal{B}} = E'_{0} + E_{1}\,.
	\end{align}
	\end{subequations}
	To simplify the notation, we now drop the primes for the strict transforms. The blow-up process prompts us to also modify the Stanley-Reisner ideal to
	\begin{equation}
		\mathscr{I}_{\hat{\mathcal{B}}} = \langle st, vw \rangle \longmapsto \mathscr{I}_{\mathcal{B}} = \langle st, vw, se_{0}, te_{1} \rangle\,.
	\end{equation}
	The total transforms of the divisors corresponding to the vanishing loci of the defining polynomials are obtained by performing the substitutions
	\begin{subequations}
	\begin{align}
		s &\longmapsto s e_{1}\,,\\
		u &\longmapsto e_{0} e_{1}\,,
	\end{align}
	\end{subequations}
	in $f$, $g$ and $\Delta$, obtaining the polynomials
	\begin{equation}
		(f,g,\Delta) \longmapsto (\tilde{f},\tilde{g},\tilde{\Delta})\,.
	\end{equation}
	Since
	\begin{equation}
		e_{1}^{4} \mid \tilde{f}\,,\qquad e_{1}^{6} \mid \tilde{g}\,,\qquad e_{1}^{12} \mid \tilde{\Delta}\,,
	\end{equation}
	we are allowed to perform the line bundle shift by prescribing the new defining polynomials
	\begin{equation}
		f_{b} := e_{1}^{-4}\tilde{f}\,,\qquad g_{b}:= e_{1}^{-6}\tilde{g}\,,\qquad \Delta_{b} := e_{1}^{-12}\tilde{\Delta}\,.
	\end{equation}
	The resulting expressions are
	\begin{subequations}
	\begin{align}
		f_{b} &= s^4 t^4 v^2 \left(e_0 e_1 v^6-3 v^4 w^2+6 v^2 w^4-3 w^6\right)\,,\\
		g_{b} &= s^5 t^5 v^3 \left(e_1 s^2 v^{16}+e_0^2 e_1 t^2 w^2-2 s t v^6 w^3+6 s t v^4 w^5-6 s t v^2 w^7+2 s t w^9\right)\,,\\
		\Delta_{b} &= s^{10} t^{10} v^6 e_1 p_{4,32,3}([s:t],[v:w:t],[s:e_{0}:e_{1}])\,.
	\end{align}
	\end{subequations}
	The generic fibers over the new base component $B^{1} = \{e_{1} = 0\}$ are singular of Kodaira type $\mathrm{I}_{1}$, as 
 follows from the vanishing orders
	\begin{equation}
	    \ord{\mathcal{Y}}(f_{b},g_{b},\Delta_{b})_{E_{1}} = (0,0,1)\,.
	\end{equation}
	In this concrete example, performing one blow-up is not enough to remove the non-minimal fibers of the family variety, since we have
	\begin{equation}
	    \ord{\mathcal{Y}}(f_{b},g_{b},\Delta_{b})_{s=e_{1}=0} = (4, 6, 12)\,.
	\end{equation}
\end{example}

Suppose that $P$ successive blow-ups of the $\mathcal{C} \cap \mathcal{U}$ locus are necessary in order to fully remove the non-minimal fibers of the family variety $\hat{\mathcal{Y}}$ and arrive at $\mathcal{Y}$. Composing the blow-up maps, we find that the total transform of the locus of the original base central fiber $\hat{B}_{0}$ is reducible with $P+1$ components, i.e.\
\begin{equation}
    \tilde{\mathcal{U}} = \sum_{p=0}^{P} E_{p}\,,
\label{eq:utilde-linear-equivalence}
\end{equation}
where the $E_{p}$ represent the strict transforms of the exceptional divisors after the composition of all blow-up maps. The components of the base are given by
\begin{equation}
    B_{0} := \mathcal{B}|_{\tilde{\mathcal{U}}}\,,\qquad\qquad B^{p} := \mathcal{B}|_{E_{p}}\,,\quad p = 0, \dotsc, P\,,
\end{equation}
and the holomorphic line bundles defined over them and associated to the elliptic fibrations $Y_{0}$ and $\{ Y^{p} \}_{0 \leq p \leq P}$ are
\begin{equation}
	\begin{aligned}
    		F_{B_{0}} &:= F|_{\tilde{\mathcal{U}}}\,,\\
    		G_{B_{0}} &:= G|_{\tilde{\mathcal{U}}}\,,\\
    		\Delta_{B_{0}} &:= \Delta|_{\tilde{\mathcal{U}}}\,,
    \end{aligned}
    \qquad\qquad
    \begin{aligned}
    		F_{p} &:= F|_{E_{p}}\,,\\
    		G_{p} &:= G|_{E_{p}}\,,\\
    		\Delta_{p} &:= \Delta|_{E_{p}}\,,
	\end{aligned}
    \quad p = 0, \dotsc, P\,.
\label{eq:definition-restricted-shifted-defining-divisors}
\end{equation}
The geometry of the base components will be discussed in \cref{sec:geometry-components-single-limit,sec:geometry-components-arbitrary-limit}, and the line bundles over them detailed in \cref{sec:line-bundles-components-single-limit,sec:line-bundles-components-arbitrary-limit}. Each base component $B^{p}$ together with the line bundles $F_{p}$, $G_{p}$ and $\Delta_{p}$ defines a Weierstrass model
\begin{equation}
	\pi^{p}: Y^{p} \longrightarrow B^{p}\,,\qquad p = 0, \dotsc, P\,,
\end{equation}
the collection of which gives the central fiber
\begin{equation}
    \pi_{0}: Y_{0} = \bigcup_{p=0}^{P} Y^{p} \longrightarrow B_{0} = \bigcup_{p=0}^{P} B^{p}
\end{equation}
of the family variety $\mathcal{Y}$. The type of codimension-zero fibers in the component $Y^{p}$ will be given by $\ord{\mathcal{Y}}(f,g,\Delta)_{E_{p}}$. As we could observe explicitly in \cref{example:modification-of-degeneration}, the type of these singularities over a component $E_{p}$ depends on the vanishing orders $\ord{\mathcal{Y}}(f,g,\Delta)_{\mathcal{C} \cap\:\! E_{p-1}}$ over the curve whose blow-up gives rise to it. If the family variety $\hat{\mathcal{Y}}$ we start with presents various curves $\{\mathcal{C}_{i} \cap \mathcal{U}\}_{1 \leq i \leq r}$ supporting non-minimal singular fibers we simply repeat the process we just described for each of them until all non-minimal fibers have been removed from the family variety.

\subsubsection{Family and component orders of vanishing}
\label{sec:orders-of-vanishing}

In the preceding sections, we have made use of the notion of the order of vanishing $\ord{\hat{\mathcal{Y}}}(f,g,\Delta)_{\mathcal{Z}}$ of the defining polynomials $f$, $g$ and $\Delta$ of a Weierstrass model $\hat{\Pi}_{\mathrm{ell}}: \hat{\mathcal{Y}} \rightarrow \hat{\mathcal{B}}$ along a given locus $\mathcal{Z}$ in the base $\hat{\mathcal{B}}$ of the elliptic fibration. While it is intuitively clear what is meant by this notation, let us be fully explicit in order to set the stage for the discussions in the upcoming sections. The examples in \cref{sec:obscured-infinite-distance-limits} illustrate the differences between the different orders of vanishing that we employ.

Assume first that $\mathcal{Z}$ is an irreducible component of the discriminant $\Delta$ of $\hat{\Pi}_{\mathrm{ell}}: \hat{\mathcal{Y}} \rightarrow \hat{\mathcal{B}}$, and hence a prime divisor of $\hat{\mathcal{B}}$. The order of vanishing of a rational function $h$ on an algebraic variety $\hat{\mathcal{B}}$ along a prime divisor is well-defined,\footnote{The local ring of a prime divisor in a normal variety is a discrete valuation ring, allowing us to define the order of vanishing of a function in its field of fractions. Higher codimension irreducible subvarieties no longer have an associated discrete valuation.} and can be obtained working in a local patch $\mathcal{A} \subset \hat{\mathcal{B}}$. Choosing local coordinates $\{a_{i}\}_{1 \leq i \leq \dim(\hat{\mathcal{B}})}$ for $\mathcal{A} \subset \hat{\mathcal{B}}$ in which $\mathcal{Z}$ is locally defined by an equation\footnote{Let us recall that, since the bases we consider are smooth, there is a one-to-one correspondence between Cartier and Weil divisors; hence, we can assume the existence of such a local equation without any further considerations.} $\{\mathcal{Z}(a_{i}) = 0\}_{\mathcal{A}}$, we simply count the factors of $\mathcal{Z}(a_{i})$ in $h$. Tate's algorithm allows us then to determine the type of elliptic fiber in the Kodaira-N\'eron list\footnote{Tate's algorithm \cite{Tate1975} can also distinguish between split, semi-split and non-split fibers, see \cite{Bershadsky:1996nh,Katz:2011qp,Grassi:2011hq} for a discussion of it in the context of F-theory.} found over the generic points of $\mathcal{Z}$ using $\ord{\hat{\mathcal{Y}}}(f,g,\Delta)_{\mathcal{Z}}$. In practice, many of the examples we work with have toric bases in which this can be done using a global description, rather than locally in a patch.

When $\mathrm{codim}_{\hat{\mathcal{B}}}(\mathcal{Z}) > 1$, e.g.\ for the intersection of various components of the discriminant, we work by restricting $\hat{\mathcal{B}}$ to a generic irreducible subvariety $\mathcal{W}$ such that $\mathcal{Z} \subset \mathcal{W}$ and
\begin{equation}
	\mathrm{codim}_{\hat{\mathcal{B}}}(\mathcal{W}) = \mathrm{codim}_{\hat{\mathcal{B}}}(\mathcal{Z}) -1 \Rightarrow \mathrm{codim}_{\mathcal{W}}(\mathcal{Z}) = 1\,.
\end{equation}
The orders of vanishing  are now again determined in codimension-one, and Tate's algorithm identifies the type of Kodaira singularity in the slice of $\hat{\Pi}_{\mathrm{ell}}: \hat{\mathcal{Y}} \rightarrow \hat{\mathcal{B}}$ given by restricting the elliptic fibration to $\mathcal{W}$. It is important that the restriction is taken such that the slice is generic; in a non-generic slice we may find singularities that are worse than those found in the generic one, see \cref{example:obscured-infinite-distance-limit-K3,example:obscured-infinite-distance-limit-F0}.

In codimension-two and higher, the Kodaira-N\'eron classification of singular fibers no longer is complete, and we can encounter non-Kodaira fibers even when the vanishing orders are minimal, see e.g.\ \cite{Miranda1983,miranda1989basic,Grassi:2000we,Esole:2011sm,Esole:2011cn,Cacciatori:2011nk,Cattaneo:2013vda}. At least for threefolds, the non-Kodaira fibers in crepant resolutions of Weierstrass models over codimension-two loci with minimal vanishing orders are contractions of the Kodaira fibers read off from Tate's algorithm for the generic surface slice of the model \cite{Cattaneo:2013vda}. In the associated singular Weierstrass model, the generic slice passing through such a codimension-two point presents a Du Val singularity in the Kodaira-N\'eron list, while the threefold presents a compound Du Val singularity at that point; locally, the threefold can be seen as a deformation of a Du Val singularity. We are interested, however, in higher codimension loci over which we have non-minimal orders of vanishing. It has been proven for threefolds that a crepant resolution of such loci, over which the singularities are still rational Gorenstein but no longer compound Du Val, would yield a non-equidimensional elliptic fibration \cite{Cattaneo:2013vda}. Hence, even in higher codimension, loci with non-minimal orders of vanishing behave very differently from their minimal counterparts. Working in F-theory, we want to preserve the equidimensional elliptic fibration structure, and therefore choose instead to resolve the codimension-two (from the point of view of $\hat{\mathcal{B}})$ non-minimal singularities in the family variety through a non-crepant base blow-up followed by a line bundle shift in order to restore the Calabi-Yau condition, as explained in \cref{sec:base-blow-ups}.

We will refer to the orders of vanishing of $f$, $g$ and $\Delta$ over loci in $\hat{\mathcal{B}}$ (respectively, for the resolved degeneration, $f_{b}$, $g_{b}$ and $\Delta_{b}$ over loci in $\mathcal{B}$) computed in this way as family orders of vanishing, this being the notion used in most of \cref{sec:base-blow-ups}.
\begin{definition}[Family vanishing orders]
    Let $\hat{\mathcal{Y}}$ (or $\mathcal{Y}$) be the elliptically fibered family variety with base $\hat{\mathcal{B}}$ (or $\mathcal{B}$) of a (resolved) degeneration, and let $\mathcal{Z}$ be an irreducible subvariety in $\hat{\mathcal{B}}$ (or $\mathcal{B}$). The family order of vanishing of a rational function $h$ of $\hat{\mathcal{B}}$ (or of $\mathcal{B}$) at $\mathcal{Z}$ is the order of vanishing computed directly through the discrete valuation associated to $\mathcal{Z}$ when $\mathcal{Z}$ is a prime divisor, and by reduction to the codimension-one problem through the generic slice passing through $\mathcal{Z}$ in higher codimension. We denote it by $\ord{\hat{\mathcal{Y}}}(h)_{\mathcal{Z}}$ (or $\ord{\mathcal{Y}}(h)_{\mathcal{Z}}$).
\end{definition}

The family orders of vanishing are important for  the resolution of the degeneration in order to remove the non-minimal singularities of the family variety. When the subvariety $\mathcal{Z}$ is taken as $E_{p} \subset \mathcal{B}$, they also define the relevant vanishing orders which determine the singular elliptic fibers over generic points in a component $B^p$ of the central fiber; we refer to these as the codimension-zero singular fibers. On the other hand, for the interpretation of the model in F-theory, what matters are the individual fibers $Y_{u}$. Therefore, to later extract the physical information corresponding to the endpoint of the infinite-distance limit, we also need to consider what we will call component orders of vanishing.
\begin{definition}[Component vanishing orders]
\label{def:component-vanishing-order}
	Let $\hat{\mathcal{Y}}$ (or $\mathcal{Y}$) be the elliptically fibered family variety with base $\hat{\mathcal{B}}$ (or $\mathcal{B}$) of a (resolved) degeneration and pick (a component of) one of its fibers, denoting it by $Y$ and its base by $B$. Let $\mathcal{Z}$ be an irreducible subvariety $\mathcal{Z} \subset B \subset \hat{\mathcal{B}}$ (or $\mathcal{Z} \subset B \subset \mathcal{B}$). The component order of vanishing of a rational function $h$ of $\hat{\mathcal{B}}$ (or $\mathcal{B}$) at $\mathcal{Z}$ is the order of vanishing  computed by reducing the problem to codimension-one through a slice contained in the restriction of the elliptic fibration $\hat{\mathcal{Y}}$ (or $\mathcal{Y}$) to the elliptic fibration $\pi_{\mathrm{ell}}: Y \rightarrow B$. We denote it by $\ord{Y}(\left. h \right|_{B})_{\mathcal{Z}}$.
\end{definition}

Note that, under the conditions of \cref{def:component-vanishing-order}, we always have
\begin{equation}
    \ord{\hat{\mathcal{Y}}}(h)_{\mathcal{Z}} \leq \ord{Y}(\left. h \right|_{B})_{\mathcal{Z}}\,,
\end{equation}
and similarly for $\mathcal{Y}$. This occurs because the definition of the component order of vanishing is the same as that of the family order of vanishing, but specifying a concrete slice that must be chosen to reduce to codimension-one, which may in particular be a non-generic slice.

Such a phenomenon gives rise to an important subtlety: It can happen that the component orders of vanishing of the defining polynomials of the Weierstrass model along a given locus are non-minimal, while the family orders of vanishing are minimal. We refer to such a situation as an obscured infinite-distance limit and analyse it in \cref{sec:obscured-infinite-distance-limits}. In particular, we will see that it is always possible to perform a base change to render the family and component vanishing orders identical. This is important to ensure that another sequence of blow-ups can be performed to remove all non-minimal elliptic fibers of the Weierstrass models $\hat{Y}_{u}$, which are the ones relevant for the F-theory interpretation, and not only those in the family variety $\hat{\mathcal{Y}}$.

Specifically for codimension-three points (from the point of view of $\mathcal{B}$) located on a curve $C_{p,q}$ along which two components $B^{p}$ and $B^{q}$ of the multi-component base central fiber $B_{0}$ intersect, it will be useful to also consider what we will call the interface order of vanishing, an even less generic specialization of the family vanishing order at the point.
\begin{definition}[Interface vanishing orders]
\label{def:interface-vanishing-order}
	Under the hypotheses of \cref{def:component-vanishing-order}, let $C \subset B$ be a curve and $\mathcal{Z}$ a point such that $\mathcal{Z} \subset C \subset B$. We call the interface order of vanishing of a rational function $h$ of $\hat{\mathcal{B}}$ (or $\mathcal{B}$) the order of vanishing of $h$ at $\mathcal{Z}$ computed by reducing the problem to codimension-one by further restricting the elliptic fibration $\pi_{\mathrm{ell}}: Y \rightarrow B$ to the elliptic fibration $\pi_{\mathrm{ell}}: S \rightarrow C$. We denote it by $\ord{S}(\left. h \right|_{C})_{\mathcal{Z}}$.
\end{definition}

Under the conditions of \cref{def:interface-vanishing-order} we have
\begin{equation} \label{rem:vanishing-restrictions}
    \ord{\hat{\mathcal{Y}}}(h)_{\mathcal{Z}} \leq \ord{Y}(\left. h \right|_{B})_{\mathcal{Z}} \leq \ord{S}(\left. h \right|_{C})_{\mathcal{Z}}\,,
\end{equation}
respectively for the resolved degeneration $\mathcal{Y}$.

Both for the component and interface orders of vanishing, if $\ord{\phat{\mathcal{Y}}}(h)_{B} \neq 0$, respectively $\ord{Y}(\left. h \right|_{B})_{C} \neq 0$, the restrictions of the rational function $h$ that we need to consider in order to compute the non-generic orders of vanishing are zero. When this occurs, we assign infinite component, respectively interface, order of vanishing to $h$ over that locus. Occasionally, we will define modified rational functions in which we remove the vanishing piece in order to obtain a finite result, see e.g.\ \cref{def:modified-discriminant}.

Note that one can in practice compute the relevant interface orders of vanishing directly in the unresolved degeneration, as long as one is careful with how base components and curves contract through the pushforward of the blow-up map. This may be useful to extract information about how the resolved degeneration will behave before attempting the resolution process, as we showcase in \cref{example:obscured-infinite-distance-limit-F0}.

\subsubsection{Class 1--5 models}
\label{sec:class-1-5-models}

In \cite{Lee:2021qkx,Lee:2021usk} the infinite-distance degenerations of elliptic K3 surfaces were sorted into five different classes depending on how the vanishing orders over the non-minimal loci exceeded the non-minimal bound. It is useful to consider an analogous classification for the codimension-one infinite-distance degenerations of elliptically fibered threefolds.
\begin{definition}[Degenerations of Class 1--5]
\label{def:class-1-5}
	Let $\hat{\mathcal{Y}}$ (or $\mathcal{Y}$) be the elliptically fibered family variety with base $\hat{\mathcal{B}}$ (or $\mathcal{B})$ of a (resolved) degeneration, and let $\mathcal{Z} \subset \phat{B}_{0}$ be a curve with
	\begin{equation}
		\ord{\phat{\mathcal{Y}}}\left( f_{(b)}, g_{(b)}, \Delta_{(b)} \right)_{\mathcal{Z}} = (4 + \alpha, 6 + \beta, 12 + \gamma)\,,\quad \alpha, \beta, \gamma \geq 0\,.
	\end{equation}
	We classify the family vanishing orders of the defining polynomials of the Weierstrass model over $\mathcal{Z}$ into
	\begin{align}
	    \textrm{Class 1:}\quad \alpha = 0\,,\; \beta = 0\,,\; \gamma = 0\,,\\
	    \textrm{Class 2:}\quad \alpha > 0\,,\; \beta = 0\,,\; \gamma = 0\,,\\
	    \textrm{Class 3:}\quad \alpha = 0\,,\; \beta > 0\,,\; \gamma = 0\,,\\
	    \textrm{Class 4:}\quad \alpha = 0\,,\; \beta = 0\,,\; \gamma > 0\,,\\
	    \textrm{Class 5:}\quad \alpha > 0\,,\; \beta > 0\,,\; \gamma > 0\,.
	\end{align}
    The degeneration $\hat{\rho}: \hat{\mathcal{Y}} \rightarrow D$ is then termed to be a Class~5 model if it presents any curves with Class~5 vanishing orders, a Class~1--4 model if it presents curves of non-minimal elliptic fibers all exhibiting Class~1--4 vanishing orders, and a finite-distance model otherwise.
\end{definition}

Since in \cref{def:class-1-5} we are assuming that the family variety $\phat{\mathcal{Y}}$ supports non-minimal elliptic fibers over $\mathcal{Z}$, the degeneration is amenable to a resolution process like the one described in \cref{sec:base-blow-ups}. Denoting by $E$ the exceptional divisor arising from the blow-up $\pi: \mathcal{B} \rightarrow \hat{\mathcal{B}}$ with centre at $\mathcal{Z}$, the codimension-zero elliptic fibers over the base component $E$ will be the ones read off from
\begin{equation}
	\ord{\mathcal{Y}}(f_{b},g_{b},\Delta_{b})_{E} = (\alpha, \beta, \gamma)\,,
\label{eq:codimension-zero-after-blow-up}
\end{equation}
c.f.\ \cref{example:modification-of-degeneration}. Hence, we see that for Class~1--4 family vanishing orders the generic fiber over the exceptional base component is of Kodaira type $\mathrm{I}_{m}$ with $m=\gamma$, while Class~5 vanishing orders lead to the remaining fiber types in the Kodaira-N\'eron list.

The type of fibers found over a base component in codimension-zero will play an important role in the physical analysis of \cite{ALWPart2}. Focusing on Class 4 vanishing orders first, we see that when the generic fiber in a component is of type $\mathrm{I}_{m \geq 1}$, the complex structure $\tau$ of the elliptic fiber attains a value for which $j(\tau) \rightarrow \infty$, implying that $\tau \rightarrow i\infty$. In F-theory, the complex structure of the elliptic fiber is identified with the Type IIB axio-dilaton, meaning that these are components in which $g_{s} \rightarrow 0$, i.e.\ regions of weak string coupling. This affects the types of 7-branes that one can encounter in said region of spacetime, which must be compatible with this background value of the axio-dilaton. Class 1--3 vanishing orders lead to codimension-zero type $\mathrm{I}_{0}$ fibers over the exceptional component; over such regions the string coupling becomes non-perturbative.

We now come to an important result: Class~5 vanishing orders can always be removed by a modification of the degeneration, possibly after a base change. This is a consequence of the classical result of Kempf, Knudsen, Mumford and Saint-Donat on semi-stable reduction that we have alluded to a few times above. In the remainder of this section, we recall this theorem and explain how it allows us to exclude Class~5 models, pointing out a few subtleties related to what we call obscured Class~5 models. As a conclusion, degenerations of Class~5 can always either be transformed into Class~1--4 degenerations, which lie at infinite distance in the moduli space, or to degenerations invoking only minimal singularities, which lie at finite distance. 
\begin{theorem}[Semi-stable Reduction Theorem \cite{Mumford1973}]
\label{thm:semi-stable-reduction}
    After a base change
    \begin{equation}
    \begin{aligned}
        \delta_{k}: D &\longrightarrow D\\
        u &\longmapsto u^{k}\,,
    \end{aligned}
    \end{equation}
    every degeneration admits a modification that is semi-stable.\footnote{The Semi-stable Reduction Theorem is actually stronger than stated, leading to a central fiber of the degeneration that has strict normal crossings. Insisting on preserving the Calabi-Yau condition may spoil this property, vide \cite{huybrechts2016}. In our case, we are happy to preserve some singularities unresolved in order to maintain a trivial canonical class, cf.\ \cref{foot:terminal-singularities}.}
\end{theorem}
For degenerations of K3 surfaces, this result can be improved. From the work of Kulikov \cite{Kulikov1981} and Persson–Pinkham \cite{PerssonPinkham1981} we know that semi-stability can be achieved while making the canonical bundle $K_{\mathcal{Y}}$ trivial. Moreover, a very complete description of the possible types of central fiber is available \cite{Kulikov1977,Persson1977,FriedmanMorrison1983}, a fact that was exploited in \cite{Lee:2021qkx,Lee:2021usk}. Once we venture into the degenerations of Calabi-Yau threefolds, we lack such strong results, but we can still invoke \cref{thm:semi-stable-reduction} to our advantage.

Components originating from the blow-up of a curve with Class~5 family vanishing orders lead to codimension-zero fibers of the types II, III, IV, $\mathrm{I}^{*}_{m}$, $\mathrm{IV}^{*}$, $\mathrm{III}^{*}$ and $\mathrm{II}^{*}$, given \eqref{eq:codimension-zero-after-blow-up} above. A degeneration presenting a component with these types of generic fibers cannot be a semi-stable degeneration,\footnote{More precisely, we mean that the degeneration $\rho: \mathcal{X} \rightarrow D$ in which even the minimal fibral singularities have been resolved will not be semi-stable.} since
\begin{itemize}
    \item Kodaira fibers of types II, III and IV present singularities that are not of normal crossing type, and
    \item Kodaira fibers of types $\mathrm{I}^{*}_{m}$, $\mathrm{IV}^{*}$, $\mathrm{III}^{*}$ and $\mathrm{II}^{*}$ contain exceptional rational curves of multiplicity bigger than one, meaning that the components over which they are fibered appear with multiplicity bigger than one as well, violating reducedness.
\end{itemize}
We hence conclude that the Semi-stable Reduction Theorem ensures that, possibly after performing a base change, Class~5 family vanishing orders can be removed from any model under consideration. The theorem is, however, not constructive, and we therefore do not know a priori what  combination of base changes and modifications of the degeneration must be taken in order to remove Class~5 vanishing orders. It is a logical possibility that, given a model presenting Class~5 vanishing orders, in the equivalent semi-stable degeneration the elliptic fibration does not extend to some components of the central fiber. We have not encountered such an example, and it may very well be that in the restricted class of degenerations that we consider this problem does not arise; were this to occur, we would interpret the model in the context of M-theory as having an obstructed F-theory limit.

While the discussion here is kept general, in \cite{ALWPart2} we mainly focus on degenerations of elliptic fibrations over Hirzebruch surfaces, both to make the discussion more explicit and to draw connections to the heterotic dual models to which some of these are related. In this explicit context, we can go substantially beyond the existence provided by the Semi-stable Reduction Theorem and determine the precise combination of base changes and modifications that transform a degeneration of Class~5 into a Class~1--4 or finite-distance model. This explicit analysis holds for degenerations on divisors of Hirzebruch surfaces and is presented in \cite{ALWClass5}.

As we discuss at length in \cref{sec:obscured-infinite-distance-limits}, in some models a curve can exhibit minimal family vanishing orders, while having non-minimal component vanishing orders, leading to an obscured infinite-distance limit. A situation similar in spirit occurs when the curve exhibits non-minimal family vanishing orders of Class~1--4, but non-minimal component vanishing orders of Class~5. The exceptional component arising from a blow-up centred at such a curve would support codimension-zero $\mathrm{I}_{m}$ fibers, as their type is determined by the family vanishing orders of the curve. The (fibral resolution of) the resolved degeneration obtained from such a model as explained in \cref{sec:base-blow-ups} will not be semi-stable, however, and therefore these obscured Class~5 models can be discarded by invoking the Semi-stable Reduction Theorem as was done above.

To see that these are not semi-stable, note that two components arising from Class~1--4 loci $Y^{p}$ and $Y^{q}$ and supporting codimension-zero $\mathrm{I}_{m}$ and $\mathrm{I}_{m'}$ fibers, respectively, intersect (either trivially or) on an elliptically fibered surface $Y^{p} \cap Y^{q}$ with codimension-zero $\mathrm{I}_{m''}$ fibers. By contrast, if one of the two components stems from an obscured Class~5 curve the codimension-zero fibers of $Y^{p} \cap Y^{q}$ will be of type II, III, IV, $\mathrm{I}^{*}_{m}$, $\mathrm{IV}^{*}$, $\mathrm{III}^{*}$ or $\mathrm{II}^{*}$. Denoting the components of the fibral resolution of the multi-component central fiber by $X_{0} = \bigcup_{p=0}^{P} \bigcup_{i_{p}=1}^{I_{p}} X^{p}_{i_{p}}$, one can see that, although the arguments given above do not apply directly to each component, they do apply to the restrictions $\left. \left( X_{0} - X^{p}_{i_{p}} \right) \right|_{X^{p}_{i_{p}}}$, meaning that the degeneration is not semi-stable.

This implies, in particular, that in the resolution process of a semi-stable degeneration the curves over which the components intersect will not exhibit non-minimal component vanishing orders, as this would imply that we are facing at least an obscured Class~5 model.

In the context of the explicit removal of Class~5 loci through the methods presented in \cite{ALWClass5}, we would first apply a base change to an obscured Class~5 model in order to make the Class~5 vanishing orders apparent at the family level, and then apply the combination of base changes and modifications required to remove the regular Class~5 models. 

As we mentioned above, the modifications of the degeneration taken as part of the application of the Semi-stable Reduction Theorem can be assumed to preserve the elliptic fibration in \mbox{F-theory}. Hence, the modifications of $\hat{\rho}: \hat{\mathcal{Y}} \rightarrow D$ induce modifications of $\left. \hat{\rho} \right|_{\hat{\mathcal{B}}}: \hat{\mathcal{B}} \rightarrow D$, over which the elliptic fibration is extended to obtain a Calabi-Yau variety. A modification of the base degeneration is a birational morphism that is an isomorphism over $D^{*}$, which is a Zariski open set. Hence, due to the Weak Factorization Theorem \cite{Włodarczyk2003,abramovich2002,wlodarczyk2009}, it can be factored into a sequence of blow-ups and blow-downs at smooth centres contained in the central fiber. In other words, the modifications of the base can be obtained by applying the process discussed in \cref{sec:base-blow-ups}. The modification of $\hat{\rho}: \hat{\mathcal{Y}} \rightarrow D$ is then obtained by taking the appropriate line bundles over the base degeneration, as explained in the same section. In some concrete examples, we will also consider flops connecting two different resolutions of the degeneration that, in agreement with the Weak Factorization Theorem, are equivalent to a sequence of blow-ups and blow-downs, \mbox{cf.\ \cref{rem:blow-up-order-flops}}.

\subsection{Single infinite-distance limits and their open-chain resolutions}	
\label{sec:curves-of-non-minimal-fibers}

We are now in a position to characterise the resolutions $\rho: \mathcal{Y} \rightarrow D$ of the infinite-distance degenerations $\hat{\rho}: \hat{\mathcal{Y}} \rightarrow D$ more precisely.

As a first result, we constrain the types of curves over which non-minimal singularities can arise. In a six-dimensional F-theory model, non-abelian gauge algebras are associated to minimal singular fibers over irreducible curves $C$, which we are assuming to be smooth, in the base $B$. With this assumption, $C$ is a Riemann surface embedded in $B$, and its  topology is completely classified by its genus $g(C)$. The choices for $C$ within a fixed base $B$ are numerous, and we can go up to very high genus. For example, in a model constructed over the base $B = \mathbb{P}^{2}$ it is easy to tune a Kodaira singularity of type $\mathrm{III}$ over a generic representative of $C = 9H$, for which the genus is $g(C) = 28$. The situation becomes more restrictive as the singular fibers that we try to tune over $C$ become worse. Hence, the choices we have for $C$ become the most constrained when we try to tune non-minimal singular fibers over it. In fact, smooth irreducible curves can only support non-minimal fibers if they have genus zero or one:
\begin{restatable}{proposition}{genusrestriction}
\label{prop:genus-restriction}
	Let $Y$ be an elliptically fibered Calabi-Yau threefold with base $\phat{B}$, where $\phat{B}$ is one of the allowed six-dimensional F-theory bases. Let $C \subset \phat{B}$ be a smooth irreducible curve of genus $g(C)$ supporting non-minimal singular fibers. Then, $g(C) \leq 1$, and $g(C) = 1$ if and only if $C = \overline{K}_{\phat{B}}$.
\end{restatable}
The proof of this result is technical, and we relegate the details to \cref{sec:genus-restriction-proof}. In fact, it is easy to see that non-minimal singularities are possible only over curves $C \leq \overline{K}_{\phat{B}}$, and the main work of \cref{sec:genus-restriction-proof} consists in showing explicitly that smooth irreducible curves on the F-theory base spaces with this property behave as in the proposition.

With this restriction in place, we only need to analyse what we will call genus-zero and genus-one degenerations, depending on the genus of the curve that supports the non-minimal fibers. Given \cref{prop:genus-restriction}, genus-one degenerations are only possible over those bases in which the anticanonical class has irreducible representatives, as otherwise tuning non-minimal singularities over the reducible curve $C = \overline{K}_{\phat{B}}$ just amounts to tuning several simultaneous genus-zero degenerations over its irreducible components.
\begin{corollary}
	With the hypothesis of \cref{prop:genus-restriction}, $g(C) = 1$ is only possible if $B = \mathbb{P}^{2}$, $B = \mathbb{F}_{n}$ with $0 \leq n \leq 2$ or a blow-up of them that does not spoil the irreducibility of $\overline{K}_{B}$.
\end{corollary}
There cannot occur genus-one degenerations over the Hirzebruch surfaces $\mathbb{F}_{n}$ with \mbox{$3 \leq n \leq 12$}, since their anticanonical class is reducible, as has been exploited in the F-theory literature in the study of non-Higgsable clusters \cite{Morrison:2012np}. Their blow-ups $\mathrm{Bl}\left( \mathbb{F}_{n} \right)$, with $3 \leq n \leq 12$, suffer the same fate. Therefore, as claimed above, we can only tune genus-one degenerations in models constructed over the bases $\mathbb{P}^{2}$, $\mathbb{F}_{n}$ with $0 \leq n \leq 2$, and  those blow-ups of them that do not spoil the irreducibility of the anticanonical class (see \cref{rem:genus-one-degeneration-over-blow-ups} for an example). Moreover, once we blow up along a genus-zero curve the line bundles over the resulting components are strictly smaller than the anticanonical class, and therefore no longer allow for the tuning of non-minimal singular fibers over an irreducible genus-one curve, even if it is present in the geometry, as will become evident in \cref{sec:line-bundles-components-single-limit,sec:line-bundles-components-arbitrary-limit}. This means that further components can then only arise through genus-zero blow-ups. Hence, the possible genus-zero degenerations outnumber the genus-one degenerations by far, and they will for this reason constitute our primary focus in what follows. We briefly comment on genus-one degenerations in \cref{sec:comments-genus-one-degenerations}.

In the upcoming sections, we will geometrically characterize genus-zero degenerations in detail. Before we come to this, however, we introduce the notion of a ``single infinite-distance limit." This should correspond to the simplest type of limit, which we can then extend further. Naively, such limits should be characterised by a single irreducible curve supporting non-minimal singular fibers in the base $\hat{B}_{0} \subset \hat{\mathcal{B}}$, as assumed to be the case in \cref{sec:modifications-of-degenerations}. As it turns out, however, this definition is not stable under base change and modifications, and the better way to define single infinite-distance limits is as follows.
\begin{restatable}[Single infinite-distance limits]{definition}{singleinfinitedistancelimit}
\label{def:single-infinite-distance-limit-original}
	Let $\hat{\rho}: \hat{\mathcal{Y}} \rightarrow D$ be a degeneration of the type described in \cref{sec:definition-of-degenerations} such that there is a collection of curves $\hat{\mathscr{C}}_{r} := \{ \mathcal{C}_{i} \cap \mathcal{U} \}_{1 \leq i \leq r}$ in $\hat{\mathcal{B}}$ with non-minimal component vanishing orders. We call the degeneration a single infinite-distance limit if
	\begin{enumerate}[label=(\roman*)]
		\item $\left( \mathcal{C}_{i} \cap \mathcal{U} \right) \cdot_{\hat{\mathcal{B}}} \left( \mathcal{C}_{j} \cap \mathcal{U} \right) = 0$ for all $1 \leq i < j \leq r$, \label{item:single-infinite-distance-limit-original-1}
		\item no point in the $\{ \mathcal{C}_{i} \cap \mathcal{U} \}_{1 \leq i \leq r}$ curves has non-minimal interface vanishing orders, and \label{item:single-infinite-distance-limit-original-2}
		\item no point in $\hat{\mathcal{B}} \setminus \left( \bigcup_{i=1}^{r} C_{i} \cap \mathcal{U} \right)$ presents infinite-distance non-minimal component vanishing orders. \label{item:single-infinite-distance-limit-original-3}
	\end{enumerate}
\end{restatable}
Indeed, the limits falling under this definition can be birationally transformed to limits with a single curve of non-minimal degenerations, hence conforming with the general intuition of what a single infinite-distance limit should be. To see this, let us first define the related notion of an open-chain resolution.
\begin{definition}[Open-chain resolution]
\label{def:single-infinite-distance-limit}
	Let $\hat{\rho}: \hat{\mathcal{Y}} \rightarrow D$ be a degeneration of the type described in \cref{sec:definition-of-degenerations} and let $\rho: \mathcal{Y} \rightarrow D$ be its modification removing the non-minimal singular fibers by repeated blow-ups and line bundle shifts as explained in \cref{sec:modifications-of-degenerations}, leading to a multi-component central fiber $Y_{0} = \bigcup_{p=0}^{P} Y^{p}$. We say that $\rho: \mathcal{Y} \rightarrow D$ is an open-chain resolution if the components of $Y_{0}$ form an open chain, i.e.\ they intersect in pairs
	\begin{equation}
		Y^{p-1} \cap Y^{p} \neq \emptyset\,,\qquad Y^{p} \cap Y^{p+1} \neq \emptyset\,,\qquad 0 < p < P\,,
	\end{equation}
	over curves, with all the other intersections vanishing.
\end{definition}

A key result of our work is that resolutions of single infinite-distance limits, in the sense of \cref{def:single-infinite-distance-limit-original}, are always open-chain resolutions: 
\begin{tcolorbox}[title=Single infinite-distance limit degenerations]
\begin{proposition}
\label{prop:single-infinite-distance-limits-and-open-resolutions}
	Let $\hat{\rho}: \hat{\mathcal{Y}} \rightarrow D$ be a single infinite-distance limit degeneration. Its resolved modifications $\rho: \mathcal{Y} \rightarrow \mathcal{B}$, obtained as explained in \cref{sec:modifications-of-degenerations}, are open-chain resolutions.
\end{proposition}
\end{tcolorbox}
This result will be proven in \cref{sec:single-infinite-distance-limits-and-open-chain-resolutions,sec:restricting-star-degenerations}, after we have explored the properties of the resolved central fiber in greater depth. As part of the proof, we will see that a resolution that deviates from the open-chain structure can be blown down to a model with non-minimal singularities over curves that fall out of the allowed class as per \cref{def:single-infinite-distance-limit-original}. Showing this explicitly requires some knowledge of the structure of base spaces of elliptic Calabi-Yau threefolds, and in particular the $\mathrm{Bl}(\mathbb{F}_{n})$, to which we devote \cref{sec:list-arbitrary-blow-ups}.

As one aspect of these results, we can rule out star-shaped resolutions: Such a structure would occur if on a base surface one could tune non-minimal elliptic fibers exclusively over three or more mutually non-intersecting curves. This, however, is not possible, as we explicitly show in \cref{sec:restricting-star-degenerations}. Tuning such non-minimal curves inevitably leads to non-minimalities over curves intersecting at least one element of the original set of curves, hence compromising the star shape of the resolution.

Now, apart from clarifying the structure of the resolutions of the degenerations falling under \cref{def:single-infinite-distance-limit-original}, \cref{prop:single-infinite-distance-limits-and-open-resolutions} also shows that we can birationally transform such a degeneration into one in which we only have one single irreducible curve supporting non-minimal fibers: This is possible simply by blowing down the open chain to one of its end-components. In this sense, \cref{def:single-infinite-distance-limit-original} does realise the ``naive" definition of single infinite-distance limit, but in a way invariant under blow-ups and blow-downs (and also under base change). This type of degenerations is the one represented in \cref{fig:semi-stable-degeneration}.

Let us now come back to \cref{def:single-infinite-distance-limit-original}  and 
gain some intuition behind Conditions \ref{item:single-infinite-distance-limit-original-2} and \ref{item:single-infinite-distance-limit-original-3} in \cref{def:single-infinite-distance-limit-original} and why they are needed for the definition to make sense.

First, the existence of an open-chain resolution as such is not necessarily invariant under base change. To understand this, note first that if $\hat{\rho}: \hat{\mathcal{Y}} \rightarrow D$ is a degeneration whose resolution leads to an open-chain central fiber as in \cref{def:single-infinite-distance-limit}, its modifications obtained by blow-ups and blow-downs of curves in $\hat{\mathcal{B}}$ will be as well. However, the existence of an open-chain resolution is no longer guaranteed after modifications involving base changes, depending on the type of obscured infinite-distance limits present in the model. \cref{example:obscured-infinite-distance-limit-F0} showcases this behaviour, where the original resolution contains some obscured infinite-distance limits which, when made apparent for the family variety through a base change, spoil the open chain structure. As we will see in \cref{sec:single-infinite-distance-limits-and-open-chain-resolutions}, Condition \ref{item:single-infinite-distance-limit-original-2} of \cref{def:single-infinite-distance-limit-original} prevents this from happening for a single infinite-distance limit degeneration; base change does not spoil its open-chain resolution structure.

Furthermore, we remark that models containing codimension-two infinite-distance non-mini\-mal points may also lead to open-chain resolutions as defined in \cref{def:single-infinite-distance-limit}, but in which some of the non-trivial intersections among components occur over points instead of curves. Since our primary focus is on codimension-one degenerations, we have excluded such models from our definition of single infinite-distance limits through Condition \ref{item:single-infinite-distance-limit-original-3} in \cref{def:single-infinite-distance-limit-original}, and from our definition of open-chain resolution by demanding that the components intersect over curves.

In \cref{sec:geometry-components-single-limit,sec:line-bundles-components-single-limit} we study the geometry of genus-zero single infinite-distance limits in detail. With this intuition at hand, we make the discussion more general in \cref{sec:geometry-components-arbitrary-limit,sec:line-bundles-components-arbitrary-limit}, exploring the geometry of those degenerations that do not allow for a simple open-chain resolution; we briefly advance some of their features in \cref{sec:analysis-general-case}.

\subsection{Geometry of the components in a single infinite-distance limit}
\label{sec:geometry-components-single-limit}

Focusing first on genus-zero single infinite-distance limit degenerations $\hat{\rho}: \hat{\mathcal{Y}} \rightarrow D$, we turn our attention to the geometry of the components of the blown up base family variety $\mathcal{B}$. As just discussed, these degenerations have open-chain resolutions, and we therefore assume in this section that the degenerations considered have such a resolution structure, relegating the general case to \cref{sec:geometry-components-arbitrary-limit}.

Since we will have to keep track of various successive blow-ups, let us denote the 
initial base variety and the one resulting from the $p$-th blow-up along a curve of non-minimal singular elliptic fibers by
\begin{equation}
    \mathrm{Bl}_{0} (\hat{\mathcal{B}}) := \hat{\mathcal{B}}\,,\qquad \mathrm{Bl}_{p}(\hat{\mathcal{B}}) := \pi^{*}_{p} \circ \cdots \circ \pi^{*}_{1} (\hat{\mathcal{B}})\,,
\end{equation}
where $\pi_{i}: \mathrm{Bl}_{i}(\hat{\mathcal{B}}) \rightarrow \mathrm{Bl}_{i-1}(\hat{\mathcal{B}})$ denotes the $i$-th blow-up map. With this notation, the final base family variety after $P$ blow-ups is
\begin{equation}
	\mathrm{Bl}_{P}(\hat{\mathcal{B}}) = \mathcal{B}\,.
\end{equation}
At each step, the central fiber $\mathrm{Bl}_{p}(\hat{\mathcal{B}})_{0}$ of the base family variety $\mathrm{Bl}_{p}(\hat{\mathcal{B}})$ has $p+1$ irreducible components,
\begin{equation}
	\mathrm{Bl}_{p}(\hat{\mathcal{B}})_{0} = \bigcup_{i=0}^{p} B^{i} = \bigcup_{i=0}^{p} E_{i}\,,
\end{equation}
which we simply denote by $B^{i} = E_{i}$ in all cases, always taking their strict transforms in the next step of the blow-up process and, hence, without risk of confusion. The elliptically fibered varieties over the successive $\{ \mathrm{Bl}_{p}(\hat{\mathcal{B}}) \}_{0 \leq p \leq P}$ bases will be denoted by $\{ \mathrm{Bl}_{p}(\hat{\mathcal{Y}}) \}_{0 \leq p \leq P}$.

We start with a curve $C_{1} = \mathcal{C}_{1} \cap \mathcal{U} \subset B^{0} \subset \mathrm{Bl}_{0}(\hat{\mathcal{B}})$ with $g(C_{1}) = 0$ that is blown-up to produce a new component $B^{1} = E_{1} \subset \mathrm{Bl}_{1}(\hat{\mathcal{B}})_{0} \subset \mathrm{Bl}_{1}(\hat{\mathcal{B}})$. The process may end here, or we may still have curves of non-minimal singular fibers. Because we are assuming to work with a degeneration whose resolution is an open-chain, we may encounter at most one such irreducible curve $C_{2}$, that must be contained either in $B^{0}$ or in $B^{1}$, and whose intersection with $C_{1}$ must be trivial
\begin{equation}
	C_{1} \cap C_{2} = \emptyset\,.
\end{equation}
Because one genus-zero blow-up has already been performed, the line bundles over the components $B^{0}$ and $B^{1}$ are such that they only allow for non-minimal singular fibers to be tuned over irreducible genus-zero curves, as we advanced above and discuss in \cref{sec:line-bundles-components-single-limit,sec:line-bundles-components-arbitrary-limit}. Hence, $g(C_{2}) = 0$. Blowing up along $C_{2}$ yields a new component $B^{2} = E_{2} \subset \mathrm{Bl}_{2}(\hat{\mathcal{B}})_{0} \subset \mathrm{Bl}_{2}(\hat{\mathcal{B}})$. Continuing in this way, we finish after $P$ blow-ups and have a central fiber $B_{0} \subset \mathcal{B}$ whose component $B^{p}$ arose from a blow-up along an irreducible curve $C_{p}$ with $g(C_{p}) = 0$ and trivial intersection with all other curves that were blown-up. Note that, because of the open-chain resolution that we are assuming, the new curves of non-minimal singular fibers at each step can only be found in the end-components of the chain.

Under the conditions described above, all components $B^{p}$ are Hirzebruch surfaces. This is intuitively clear since the blow-up operation leading to the $B^{p}$ component is placing a $\mathbb{P}^{1}$ factor on top of each point of the genus-zero curve $C_{p} \cong \mathbb{P}^{1}$, ultimately resulting in a $\mathbb{P}^{1}$-fibration over $\mathbb{P}^{1}$. Running this argument more carefully, we can see that it is, in fact, a $\mathbb{P}^{1}$-bundle over $\mathbb{P}^{1}$, and hence a Hirzebruch surface $B^{p} = F_{n_{p}}$.
\begin{tcolorbox}[title=Hirzebruch surfaces as components of open-chain resolutions]
\begin{proposition}
\label{prop:component-geometry-single}
	Let $\mathrm{Bl}_{p-1}(\hat{B})$ be the result of blowing up $p-1$ times the base family variety $\hat{\mathcal{B}}$ of a genus-zero degeneration $\hat{\rho}: \hat{\mathcal{Y}} \rightarrow D$ with an open-chain resolution. Let $C_{p}$ be an irreducible curve supporting non-minimal singular fibers and contained in the component $B^{i}$. Then, the exceptional component $B^{p} = E_{p}$ arising from the blow-up of $\mathrm{Bl}_{p-1}(\hat{B})$ along $C_{p}$ is the Hirzebruch surface $\mathbb{F}_{|C_{p} \cdot_{B^{i}} C_{p}|}$.
\end{proposition}
\end{tcolorbox}
\begin{proof}
	Denoting by $\mathcal{C}_{C_{p}/\mathrm{Bl}_{p-1}(\hat{B})}$ the normal cone of $C_{p}$ in $\mathrm{Bl}_{p-1}(\hat{B})$ and by $\mathcal{N}_{C_{p}/\mathrm{Bl}_{p-1}(\hat{B})}$ its normal bundle, the exceptional divisor of the blow-up along $C_{p}$ is the projectivization
	\begin{equation}
		E_{p} = \mathbb{P}\left( \mathcal{C}_{C_{p}/\mathrm{Bl}_{p-1}(\hat{B})} \right) = \mathbb{P}\left( \mathcal{N}_{C_{p}/\mathrm{Bl}_{p-1}(\hat{B})} \right) \,,
	\end{equation}
    where the second equality is a consequence of the smoothness of $C_{p}$ and $\mathrm{Bl}_{p-1}(\hat{B})$. Altogether, $E_{p}$ is therefore the projectivization of a rank 2 vector bundle over the genus-zero curve $C_{p} \cong \mathbb{P}^{1}$, which is by definition (as we reviewed in \cref{sec:six-dimensional-F-theory-bases}) a Hirzebruch surface $\mathbb{F}_{n_{p}}$ for some $n_{p}$.
	
	Let us now determine the concrete Hirzebruch surface that we obtain. We start by noting that, due to the inclusions $C_{p} \subset B^{i} \subset \mathrm{Bl}_{p-1}(\hat{B})$, we have the short exact sequence
	\begin{equation}
		0 \longrightarrow \mathcal{N}_{C_{p}/B^{i}} \longrightarrow \mathcal{N}_{C_{p}/\mathrm{Bl}_{p-1}(\hat{B})} \longrightarrow \left. \mathcal{N}_{B^{i}/\mathrm{Bl}_{p-1}(\hat{B})} \right|_{C_{p}} \longrightarrow 0\,.
	\label{eq:normal-bundle-short-exact-sequence}
	\end{equation}
	These are all holomorphic vector bundles over $C_{p} \cong \mathbb{P}^{1}$, and due to Grothendieck's splitting theorem we can therefore conclude that
    \begin{equation}
		\mathcal{N}_{C_{p}/\mathrm{Bl}_{p-1}(\hat{B})} = \mathcal{N}_{C_{p}/B^{i}} \oplus \left.\mathcal{N}_{B^{i}/\mathrm{Bl}_{p-1}(\hat{B})} \right|_{C_{p}} \,.
	\end{equation}
    Furthermore, due to the smoothness of the curve $C_{p}$, the component $B^{i} = E_{i}$ and $\mathrm{Bl}_{p-1}(\hat{B})$,
    \begin{subequations}
	\begin{align}
	   \mathcal{N}_{C_{p}/B^{i}} &= \mathcal{O}_{C_{p}}\left( C_{p} \right) := \left. \mathcal{O}_{B^{i}} \left( C_{p} \right) \right|_{C_{p}} = \mathcal{O}_{\mathbb{P}^{1}} \left( C_{p} \cdot_{B^{i}} C_{p} \right)    \,,\\
		 \mathcal{N}_{E_{i}/\mathrm{Bl}_{p-1}(\hat{B})} &= \mathcal{O}_{E_{i}}\left( E_{i} \right) := \left. \mathcal{O}_{\mathrm{Bl}_{p-1}(\hat{B})}\left( E_{i} \right) \right|_{E_{i}}\,.
	\end{align}
	\end{subequations}
	To evaluate the last expression, note that due to the open-chain resolution structure, the component $B^{i}$ where $C_{p}$ lies must be one of the end-components of the open chain, and hence
	\begin{equation}
		E_{i} \cdot_{\mathrm{Bl}_{p-1}(\hat{B})} E_{j} = C_{q} \neq 0
	\end{equation}
	for some particular value of $j$ and $q$, with the intersections with the components $E_{k \neq i,j}$ vanishing. Then, we have
	\begin{equation}
		\mathcal{N}_{E_{i}/\mathrm{Bl}_{p-1}(\hat{B})} = E_{i} \cdot_{\mathrm{Bl}_{p-1}(\hat{B})} E_{i} = - E_{i} \cdot_{\mathrm{Bl}_{p-1}(\hat{B})} E_{j} = -C_{q}\,,
	\end{equation}
	and therefore
	\begin{equation}
		\left. \mathcal{N}_{B^{i}/\mathrm{Bl}_{p-1}(\hat{B})} \right|_{C_{p}} = \mathcal{O}_{\mathbb{P}^{1}} \left( -C_{q} \cdot_{B^{i}} C_{p} \right) = \mathcal{O}_{\mathbb{P}^{1}}\,,
	\end{equation}
	where the intersection is vanishing due to the open-chain resolution assumption. Altogether, this leads to
	\begin{equation}
		E_{p} = \mathbb{P} \left( \mathcal{O}_{\mathbb{P}^{1}} \oplus \mathcal{O}_{\mathbb{P}^{1}} \left( C_{p} \cdot_{B^{i}} C_{p} \right) \right)\,.
	\end{equation}
	Finally, due to the invariance of the projectivization of a bundle under twists by Abelian line bundles, we can always take
	\begin{equation}
		E_{p} = \mathbb{P} \left( \mathcal{O}_{\mathbb{P}^{1}} \oplus \mathcal{O}_{\mathbb{P}^{1}} \left( |C_{p} \cdot_{B^{i}} C_{p}| \right) \right) = \mathbb{F}_{|C_{p} \cdot_{B^{i}} C_{p}|}\,.
	\end{equation}
\end{proof}

For smooth divisors $D_{1}$, $D_{2}$ and $D_{3}$ in a smooth threefold $X$, the intersection product satisfies
\begin{equation}
    D_{1} \cdot_{X} D_{2} \cdot_{X} D_{3} = \left. D_{1} \right|_{D_{3}} \cdot_{D_{3}} \left. D_{2} \right|_{D_{3}}\,.
\label{eq:single-infinite-distance-limit-blow-up-curves-Fn}
\end{equation}
Applying this to $E_{i}$ and $E_{p}$ in $\mathrm{Bl}_{p}(\hat{\mathcal{B}})$, we have
\begin{equation}
\begin{aligned}
    C_{p} \cdot_{B^{i}} C_{p} &= -\left. E_{i} \right|_{E_{i}} \cdot_{E_{i}} \left. E_{p} \right|_{E_{i}} = -E_{i} \cdot_{\mathrm{Bl}_{p}(\hat{\mathcal{B}})} E_{p} \cdot_{\mathrm{Bl}_{p}(\hat{\mathcal{B}})} E_{i} = E_{i} \cdot_{\mathrm{Bl}_{p}(\hat{\mathcal{B}})} E_{p} \cdot_{\mathrm{Bl}_{p}(\hat{\mathcal{B}})} E_{p}\\
    &= \left. E_{i} \right|_{E_{p}} \cdot_{E_{p}} \left. E_{p} \right|_{E_{p}} = -C_{p} \cdot_{B^{p}} C_{p}\,.
\end{aligned}
\end{equation}
Hence, $C_{p} \subset B^{p} \cong \mathbb{F}_{|C_{p} \cdot_{B^{i}} C_{p}|}$ is the $(-C_{p} \cdot_{B^{i}} C_{p})$-curve of the Hirzebruch surface.

The geometry of the central fiber $B_{0}$ of the base variety $\mathcal{B}$ of the open-chain resolution, as described in \cref{prop:component-geometry-single}, is schematically represented in \cref{fig:single-infinite-distance-limit-open-chain}.
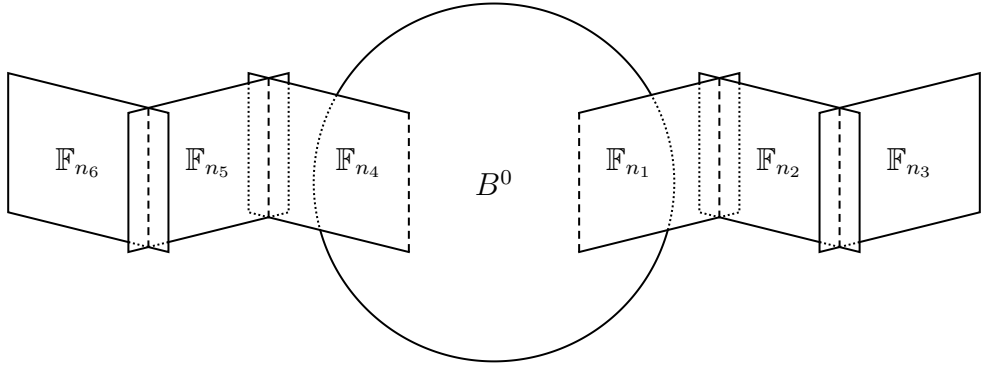
\begin{figure}[t!]
    \centering
	\begin{tikzpicture}[x=0.75pt,y=0.75pt]
		% Labels
		\node [] at (142.5-65-10,125) {$\mathbb{F}_{n_{1}}$};
		\node [] at (142.5+0,125) {$\mathbb{F}_{n_{2}}$};		
		\node [] at (142.5+65,125) {$\mathbb{F}_{n_{3}}$};
		
		\node [] at (-142.5+65+10,125) {$\mathbb{F}_{n_{4}}$};
		\node [] at (-142.5-0,125) {$\mathbb{F}_{n_{5}}$};		
		\node [] at (-142.5-65,125) {$\mathbb{F}_{n_{6}}$};
		
		\node [] at (0,115) {$B^{0}$};

		% Original component
		\draw [style=black fine-line, domain=29.2694:150.731, samples=250] plot ({142.5-100-42.5+90*cos(\x)}, {115+90*sin(\x)});
		\draw [style=black fine-line, domain=195.423:344.577, samples=250] plot ({142.5-100-42.5+90*cos(\x)}, {115+90*sin(\x)});
		\draw [style=densely dotted, line width=0.75pt, domain=150.731:195.423, samples=250] plot ({142.5-100-42.5+90*cos(\x)}, {115+90*sin(\x)});
		\draw [style=densely dotted, line width=0.75pt, domain=-15.4233:29.2694, samples=250] plot ({142.5-100-42.5+90*cos(\x)}, {115+90*sin(\x)});
		
		% Component 1
		\draw [style=densely dotted, line width=0.75pt] (142.5-20,165) -- (142.5-20,100) -- (142.5-30,195/2);
		\draw [style=black fine-line] (142.5-30,195/2) -- (142.5-100,80);
		\draw [style=black fine-line] (142.5-100,150) -- (142.5-20,170) -- (142.5-20,165);
		\draw [style=densely dashed, line width=0.75pt] (142.5-100,80) -- (142.5-100,150);
		
		% Component 2
		\draw [style=black fine-line] (142.5+30,165/2) -- (142.5+40,80) -- (142.5+40,150) -- (142.5-40,170) -- (142.5-40,165);
		\draw [style=densely dotted, line width=0.75pt] (142.5-40,165) -- (142.5-40,100) -- (142.5-30,195/2);
		\draw [style=black fine-line] (142.5-30,195/2) -- (142.5+20,85);
		\draw [style=densely dotted, line width=0.75pt] (142.5+20,85) -- (142.5+30,165/2);
		
		% Component 3
		\draw [style=black fine-line] (142.5+30,165/2) -- (142.5+20,80) -- (142.5+20,150) -- (142.5+100,170) -- (142.5+100,100) -- (142.5+40,85);
		\draw [style=densely dotted, line width=0.75pt] (142.5+40,85) -- (142.5+30,165/2);
		
		% Intersection lines
		\draw [style=densely dashed, line width=0.75pt] (142.5+30,70+165/2) -- (142.5+30,165/2);
		\draw [style=densely dashed, line width=0.75pt] (142.5-30,70+195/2) -- (142.5-30,195/2);
		
		% Component 4
		\draw [style=densely dotted, line width=0.75pt] (-142.5+20,165) -- (-142.5+20,100) -- (-142.5+30,195/2);
		\draw [style=black fine-line] (-142.5+30,195/2) -- (-142.5+100,80);
		\draw [style=black fine-line] (-142.5+100,150) -- (-142.5+20,170) -- (-142.5+20,165);
		\draw [style=densely dashed, line width=0.75pt] (-142.5+100,80) -- (-142.5+100,150);
		
		% Component 5
		\draw [style=black fine-line] (-142.5-30,165/2) -- (-142.5-40,80) -- (-142.5-40,150) -- (-142.5+40,170) -- (-142.5+40,165);
		\draw [style=densely dotted, line width=0.75pt] (-142.5+40,165) -- (-142.5+40,100) -- (-142.5+30,195/2);
		\draw [style=black fine-line] (-142.5+30,195/2) -- (-142.5-20,85);
		\draw [style=densely dotted, line width=0.75pt] (-142.5-20,85) -- (-142.5-30,165/2);
		
		% Component 6
		\draw [style=black fine-line] (-142.5-30,165/2) -- (-142.5-20,80) -- (-142.5-20,150) -- (-142.5-100,170) -- (-142.5-100,100) -- (-142.5-40,85);
		\draw [style=densely dotted, line width=0.75pt] (-142.5-40,85) -- (-142.5-30,165/2);
		
		% Intersection lines
		\draw [style=densely dashed, line width=0.75pt] (-142.5-30,70+165/2) -- (-142.5-30,165/2);
		\draw [style=densely dashed, line width=0.75pt] (-142.5+30,70+195/2) -- (-142.5+30,195/2);
	\end{tikzpicture}
    \caption{The open-chain of $B^{p}$ components arising for the central fiber $B_{0}$ of the base family variety $\mathcal{B}$ of the (open-chain) resolution of a single infinite-distance limit degeneration $\hat{\rho}: \hat{\mathcal{Y}} \rightarrow D$. The strict transform $B^{0}$ of the original base $\hat{B}$ preserves its original geometry $B^{0} \cong \hat{B}_{0}$, while the rest of the components $B^{p}$ for $1 \leq p \leq P$ are Hirzebruch surfaces, as proved in \cref{prop:component-geometry-single}.}
    \label{fig:single-infinite-distance-limit-open-chain}
\end{figure}
In the proof of \cref{prop:component-geometry-single}, the fact that we are dealing with an open-chain resolution was only used to determine the fact that $\left( E_{i} \cdot_{\mathrm{Bl}_{p-1}(\hat{\mathcal{B}})} E_{j} \right) \cap C_{p} = \emptyset$ for all $i \neq j \in \{0, \dotsc, p-1\}$. Hence, the same proof applies to the following rephrased proposition.
\begin{proposition}
\label{prop:component-geometry-single-bis}
	Let $\hat{\mathcal{B}}$ be the base family variety of a genus-zero degeneration $\hat{\rho}: \hat{\mathcal{Y}} \rightarrow D$ and $\mathrm{Bl}_{p-1}(\hat{B})$ be the result of $p-1$ blow-ups of $\hat{\mathcal{B}}$ with centres along the irreducible curves $\{C_{q}\}_{1 \leq q \leq p-1}$. Let $C_{p}$ be a smooth irreducible curve over which $\mathrm{Bl}_{p-1}(\hat{\mathcal{Y}})$ presents non-minimal singular fibers, contained in the component $B^{i}$ and with trivial intersections $C_{p} \cdot_{\mathrm{Bl}_{p-1}(\hat{B})} C_{q}$ for all $1 \leq q \leq p-1$. Then, the exceptional component $B^{p} = E_{p}$ arising from the blow-up of $\mathrm{Bl}_{p-1}(\hat{B})$ along $C_{p}$ is the Hirzebruch surface $\mathbb{F}_{|C_{p} \cdot_{B^{i}} C_{p}|}$.
\end{proposition}
This slightly modified form of the proposition will be the one that we will generalize in \cref{sec:geometry-components-arbitrary-limit} when we go beyond the study of open-chain resolutions. It tells us that, if we keep blowing up along genus-zero curves that do not intersect each other, we only produce chains of Hirzebruch surfaces\footnote{We can have at most two such chains of Hirzebruch surfaces attached to the strict transform of the original base component, see the discussion in \cref{sec:restricting-star-degenerations}.} attached to the strict transform $B^{0}$ of the original component $\hat{B}_{0}$. Each Hirzebruch surface $\mathbb{F}_{n_{p}}$ intersects its predecessor in the chain along the $(\pm n_{p})$-curve, and hence its successor along the $(\mp n_{p})$-curve, if we are to avoid intersections among the blow-up centres. For the end-components of such chains of Hirzebruch surfaces, we will refer to the $(\pm n_{p})$-curves over which the end-component does not intersect its predecessor in the chain as the end-curves.

We conclude the section with an illustrative example making the above discussion concrete.
\begin{example}
\label{example:geometry-components-single}
	Before concluding \cref{example:modification-of-degeneration}, we noticed that the single blow-up that we had performed did not completely remove the non-minimal singular fibers from the family variety $\mathrm{Bl}_{1}(\hat{\mathcal{B}})$, since
	\begin{equation}
	    \ord{\mathcal{Y}}(f_{b},g_{b},\Delta_{b})_{s=e_{1}=0} = (4, 6, 12)\,.
	\end{equation}
	Let us finish the resolution process and analyse the geometry of the resulting components. We carry out a second and final (toric) blow-up leading to $\mathrm{Bl}_{2}(\hat{\mathcal{B}}) = \mathcal{B}$, this time with centre $\mathcal{S} \cap E_{1}$. To this end, we add a new (exceptional) coordinate $e_{2}$ to the (total) coordinate ring of $\mathrm{Bl}_{1}(\hat{\mathcal{B}})$, together with a new $\mathbb{C}^{*}$-action
	\begin{equation}
	\begin{aligned}
		\mathbb{C}^{*}_{\mu_{2}}: \mathbb{C}^{*}_{(s,t,v,w,e_{0},e_{1},e_{2})} &\longrightarrow \mathbb{C}^{*}_{(s,t,v,w,e_{0},e_{1},e_{2})}\\
		(s,t,v,w,e_{0},e_{1},e_{2}) &\longmapsto (\mu_{2} s,t,v,w, e_{0},\mu_{2} e_{1}, \mu^{-1}_{2} e_{2})\,,
	\end{aligned}
	\label{eq:example1-Cstar-mu-2}
	\end{equation}
	and we modify the Stanley-Reisner ideal to be
	\begin{equation}
		\mathscr{I}_{\mathcal{B}} = \langle st, vw, se_{0}, se_{1}, te_{1}, te_{2}, e_{0}e_{2}  \rangle\,.
	\label{eq:example1-sr-ideal}
	\end{equation}
	In the defining polynomials of the Weierstrass model we perform the substitutions
	\begin{subequations}
		\begin{align}
			s &\longmapsto se_{2}\,,\\
			e_{1} &\longmapsto e_{1}e_{2}\,,
		\end{align}
	\end{subequations}
	and then divide them by the appropriate powers of $e_{2}$. Altogether, we arrive at the Weierstrass model given by
	\begin{subequations}
	\begin{align}
		f_{b} &= s^4 t^4 v^2 \left(e_0 e_1 e_2 v^6-3 v^4 w^2+6 v^2 w^4-3 w^6\right)\,,\\
		g_{b} &= s^5 t^5 v^3 \left(e_1 e_2^2 s^2 v^{16}+e_0^2 e_1 t^2 w^2-2 s t v^6 w^3+6 s t v^4 w^5-6 s t v^2 w^7+2 s t w^9\right)\,,\\
		\Delta_{b} &= s^{10} t^{10} v^6 e_1 p_{4,32,3,1}([s:t],[v:w:t],[s:e_{0}:e_{1}],[s:e_{1}:e_{2}])\,,
	\end{align}
	\label{eq:example1-defining-polynomials-blow-up}%
	\end{subequations}
	in which all non-minimal singular fibers have been removed.\footnote{All infinite-distance non-minimal singularities have been removed, but there still exist non-minimal singular fibers in codimension-two corresponding to SCFT points, that we keep unresolved.}
	
	Since we have blown up a total of two times, the central fiber $Y_{0}$ of the resolved degeneration $\rho: \mathcal{Y} \rightarrow D$ consists of three components $\{Y^{p}\}_{0 \leq p \leq 2}$ with base $\{B^{p}\}_{0 \leq p \leq 2}$. The base of the central fiber of the original degeneration was $\hat{B}_{0} = \mathbb{F}_{7}$, and therefore its strict transform is also $B^{0} = \mathbb{F}_{7}$. Since this is an open-chain resolution (of, in this case, a single infinite-distance limit), we can now apply \cref{prop:component-geometry-single} to determine the geometry of the other components. The first blow-up was along the curve $C_{1} = \mathcal{S} \cap \mathcal{U}$, which is the $(-7)$-curve of $\hat{B}_{0}$. This makes the associated exceptional component $B^{1} = \mathbb{F}_{7}$ again. The next blow-up occurs along the curve $C_{2} = \mathcal{S} \cap E_{1}$, which is the $(-7)$-curve of $B^{1}$, and therefore we also find $B^{2} = \mathbb{F}_{7}$. Hence, we expect that all components are
	\begin{equation}
		B^{0} \cong B^{1} \cong B^{2} \cong \mathbb{F}_{7}\,.
	\end{equation}
	Let us check this explicitly only for $B^{1}$, since the computation is totally analogous for $B^{2}$. Given that $B^{1} = \{e_{1} = 0\}$, the Stanley-Reisner ideal \eqref{eq:example1-sr-ideal} tells us that in this component $s$ and $t$ cannot vanish. We can then first use the $\mathcal{C}^{*}_{\lambda_{1}}$-action of \eqref{eq:Cstar-action-weights} with $\lambda_{1} = 1/s$ and then the $\mathbb{C}^{*}_{\mu_{1}}$-action of \eqref{eq:example1-Cstar-mu-1} with $\mu_{1} = t/s$ to set the coordinates to
	\begin{equation}
		(s,t,v,w,e_{0},0,e_{2}) \longmapsto (1,1,v,w,e_{0} t/s,0,e_{2})\,.
	\end{equation}
	Dropping the coordinates that are fixed, renaming $\tilde{e}_{0} := e_{0} t/s$, and relabelling the remaining $\mathbb{C}^{*}$-actions in \eqref{eq:Cstar-action-weights} and \eqref{eq:example1-Cstar-mu-2} by defining $\eta_{1} := \mu^{-1}_{2}$ and $\eta_{2} := \lambda_{2}$, we have
	\begin{subequations}
	\begin{align}
	\begin{aligned}
		\mathbb{C}^{*}_{\eta_{2}}: \mathbb{C}^{*}_{(v,w,\tilde{e}_{0},e_{2})} &\longrightarrow \mathbb{C}^{*}_{(v,w,\tilde{e}_{0},e_{2})}\\
		(v,w,\tilde{e}_{0},e_{2}) &\longmapsto (v,w,\eta_{1}\tilde{e}_{0},\eta_{1}e_{2})\,,
	\end{aligned}
	\\[0.5em]
	\begin{aligned}
		\mathbb{C}^{*}_{\eta_{2}}: \mathbb{C}^{*}_{(v,w,\tilde{e}_{0},e_{2})} &\longrightarrow \mathbb{C}^{*}_{(v,w,\tilde{e}_{0},e_{2})}\\
		(\eta_{2}v,\eta_{2}w,\eta_{2}^{7}\tilde{e}_{0},e_{2}) &\longmapsto (v,w,\tilde{e}_{0},e_{2})\,.
	\end{aligned}
	\end{align}
	\end{subequations}
	Altogether, we find that, indeed
	\begin{subequations}
	\begin{align}
		B^{1} &= \bigslant{\left( \mathbb{C}_{(v,w,\tilde{e}_{0},e_{2})}  \setminus Z \right)}{\mathbb{C}^{*}_{\eta_{1}} \times \mathbb{C}^{*}_{\eta_{2}}} \cong \mathbb{F}_{7}\,,\\
		Z &:= \{v=w=0\} \cup \{\tilde{e}_{0} = e_{2} = 0\}\,,
	\end{align}
	\end{subequations}
	as expected from \cref{prop:component-geometry-single}. Moreover, we see that $C_{1} = \{\tilde{e}_{0} = 0\}_{B^{1}} = \{e_{0} = e_{1} = 0\}_{\mathcal{B}}$ is the $(+7)$-curve of $B^{1}$, as was explained after \eqref{eq:single-infinite-distance-limit-blow-up-curves-Fn}.
\end{example}

\subsection{Weierstrass models and log Calabi-Yau structure}
\label{sec:line-bundles-components-single-limit}

Having described the geometry of the base components of an open-chain resolution, we now turn to the Weierstrass model that is defined over them. In other words, we need to see how the holomorphic line bundle $\mathcal{L}$ over $\mathcal{B}$ defining the elliptic fibration for the family variety $\mathcal{Y}$ restricts to the individual base components $\{B^{p}\}_{0 \leq p \leq P}$ to define the elliptic fibrations $\{Y^{p}\}_{0 \leq p \leq P}$. We recall that the analysis of this section will apply to single infinite-distance limits in particular, leaving the general case for \cref{sec:line-bundles-components-arbitrary-limit}. As we will see, the individual components of the resolved central fiber give rise to a log Calabi-Yau structure.

Our starting point is given by the degenerations $\hat{\rho}: \hat{\mathcal{Y}} \rightarrow D$ of the type described in \cref{sec:definition-of-degenerations}, where the base family variety is $\hat{\mathcal{B}} = \hat{B} \times D$. By construction, both every fiber $\hat{Y}_{u}$ of the family and the family variety $\hat{\mathcal{Y}}$ itself are elliptically fibered Calabi-Yau varieties, since the holomorphic line bundle defining their respective Weierstrass models satisfies
\begin{equation}
	\mathcal{L}_{B_{u}} = \overline{K}_{B_{u}} = \overline{K}_{B}\,,\qquad\qquad \mathcal{L}_{\hat{\mathcal{B}}} = \overline{K}_{\hat{\mathcal{B}}}\,.
\end{equation}
The resolution process described in \cref{sec:modifications-of-degenerations} and leading to the modification of the degeneration $\rho: \mathcal{Y} \rightarrow D$ consists of a series of blow-ups, each of them followed by a line bundle shift that ensures that the Calabi-Yau condition is satisfied at each step. This  means that at the end we also have a Calabi-Yau family variety, since
\begin{equation}
	\mathcal{L}_{\mathcal{B}} = \overline{K}_{\mathcal{B}}\,.
\end{equation}
The modification of the degeneration is an isomorphism over $D^{*}$, and we therefore do not need to worry about the generic fibers $Y_{u \neq 0}$. The geometrical representative of the endpoint of the limit is now, however, the multi-component central fiber $Y_{0} = \bigcup_{p=0}^{P} Y^{p}$ with base $B_{0} = \bigcup_{p=0}^{P} B^{p}$. The elliptic fibration $\pi_{p}: Y^{p} \rightarrow B^{p}$ is given by a Weierstrass model that is obtained as the restriction of the Weierstrass model of $\mathcal{Y}$ to $Y^{p}$. Hence, and since the base component $B^{p}$ is the exceptional divisor $E_{p}$, the defining holomorphic line bundle of $\pi_{p}: Y^{p} \rightarrow B^{p}$ is the restriction $\left. \mathcal{L}_{\mathcal{B}} \right|_{E_{p}}$.

In an open-chain resolution, shown in \cref{fig:single-infinite-distance-limit-open-chain}, the base components only intersect the adjacent members of the chain, i.e.\ we have the non-trivial intersections
\begin{equation}
	B^{p-1} \cap B^{p} \neq \emptyset\,,\qquad B^{p} \cap B^{p+1} \neq \emptyset\,,\qquad 0 < p < P\,,
\end{equation}
with all the other intersections vanishing. Here, we have relabelled the components such that $p$ runs in sequential order in the open chain. The intersections occur over the $\{C_{p}\}_{1 \leq p \leq P}$ curves that acted as the centres for the blow-ups, having
\begin{equation}
	E_{p-1} \cdot_{\mathcal{B}} E_{p} = C_{p}\,,\qquad 1 \leq p \leq P\,.
\label{eq:intersections-exceptional-divisors-single-infinite-distance-limit}
\end{equation}
With this in mind, let us compute the restriction $\left. \mathcal{L}_{\mathcal{B}} \right|_{E_{p}}$.
\begin{tcolorbox}[title=Weierstrass models over open-chain resolutions]
\begin{proposition}
\label{prop:component-line-bundle-single-infinite-distance-limit}
	Let $\{B^{p}\}_{0 \leq p \leq P}$ be the base components of the central fiber $Y_{0}$ of the modification $\rho: \mathcal{Y} \rightarrow D$ giving an open-chain resolution of a degeneration $\hat{\rho}: \hat{\mathcal{Y}} \rightarrow D$. Then, the holomorphic line bundles $\{\mathcal{L}_{p}\}_{0 \leq p \leq P} := \{\mathcal{L}_{B^{p}}\}_{0 \leq p \leq P}$ defining the Weierstrass models over the $\{B^{p}\}_{0 \leq p \leq P}$ are
	\begin{subequations}
	\begin{align}
		\mathcal{L}_{0} &= \overline{K}_{B^{0}} - C_{1}\,,\\
		\mathcal{L}_{p} &= \overline{K}_{B^{p}} - C_{p-1} - C_{p+1}\,,\qquad 1 \leq p \leq P-1\,,\\
		\mathcal{L}_{P} &=\overline{K}_{B^{P}} - C_{P-1}\,.
	\end{align}
	\end{subequations}
\end{proposition}
\end{tcolorbox}
\begin{proof}
	The line bundles are given by
		\begin{equation}
			\mathcal{L}_{p} := \mathcal{L}_{B^{p}} = \left. \mathcal{L}_{\mathcal{B}} \right|_{E_{p}}\,.
		\end{equation}
		To compute this restriction let us recall the adjunction formula for the smooth divisors $\{E_{p}\}_{0 \leq p \leq P}$ in the smooth variety $\mathcal{B}$, which takes the form
		\begin{equation}
			K_{E_{p}} = \left. \left( K_{\mathcal{B}} + E_{p} \right) \right|_{E_{p}}\,.
		\end{equation}
		Using then the fact that the family variety $\mathcal{Y}$ is Calabi-Yau, and therefore $\mathcal{L}_{B} = \overline{K}_{\mathcal{B}}$, we obtain
		\begin{equation}
			\mathcal{L}_{p} = \overline{K}_{B^{p}} + \left. E_{p} \right|_{E_{p}}\,.
		\end{equation}
		From the relation
		\begin{equation}
			\tilde{\mathcal{U}} = \sum_{p=0}^{P} E_{p}\,,
		\end{equation}
		the triviality of the $\tilde{\mathcal{U}}$ class and the intersections \eqref{eq:intersections-exceptional-divisors-single-infinite-distance-limit}, the result follows.
\end{proof}
We can be even more concrete if we remember that the  components of $B_{0}$, besides the strict transform of $\hat{B}_{0}$, are Hirzebruch surfaces, see \cref{prop:component-geometry-single}. As pointed out after \eqref{eq:single-infinite-distance-limit-blow-up-curves-Fn}, the blow-up centres, as seen from the Hirzebruch surfaces $\mathbb{F}_{n_{p}}$, are just the \mbox{$(\pm n_{p})$-curves}, and therefore the holomorphic line bundles over them are simply the anticanonical class of the Hirzebruch surface with some section classes subtracted. In particular, intermediate $\mathbb{F}_{n_{p}}$ components of the chain have line bundles consisting only of vertical classes. Using the notation of \eqref{eq:Hirzebruch-toric-divisors}, the holomorphic line bundles over the $\mathbb{F}_{n_{p}}$ components are
\begin{subequations}
\begin{align}
    \mathcal{L}_{p} = 2V_{p}\,,\qquad &\text{if }B^{p}\text{ is an intermediate component,}\\
    \mathcal{L}_{p} = S_{p} + 2V_{p}\quad \text{or}\quad \mathcal{L}_{p} = T_{p} + 2V_{p}\,,\qquad &\text{if }B^{p}\text{ is an end-component.}
\end{align}
\label{eq:line-bundles-Fn-components-open-chain}%
\end{subequations}
Since for all components we have that $\mathcal{L}_{p} \leq \overline{K}_{B^{p}}$, and given \cref{prop:genus-restriction}, tuning non-minimal elliptic fibers to appear over a genus-one curve is not possible.

From \cref{prop:component-line-bundle-single-infinite-distance-limit}, we observe that the Weierstrass models $\pi_{p}: Y^{p} \rightarrow B^{p}$ are describing varieties that are not Calabi-Yau, since $\mathcal{L}_{p} \neq \overline{K}_{B^{p}}$. While each component of the central fiber $Y_{0}$ is not Calabi-Yau, the multi-component central fiber $Y_{0} = \bigcup_{p=0}^{P} Y^{p}$ still is, as we can see from
\begin{equation}
	\mathcal{L}_{B_{0}} = \left. \mathcal{L}_{\mathcal{B}} \right|_{\tilde{\mathcal{U}}} = \overline{K}_{\tilde{\mathcal{U}}} + \left. \tilde{\mathcal{U}}\right|_{\tilde{\mathcal{U}}} = \overline{K}_{\tilde{\mathcal{U}}} = \sum_{p=0}^{P} \mathcal{L}_{p}\,.
\label{eq:line-bundle-central-fiber}
\end{equation}
Let us abuse notation and denote the pullback divisors $\pi_{p}^{*}(C_{q})$ in $Y^{p}$ also by $C_{q}$. Rephrasing the above discussion, the pairs
\begin{equation}
	\left( Y^{0}, C_{1} \right)\,,\qquad \left( Y^{p}, C_{p-1} + C_{p+1} \right)\,, \quad p = 1, \dotsc, P-1\,,\qquad \left( Y^{P}, C_{P-1} \right)
\end{equation}
are log Calabi-Yau spaces,\footnote{A log Calabi-Yau space is a pair $(X,D)$, where $X$ is a variety and $D$ an effective divisor in $X$, called the boundary, such that $K_{(X,D)} = K_{X} + D$ is trivial. For some applications of this notion in the context of physics, and F-theory in particular, see \cite{Donagi:2012ts}.} their union along the boundaries
\begin{equation}
	Y_{0} = Y^{0} \cup_{C_{1}} Y^{1} \cup_{C_{2}} \cdots \cup_{C_{P-1}} Y^{P-1} \cup_{C_{P}} Y^{P}
\end{equation}
giving a Calabi-Yau variety.

The divisors associated to the defining polynomials and discriminant
\begin{equation}
	f_{p} := \left. f_{b} \right|_{e_{p}=0}\,,\qquad g_{p} := \left. g_{b} \right|_{e_{p}=0}\,,\qquad \Delta_{p} := \left. \Delta_{b} \right|_{e_{p}=0}\,,
\label{eq:defining-polynomials-component-restrictions}
\end{equation}
of the Weierstrass model of the component $Y_{p}$ are in the classes
\begin{subequations}
\begin{align}
	F_{p} &:= \left. F \right|_{E_{p}} = 4\mathcal{L}_{p}\,,\\
	G_{p} &:= \left. G \right|_{E_{p}} = 6\mathcal{L}_{p}\,,\\
	\Delta_{p} &:= \left. \Delta \right|_{E_{p}} = 12\mathcal{L}_{p}\,.
\end{align}
\end{subequations}
Depending on the type of codimension-zero fibers over a given component, the restrictions \eqref{eq:defining-polynomials-component-restrictions} may vanish. In particular, exceptional components stemming from a blow-up along a curve with Class~5 family vanishing orders will yield trivial restrictions for all defining polynomials. However, we do not need to consider this case since, as discussed in \cref{sec:class-1-5-models}, it can always be eliminated by a series of base changes and modifications of the degeneration. Focusing then on models with Class~1--4 family vanishing orders and an open-chain resolution, the information about the geometry of the base components and the type of codimension-zero elliptic fibers found over them can be encapsulated in a diagram
\begin{equation}
    \begin{tikzcd}[column sep=1em]
        \mathrm{I}_{n_{0}} \arrow[r, dash] \arrow[d, dash] & \cdots \arrow[r, dash] & \mathrm{I}_{n_{p}} \arrow[r, dash] \arrow[d, dash] & \cdots \arrow[r, dash] & \mathrm{I}_{n_{P}} \arrow[d, dash]\\
        B_{0} \arrow[r, dash] & \cdots \arrow[r, dash] & B_{p} \arrow[r, dash] & \cdots \arrow[r, dash] & B_{P}\mathrlap{\,.}
    \end{tikzcd}
\end{equation}
Out of the Class~1--4 family vanishing orders, Classes 2 and 3 will lead to vanishing $f_{p}$ and $g_{p}$ restrictions for the exceptional component of the blow-up, respectively, to which we assign infinite component vanishing orders, see the paragraph after \eqref{rem:vanishing-restrictions}. For Class 4 family vanishing orders, we obtain a vanishing $\Delta_{p}$ restriction for the exceptional component of the blow-up. In analogy with what was done in \cite{Lee:2021qkx,Lee:2021usk} for the eight-dimensional F-theory analysis, it is convenient to define a modified discriminant divisor $\Delta'$ of the family variety $\mathcal{Y}$ that restricts to the components $Y^{p}$ non-trivially. We will use this divisor in \cref{sec:extraction-codimension-one-information} to read the physical 7-brane content of the components. With this application in mind, we define the modified discriminant $\Delta'$ to be the one in which we remove the discriminant components responsible for the codimension-zero singular fibers of the components $Y^{p}$, i.e.\ we ``subtract the background value of the axio-dilaton" before reading the 7-brane content in said components.
\begin{definition}
\label{def:modified-discriminant}
	Let $\{B^{p}\}_{0 \leq p \leq P}$ be the base components of the central fiber $Y_{0}$ of the open-chain resolution $\rho: \mathcal{Y} \rightarrow D$ of a degeneration $\hat{\rho}: \hat{\mathcal{Y}} \rightarrow D$, and let
	\begin{equation}
		 \ord{\mathcal{Y}}(f_{b},g_{b},\Delta_{b})_{E_{p}} = (0,0,n_{p})\,,\qquad 0 \leq p \leq P\,,
	\end{equation}
	be the vanishing orders associated to the codimension-zero singular fibers in said components. We define the divisor class of the modified discriminant in $\mathcal{B}$ to be
	\begin{equation}
	\Delta' := \Delta - \sum_{p=0}^{P} n_{p}E_{p}\,.
\end{equation}
\end{definition}
At the level of the defining polynomials of the Weierstrass model of $\mathcal{Y}$, the modified discriminant is defined by the equation
\begin{equation}
    \Delta = e_{0}^{n_{0}} \cdots e_{P}^{n_{P}} \Delta'\,,
\label{eq:def-modified-discriminant}
\end{equation}
and therefore it no longer restricts to zero in the components.

Whenever we analyse the endpoint of an infinite-distance limit component by component, we will always work in what follows with the set of polynomials $\{f_{p}, g_{p}, \Delta'_{p}\}_{0 \leq p \leq P}$. For future reference, we collect their associated divisor classes.
\begin{proposition}
\label{prop:component-divisor-classes-single-infinite-distance-limit}
	Let $\{B^{p}\}_{0 \leq p \leq P}$ be the base components of the central fiber $Y_{0}$ of the open-chain resolution $\rho: \mathcal{Y} \rightarrow D$ of a degeneration $\hat{\rho}: \hat{\mathcal{Y}} \rightarrow D$, and let
	\begin{equation}
		 \ord{\mathcal{Y}}(f_{b},g_{b},\Delta_{b})_{E_{p}} = (0,0,n_{p})\,,\qquad 0 \leq p \leq P\,,
	\end{equation}
	be the vanishing orders associated to the codimension-zero singular fibers in said components. The (modified) divisor classes associated to the Weierstrass models in the components are
	\begin{subequations}
	\begin{align}
		F_{p} &= 4\overline{K}_{B^{p}} - (1-\delta_{p,0})4C_{p-1} - (1-\delta_{p,P})4C_{p+1}\,,\\
		G_{p} &= 6\overline{K}_{B^{p}} - (1-\delta_{p,0})6C_{p-1} - (1-\delta_{p,P})6C_{p+1}\,,\\
		\Delta'_{p} &= 12\overline{K}_{B^{p}} + (1-\delta_{p,0})(n_{p-1} - 12)C_{p-1} + (1-\delta_{p,P})(n_{p+1} - 12)C_{p+1}\,,
	\end{align}
	\end{subequations}
	for $0 \leq p \leq P$.
\end{proposition}
\begin{proof}
	It follows from \cref{prop:component-line-bundle-single-infinite-distance-limit} and \cref{def:modified-discriminant}.
\end{proof}
We illustrate these computations with an example.

\begin{example}
	Continuing with \cref{example:geometry-components-single}, we see from the defining polynomials \eqref{eq:example1-defining-polynomials-blow-up} that the codimension-zero singular fibers in the components $\{Y^{p}\}_{0 \leq p \leq 2}$ are associated to the vanishing orders
	\begin{subequations}
	\begin{align}
	    \ord{\mathcal{Y}}(f_{b},g_{b},\Delta_{b})_{E_{0}} = (0,0,0)\,,\\
	     \ord{\mathcal{Y}}(f_{b},g_{b},\Delta_{b})_{E_{1}} = (0,0,1)\,,\\
	      \ord{\mathcal{Y}}(f_{b},g_{b},\Delta_{b})_{E_{2}} = (0,0,0)\,,
	\end{align}
	\end{subequations}
	i.e.\ the geometry of the model is
	\begin{equation}
	    \begin{tikzcd}[column sep=1em]
	        \mathrm{I}_{0} \arrow[rr, dash] \arrow[d, dash] & & \mathrm{I}_{1} \arrow[rr, dash] \arrow[d, dash] & & \mathrm{I}_{0} \arrow[d, dash]\\
	        \mathbb{F}_{7} \arrow[rr, dash] & & \mathbb{F}_{7} \arrow[rr, dash] & & \mathbb{F}_{7}\mathrlap{\,.}
	    \end{tikzcd}
	\end{equation}
	Using \cref{prop:component-line-bundle-single-infinite-distance-limit} we can compute the defining holomorphic line bundle over each component to be
	\begin{subequations}
	\begin{align}
		\mathcal{L}_{0} &= S_{0} + 9V_{0}\,,\\
		\mathcal{L}_{1} &= 2V_{1}\,,\\
		\mathcal{L}_{2} &= S_{2} + 2V_{2}\,,
	\end{align}
	\end{subequations}
	where we are using the notation of \eqref{eq:Hirzebruch-toric-divisors} with the added subscripts indicating the base component that we are referring to. The divisors associated to the defining polynomials of the component Weierstrass models, as well as the restrictions of the modified discriminant, are
	\begin{subequations}
	\begin{align}
		F_{0} &= 4S_{0} + 36V_{0}\,,\\
		G_{0} &= 6S_{0} + 54V_{0}\,,\\
		\Delta'_{0} &= 11S_{0} + 108V_{0}\,,
	\end{align}
	\label{eq:example1-divisor-classes-B0}%
	\end{subequations}
	in the $B^{0}$ component,
	\begin{subequations}
	\begin{align}
		F_{1} &= 8V_{1}\,,\\
		G_{1} &= 12V_{1}\,,\\
		\Delta'_{1} &= 2S_{1} + 31V_{1}\,,
	\end{align}
	\label{eq:example1-divisor-classes-B1}%
	\end{subequations}
	in the $B^{1}$ component, and
	\begin{subequations}
	\begin{align}
		F_{2} &= 4S_{2} + 8V_{2}\,,\\
		G_{2} &= 6S_{2} + 12V_{2}\,,\\
		\Delta'_{2} &= 11S_{2} + 17V_{2}\,,
	\end{align}
	\label{eq:example1-divisor-classes-B2}%
	\end{subequations}
	in the $B^{2}$ component. Computing the polynomials $\{f_{p}, g_{p}, \Delta'_{p}\}_{0 \leq p \leq 2}$ starting from \eqref{eq:example1-defining-polynomials-blow-up} we find
	\begin{subequations}
	\begin{align}
		f_{0} &= -3 t^4 v^2 w^2 (v-w)^2 (v+w)^2\,,\\
		g_{0} &= t^5 v^3 \left(e_1 v^{16}-2 t v^6 w^3+6 t v^4 w^5-6 t v^2 w^7+2 t w^9\right)\,,\\
		\Delta'_{0} &= 27 t^{10} v^{22} \left(e_1 v^{16}-4 t v^6 w^3+12 t v^4 w^5-12 t v^2 w^7+4 t w^9\right)\,,
	\end{align}
	\label{eq:example1-defining-polynomials-B0}%
	\end{subequations}
	for the $B^{0}$ component,
	\begin{subequations}
	\begin{align}
		f_{1} &= -3 v^2 w^2 (v-w)^2 (v+w)^2\,,\\
		g_{1} &= -2 v^3 w^3 (v-w)^3 (v+w)^3\,,\\
		\Delta'_{1} &= -108 v^6 w^3 (v-w)^3 (v+w)^3 \left(e_2^2 v^{16}-e_0 e_2 v^8 w+e_0 e_2 v^6 w^3+e_0^2 w^2\right)\,,
	\end{align}
	\label{eq:example1-defining-polynomials-B1}%
	\end{subequations}
	for the $B^{1}$ component, and
	\begin{subequations}
	\begin{align}
		f_{2} &= -3 s^4 v^2 w^2 (v-w)^2 (v+w)^2\,,\\
		g_{2} &= s^5 v^3 w^2 \left(e_1-2 s v^6 w+6 s v^4 w^3-6 s v^2 w^5+2 s w^7\right)\,,\\
		\Delta'_{2} &= 27 s^{10} v^6 w^4 \left(e_1-4 s v^6 w+12 s v^4 w^3-12 s v^2 w^5+4 s w^7\right)\,,
	\end{align}
	\label{eq:example1-defining-polynomials-B2}%
	\end{subequations}
	for the $B^{2}$ component, where we have used the available $\mathbb{C}^{*}$-actions to set to one those coordi\-nates that are not allowed to vanish in a given component in view of the Stanley-Reisner ideal \eqref{eq:example1-sr-ideal}. We see that, indeed, the zero loci of \eqref{eq:example1-defining-polynomials-B0}, \eqref{eq:example1-defining-polynomials-B1} and \eqref{eq:example1-defining-polynomials-B2} are in the divisor classes \eqref{eq:example1-divisor-classes-B0}, \eqref{eq:example1-divisor-classes-B1} and \eqref{eq:example1-divisor-classes-B2}, respectively.
\end{example}

\subsection{General degenerations and their resolution trees}
\label{sec:analysis-general-case}

In \cref{sec:geometry-components-single-limit,sec:line-bundles-components-single-limit} we have analysed the geometry of the base components of an open-chain resolution $\rho: \mathcal{Y} \rightarrow D$ of a genus-zero degeneration $\hat{\rho}: \hat{\mathcal{Y}} \rightarrow D$, as well as the line bundles associated to the Weierstrass model describing the elliptic fibrations over them. Restricting our attention to open-chain resolutions meant that in the geometrical study we could assume that the centres of the successive blow-ups were non-intersecting. Under these conditions, we only produce Hirzebruch surfaces as exceptional base components, as explained in \cref{prop:component-geometry-single-bis} and depicted in \cref{fig:single-infinite-distance-limit-open-chain}, with the line bundles printed in \cref{prop:component-line-bundle-single-infinite-distance-limit} defined over them.

More generally, one can relax this condition by allowing the blow-up centres to intersect, which corresponds to the most general degenerations of the type we consider, see the discussions in \cref{sec:definition-of-degenerations} and \cref{sec:modifications-of-degenerations}. The central fiber $B_{0}$ of $\mathcal{B}$ no longer has to be an open chain of surfaces, but can consist of a central component $B^{0}$, stemming from the original component $\hat{B}_{0}$, with ``branches" of intersecting surfaces attached to it that can split, resembling a tree. Moreover, $B^{0}$ may be a blow-up of $\hat{B}_{0}$ at points, and not always the same type of surface as occurred for the open-chain resolutions. We depict this in \cref{fig:general-limit-tree-degeneration}.
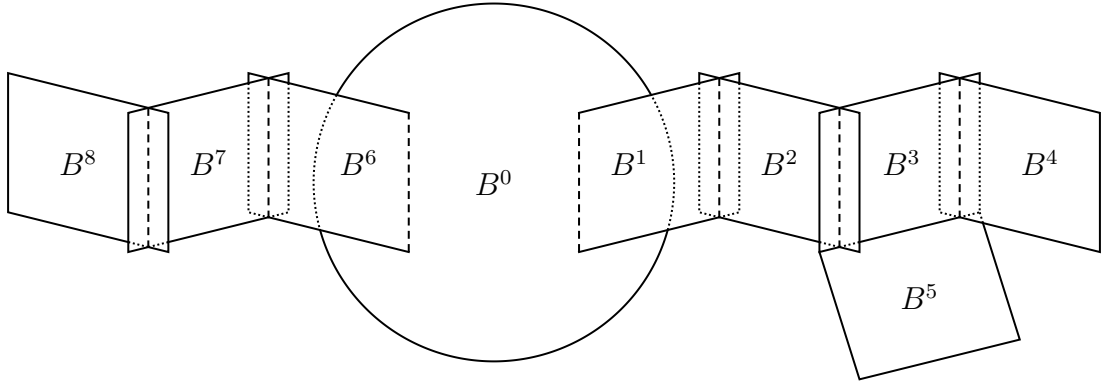
\begin{figure}[t!]
    \centering
	\begin{tikzpicture}[x=0.75pt,y=0.75pt]
		% Labels
		\node [] at (142.5-65-10,125) {$B^{1}$};
		\node [] at (142.5+0,125) {$B^{2}$};		
		\node [] at (142.5+65-4,125) {$B^{3}$};
		\node [] at (142.5+130,125) {$B^{4}$};	
		
		\node [] at (212.5,57.5) {$B^{5}$};
		
		\node [] at (-142.5+65+10,125) {$B^{6}$};
		\node [] at (-142.5-0,125) {$B^{7}$};		
		\node [] at (-142.5-65,125) {$B^{8}$};
		
		\node [] at (0,115) {$B^{0}$};

		% Original component
		\draw [style=black fine-line, domain=29.2694:150.731, samples=250] plot ({142.5-100-42.5+90*cos(\x)}, {115+90*sin(\x)});
		\draw [style=black fine-line, domain=195.423:344.577, samples=250] plot ({142.5-100-42.5+90*cos(\x)}, {115+90*sin(\x)});
		\draw [style=densely dotted, line width=0.75pt, domain=150.731:195.423, samples=250] plot ({142.5-100-42.5+90*cos(\x)}, {115+90*sin(\x)});
		\draw [style=densely dotted, line width=0.75pt, domain=-15.4233:29.2694, samples=250] plot ({142.5-100-42.5+90*cos(\x)}, {115+90*sin(\x)});
		
		% Component 1
		\draw [style=densely dotted, line width=0.75pt] (142.5-20,165) -- (142.5-20,100) -- (142.5-30,195/2);
		\draw [style=black fine-line] (142.5-30,195/2) -- (142.5-100,80);
		\draw [style=black fine-line] (142.5-100,150) -- (142.5-20,170) -- (142.5-20,165);
		\draw [style=densely dashed, line width=0.75pt] (142.5-100,80) -- (142.5-100,150);
		
		% Component 2
		\draw [style=black fine-line] (142.5+30,165/2) -- (142.5+40,80) -- (142.5+40,150) -- (142.5-40,170) -- (142.5-40,165);
		\draw [style=densely dotted, line width=0.75pt] (142.5-40,165) -- (142.5-40,100) -- (142.5-30,195/2);
		\draw [style=black fine-line] (142.5-30,195/2) -- (142.5+20,85);
		\draw [style=densely dotted, line width=0.75pt] (142.5+20,85) -- (142.5+30,165/2);
		
		% Component 3
		\draw [style=densely dotted, line width=0.75pt] (142.5+100,165) -- (142.5+100,100) -- (142.5+90,195/2);
		\draw [style=black fine-line] (142.5+90,195/2) -- (142.5+40,85);
		\draw [style=black fine-line] (142.5+30,165/2) -- (142.5+20,80) -- (142.5+20,150) -- (142.5+100,170) -- (142.5+100,165);
		\draw [style=densely dotted, line width=0.75pt] (142.5+40,85) -- (142.5+30,165/2);
		
		% Component 4
		\draw [style=black fine-line] (142.5+90,195/2) -- (142.5+160,80) -- (142.5+160,150) -- (142.5+80,170) -- (142.5+80,165);
		\draw [style=densely dotted, line width=0.75pt] (142.5+80,165) -- (142.5+80,100) -- (142.5+90,195/2);
		
		% Component 5
		\draw [style=black fine-line] (142.5+20,80) -- (142.5 + 40,15.969) -- (142.5+120,35.969) -- (244.111,94.597);
		\draw [style=densely dotted, line width=0.75pt] (142.5+100,100) -- (244.111,94.597);
		
		% Intersection lines
		\draw [style=densely dashed, line width=0.75pt] (142.5+30,70+165/2) -- (142.5+30,165/2);
		\draw [style=densely dashed, line width=0.75pt] (142.5-30,70+195/2) -- (142.5-30,195/2);
		\draw [style=densely dashed, line width=0.75pt] (142.5+90,70+195/2) -- (142.5+90,195/2);
		
		% Component 6
		\draw [style=densely dotted, line width=0.75pt] (-142.5+20,165) -- (-142.5+20,100) -- (-142.5+30,195/2);
		\draw [style=black fine-line] (-142.5+30,195/2) -- (-142.5+100,80);
		\draw [style=black fine-line] (-142.5+100,150) -- (-142.5+20,170) -- (-142.5+20,165);
		\draw [style=densely dashed, line width=0.75pt] (-142.5+100,80) -- (-142.5+100,150);
		
		% Component 7
		\draw [style=black fine-line] (-142.5-30,165/2) -- (-142.5-40,80) -- (-142.5-40,150) -- (-142.5+40,170) -- (-142.5+40,165);
		\draw [style=densely dotted, line width=0.75pt] (-142.5+40,165) -- (-142.5+40,100) -- (-142.5+30,195/2);
		\draw [style=black fine-line] (-142.5+30,195/2) -- (-142.5-20,85);
		\draw [style=densely dotted, line width=0.75pt] (-142.5-20,85) -- (-142.5-30,165/2);
		
		% Component 8
		\draw [style=black fine-line] (-142.5-30,165/2) -- (-142.5-20,80) -- (-142.5-20,150) -- (-142.5-100,170) -- (-142.5-100,100) -- (-142.5-40,85);
		\draw [style=densely dotted, line width=0.75pt] (-142.5-40,85) -- (-142.5-30,165/2);
		
		% Intersection lines
		\draw [style=densely dashed, line width=0.75pt] (-142.5-30,70+165/2) -- (-142.5-30,165/2);
		\draw [style=densely dashed, line width=0.75pt] (-142.5+30,70+195/2) -- (-142.5+30,195/2);
	\end{tikzpicture}
    \caption{Central fiber $B_{0}$ of the base $\mathcal{B}$ of the tree resolution of a general degeneration \mbox{$\hat{\rho}: \hat{\mathcal{Y}} \rightarrow D$} not falling under the single infinite-distance limit category. Note that a resolution with the configuration of components depicted here will always present obscured infinite-distance limits, see the discussions in \cref{sec:obscured-infinite-distance-limits,sec:single-infinite-distance-limits-and-open-chain-resolutions}.}
    \label{fig:general-limit-tree-degeneration}
\end{figure}
Note that allowing the branches to split automatically implies that the blow-up centres must intersect, since in a Hirzebruch surface we cannot tune more than two non-intersecting curves of non-minimal singular fibers, see \cref{prop:star-degenerations-P2-Fn}.

We relegate a detailed description of the structure of such resolution trees and the Weierstrass models over them  to Appendix \ref{sec:res-trees}. As we will see, the study of the open-chain resolutions associated to single infinite-distance limit degenerations already allowed us to discuss most of the needed concepts, and only minor modifications of the analysis will be necessary in order to include the new cases. The structure of the more general resolution trees is, in fact, summarised in the following proposition.
\begin{restatable}{proposition}{componentgeometrygeneral}
\label{prop:component-geometry-general}
	Let $\hat{\mathcal{B}}$ be the base family variety of a genus-zero degeneration $\hat{\rho}: \hat{\mathcal{Y}} \rightarrow D$, and $\mathrm{Bl}_{p-1}(\hat{B})$ be the result of $p-1$ blow-ups of $\hat{\mathcal{B}}$. Let $C_{p} \subset B^{i}$ be a smooth irreducible curve over which $\mathrm{Bl}_{p-1}(\hat{\mathcal{Y}})$ presents non-minimal singular fibers. Then, the exceptional component $B^{p} = E_{p}$ arising from the blow-up of $\mathrm{Bl}_{p-1}(\hat{B})$ along $C_{p}$ is the Hirzebruch surface
	\begin{equation}
		B^{p} = \mathbb{F}_{\left| n_{p} \right|}\,,\qquad n_{p} := C_{p} \cdot_{B^{i}} C_{p} + \sum_{\substack{q = 0\\q \neq i}}^{p-1} \left. E_{q} \right|_{E_{i}} \cdot_{B^{i}} C_{p}\,.
	\label{eq:general-blow-up-Fn}
	\end{equation}
	Moreover, define the set of components $\{B^{q}\}_{q \in \mathcal{I}}$ to be comprised by those elements in $\{B^{q}\}_{0 \leq q \leq p-1}$ such that
	\begin{equation}
		\mathrm{codim}_{B^{i}} \left( \left. E_{q} \right|_{E_{i}} \cdot_{B^{i}} C_{p} \right) = 2\,.
	\end{equation}
	After the blow-up along $C_{p}$, the old components $\{B^{q}\}_{q \in \mathcal{I}}$ must be substituted for their blow-ups $\{\mathrm{Bl}_{\left. E_{q} \right|_{E_{i}} \cdot_{B^{i}} C_{p}}(B^{q})\}_{q \in \mathcal{I}}$.
\end{restatable}

With this understanding of the more general resolution trees, we can finally tackle the proof of \cref{prop:single-infinite-distance-limits-and-open-resolutions}, which guarantees that single infinite-distance limits only give rise to open-chain resolutions, rather than to resolutions trees. This analysis can be found in \cref{sec:single-infinite-distance-limits-and-open-chain-resolutions}.

\subsection{Comments on genus-one degenerations}
\label{sec:comments-genus-one-degenerations}

We saw in \cref{sec:curves-of-non-minimal-fibers} that genus-one degenerations are much more constrained than their genus-zero counterparts. In fact, they can only occur if the curve $C$ in the central fiber $\hat{B}_{0}$ of $\hat{\mathcal{B}}$ is in the anticanonical class $C = \overline{K}_{\hat{B}_{0}}$, see \cref{prop:genus-restriction}. We leave the study of these highly non-generic class of infinite-distance limits for future works, merely offering some comments on them before we conclude the section.

Let us look at the geometry of the base components of the resolved degeneration. A study analogous to the one carried out in \cref{sec:geometry-components-single-limit,sec:geometry-components-arbitrary-limit} would show that the exceptional components of the base blow-ups would be the projectivization of rank 2 vector bundles over a genus one curve. This is the first non-trivial case of algebraic vector bundles over a curve, where Grothendieck's splitting theorem no longer applies. Algebraic vector bundles over an elliptic curve defined over an algebraically closed field have been classified by Atiyah \cite{Atiyah1957}.

Regarding the holomorphic line bundles defined over the components and associated to the Weierstrass models giving the elliptic fibrations of which they are the bases, one can perform an analysis similar to the one in \cref{sec:line-bundles-components-single-limit,sec:line-bundles-components-arbitrary-limit}. Considering a two-component model arising from a genus-one degeneration, this would lead us to conclude that the line bundle defined over the strict transform of the original base components is trivial, meaning in physical terms that it contains no local 7-brane content. Hence, the resolution of genus-one degenerations will present some spectator tails of the base geometry over which the elliptic fibration is trivial and the Type IIB axio-dilaton constant. This structure also implies that repeatedly blowing up along genus-one curves leads to open-chain resolutions.
%auto-ignore

\section{Degenerations of Hirzebruch models}
\label{sec:degenerations-Hirzebruch-models}

While the previous section analysed single infinite-distance limits for arbitrary base spaces, we now specialise the discussion to elliptic fibrations over Hirzebruch surfaces $\hat{B} = \mathbb{F}_{n}$, with \mbox{$0 \leq n \leq 12$} (see \cref{sec:P2-and-Fn}). These models are of particular interest because of heterotic duality. For this reason, we use this section to lay out the properties of this subclass of degenerations as explicitly as possible. 

After briefly reviewing their relation to the heterotic string, we classify the possible genus-zero (and genus-one) curves over which non-minimal elliptic fibers can be supported in a single infinite-distance limit. We then explicitly describe the open-chain resolutions for the different single infinite-distance limits involving genus-zero non-minimal curves. This gives rise to so-called horizontal, vertical, or mixed degenerations (and their resolutions), where the terminology refers to the location of the non-minimal curve(s) with respect to the rational fibration of the Hirzebruch surface $\hat{B}_{0} = \mathbb{F}_{n}$. The base blow-ups taken during the resolution process yield chains of intersecting Hirzebruch surfaces of special types acting as the base of the elliptic fibration of the central fiber of the degeneration; the log Calabi-Yau structure of the elliptic fibrations in its components can be made very explicit. These results are a direct application of the general discussion in \cref{sec:geometric-description-6D-F-theory-limits}. However, toric methods can be used in the analysis of some degenerations of Hirzebruch models, allowing us to rederive a subset of the results obtained in \cref{sec:geometric-description-6D-F-theory-limits} in a succinct and alternative way and therefore serving as further examples for the general discussion.

To set the stage, recall that six-dimensional F-theory models over Hirzebruch surfaces have $n_{T} = 1$ tensors and are of particular interest due to their connection to perturbative heterotic dual models \cite{Morrison:1996na,Morrison:1996pp,Bershadsky:1996nh,Friedman:1997yq}. This is reflected in the appearance of a compatible K3 fibration structure which extends to the degeneration $\hat{\rho}: \hat{\mathcal{Y}} \rightarrow D$ introduced in \cref{sec:definition-of-degenerations} in the following way. For $\hat{B} = \mathbb{F}_{n}$, the $\mathbb{P}^{1}$-fibration in $\mathbb{F}_{n}$ implies that a fixed member $\hat{Y}_{u}$ of the family $\hat{\mathcal{Y}}$ can be seen as an elliptic fibration over $\mathbb{F}_{n}$, or as a K3-fibration over the base $\mathbb{P}^{1}_{b}$ of ${\mathbb F_n}$, i.e.\
\begin{equation}
    \begin{tikzcd}
        \mathcal{E} \arrow[r] & \hat{Y}_{u} \arrow[d, "\pi_{\mathrm{ell}}"]\\
         & \mathbb{F}_{n}
    \end{tikzcd}
    \qquad
    \textrm{and}
    \qquad
    \begin{tikzcd}
        \mathrm{K3} \arrow[r] & \hat{Y}_{u} \arrow[d, "\pi_{\mathrm{K3}}"]\\
         & \mathbb{P}^{1}_{b}\mathrlap{\,.}
    \end{tikzcd}
\end{equation}
The two fibrations naturally extend to $\hat{\mathcal{Y}}$, where we have
\begin{equation}
    \begin{tikzcd}
        \mathcal{E} \arrow[r] & \hat{\mathcal{Y}} \arrow[d, "\Pi_{\mathrm{ell}}"]\\
         & \mathbb{F}_{n} \times D
    \end{tikzcd}
    \qquad
    \textrm{and}
    \qquad
    \begin{tikzcd}
        \mathrm{K3} \arrow[r] & \hat{\mathcal{Y}} \arrow[d, "\Pi_{\mathrm{K3}}"]\\
         & \mathbb{P}^{1}_{b} \times D\mathrlap{\,,}
    \end{tikzcd}
\end{equation}
with $\pi_{\mathrm{ell}} = \hat{\rho} \circ \Pi_{\mathrm{ell}}$ and $\pi_{\mathrm{K3}} = \hat{\rho} \circ \Pi_{\mathrm{K3}}$. The K3-fibration compatible with the elliptic fibration over $\hat{B} = \mathbb{F}_{n}$ yields a Kulikov degeneration of the K3 fibers that can be interpreted in the heterotic dual model as a (possibly infinite-distance) limit; this is the dual of the infinite-distance limit in the complex structure moduli space of six-dimensional F-theory represented by the degeneration $\hat{\rho}: \hat{\mathcal{Y}} \rightarrow D$. While this duality is expected to hold in general, the precise map between the theories is available only under certain conditions, which we recall in \cite{ALWPart2}.

The heterotic interpretation for (at least some of) the degenerations of Hirzebruch models will prove most useful in \cite{ALWPart2} when we try to gain some intuition for the physics obtained at the endpoints of the infinite-distance limits under study. This motivates a detailed study of the underlying geometry.

\begin{sidewaystable}[p!]
    \centering
    \resizebox{\linewidth}{!}{
    \begin{tblr}{columns={c,m},
				row{1}={abovesep=3pt,belowsep=3pt},
    				row{2-5}={abovesep=7.5pt,belowsep=7.5pt},
				hline{1}={2-4}{solid},
				hline{2-6}={solid},
				vline{1}={2-5}{solid},
				vline{2-5}={solid},
				}
         & Non-minimal curves & Central component structure & Component line bundles and discriminants\\
        \begin{tabular}{c}Horizontal\\ (Case A)\end{tabular} &
        $\begin{aligned}
            \hat{\mathscr{C}}_{1} &= \{h\}\\
            \hat{\mathscr{C}}_{1} &= \{h + nf\}\\
            \hat{\mathscr{C}}_{2} &= \{h, h + nf\}
        \end{aligned}$ &
        $\begin{tikzcd}[column sep=1em, ampersand replacement=\&]
            \mathrm{I}_{n_{0}} \arrow[r, dash] \arrow[d, dash] \& \cdots \arrow[r, dash] \& \mathrm{I}_{n_{p}} \arrow[r, dash] \arrow[d, dash] \& \cdots \arrow[r, dash] \& \mathrm{I}_{n_{P}} \arrow[d, dash]\\
            \mathbb{F}_{n} \arrow[r, dash] \& \cdots \arrow[r, dash] \& \mathbb{F}_{n} \arrow[r, dash] \& \cdots \arrow[r, dash] \& \mathbb{F}_{n}
        \end{tikzcd}$ &
        {$\begin{aligned}
            \mathcal{L}_{0} &= S_{0} + (2+n)V_{0}\\
            \mathcal{L}_{p} &= 2V_{p}\\
            \mathcal{L}_{P} &= S_{P} + 2V_{P}
        \end{aligned}$\\[0.5cm]
        $\begin{aligned}
            \Delta'_{0} &= (12 + n_{0} - n_{1})S_{0} + (24 + 12n)V_{0}\\
            \Delta'_{p} &= (2n_{p} - n_{p-1} - n_{p+1})S_{p} + (24 + n(n_{p} - n_{p-1}))V_{p}\\
            \Delta'_{P} &= (12 + n_{P} - n_{P-1})S_{P} + (24 + n(n_{P} - n_{P-1}))V_{P}
        \end{aligned}$}\\
         \begin{tabular}{c}Vertical\\ (Case B)\end{tabular} &
        $\hat{\mathscr{C}}_{1} = \{f\}$ &
        $\begin{tikzcd}[column sep=1em, ampersand replacement=\&]
            \mathrm{I}_{n_{0}} \arrow[r, dash] \arrow[d, dash] \& \cdots \arrow[r, dash] \& \mathrm{I}_{n_{p}} \arrow[r, dash] \arrow[d, dash] \& \cdots \arrow[r, dash] \& \mathrm{I}_{n_{P}} \arrow[d, dash]\\
            \mathbb{F}_{n} \arrow[r, dash] \& \cdots \arrow[r, dash] \& \mathbb{F}_{0} \arrow[r, dash] \& \cdots \arrow[r, dash] \& \mathbb{F}_{0}
        \end{tikzcd}$ &
        {$\begin{aligned}
            \mathcal{L}_{0} &= 2S_{0} + (1+n)W_{0}\\
            \mathcal{L}_{p} &= 2S_{p}\\
            \mathcal{L}_{P} &= 2S_{P} + W_{P}
        \end{aligned}$\\[0.5cm]
        $\begin{aligned}
            \Delta'_{0} &= 24S_{0} + (12 + 12n + n_{0} - n_{1})W_{0}\\
            \Delta'_{p} &= 24S_{p} + (2n_{p} - n_{p-1} - n_{p+1})W_{p}\\
            \Delta'_{P} &= 24S_{P} + (12 + n_{P} - n_{P-1})W_{P}
        \end{aligned}$}\\
        \begin{tabular}{c}Mixed sectional\\ (Case C)\end{tabular} &
        \begin{tabular}{c}
            $\hat{\mathscr{C}}_{1} = \{h+(n+\alpha)f\}$\\[0.25cm]
            $\alpha = 1\quad \textrm{with}\quad n \leq 6$\\
            $\alpha = 2\quad \textrm{with}\quad n = 0$
        \end{tabular} &
        $\begin{tikzcd}[column sep=1em, ampersand replacement=\&]
            \mathrm{I}_{n_{0}} \arrow[r, dash] \arrow[d, dash] \& \cdots \arrow[r, dash] \& \mathrm{I}_{n_{p}} \arrow[r, dash] \arrow[d, dash] \& \cdots \arrow[r, dash] \& \mathrm{I}_{n_{P}} \arrow[d, dash]\\
            \mathbb{F}_{n+2\alpha} \arrow[r, dash] \& \cdots \arrow[r, dash] \& \mathbb{F}_{n+2\alpha} \arrow[r, dash] \& \cdots \arrow[r, dash] \& \mathbb{F}_{n}
        \end{tikzcd}$ &
        {$\begin{aligned}
            \mathcal{L}_{0} &= S_{0} + (2+(n+2\alpha))V_{0}\\
            \mathcal{L}_{p} &= 2V_{p}\\
            \mathcal{L}_{P} &= S_{P} + (2-\alpha)V_{P}
        \end{aligned}$\\[0.5cm]
        $\begin{aligned}
            \Delta'_{0} &= (12 + n_{0} - n_{1})S_{0} + (24 + 12(n+2\alpha))V_{0}\\
            \Delta'_{p} &= (2n_{p} - n_{p-1} - n_{p+1})S_{p} + (24 + (n+2\alpha)(n_{p} - n_{p-1}))V_{p}\\
            \Delta'_{P} &= (12 + n_{P} - n_{P-1})S_{P} + ((24-12\alpha) + (n+\alpha)(n_{P} - n_{P-1}))V_{P}
        \end{aligned}$}\\
        \begin{tabular}{c}Mixed bisectional\\ (Case D)\end{tabular}
         &
        \begin{tabular}{c}
            $\hat{\mathscr{C}}_{1} = \{ 2h + bf \}$\\[0.25cm]
            $(n,b) = (0,1)$\\
            $(n,b) = (1,2)$
        \end{tabular} &
        \begin{tikzcd}[column sep=1em, ampersand replacement=\&]
            \mathrm{I}_{n_{0}} \arrow[r, dash] \arrow[d, dash] \& \cdots \arrow[r, dash] \& \mathrm{I}_{n_{p}} \arrow[r, dash] \arrow[d, dash] \& \cdots \arrow[r, dash] \& \mathrm{I}_{n_{P}} \arrow[d, dash]\\
            \mathbb{F}_{4} \arrow[r, dash] \& \cdots \arrow[r, dash] \& \mathbb{F}_{4} \arrow[r, dash] \& \cdots \arrow[r, dash] \& \mathbb{F}_{n}
        \end{tikzcd} &
        {$\begin{aligned}
            \mathcal{L}_{0} &= S_{0} + (2+4)V_{0}\\
            \mathcal{L}_{p} &= 2V_{p}\\
            \mathcal{L}_{P} &= V_{P}
        \end{aligned}$\\[0.5cm]
        $\begin{aligned}
            \Delta'_{0} &= (12 + n_{0} - n_{1})S_{0} + (24 + 12 \cdot 4)V_{0}\\
            \Delta'_{p} &= (2n_{p} - n_{p-1} - n_{p+1})S_{p} + (24 + 4(n_{p} - n_{p-1}))V_{p}\\
            \Delta'_{P} &= 2(n_{P} - n_{P-1})S_{P} + (12 + (n+1)(n_{P} - n_{P-1}))V_{P}
        \end{aligned}$}
    \end{tblr}
    }
    \caption{Genus-zero single infinite-distance limit degenerations of Hirzebruch models.}
    \label{tab:genus-zero-Hirzebruch-summary}
\end{sidewaystable}

\subsection{Single infinite-distance limits in Hirzebruch models}
\label{sec:single-Hirzebruch}

The simplest kind of degenerations of Hirzebruch models are those in which a single curve supports non-minimal singular elliptic fibers, rather than a collection of them, meaning that we are facing a single infinite-distance limit instead of a simultaneous superposition of more of them. This intuition was more carefully encapsulated in the notion of a single infinite-distance limit degeneration, see \cref{def:single-infinite-distance-limit-original}, which allows for slightly more general configurations that would nonetheless arise as equivalent degenerations to the ones just described.

We will mostly concern ourselves with this class of degenerations in \cite{ALWPart2}, and it would therefore be useful to precisely determine over which curves they can occur. We already know from \cref{prop:genus-restriction} that they must be curves of genus-zero or genus-one. Restricting our attention to Hirzebruch models, we can be more explicit and give a complete list of the curves that can support non-minimal elliptic fibers as part of a single infinite-distance limit. To derive such a list, we analyse below the non-minimal fibers over curves in a Weierstrass model \mbox{$\pi: Y \rightarrow B$} with $B = \mathbb{F}_{n}$. This applies, in particular, to the non-minimal elliptic fibers in the central fiber of a degeneration of a Hirzebruch model, meaning that the list is relevant both for conventional and obscured single infinite-distance limits (as defined and discussed in \cref{sec:obscured-infinite-distance-limits}). 

In the sequel, we denote by $h$ and $f$ the class of the exceptional section and the fiber of the Hirzebruch surface $B = \mathbb F_n$, respectively, with intersection products $h \cdot h = -n$, $h \cdot f = 1$ and $f \cdot f =0$. For more details on the notation, we refer to \cref{sec:six-dimensional-F-theory-bases}.
\begin{proposition}
\label{prop:non-minimal-curves}
	Let $\pi: Y \rightarrow B$ be a Calabi-Yau Weierstrass model over $B = \mathbb{F}_{n}$. The smooth, irreducible curves $C$ that can support non-minimal fibers are the following:
    \begin{itemize}
        \item $C = f$, with $g(C) = 0$;
        \item $C = h + bf$, with $b = 0$, $n$, $n+1$, $n+2$ and $g(C) = 0$;
        \item $C = 2h + bf$, with $(n,b) = (0,1)$, $(1,2)$ and $g(C) = 0$; and
        \item $C = 2h + bf$, with $(n,b) = (0,2)$, $(1,3)$, $(2,4)$ and $g(C) = 1$.
    \end{itemize}
\end{proposition}
\begin{proof}
    Let us express the divisor class of $C$ as
    \begin{equation}
        C = ah + bf\,,\quad a,b \in \mathbb{Z}_{\geq 0}\,.
    \end{equation}
    From the effectiveness constraint $C \leq \overline{K}_{B}$, see \cref{prop:K-C-effectiveness}, we find the bounds
    \begin{equation}
        a \leq 2\,,\quad b \leq 2+n\,,
    \end{equation}
    to which we add $a+b > 0$ to avoid the trivial case. We study the possible curves for the three allowed values of $a = 0$, $1$, $2$ separately.
    \begin{itemize}
        \item $a=0$: For this case, the curves
        \begin{equation}
            C = bf\,,\quad 0 < b \leq 2+n \,,
        \end{equation}
        are generically reducible, unless $b = 1$. Hence, $C = f$, and from the Adjunction formula of \cref{prop:adjunction-genus} we read $g(C) = 0$.

        \item $a = 1$: We have the curves
        \begin{equation}
            C = h + bf\,,\quad 0 \leq b \leq 2+n\,.
        \end{equation}
        From the intersection with $h$ we see, applying \cref{prop:negative-intersection}, that, for $b \neq 0$, $C$ is generically irreducible when
        \begin{equation}
            C \cdot h \geq 0 \Leftrightarrow b \geq n\,.
        \end{equation}
        Additionally, the curve $C = h$ is also irreducible. Altogether, we have the curves $C = h + bf$ with $b = 0$, $n$, $n+1$, $n+2$. From the Adjunction formula, we see that $g(C) = 0$ for all of them.

        \item $a = 2$: Computing the intersection product of the curves
        \begin{equation}
            C = 2h + bf\,, \quad 0 \leq b \leq 2+n \,,
        \end{equation}
        with $h$ we see, applying \cref{prop:negative-intersection}, that, for $b \neq 0$, $C$ is generically irreducible when
        \begin{equation}
            C \cdot h \geq 0 \Leftrightarrow b \geq 2n\,.
        \end{equation}
        Writing $b = 2n + \beta$, we observe that both inequalities can only be fulfilled in a handful of cases. These are listed in \cref{tab:2h-irreducible-curves}, with their genus computed through the Adjunction formula. The curves with $g(C) = 1$ are those in the anti-canonical class. When $b = 0$, which can only occur for $n=0$, the curve $C = 2h$ is also reducible.
        \begin{table}[ht!]
            \centering
            \begin{tblr}{columns = {c},
                        hline{1,2,3}={solid},
                        hline{4}={1-3}{solid},
                        hline{5}={1-2}{solid},
                        vline{1,2,3}={solid},
                        vline{4}={1-3}{solid},
                        vline{5}={1-2}{solid}
                        }
                \diagbox{$\beta$}{$n$} & $0$ & $1$ & $2$ \\
                $0$ & $b=0$ & $g(C) = 0$ & $g(C) = 1$ \\
                $1$ & $g(C) = 0$ & $g(C) = 1$ &  \\
                $2$ & $g(C) = 1$ &  &  \\
            \end{tblr}
            \caption{Irreducible $C = 2h + (2n + \beta)f$ curves and their genus.}
            \label{tab:2h-irreducible-curves}
        \end{table}
    \end{itemize}
\end{proof}

A degeneration $\hat{\rho}: \hat{\mathcal{Y}} \rightarrow D$ of Hirzebruch models in which the set curves $\hat{\mathscr{C}}_{r}$ in $\hat{\mathcal{B}}$ with non-minimal component vanishing orders is $\hat{\mathscr{C}}_{1} = \{ C \}$ with $C \subset \hat{B}_{0}$ in the list of \cref{prop:non-minimal-curves} can be a single infinite-distance limit degeneration if the other conditions in \cref{def:single-infinite-distance-limit-original} are met. Beyond these candidates, the only other possible single infinite-distance limit degenerations of Hirzebruch models are those in which $\hat{\mathscr{C}}_{2} = \{ C_{0}, C_{\infty} \}$. All other choices of $\hat{\mathscr{C}}_{r}$ violate \mbox{Condition \ref{item:single-infinite-distance-limit-original-1}} of \cref{def:single-infinite-distance-limit-original}. We show this in \cref{prop:star-degenerations-P2-Fn}, to which we refer for details. In particular, one cannot engineer non-minimal fibers over two (unless $n=0$) or more distinct representatives of $f$ without tuning non-minimalities also over $h$.

Focusing on genus-zero degenerations of Hirzebruch models, we can classify them into four cases sharing similar properties, depending on which of the curves in \cref{prop:non-minimal-curves} support the non-minimal elliptic fibers. We will use the following nomenclature.
\begin{definition}
\label{def:classification-Hirzebruch-single-infinite-distance-limits}
	Let $\hat{\rho}: \hat{\mathcal{Y}} \rightarrow D$ with $\hat{\mathcal{B}} = \mathbb{F}_{n} \times D$ and $0 \leq n \leq 12$ be a single (obscured) infinite-distance limit degeneration of Hirzebruch models. The curves in $\hat{B}_{0}$ that present non-minimal (component) vanishing orders can be classified as follows: 
	\begin{itemize}
		\item Case A:\quad $C = h$, or $C = h + nf$, or both (horizontal model).
		
		\item Case B:\quad $C = f$ (vertical model).
		
		\item Case C:\quad $C = h + (n+1)f$, or $C = h + (n+2)f$ (mixed section model).
		
		\item Case D:\quad $C = 2h + bf$, with $(n,b) = (0,1)$, $(1,2)$ (mixed bisection model).
	\end{itemize}
\end{definition}
For Cases A and B there are representatives of the divisor classes of the non-minimal curves that coincide with coordinate divisors of the toric description of $\hat{\mathcal{B}}$, which makes the blow-up process of the base very explicit thanks to its global coordinate description. This was already used in \cref{example:modification-of-degeneration,example:geometry-components-single}, which refer to a Case A degeneration of Hirzebruch models, and will be exploited in \cref{sec:geometry-components-horizontal,sec:geometry-components-horizontal} to rederive the results of \cref{prop:component-geometry-single} for this particular set of degenerations. Cases C and D, on the other hand, cannot be treated in such a convenient way, but may still be analysed following the general discussion of single infinite-distance limits of \cref{sec:geometry-components-single-limit,sec:line-bundles-components-single-limit}.

\subsection{Horizontal models}
\label{sec:geometry-components-horizontal}

Consider a single infinite-distance limit degeneration $\hat{\rho}: \hat{\mathcal{Y}} \rightarrow D$ of Hirzebruch models corresponding to a horizontal model, i.e.\ in Case A. This means that the set of non-minimal curves in $\hat{B}_{0} = \mathbb{F}_{n} \subset \hat{\mathcal{B}}$ is either $\hat{\mathscr{C}}_{1} = \{ h \}$, $\hat{\mathscr{C}}_{1} = \{ h+nf \}$ or $\hat{\mathscr{C}}_{2} = \{ h, h+nf \}$. Through a sequence of blow-ups and blow-downs, we may assume without loss of generality that it is $\hat{\mathscr{C}}_{1} = \{ h \}$, see \cref{sec:curves-of-non-minimal-fibers}. Let $\rho: \mathcal{Y} \rightarrow D$ be the open-chain resolution of $\hat{\rho}: \hat{\mathcal{Y}} \rightarrow D$. Since $\hat{B}_{0} = \mathbb{F}_{n}$ and $\hat{\mathscr{C}}_{1} = \{ h \}$, we know from \cref{prop:component-geometry-single} that the components $\{B^{p}\}_{0 \leq p \leq P}$ of the base of the central fiber $Y_{0} = \bigcup_{p=0}^{P} Y^{p}$ of the open chain resolution will be $B^{p} = \mathbb{F}_{n}$ for $0 \leq p \leq P$. Let us rederive this succinctly using toric methods.

The base $\hat{\mathcal{B}}$ of the original family $\hat{\mathcal{Y}}$ is $\hat{\mathcal{B}} = \mathbb{F}_{n} \times D \simeq \mathbb{F}_{n} \times \mathbb{C}$. As a toric variety, it can be described
in the lattice
\begin{equation}
    N := \langle (1,0,0),\, (0,1,0),\, (0,0,1) \rangle_{\mathbb{Z}}
\end{equation}
using the fan $\Sigma_{\hat{\mathcal{B}}}$ given by the edges
\begin{equation}
    v = (1,0,0)\,,\quad t = (0,1,0)\,,\quad w = (-1,-n,0)\,,\quad s = (0,-1,0)\,,\quad u = (0,0,1)\,,
\label{eq:F1-fan-edges}
\end{equation}
which we represent in \cref{fig:F1xDfan}. We use the notation for the coordinates introduced in \cref{sec:six-dimensional-F-theory-bases}.
\begin{figure}[t!]
     \centering
     \begin{subfigure}[b]{0.45\textwidth}
         \centering
         \includegraphics[width=\textwidth]{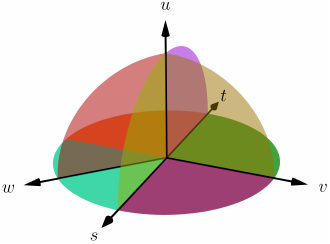}
         \caption{Toric fan of $\mathbb{F}_{n} \times \mathbb{C}$.}
         \label{fig:F1xDfan}
     \end{subfigure}
     \hfill
     \begin{subfigure}[b]{0.45\textwidth}
         \centering
         \includegraphics[width=\textwidth]{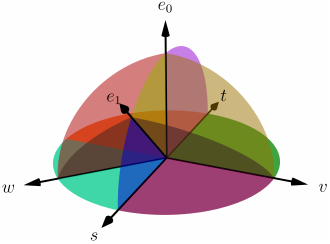}
         \caption{Toric fan of $\mathbb{F}_{n} \times \mathbb{C}$ blown up along $S \cap \mathcal{U}$.}
         \label{fig:F1xDbhorfan}
     \end{subfigure}
     \caption{Toric fans associated to the family base of a horizontal model.}
     \label{fig:toric-fans-horizontal}
\end{figure}
Performing a toric blow-up with centre the curve $S \cap \mathcal{U} = \{s=u=0\}_{\hat{\mathcal{B}}}$ subdivides the fan by adding a new edge and the appropriate 2-cones. The resulting fan $\Sigma_{\mathcal{B}}$ is given in \cref{fig:F1xDbhorfan}. If $P$ such blow-ups are necessary in order to arrive at the $\mathcal{Y}$ family variety, the corresponding family base $\mathcal{B}$ will be described by the fan with edges \eqref{eq:F1-fan-edges}, in which we rename $e_{0} := u$, and to which we add the set of edges
\begin{equation}
	\left\{ e_{p} = (0,-p,1) \right\}_{1 \leq p \leq P}\,,
\end{equation}
as well as the necessary 2-cones. The toric divisors associated to the edges $v$, $t$, $w$ and $s$ of the blown-up fan are the strict transforms under the composition of the blow-up maps of the original toric divisors.
The different components $\{ B^{p} \}_{0 \leq p \leq P}$ of the central fiber $B_{0}$ of $\mathcal{B}$ correspond to the toric divisors given by the original $u$ edge and the exceptional edges, i.e.\
\begin{equation}
    B^{p} = \{e_{p} = 0\}_{\mathcal{B}}\,,\quad p = 0, \dotsc, P\,.
\end{equation}
The toric fan $\Sigma_{B^{p}}$ of $B^{p}$ can be computed using the orbit closure theorem for the $e_{p}$ edge. The 2-cones that contain $e_{p}$ as a face, and will therefore become the edges of the orbit closure fan, are
\begin{equation}
    (v, e_{p})\,,\quad (e_{p-1},e_{p})\,,\quad (w,e_{p})\,,\quad (e_{p+1},e_{p})\,,
\end{equation}
where we use the notation
\begin{equation}
    e_{-1} := t\,,\quad e_{P+1} := s\,.
\end{equation}
Taking then the quotient by $N_{e_{p}} := \langle (0,-p,1) \rangle_{\mathbb{Z}}$, we obtain the lattice $N(e_{p}) := N/N_{e_{p}}$, in which we have the fan $\Sigma_{B^{p}}$ is given by the edges
\begin{equation}
    \begin{aligned}
        v = (1,0,0) &\sim v_{p} = (1,0,0) \mod (0,-p,1)\,,\\
        e_{p-1} = (0,-(p-1),1) &\sim t_{p} = (0,1,0) \mod (0,-p,1)\,,\\
        w = (-1,-n,0) &\sim w_{p} = (-1,-n,0) \mod (0,-p,1)\,,\\
        e_{p+1} = (0,-(p+1),1) &\sim s_{p} = (0,-1,0) \mod (0,-p,1)\,.
    \end{aligned}
\end{equation}
This is the fan of the Hirzebruch surface $\mathbb{F}_{n}$, as we already knew from \cref{prop:component-geometry-single}. With the embedding map implied above, we have the divisor restrictions
\begin{equation}
    E_{p+1}|_{E_{p}} = S_{p}\,,\quad E_{p-1}|_{E_{p}} = T_{p}\,,\quad V|_{E_{p}} = V_{p}\,,\quad W|_{E_{p}} = W_{p}\,.
\end{equation}

We now describe the Weierstrass models associated with the components $\{Y^{p}\}_{0 \leq p \leq P}$ of the central fiber $Y_{0}$. The anti-canonical class of each of the $P$ components of $B_{0}$, which are $\mathbb{F}_{n}$ surfaces by themselves, is 
\begin{equation}
    \overline{K}_{B^{p}} = 2S_{p} + (2 + n)V_{p}\,,\qquad p = 0, \dotsc, P\,.
\end{equation}
We know from \cref{sec:line-bundles-components-single-limit} that the total space of the elliptic fibration over each of these components is a log Calabi-Yau space, meaning that the holomorphic line bundle defining it differs from $\overline{K}_{B^{p}}$. Let us particularise, for future reference, the line bundle computations of that section to the case of horizontal models.

Performing the sequence of toric blow-ups necessary to arrive at the open-chain resolution of a Class~1--4 horizontal model, we will have at the $p$-th step of the blow-up process the family vanishing orders
\begin{equation}
	\ord{\mathrm{Bl}_{p}\left(\hat{\mathcal{Y}}\right)}(\Delta)_{s=e_{p}=0} = 12 + n_{p+1}\,,
\end{equation}
ultimately leading to a central fiber $Y_{0}$ with the structure
\begin{equation}
    \begin{tikzcd}[column sep=1em]
        \mathrm{I}_{n_{0}} \arrow[r, dash] \arrow[d, dash] & \cdots \arrow[r, dash] & \mathrm{I}_{n_{p}} \arrow[r, dash] \arrow[d, dash] & \cdots \arrow[r, dash] & \mathrm{I}_{n_{P}} \arrow[d, dash]\\
        \mathbb{F}_{n} \arrow[r, dash] & \cdots \arrow[r, dash] & \mathbb{F}_{n} \arrow[r, dash] & \cdots \arrow[r, dash] & \mathbb{F}_{n}\mathrlap{\,.}
    \end{tikzcd}
\end{equation}
The relation between the total and strict transforms under the composition of blow-up maps of the toric divisors of $\hat{\mathcal{B}}$ is
\begin{equation}
    \tilde{\mathcal{S}} = \mathcal{S} + \sum_{p=1}^{P} E_{p}\,,\quad \tilde{\mathcal{T}} = \mathcal{T}\,,\quad \tilde{\mathcal{V}} = \mathcal{V}\,,\quad \tilde{\mathcal{W}} = \mathcal{W}\,,\quad \tilde{\mathcal{U}} = \sum_{p=0}^{P} E_{p}\,,
\label{eq:total-transforms-horizontal-model}
\end{equation}
where here and in what follows we denote the strict transforms without a prime.

The holomorphic line bundle defining the elliptic fibration $\Pi_{\mathrm{ell}}: \mathcal{Y} \rightarrow \mathcal{B}$ is
\begin{equation}
	\mathcal{L} = \overline{K}_{\mathcal{B}} = 2\mathcal{S} + (2+n)\mathcal{V} + \sum_{p=1}^{P} pE_{p}\,.
\end{equation}
From its restrictions we obtain the holomorphic line bundles over the $B^{p}$ components of $B_{0}$ defining the Weierstrass models $\pi^{p}: Y^{p} \rightarrow B^{p}$, which are
\begin{subequations}
\begin{alignat}{2}
    \mathcal{L}_{0} &:= \mathcal{L}|_{E_{0}} = S_{0} + (2+n)V_{0}\,, & \label{eq:line-bundle-hor-0}\\
    \mathcal{L}_{p} &:= \mathcal{L}|_{E_{p}} = 2V_{p}\,, &\qquad p = 1, \dotsc, P-1\,, \label{eq:line-bundle-hor-1}\\
    \mathcal{L}_{P} &:= \mathcal{L}|_{E_{P}} = S_{P} + 2V_{P}\,, & \label{eq:line-bundle-hor-2}
\end{alignat}
\end{subequations}
cf.\ \eqref{eq:line-bundles-Fn-components-open-chain}. The divisors associated to the defining polynomials and the discriminant are obtained as appropriate powers of these line bundles, while the modified discriminant in each component is in the divisor class
\begin{subequations}
\begin{alignat}{2}
    \Delta'_{0} &= (12 + n_{0} - n_{1})S_{0} + (24 + 12n)V_{0}\,, & \label{eq:Delta-hor-1}\\
    \Delta'_{p} &= (2n_{p} - n_{p-1} - n_{p+1})S_{p} + (24 + n(n_{p} - n_{p-1}))V_{p}\,, &\qquad p = 1, \dotsc, P-1\,, \label{eq:Delta-hor-2}\\
    \Delta'_{P} &= (12 + n_{P} - n_{P-1})S_{P} + (24 + n(n_{P} - n_{P-1}))V_{P}\,. & \label{eq:Delta-hor-3}
\end{alignat}
\label{eq:Delta-hor}%
\end{subequations}
Horizontal models of this type are a relative version of the K3 degenerations discussed in \cite{Lee:2021qkx,Lee:2021usk}. By taking the intersection of $\Delta'_{p}$ with the fiber class $V_{p}$ of $B^{p}$ we see that it indeed agrees with the distribution of 7-branes in components obtained for the complex structure infinite-distance limits of eight-dimensional F-theory computed in \cite{Lee:2021qkx}.

The Calabi-Yau nature of the central fiber $Y_{0}$, obtained as the union of the log Calabi-Yau components $Y^{p}$ along their boundaries, can be seen explicitly by using \eqref{eq:line-bundle-central-fiber} and the relations between the divisors in $Y_{0}$ and those in the individual components $Y^{p}$. These are given by
\begin{equation}
    \begin{aligned}
        \mathcal{S}|_{\tilde{\mathcal{U}}} &= S_{P}\,,\quad \mathcal{T}|_{\tilde{\mathcal{U}}} = T_{0}\,, \quad \mathcal{V}|_{\tilde{\mathcal{U}}} = \sum_{p=0}^{P} V_{p}\,,\quad \mathcal{W}|_{\tilde{\mathcal{U}}} = \sum_{p=0}^{P} W_{p}\,,\\
        E_{0}|_{\tilde{\mathcal{U}}} &= T_{1} - S_{0}\,,\quad E_{p}|_{\tilde{\mathcal{U}}} = (T_{p+1} - T_{p}) - (S_{p} - S_{p-1})\,,\quad E_{P}|_{\tilde{\mathcal{U}}} = -(T_{P} - S_{P-1})\,,
    \end{aligned}
\label{eq:horizontal-model-family-divisor-restrictions}
\end{equation}
where $p = 1, \dotsc, P-1$.

\subsection{Vertical models}
\label{sec:geometry-components-vertical}

Consider now a vertical single infinite-distance limit degeneration $\hat{\rho}: \hat{\mathcal{Y}} \rightarrow D$ of Hirzebruch models, i.e.\ a degeneration in Case B. The set of non-minimal curves in $\hat{B}_{0} = \mathbb{F}_{n} \subset \hat{\mathcal{B}}$ is then $\hat{\mathscr{C}}_{1} = \{ f \}$. Without loss of generality, we can take the non-minimal curve to be the representative $\mathcal{V} \cap \mathcal{U}$ of $f$; the resolution process leading to the open-chain resolution $\rho: \mathcal{Y} \rightarrow D$ then consists of toric blow-ups.

The starting point is the fan $\Sigma_{\hat{\mathcal{B}}}$, with the edges given in \eqref{eq:F1-fan-edges}. To arrive at the fan $\Sigma_{\mathcal{B}}$ describing the resolved family base $\mathcal{B}$, we rename $e_{0} := u$ and add to $\Sigma_{\hat{\mathcal{B}}}$ the set of edges
\begin{equation}
	\left\{ e_{p} = (p,0,1) \right\}_{1 \leq p \leq P}\,
\end{equation}
and the necessary 2-cones to complete the fan.

The toric fan $\Sigma_{B^{p}}$ of the component $B^{p}$ can be obtained applying the orbit closure theorem to the edge $e_{p}$. The 2-cones that contain $e_{p}$ as a face are
\begin{equation}
    (e_{p+1},e_{p})\,,\quad (t,e_{p})\,,\quad (e_{p-1},e_{p})\,,\quad (s,e_{p})\,,
\end{equation}
where we use the notation
\begin{equation}
    e_{-1} := w\,,\quad e_{P+1} := v\,.
\end{equation}
Let us apply this to the components $B^{p}$ with $p = 1, \dotsc, P$. The quotient by $N_{e_{p}} := \langle (p,0,1) \rangle_{\mathbb{Z}}$ leads to the lattice $N(e_{p}) = N/N_{e_{p}}$, in which the fan $\Sigma_{B^{p}}$ is given by the edges
\begin{equation}
    \begin{aligned}
        e_{p+1} = (p+1,0,1) &\sim v_{p} = (1,0,0) \mod (p,0,1)\,,\\
        t = (0,1,0) &\sim t_{p} = (0,1,0) \mod (p,0,1)\,,\\
        e_{p-1} = (p-1,0,1) &\sim w_{p} = (-1,0,0) \mod (p,0,1)\,,\\
        s = (0,-1,0) &\sim s_{p} = (0,-1,0) \mod (p,0,1)\,.
    \end{aligned}
\end{equation}
This is the fan of the Hirzebruch surface $\mathbb{F}_{0}$. If one computes the orbit closure of $e_{0}$, the resulting fan is instead that of $\mathbb{F}_{n}$, i.e.\ $B^{0}$ is a Hirzebruch surface of the same type as $\hat{B}_{0}$. These results agree with \cref{prop:component-geometry-single}.

While blowing the open-chain resolution down to the $B^{0}$ component can be done directly, doing so to one of the components $B^{p}$, with $p = 1, \dotsc, P$, requires flopping some curves first. We comment on this fact in \cref{sec:blowing-down-vertical-components}.

Let us now collect the holomorphic line bundles associated to the Weierstrass models over vertical models. At the $p$-th step of the sequence of blow-ups necessary to arrive at the open-chain resolution of a Class~1--4 vertical model, we have the family vanishing orders
\begin{equation}
	\ord{\mathrm{Bl}_{p}\left(\hat{\mathcal{Y}}\right)}(\Delta)_{v=e_{p}=0} = 12 + n_{p+1}\,.
\end{equation}
This leads to a central fiber $Y_{0}$ with the structure
\begin{equation}
    \begin{tikzcd}[column sep=1em]
        \mathrm{I}_{n_{0}} \arrow[r, dash] \arrow[d, dash] & \cdots \arrow[r, dash] & \mathrm{I}_{n_{p}} \arrow[r, dash] \arrow[d, dash] & \cdots \arrow[r, dash] & \mathrm{I}_{n_{P}} \arrow[d, dash]\\
        \mathbb{F}_{n} \arrow[r, dash] & \cdots \arrow[r, dash] & \mathbb{F}_{0} \arrow[r, dash] & \cdots \arrow[r, dash] & \mathbb{F}_{0}\mathrlap{\,.}
    \end{tikzcd}
\end{equation}
The relation between the total and strict transforms under the composition of the blow-up maps of the toric divisors of $\hat{\mathcal{B}}$ is
\begin{equation}
    \tilde{\mathcal{S}} = \mathcal{S}\,, \quad \tilde{\mathcal{T}} = \mathcal{T}\,,\quad \tilde{\mathcal{V}} = \mathcal{V} + \sum_{p=1}^{P} E_{p}\,,\quad \tilde{\mathcal{W}} = \mathcal{W}\,,\quad \tilde{\mathcal{U}} = \sum_{p=0}^{P} E'_{p}\,.
\end{equation}

The holomorphic line bundle defining the elliptic fibration $\Pi_{\mathrm{ell}}: \mathcal{Y} \rightarrow \mathcal{B}$ is, expressed in terms of the strict transforms,
\begin{equation}
    \mathcal{L} = \overline{K}_{\mathcal{B}} = 2\mathcal{S} + (2+n)\mathcal{W} - \sum_{p=1}^{P} p E_{p}\,.
\end{equation}
Taking the pertinent restrictions of it we obtain the holomorphic line bundles defining the Weierstrass models $\pi^{p}: Y^{p} \rightarrow B^{p}$, which are
\begin{subequations}
\begin{alignat}{2}
    \mathcal{L}_{0} &:= \mathcal{L}|_{E_{0}} = 2S_{0} + (1+n)W_{0}\,, & \label{eq:component-line-bundles-vertical-one-sided-0}\\
    \mathcal{L}_{p} &:= \mathcal{L}|_{E_{p}} = 2S_{p}\,, &\qquad p = 1, \dotsc, P-1\,, \label{eq:component-line-bundles-vertical-one-sided-1}\\
    \mathcal{L}_{P} &:= \mathcal{L}|_{E_{P}} = 2S_{P} + W_{P}\,. & \label{eq:component-line-bundles-vertical-one-sided-2}
\end{alignat}
\end{subequations}
Appropriate powers of these line bundles yield the divisor classes associated to the defining polynomials and the discriminant, while the modified discriminant in each component is in the divisor class
\begin{subequations}
\begin{alignat}{2}
    \Delta'_{0} &= 24S_{0} + (12 + 12n + n_{0} - n_{1})W_{0}\,, & \label{eq:component-discriminant-vertical-one-sided-0}\\
    \Delta'_{p} &= 24S_{p} + (2n_{p} - n_{p-1} - n_{p+1})W_{p}\,, &\qquad p = 1, \dotsc, P-1\,, \label{eq:component-discriminant-vertical-one-sided-1}\\
    \Delta'_{P} &= 24S_{P} + (12 + n_{P} - n_{P-1})W_{P}\,. & \label{eq:component-discriminant-vertical-one-sided-2}
\end{alignat}
\label{eq:component-discriminant-vertical-one-sided}%
\end{subequations}

As occurred for horizontal models, the fact that the log Calabi-Yau components $Y^{p}$ give, when glued along their boundaries, a Calabi-Yau central fiber $Y_{0}$ can be seen explicitly by using \eqref{eq:line-bundle-central-fiber} and the relations between the divisors in $Y_{0}$ and those in the individual components $Y^{p}$, that now are given by
\begin{equation}
    \begin{aligned}
        \mathcal{S}|_{\tilde{\mathcal{U}}} &= \sum_{p=0}^{P} S_{p}\,,\quad \mathcal{T}|_{\tilde{\mathcal{U}}} = \sum_{p=0}^{P} T_{p}\,,\quad \mathcal{V}|_{\tilde{\mathcal{U}}} = V_{P}\,,\quad \mathcal{W}|_{\tilde{\mathcal{U}}} = W_{0}\,,\\
        E_{0}|_{\tilde{\mathcal{U}}} &= W_{1} - V_{0}\,,\quad E_{p}|_{\tilde{\mathcal{U}}} = (W_{p+1} - W_{p}) - (V_{p} - V_{p-1})\,,\quad E_{P}|_{\tilde{\mathcal{U}}} = -(W_{P} - V_{P-1})\,,
    \end{aligned}
\end{equation}
where $p = 1, \dotsc, P-1$.

\subsection{Mixed genus-zero degenerations}

The remaining single infinite-distance limit genus-zero degenerations of Hirzebruch models are, according to the classification of \cref{def:classification-Hirzebruch-single-infinite-distance-limits}, those in the Cases C and D. We recall that the non-minimal curves for these two cases are:
\begin{itemize}
	\item Case C:\quad $C = h + (n+1)f$, or $C = h + (n+2)f$  (mixed section model).
		
	\item Case D:\quad $C = 2h + bf$, with $(n,b) = (0,1)$, $(1,2)$  (mixed bisection model).
\end{itemize}
While for Cases A and B we could choose toric divisors as the representatives of the curve class $C$, leading to a very explicit resolution process based on the description in terms of global homogeneous coordinates, the same is not true for Cases C and D. Nonetheless, we can proceed by employing the general machinery for single infinite-distance limits, discussed in \cref{sec:geometry-components-single-limit,sec:line-bundles-components-single-limit}. Before we do so, let us further restrict the realization of Cases C and D.

\subsubsection{Restriction of Cases C and D}
\label{sec:restriction-cases-C-D}

In \cref{prop:non-minimal-curves} we obtained a list of the smooth, irreducible curves that can support non-minimal elliptic fibers in a single infinite-distance limit degeneration of Hirzebruch models. We did this by demanding that $C \leq \overline{K}_{\hat{B}_{0}}$, see \cref{prop:K-C-effectiveness}, and checking which curves satisfying this condition were smooth and irreducible. We did not, however, check if tuning said curves would force a second curve $C'$ to factorize with non-minimal vanishing orders and such that $C \cdot C' \neq 0$. This would violate Condition \ref{item:single-infinite-distance-limit-original-1} of \cref{def:single-infinite-distance-limit-original}, hence not corresponding to a single infinite-distance limit degeneration. We now study when this occurs, further constraining Cases C and D.

Focusing first on Case C, the two possible curve classes are
\begin{equation}
	C = h + (n+\alpha)f\,,\quad \alpha = 1, 2\,,
\end{equation}
leading to
\begin{align}
	F - 4C &= 4\left[ h+(2-\alpha)f \right]\,,\\
	G - 6C &= 6\left[ h+(2-\alpha)f \right]\,.
\end{align}

When $\alpha = 2$, this means that
\begin{equation}
	F - 4C = 4h\,,\quad G - 6C = 6h\,.
\end{equation}
If $n > 0$, the class $h$ has the unique representative $S = \{ s= 0\}_{\hat{B}_{0}}$, and therefore the above factorization implies
\begin{equation}
	\mathrm{ord}_{\hat{Y}_{0}}(f_{0},g_{0})_{s=0} = (4,6)\,.
\end{equation}
Since $C \cdot h = 2$, the model does not correspond to a single infinite-distance limit. When $n = 0$, the classes $4h$ and $6h$ factorizing in $F$ and $G$, respectively, yield a residual discriminant $12h$. Since $h$ now moves in a linear system, this generically leads to $12$ horizontal curves of type $\mathrm{I}_{1}$ fibers, and we obtain a valid single infinite-distance limit degeneration.

For $\alpha = 1$ we find, instead,
\begin{equation}
	F - 4C = 4h + 4f\,,\quad G - 6C = 6h + 6f\,.
\end{equation}
As we increase the value of $n$, these classes will become generically reducible with components $h$, eventually forcing non-minimal fibers over this curve. This occurs when
\begin{equation}
	\left( G - 6C - 5h \right) \cdot h = n + 6 < 0 \Leftrightarrow n \geq 7\,.
\end{equation}
Since $C \cdot h = 1$, this means that the model is beyond the single infinite-distance limit case. In conclusion, Case C can be restricted to
\begin{equation}
	\alpha = 1\quad \textrm{with}\quad n \leq 6\,,\qquad \textrm{or}\qquad \alpha = 2\quad \textrm{with}\quad n = 0\,.
\end{equation}

Case D does not suffer from any additional restrictions. The divisors $F$ and $G$ are, after factoring out the appropriate number of copies of $C$,
\begin{subequations}
\begin{align}
	F - 4C &= 4(2+n-b)f = 4f\,,\\
	G - 6C &= 6(2+n-b)f = 6f\,.
\end{align}
\label{eq:Case-D-remainder-F-G}
\end{subequations}
This means that the residual discriminant in such a model is purely vertical and reducible, generically leading to $12$ vertical curves of type $\mathrm{I}_{1}$ fibers.

\subsubsection{Study of Case C}

Let $\hat{\rho}: \hat{\mathcal{Y}} \rightarrow D$ be a single infinite-distance limit degeneration of Hirzebruch models in Case~C, and let $\rho: \mathcal{Y} \rightarrow D$ be its open-chain resolution. To determine the geometry of the base components $B^{p}$ of the central fiber $B_{0}$ of $\mathcal{B}$, we apply \cref{prop:component-geometry-single}. To this end, note that the self-intersection of $C$ is
\begin{equation}
	C \cdot C = n + 2\alpha\,,
\end{equation}
meaning that the first base blow-up produces an exceptional component $B^{1} = \mathbb{F}_{n+2\alpha}$. The curve $E_{0} \cap E_{1}$ over which the components $B^{0}$ and $B^{1}$ meet is the $(-(n+2\alpha))$-curve of the $B^{1}$ component. Therefore, within the single infinite-distance limit class of degenerations, we can tune further non-minimal elliptic fibers along the $(+(n+2\alpha))$-curve of $B^{1}$, prompting us to perform a second blow-up giving rise to a new component $B^{2} = \mathbb{F}_{n+2\alpha}$. We can continue iterating this until we tune a model with, say, $P+1$ components. The central fiber of the open-chain resolution of a Class~1--4 single infinite-distance degeneration of Hirzebruch models in Case C has therefore the structure
\begin{equation}
    \begin{tikzcd}[column sep=1em]
        \mathrm{I}_{n_{0}} \arrow[r, dash] \arrow[d, dash] & \cdots \arrow[r, dash] & \mathrm{I}_{n_{p}} \arrow[r, dash] \arrow[d, dash] & \cdots \arrow[r, dash] & \mathrm{I}_{n_{P}} \arrow[d, dash]\\
        \mathbb{F}_{n+2\alpha} \arrow[r, dash] & \cdots \arrow[r, dash] & \mathbb{F}_{n+2\alpha} \arrow[r, dash] & \cdots \arrow[r, dash] & \mathbb{F}_{n}\mathrlap{\,,}
    \end{tikzcd}
\end{equation}
where we have inverted the labelling of the components to more effectively draw comparisons with the results of \cref{sec:geometry-components-horizontal}. We see that the $P$ first components are the same as in a horizontal model constructed, not over the Hirzebruch surface $\mathbb{F}_{n}$, but of $\mathbb{F}_{n+2\alpha}$. The same is true for the holomorphic line bundles defining the Weierstrass models $\pi^{p}: Y^{p} \rightarrow B^{p}$, which are
\begin{subequations}
\begin{alignat}{2}
    \mathcal{L}_{0} &:= S_{0} + (2+(n+2\alpha))V_{0}\,, & \label{eq:line-bundle-C-0}\\
    \mathcal{L}_{p} &:= 2V_{p}\,, &\qquad p = 1, \dotsc, P-1\,, \label{eq:line-bundle-C-1}\\
    \mathcal{L}_{P} &:= S_{P} + (2-\alpha)V_{P}\,. & \label{eq:line-bundle-C-2}
\end{alignat}
\end{subequations}
The divisors associated to the defining polynomials and the discriminant are obtained as appropriate powers of these line bundles, while the modified discriminant in each component is in the divisor class
\begin{subequations}
\begin{align}
    \Delta'_{0} &= (12 + n_{0} - n_{1})S_{0} + (24 + 12(n+2\alpha))V_{0}\,, \label{eq:Delta-C-1}\\
    \Delta'_{p} &= (2n_{p} - n_{p-1} - n_{p+1})S_{p} + (24 + (n+2\alpha)(n_{p} - n_{p-1}))V_{p}\,,\qquad p = 1, \dotsc, P-1\,, \label{eq:Delta-C-2}\\
    \Delta'_{P} &= (12 + n_{P} - n_{P-1})S_{P} + ((24-12\alpha) + (n+\alpha)(n_{P} - n_{P-1}))V_{P}\,. \label{eq:Delta-C-3}
\end{align}
\label{eq:Delta-C}%
\end{subequations}
As in the other cases, each component $Y^{p}$ of the central fiber $Y_{0}$ is a log Calabi-Yau space, while $Y_{0}$ is the Calabi-Yau space obtained by taking their union along the boundaries.

\subsubsection{Study of Case D}

Assume now instead that $\hat{\rho}: \hat{\mathcal{Y}} \rightarrow D$ is a single infinite-distance limit degeneration of Hirzebruch models in Case D. Since
\begin{equation}
	C \cdot C = 4(b-n) = 4\,,
\end{equation}
the exceptional component obtained after the first blow-up is $B^{1} = \mathbb{F}_{4}$. Arguing as we did for Case C, and inverting the labelling of the components, a Class~1--4 single infinite-distance limit degeneration of Hirzebruch models in Case D leads to a central fiber with the structure
\begin{equation}
    \begin{tikzcd}[column sep=1em]
        \mathrm{I}_{n_{0}} \arrow[r, dash] \arrow[d, dash] & \cdots \arrow[r, dash] & \mathrm{I}_{n_{p}} \arrow[r, dash] \arrow[d, dash] & \cdots \arrow[r, dash] & \mathrm{I}_{n_{P}} \arrow[d, dash]\\
        \mathbb{F}_{4} \arrow[r, dash] & \cdots \arrow[r, dash] & \mathbb{F}_{4} \arrow[r, dash] & \cdots \arrow[r, dash] & \mathbb{F}_{n}\mathrlap{\,.}
    \end{tikzcd}
\end{equation}
The holomorphic line bundles defining the Weierstrass models $\pi^{p}: Y^{p} \rightarrow B^{p}$ are
\begin{subequations}
\begin{alignat}{2}
    \mathcal{L}_{0} &:= S_{0} + (2+4)V_{0}\,, & \label{eq:line-bundle-D-0}\\
    \mathcal{L}_{p} &:= 2V_{p}\,, &\qquad p = 1, \dotsc, P-1\,, \label{eq:line-bundle-D-1}\\
    \mathcal{L}_{P} &:= V_{P}\,. & \label{eq:line-bundle-D-2}
\end{alignat}
\end{subequations}
Adequate powers of these give the divisor classes associated to the defining polynomials and the discriminant, while the modified discriminant in each component is in the divisor class\\
\begin{subequations}
\begin{alignat}{2}
    \Delta'_{0} &= (12 + n_{0} - n_{1})S_{0} + (24 + 12 \cdot 4)V_{0}\,, & \label{eq:Delta-D-1}\\
    \Delta'_{p} &= (2n_{p} - n_{p-1} - n_{p+1})S_{p} + (24 + 4(n_{p} - n_{p-1}))V_{p}\,, &\qquad p = 1, \dotsc, P-1\,, \label{eq:Delta-D-2}\\
    \Delta'_{P} &= 2(n_{P} - n_{P-1})S_{P} + (12 + (n+1)(n_{P} - n_{P-1}))V_{P}\,. & \label{eq:Delta-D-3}
\end{alignat}
\label{eq:Delta-D}%
\end{subequations}
In this case, the first $P$ components behave like in a horizontal model constructed over $\mathbb{F}_{4}$. As in the previous cases, the union of the $Y^{p}$ log Calabi-Yau spaces along their boundary yields the Calabi-Yau central fiber $Y_{0}$.
%auto-ignore

\section{Extracting the codimension-one information}
\label{sec:extraction-codimension-one-information}

In this section, we analyse the discriminant structure of the central fiber of the semi-stable degeneration. This is required in order to 
read off the gauge algebra of the effective action in F-theory, prior to taking the decompactification limit, as we will explain.

According to the usual F-theory dictionary, the gauge algebras found in space-time are encoded in the types of singularities of the elliptic fiber of the internal Calabi-Yau space  over codimension-one loci in the base; these loci are the spacetime regions wrapped by 7-branes, and the singularity in the elliptic fiber captures the singular nature of the Type IIB axio-dilaton profile on top of its sources.

Let $\mathcal{D}$ be an irreducible divisor in the base of an elliptic Calabi-Yau $n$-fold $Y$. The gauge algebra\footnote{In six dimensions or fewer, the unambiguous determination of the gauge algebra requires also analysing the monodromy cover. Obtaining the global form of the gauge group requires additional information as well.} supported on $\mathcal{D}$ can be read off from the vanishing orders $\ord{Y}(f,g,\Delta)_{\mathcal{D}}$ with the help of the Kodaira-N\'eron classification of singular elliptic fibers. From the point of view of M-theory, the type of singular elliptic fibers over $\mathcal{D}$ determines the pattern of exceptional curves that shrink as we take the F-theory limit (going to the origin of the Coulomb branch), with the gauge algebras furnished by the light states arising from M2-branes wrapped on these curves. The relation between the type of singular elliptic fibers over a certain locus and the vanishing orders of the defining polynomials over it was discussed in \cref{sec:orders-of-vanishing} (see also, e.g.,\ \cite{Weigand:2018rez} for a review on F-theory).

Given a degeneration $\hat{\rho}: \hat{\mathcal{Y}} \rightarrow D$ of the type described in \cref{sec:definition-of-degenerations}, this is how we would determine the gauge algebras for the F-theory models described by the generic elements $\hat{Y}_{u \neq 0}$ of the family, which do not exhibit any infinite-distance non-minimal fibral singularities. The situation is more subtle for the central fiber $\hat{Y}_{0}$ of the degeneration due to the appearance of non-minimal elliptic fibers. To read off the gauge algebra, we therefore first have to apply the procedures explained in \cref{sec:modifications-of-degenerations} to obtain the resolved degeneration $\rho: \mathcal{Y} \rightarrow D$ (which must be free of obscured infinite-distance limits, see \cref{sec:obscured-infinite-distance-limits}). Its central fiber $Y_{0}$ presents then no infinite-distance non-minimal singularities.

However, $Y_{0}$ factors into multiple log Calabi-Yau components $\{Y^{p}\}_{0 \leq p \leq P}$. This gives rise to interesting phenomena not occurring in conventional six-dimensional F-theory models. When the base of the fibration is an irreducible surface, the 7-branes correspond to irreducible curves in it. Two irreducible curves in an irreducible surface will either coincide, leading to a gauge enhancement, or intersect over points if they are distinct, resulting in localized matter. When the base of the elliptic fibration is instead a reducible surface, as occurs for $Y_{0}$, two distinct 7-branes in a component may overlap in a different component, leading to what looks like a local gauge enhancement. This gives rise to subtleties in determining the gauge algebra content from the component vanishing orders, whose definition was given in \cref{def:component-vanishing-order}.

More precisely, there are two types of complications that obscure the interpretation of the vanishing orders of the Weierstrass model at first sight:
\begin{enumerate}
    \item \textbf{Locally coincident discriminant components:} Certain components of the discriminant locus may coincide in some components of the base $B_{0}$ while being separated in others. An example is shown in \cref{fig:illustrative-example}. 

    \item \textbf{Locally reducible discriminant components:} Two or more divisors in some component of $B_0$ may in fact be part of a single connected divisor extending over the entire base $B_0$, as illustrated in \cref{fig:overcounting-example}.
\end{enumerate}

In \cref{sec:illustrative-example,sec:overcounting-example}, we analyse these two complications in turn. We first exhibit them in some examples, and then propose a way to unambiguously read off the global gauge algebra content of the F-theory model represented by $Y_{0}$. 
A first step in this direction was the definition of the modified discriminant $\Delta'$, see \cref{def:modified-discriminant}, in which we subtract the background value of the axio-dilaton in the components. The general algorithm to unambiguously extract the vanishing orders of the Weierstrass model is then presented in \cref{sec:physical-discriminant}. The gauge algebra assignment also needs to take into account possible local monodromies in the fiber over some components, as pointed out in \cref{sec:monodromy-cover}. The final algorithm to determine the gauge algebra as encoded in the Weierstrass model is summarised in \cref{sec:algorithm-gauge-algebra}.

It is worth noting that the subsequent analysis maintains a six-dimensional standpoint; the gauge algebra extracted from $Y_0$ is that of a six-dimensional effective theory prior to partial decompactification. The partial decompactification may be a consequence of taking the infinite-distance limit, possibly leading to partial gauge enhancements as viewed from the point of view of the asymptotic, higher-dimensional gauge theory. When we refer to the ``gauge algebra" in this section, this important effect has not been considered yet. It will be discussed in \cite{ALWPart2}.

Some of the considerations below also apply to the study of open-moduli infinite-distance limits in eight-dimensional F-theory performed in \cite{Lee:2021qkx,Lee:2021usk}, but are not crucial for the correct identification of the codimension-one physics. The 7-branes in an eight-dimensional model correspond to points in the base, and they hence cannot extend between the components of the central fiber of the resolved degeneration as it occurs in six dimensions.

As emphasized in \cref{sec:orders-of-vanishing}, the gauge algebras can be determined using a global description of the F-theory model, if available like in toric examples, or by working in local patches. In the remainder of the section, we frame the discussion using the global picture for ease of exposition, with analogous considerations following \textit{mutatis mutandis} for a local analysis.

\subsection{Locally coincident discriminant components}
\label{sec:illustrative-example}

In this section, we illustrate the phenomenon where components of the discriminant may locally overlap in certain parts of the base $B_{0}$, and how this, naively, leads to ambiguities in the vanishing orders of the components in the Weierstrass model. 

Consider, for instance, a degeneration $\hat{\rho}: \hat{\mathcal{Y}} \rightarrow D$ with $\mathcal{B} = \mathbb{F}_{n} \times D$ and non-minimal family vanishing orders over the curve $\mathcal{S} \cap \mathcal{U}$,
\begin{equation}
	\ord{\hat{\mathcal{Y}}}(f,g,\Delta)_{\mathcal{S} \cap \mathcal{U}} = (4,6,12)\,.
\end{equation}
This requires at least one base blow-up along $\mathcal{S} \cap \mathcal{U}$, followed by an appropriate line bundle shift, in order to arrive at the resolved degeneration $\rho: \mathcal{Y} \rightarrow D$. The curve of the central fiber $\mathcal{V} \cap \tilde{\mathcal{U}} = \{ v=0 \}_{B_{0}}$ traverses both the $B^{0}$ and $B^{1}$ components. It may occur, however, that the component vanishing orders differ between the two components,
\begin{equation}
    \ord{Y^{0}}\left( f_{0},g_{0},\Delta'_{0} \right)_{v=0} \neq \ord{Y^{1}}\left( f_{1},g_{1},\Delta'_{1} \right)_{v=0}\,.
\end{equation}
This prompts us to question how they are related to the gauge algebra content and to each other. We can see this realised explicitly in the following degeneration.
\begin{example}
\label{example:illustrative-example}
Consider the Weierstrass model
\begin{subequations}
\begin{align}
    f &= s^3 t^3 (s v+t u) \left(s u v^8+t u w^7+t v^3 w^4+t v^2 w^5+t v w^6\right)\,, \label{eq:f-naive-example}\\
    g &= s^4 t^5 v w^5 (s v+t u)^2 \left(s w^5+t v^4+t v^3 w+t v^2 w^2+t v w^3\right)\,, \label{eq:g-naive-example}\\
    \Delta &= s^8 t^9 (s v+t u)^3 p_{4,24}([s:t],[v:w],u)\,, \label{eq:Delta-naive-example}
\end{align}
\end{subequations}
defining an elliptically fibered variety $\hat{\mathcal{Y}}$ over the base $\hat{\mathcal{B}} = \mathbb{F}_{1} \times D$ with non-minimal singular elliptic fibers over $\mathcal{S} \cap \mathcal{U}$,
\begin{equation}
    \ord{\hat{\mathcal{Y}}}\left( f,g,\Delta \right)_{s = u = 0} = (4,6,12)\,.
\label{eq:non-minimal-vanising-orders-naive-example}
\end{equation}
This is a single infinite-distance limit horizontal model according to the classification in \cref{def:classification-Hirzebruch-single-infinite-distance-limits}. After performing a (toric) blow-up along $\mathcal{S} \cap \mathcal{U}$ and the necessary line bundle shift, we arrive at the open-chain resolution
\begin{subequations}
\begin{align}
    f &= s^3 t^3 \left(e_0 t+s v\right) \left(e_0 e_1^2 s v^8+e_0 e_1 t w^7+t v^3 w^4+t v^2 w^5+t v w^6\right)\,, \label{eq:fb-naive-example}\\
    g &= s^4 t^5 v w^5 \left(e_0 t+s v\right){}^2 \left(e_1 s w^5+t v^4+t v^3 w+t v^2 w^2+t v w^3\right)\,, \label{eq:gb-naive-example}\\
    \Delta &= s^8 t^9 \left(s v+t e_0\right)^3 p_{4,24,1}([s:t],[v:w:t],[s:e_{0}:e_{1}])\,, \label{eq:Deltab-naive-example}
\end{align}
\end{subequations}
with Stanley-Reisner ideal
\begin{equation}
	\mathscr{I}_{\mathcal{B}} = \langle st, vw, se_{0}, te_{1} \rangle\,.
\end{equation}
The central fiber of the resolved degeneration has the following pattern of codimension-zero singularities:
\begin{equation}
    \begin{tikzcd}[column sep=1em]
        \mathrm{I}_{0} \arrow[rr, dash] \arrow[d, dash] & & \mathrm{I}_{0} \arrow[d, dash]\\
        \mathbb{F}_{1} \arrow[rr, dash] & & \mathbb{F}_{1}\mathrlap{\,.}
    \end{tikzcd}
\end{equation}
Computing the restricted polynomials \eqref{eq:defining-polynomials-component-restrictions} and modified discriminant \eqref{eq:def-modified-discriminant}, we find for the $B^{0}$ component
\begin{subequations}
\begin{align}
    f_{0} &= t^4 v^2 w^4 \left(v^2+v w+w^2\right)\,,\\
    g_{0} &= t^5 v^3 w^5 \left(e_1 w^5+t v^4+t v^3 w+t v^2 w^2+t v w^3\right)\,,\\
    \Delta'_{0} &= t^{10} v^6 w^{10} p_{2,10}([t:e_{1}],[v:w])\,,
\end{align}
\end{subequations}
while for the $B^{1}$ component we obtain
\begin{subequations}
\begin{align}
    f_{1} &= s^3 v w^4 \left(v^2+v w+w^2\right) \left(e_0+s v\right)\,,\\
    g_{1} &= s^4 v^2 w^5 (v+w) \left(v^2+w^2\right) \left(e_0+s v\right){}^2\,,\\
    \Delta'_{1} &= s^8 v^3 w^{10} \left(e_0+s v\right)^3 p_{1,8}([e_{0}:s],[v:w])\,.
\end{align}
\end{subequations}
Here we have used the available $\mathbb{C}^{*}$-actions to set the redundant coordinates to one.

The model presents a discrepancy between the family and component vanishing orders
\begin{align}
	(1,5,3) = \ord{\mathcal{Y}}(f,g,\Delta)_{w=e_{0}=0} &\leq \ord{Y^{0}}(f_{0},g_{0},\Delta'_{0})_{w=0} = (4,5,10)\,,\\
	(2,5,6) = \ord{\mathcal{Y}}(f,g,\Delta)_{w=e_{1}=0} &\leq \ord{Y^{1}}(f_{1},g_{1},\Delta'_{1})_{w=0} = (4,5,10)\,,
\end{align}
which is expected to occur, as discussed in \cref{sec:orders-of-vanishing}. This is not a problem for determining the gauge algebra because, as is emphasized in \cref{sec:obscured-infinite-distance-limits}, it is the component vanishing orders that are expected to be of physical relevance for the codimension-one physics; the family vanishing orders, by contrast, play a role during the resolution process and for identifying the codimension-zero singular elliptic fibers over the components of the central fiber. More importantly, we observe a discrepancy between the component vanishing orders
\begin{equation}
    (2,3,6) = \ord{Y^{0}}\left( f_{0},g_{0},\Delta'_{0} \right)_{v=0} \neq \ord{Y^{1}}\left( f_{1},g_{1},\Delta'_{1} \right)_{v=0} = (1,2,3)\,.
\end{equation}
\end{example}
Hence reading off the gauge algebra supported on $\mathcal{V} \cap \tilde{\mathcal{U}} \subset B_{0}$ from the component vanishing orders in $B^{0}$ seemingly results in a bigger gauge factor than the one read in the $B^{1}$ component. 

To resolve this puzzle, note that in the multi-component central fiber $Y_{0}$ of a resolved degeneration $\rho: \mathcal{Y} \rightarrow D$, the component vanishing orders $\ord{Y^{p}}(f_{p},g_{p},\Delta'_{p})_{C_{p}}$ only reflect the information available in a given component $Y^{p}$. The physics associated with the endpoints of open-moduli infinite-distance limits in six-dimensional F-theory involves, however, the central fiber $Y_{0}$ taken as a whole.

% In this section, we argue how the global and local pictures are to be interpreted.
% To this end, let us further exploit \cref{example:illustrative-example} as a concrete discussion arena for this general question. 

To exploit this observation, a more revealing way to look at \cref{example:illustrative-example} is by graphically representing the restrictions $\Delta'_{0}$ and $\Delta'_{1}$ of its (modified) discriminant $\Delta'$. We do so in \cref{fig:illustrative-example}, omitting the residual discriminant in both components. We also print the component vanishing orders computed in $Y^{0}$ and $Y^{1}$.
\begin{figure}[t!]
    \centering
    \begin{tikzpicture}
		\node [] (0) at (0, 2.5) {};
		\node [] (1) at (0, -2.5) {};
		\node [] (2) at (6, 2.5) {};
		\node [] (3) at (6, -2.5) {};
		\node [] (4) at (-6, -2.5) {};
		\node [] (5) at (-6, 2.5) {};
		\node [] (6) at (-6, 1.5) {};
		\node [label={[align=center, xshift=1cm, yshift=-0.7cm]\textcolor{diagLightGreen}{$w=0$}\\ \textcolor{diagLightGreen}{$(4,5,10)$}}] (7) at (6, 1.5) {};
		\node [] (8) at (0, 1.5) {};
		\node [] (9) at (-6, -1.5) {};
		\node [] (10) at (0, -1.5) {};
		\node [] (11) at (6, -1.5) {};
		\node [] (12) at (5, 2.5) {};
		\node [label={[align=center, yshift=-1.5cm]\textcolor{diagLightBlue}{$s=0$}\\ \textcolor{diagLightBlue}{$(3,4,8)$}}] (13) at (5, -2.5) {};
		\node [label={[align=center, yshift=-1.5cm]\textcolor{diagLightPurple}{$t=0$}\\ \textcolor{diagLightPurple}{$(4,5,10)$}}] (14) at (-5, -2.5) {};
		\node [] (15) at (-5, 2.5) {};
		\node [label={[align=center, xshift=-0.25cm, yshift=-2.75cm]\textcolor{diagLightYellow}{$e_{0} + s v=0$}\\ \textcolor{diagLightYellow}{$(1,2,3)$}}] (16) at (4, 2.5) {};
        \node [label={[align=center, xshift=-0.25cm, yshift=-3.95cm]\textcolor{diagMediumRed}{$v=0$}\\ \textcolor{diagMediumRed}{$(1,2,3)$}}] (17) at (4, 2.5) {};
        \node [label={[align=center, xshift=-6.5cm, yshift=-3.95cm]\textcolor{diagDarkYellow}{$v=0$}\\ $\textcolor{diagDarkYellow}{(2,3,6)=}\; \textcolor{diagMediumRed}{(1,2,3)}\, \textcolor{diagDarkYellow}{+}\, \textcolor{diagLightYellow}{(1,2,3)}$}] (18) at (4, 2.5) {};
        \node [label={[yshift=0cm]$\{e_{0} = 0\}_{\makebox[0pt]{$\scriptstyle \;\;\mathcal{B}$}}$}] (19) at (-3, 2.5) {};
        \node [label={[yshift=0cm]$\{e_{1} = 0\}_{\makebox[0pt]{$\scriptstyle \;\;\mathcal{B}$}}$}] (20) at (3, 2.5) {};

        \draw [style=light-green line] (6.center) to (8.center);
		\draw [style=light-green line] (8.center) to (7.center);
		\draw [style=medium-red line] (9.center) to (10.center);
		\draw [style=medium-red line] (10.center) to (11.center);
		\draw [style=light-purple line] (15.center) to (14.center);
		\draw [style=light-blue line] (12.center) to (13.center);
        \draw [style=light-yellow snake] (9.center) to (10.center);
		\draw [style=light-yellow snake] (16.center) to (10.center);

        \draw [style=black line] (0.center) to (1.center);
		\draw [style=black line] (5.center) to (0.center);
		\draw [style=black line] (0.center) to (2.center);
		\draw [style=black line] (2.center) to (3.center);
		\draw [style=black line] (3.center) to (1.center);
		\draw [style=black line] (1.center) to (4.center);
		\draw [style=black line] (4.center) to (5.center);
\end{tikzpicture}
    \caption{Restrictions $\Delta'_{0}$ and $\Delta'_{1}$ of the (modified) discriminant for \cref{example:illustrative-example}, with the residual discriminant omitted for clarity. The printed vanishing orders correspond to the component vanishing orders in each component.}
    \label{fig:illustrative-example}
\end{figure}
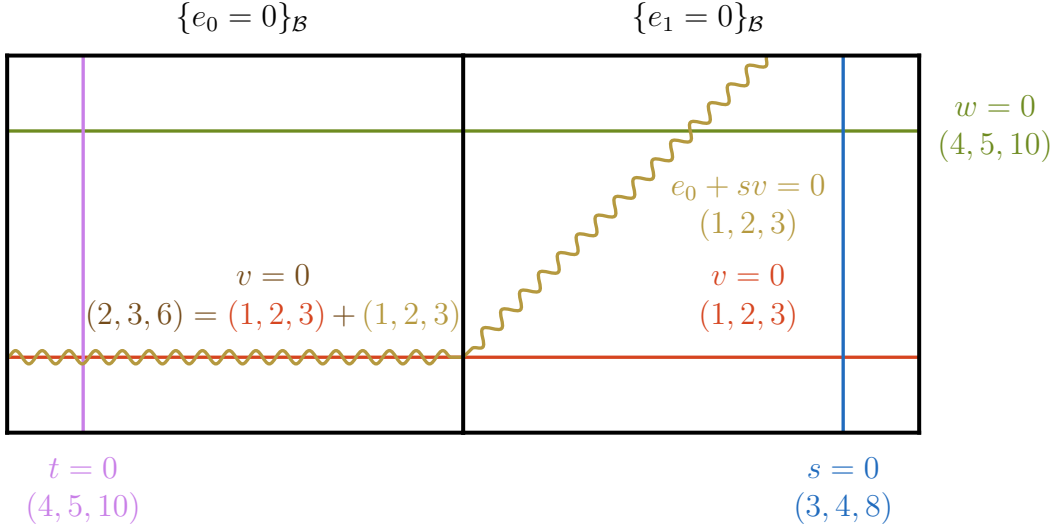
From the point of view of the $Y^{1}$ component, a gauge enhancement occurs over four curves. In terms of their divisor classes within $B^{1}$, these curves and their associated component vanishing orders are
\begin{equation}
    S_{1}: (3,4,8)\,,\quad T_{1}: (1,2,3)\,,\quad V_{1}: (1,2,3)\,,\quad W_{1}: (4,5,10)\,.
\end{equation}
From the perspective of the $Y^{0}$ component, we have instead three curves supporting a gauge enhancement, with vanishing orders
\begin{equation}
    T_{0}: (4,5,10)\,,\quad V_{0}: (2,3,6)\,,\quad W_{0}: (4,5,10)\,.
\end{equation}
From \cref{fig:illustrative-example}, we observe that the additional enhancement over $V_{0}$ with respect to $V_{1}$ occurs because the representative $\{ e_{0} + sv = 0 \}_{B^{1}}$ of the class $T_{1}$ intersects the curve $\{ e_{0} = 0 \}_{B^{1}}$ at $\{ e_{0}=v=0 \}_{B^{1}}$ and extends to the component $B^{0} = \{e_{0} = 0\}_{\mathcal{B}}$ in such a way that it overlaps the representative $\{ v=0 \}_{B^{0}}$ of $V_{0}$. This is the analysis of local gauge enhancements as read off from the component vanishing orders.

Let us now take a global perspective and express the same vanishing orders in terms of the restrictions of the divisors of the family variety $\mathcal{Y}$ to $\tilde{\mathcal{U}}$. For a single infinite-distance limit horizontal model such as \cref{example:illustrative-example}, these restrictions are given in \eqref{eq:horizontal-model-family-divisor-restrictions}. In this way, we see that
\begin{align}
    \mathcal{T}|_{\tilde{\mathcal{U}}} = T_{0}&: (4,5,10)\,,\\
    \mathcal{S}|_{\tilde{\mathcal{U}}} = S_{2}&: (3,4,8)\,,\\
    (\mathcal{T} + E_{0})|_{\tilde{\mathcal{U}}} = (\mathcal{S} + \mathcal{V})|_{\tilde{\mathcal{U}}} = V_{0} + T_{1}&: (1,2,3)\,,\\
    \mathcal{V}|_{\tilde{\mathcal{U}}} = V_{0} + V_{1}&: (1,2,3)\,,\\
    \mathcal{W}|_{\tilde{\mathcal{U}}} = W_{0} + W_{1}&: (4,5,10)\,.
\end{align}
From this point of view, there is no ambiguity in assigning a gauge factor to the divisor \mbox{$\mathcal{V}|_{\tilde{\mathcal{U}}} = \{v=0\}_{B_{0}}$}: The vanishing orders are $(1,2,3)$, and the associated gauge algebra is $\mathfrak{su}(2)$.

This enhancement observed for the central fiber of \cref{example:illustrative-example} is not present at finite distance, as one can check for the base divisor $\{ v=0 \}_{B_{\tilde{u}}}$ of the generic fibers $\hat{Y}_{u \neq 0} \simeq Y_{\tilde{u} \neq 0}$ of the degeneration. The tuning over $\{v=0\}_{B_{\tilde{u}}}$ occurs at the same time as we take the infinite-distance limit. In contrast to this, the enhancement over \mbox{$\pi^{*} \left(\{tu + sv = 0\}_{\hat{B}_{u}}\right) = \{t e_{0} + sv = 0\}_{B_{\tilde{u}}}$} is present both at finite distance and once we reach the endpoint of the limit, appearing with vanishing orders $(1,2,3)$ throughout. The local enhancement over $\{v=0\}_{B^{0}}$ in the $B^{0}$ component, as compared to the gauge algebra read for $\{v=0\}_{B^{1}}$, is a consequence of the fact that the restrictions of $\{ v=0 \}_{B_{0}}$ and $\{t e_{0} + sv = 0\}_{B_{0}}$ to $B^{0}$ coincide, needing the global picture to distinguish between the two.

Note that working with the resolved degeneration $\rho: \mathcal{Y} \rightarrow D$ is crucial to correctly read the gauge algebra content, from the global picture, at the endpoint of the limit. If we try to avoid the multi-component central fiber $Y_{0}$ of the resolved degeneration by working with the original degeneration $\hat{\rho}: \hat{\mathcal{Y}} \rightarrow D$ while ignoring the non-minimal elliptic fibers of $\hat{Y}_{0}$, we potentially run into the same problems that can occur for $Y_{0}$ when only looking at the component vanishing orders without taking the global picture into account.

Let us use \cref{example:illustrative-example} to showcase such a behaviour explicitly. Denote by $\breve{\rho}: \breve{\mathcal{Y}} \rightarrow D$ the blow-down of $\rho: \mathcal{Y} \rightarrow D$ in which we contract the $Y^{0}$ component. Both $\hat{\rho}: \hat{\mathcal{Y}} \rightarrow D$ and $\breve{\rho}: \breve{\mathcal{Y}} \rightarrow D$ have $\rho: \mathcal{Y} \rightarrow D$ as their open-chain resolution, the process of reaching it involving the base blow-ups $\hat{\pi}: \mathcal{B} \rightarrow \hat{\mathcal{B}}$ with centre $\mathcal{S} \cap \mathcal{U}$ and $\breve{\pi}: \mathcal{B} \rightarrow \breve{\mathcal{B}}$ with centre $\mathcal{T} \cap \breve{\mathcal{U}}$, respectively. These three degenerations are equivalent, and therefore represent the same infinite-distance limit in the complex structure moduli space of six-dimensional F-theory. Given the relations
\begin{subequations}
\begin{align}
	\hat{\pi}_{*}(F_{0}) &= F \big|_{\mathcal{U}} - 4 \big( \mathcal{S} \cap \mathcal{U} \big)\,,\\
	\hat{\pi}_{*}(G_{0}) &= G \big|_{\mathcal{U}} - 6 \big( \mathcal{S} \cap \mathcal{U} \big)\,,\\
	\hat{\pi}_{*}(\Delta_{0}) &= \Delta \big|_{\mathcal{U}} - 12 \big( \mathcal{S} \cap \mathcal{U} \big)\,,
\end{align}
\end{subequations}
and
\begin{subequations}
\begin{align}
	\breve{\pi}_{*}(F_{1}) &= {F} \big|_{\breve{\mathcal{U}}} - 4 \big( \mathcal{T} \cap \breve{\mathcal{U}} \big)\,,\\
	\breve{\pi}_{*}(G_{1}) &= {G} \big|_{\breve{\mathcal{U}}} - 6 \big( \mathcal{T} \cap \breve{\mathcal{U}} \big)\,,\\
	\breve{\pi}_{*}(\Delta_{1}) &= {\Delta} \big|_{\breve{\mathcal{U}}} - 12 \big( \mathcal{T} \cap \breve{\mathcal{U}} \big)\,,
\end{align}
\end{subequations}
see \eqref{eq:F0-F-G0-G-relations}, it is clear that
\begin{align}
    \ord{\hat{Y}_{0}}(f|_{u=0},g|_{u=0},\Delta'|_{u=0})_{v=0} = \ord{Y^{0}}\left( f_{0},g_{0},\Delta'_{0} \right)_{v=0} &= (2,3,6)\,,\\
    \ord{\breve{Y}_{0}}({f}|_{\breve{u}=0},{g}|_{\breve{u}=0},{\Delta}'|_{\breve{u}=0})_{v=0} = \ord{Y^{1}}\left( f_{1},g_{1},\Delta'_{1} \right)_{v=0} &= (1,2,3)\,.
\end{align}
Hence, reading off the gauge algebra from the component vanishing orders found in the central fiber of $\hat{\rho}: \hat{\mathcal{Y}} \rightarrow D$ is equivalent to performing a local analysis of the resolved degeneration taking only into account the $Y^{0}$ component; it gives the wrong impression that the local enhancement occurring in this component is a global feature, and only the resolution process resolves the ambiguity. In this example, the vanishing orders over the curve $\mathcal{T} \cap \breve{\mathcal{U}}$ of the central fiber of $\breve{\rho}: \breve{\mathcal{Y}} \rightarrow D$ happen to coincide with those read off from the resolved degeneration.

\subsection{Locally reducible discriminant components}
\label{sec:overcounting-example}

The preceding discussion showed how studying the multi-component central fiber of the resolved degeneration as a whole can reveal certain enhancements occurring over components to be a local phenomenon, which proved important in correctly identifying the types of gauge factors in the model, as illustrated using \cref{example:illustrative-example}. We now focus on another phenomenon that also highlights the importance of studying the global picture. Namely, a single irreducible curve in one component may intersect the interface curve with an adjacent component more than once, extending to said component as a collection of irreducible curves. A local analysis in the adjacent component would prompt us then to overcount the gauge factors, while the global picture shows that these irreducible curves all join together as a single irreducible curve in the component originally considered, and should therefore be counted as a single gauge factor contribution. We showcase this behaviour in the following example.
\begin{example}
\label{example:overcounting-gauge-factors}
Consider the Weierstrass model of the family variety $\hat{\mathcal{Y}}$, elliptically fibered over the base $\hat{\mathcal{B}} = \mathbb{F}_{1} \times D$, given by the defining polynomials
\begin{subequations}
\begin{align}
    f &= -s^3 t^4 v^4 w \left(u^{2} t (v-2w) + s (v+w)(v-w)\right) \left(4 u^2 v+3 u^2 w+2 v+w\right)\,, \label{eq:f-overcounting-example}\\
    g &= s^4 t^5 v^5 w^2 \left(u^{2} t (v-2w) + s (v+w)(v-w)\right)^2 \left(23 s v^2+8 s v w+7 s w^2+6 t v+5 t w\right)\,, \label{eq:g-overcounting-example}\\
    \Delta &= s^8 t^{10} v^{10} w^3 \left(u^{2} t (v-2w) + s (v+w)(v-w)\right)^3 p_{3,7}([s:t],[v:w],u)\,. \label{eq:Delta-overcounting-example}
\end{align}
\end{subequations}
It is a single infinite-distance limit horizontal model presenting non-minimal elliptic fibers over the curve $\mathcal{S} \cap \mathcal{U}$ with
\begin{equation}
	\ord{\hat{\mathcal{Y}}}(f,g,\Delta))_{\mathcal{S} \cap \mathcal{U}} = (4,6,12)\,.
\end{equation}
The resolved degeneration $\rho: \mathcal{Y} \rightarrow D$,  reached after $P=2$ (toric) blow-ups, is the three-component model given by the defining polynomials
\begin{subequations}
\begin{align}
    f &= -s^3 t^4 v^4 w \left(e_0^2 e_1 t (v-2w) + s (v+w)(v-w)\right) \left(e_0^2 e_1^2 e_2^2 (4v+3w) + 2v + w\right)\,,\\
    \begin{split}
    g &= s^4 t^5 v^5 w^2 \left(e_0^2 e_1 t (v-2w) + s (v+w)(v-w)\right)^2\\
    &\quad \times \left(e_1 e_2^2 s(23 v^2+8 v w+7 w^2) +6 t v+5 t w\right)\,,
    \end{split}\\
    \begin{split}
    \Delta &= s^8 t^{10} v^{10} w^3 \left(e_0^2 e_1 t (v-2w) + s (v+w)(v-w)\right)^3\\
    &\quad \times p_{3,7,1,1}([s:t],[v:w],[s:e_{0}:e_{1}],[s:e_{1}:e_{2}])\,,
    \end{split}
\end{align}
\end{subequations}
with Stanley-Reisner ideal
\begin{equation}
	\mathscr{I}_{\mathcal{B}} = \langle st, vw, se_{0}, se_{1}, te_{1}, te_{2}, e_{0}e_{2}  \rangle\,.
\end{equation}
The pattern of codimension-zero fibers over the components of the central fiber of the resolved degeneration is
\begin{equation}
    \begin{tikzcd}[column sep=1em]
        \mathrm{I}_{0} \arrow[rr, dash] \arrow[d, dash] & & \mathrm{I}_{0} \arrow[rr, dash] \arrow[d, dash] & & \mathrm{I}_{0} \arrow[d, dash]\\
        \mathbb{F}_{1} \arrow[rr, dash] & & \mathbb{F}_{1} \arrow[rr, dash] & & \mathbb{F}_{1}\mathrlap{\,.}
    \end{tikzcd}
\end{equation}
Computing the restricted polynomials \eqref{eq:defining-polynomials-component-restrictions} and modified discriminant \eqref{eq:def-modified-discriminant} we obtain for the $B^{0}$ component
\begin{subequations}
\begin{align}
    f_{0} &= -t^4 v^4 w (v+w) (v-w) (2 v+w)\,,\\
    g_{0} &= t^5 v^5 w^2 (v+w)^2 (v-w)^2 \left(23 e_1 v^2+8 e_1 v w+7 e_1 w^2+6 t v+5 t w\right)\,,\\
    \Delta'_{0} &= t^{10} v^{10} w^3 (v+w)^3 (v-w)^3 p_{2,7}([t:e_{1}],[v:w])\,,
\end{align}
\end{subequations}
for the $B^{1}$ component
\begin{subequations}
\begin{align}
    f_{1} &= -v^4 w (v+w) (v-w) (2 v+w)\,,\\
    g_{1} &= v^5 w^2 (v+w)^2 (v-w)^2 (6 v+5 w)\,,\\
    \Delta'_{1} &= -v^{10} w^3 (v+w)^3 (v-w)^3 p_{0,5}([e_{0},e_{1}],[v:w])\,,
\end{align}
\end{subequations}
and for the $B^{2}$ component
\begin{subequations}
\begin{align}
    f_{2} &= -s^3 v^4 w \left(e_1 (v-2 w) +s (v+w)(v-w)\right) (2 v+w) \,,\\
    g_{2} &= s^4 v^5 w^2 \left(e_1 (v-2 w) +s (v+w)(v-w)\right)^2 (6 v+5 w)\,,\\
    \Delta'_{2} &= -s^8 v^{10} w^3 \left(e_1 (v-2 w) +s (v+w)(v-w)\right)^3 p_{1,5}([e_{1},s],[v:w])\,.
\end{align}
\end{subequations}

From the point of view of the $B^{0}$ and $B^{1}$ components, we observe two local gauge enhancements along
\begin{align}
    \ord{Y^{0}}(f_{0},g_{0},\Delta'_{0})_{v+w=0} = \ord{Y^{1}}(f_{1},g_{1},\Delta'_{1})_{v+w=0} &= (1,2,3)\,,\\
    \ord{Y^{0}}(f_{0},g_{0},\Delta'_{0})_{v-w=0} = \ord{Y^{1}}(f_{1},g_{1},\Delta'_{1})_{v-w=0} &= (1,2,3)\,,
\end{align}
which extend into the $B^{2}$ component as a single irreducible curve
\begin{equation}
    \ord{Y^{2}}(f_{2},g_{2},\Delta'_{2})_{e_1 (v-2 w) +s (v+w)(v-w)=0} = (1,2,3)\,.
\end{equation}
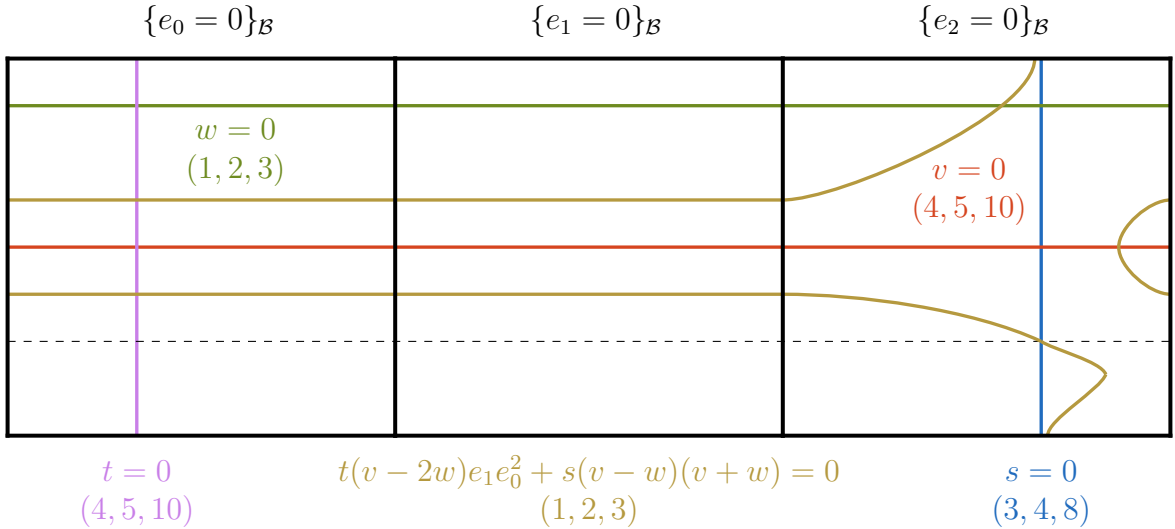
\begin{figure}[t!]
    \centering
    \begin{tikzpicture}
    	\def\scalethreecomphor{1.7}
		\def\scalethreecompver{2.5}

		\node [] (2) at (1.5*\scalethreecomphor, 1*\scalethreecompver) {};
		\node [] (3) at (1.5*\scalethreecomphor, -1*\scalethreecompver) {};
		\node [] (4) at (-1.5*\scalethreecomphor, -1*\scalethreecompver) {};
		\node [] (5) at (-1.5*\scalethreecomphor, 1*\scalethreecompver) {};
		\node [] (6) at (4.5*\scalethreecomphor, 1*\scalethreecompver) {};
		\node [] (7) at (4.5*\scalethreecomphor, -1*\scalethreecompver) {};
		\node [] (8) at (-4.5*\scalethreecomphor, -1*\scalethreecompver) {};
		\node [] (9) at (-4.5*\scalethreecomphor, 1*\scalethreecompver) {};
		\node [] (10) at (3.5*\scalethreecomphor, 1*\scalethreecompver) {};
		\node [label={[align=center, yshift=-1.5cm]\textcolor{diagLightBlue}{$s=0$}\\ \textcolor{diagLightBlue}{$(3,4,8)$}}] (11) at (3.5*\scalethreecomphor, -1*\scalethreecompver) {};
		\node [] (12) at (-4.5*\scalethreecomphor, 0*\scalethreecompver) {};
		\node [label={[align=center, xshift=-2.65cm, yshift=0cm]\textcolor{diagMediumRed}{$v=0$}\\ \textcolor{diagMediumRed}{$(4,5,10)$}}] (13) at (4.5*\scalethreecomphor, 0*\scalethreecompver) {};
		\node [label={[align=center, xshift=3cm, yshift=-1.35cm]\textcolor{diagLightGreen}{$w=0$}\\ \textcolor{diagLightGreen}{$(1,2,3)$}}] (14) at (-4.5*\scalethreecomphor, 0.75*\scalethreecompver) {};
		\node [] (15) at (4.5*\scalethreecomphor, 0.75*\scalethreecompver) {};
		\node [] (16) at (-4.5*\scalethreecomphor, 0.25*\scalethreecompver) {};
		\node [] (17) at (-4.5*\scalethreecomphor, -0.25*\scalethreecompver) {};
		\node [] (18) at (1.5*\scalethreecomphor, 0.25*\scalethreecompver) {};
		\node [] (19) at (1.5*\scalethreecomphor, -0.25*\scalethreecompver) {};
		\node [] (20) at (3.5*\scalethreecomphor, -0.5*\scalethreecompver) {};
		\node [] (21) at (-3.5*\scalethreecomphor, 1*\scalethreecompver) {};
		\node [label={[align=center, yshift=-1.5cm]\textcolor{diagLightPurple}{$t=0$}\\ \textcolor{diagLightPurple}{$(4,5,10)$}}] (22) at (-3.5*\scalethreecomphor, -1*\scalethreecompver) {};
		\node [] (23) at (3.45*\scalethreecomphor, 1*\scalethreecompver) {};
		\node [] (24) at (4*\scalethreecomphor, -0.675*\scalethreecompver) {};
		\node [] (25) at (3.55*\scalethreecomphor, -1*\scalethreecompver) {};
		\node [] (26) at (4.5*\scalethreecomphor, 0.25*\scalethreecompver) {};
		\node [] (27) at (4.5*\scalethreecomphor, -0.25*\scalethreecompver) {};
		\node [] (28) at (4.1*\scalethreecomphor, 0*\scalethreecompver) {};
		\node [] (29) at (-4.5*\scalethreecomphor, -0.5*\scalethreecompver) {};
		\node [] (30) at (4.5*\scalethreecomphor, -0.5*\scalethreecompver) {};
        \node [label={[yshift=0cm]$\{e_{0} = 0\}_{\makebox[0pt]{$\scriptstyle \;\;\mathcal{B}$}}$}] (31) at (-3*\scalethreecomphor, 1*\scalethreecompver) {};
        \node [label={[yshift=0cm]$\{e_{1} = 0\}_{\makebox[0pt]{$\scriptstyle \;\;\mathcal{B}$}}$}] (32) at (0*\scalethreecomphor, 1*\scalethreecompver) {};
        \node [label={[yshift=0cm]$\{e_{2} = 0\}_{\makebox[0pt]{$\scriptstyle \;\;\mathcal{B}$}}$}] (33) at (3*\scalethreecomphor, 1*\scalethreecompver) {};
        \node [label={[align=center, yshift=-1.5cm]\textcolor{diagLightYellow}{$t(v-2w)e_{1}e_{0}^{2} + s(v-w)(v+w)=0$}\\ \textcolor{diagLightYellow}{$(1,2,3)$}}] (34) at (0*\scalethreecomphor, -1*\scalethreecompver) {};
  
		\draw [style=light-blue line] (10.center) to (11.center);
		\draw [style=medium-red line] (12.center) to (13.center);
		\draw [style=light-green line] (14.center) to (15.center);
		\draw [style=light-purple line] (21.center) to (22.center);
		\draw [style=light-yellow line] (16.center) to (18.center);
		\draw [style=light-yellow line] (17.center) to (19.center);
        \draw [style=dashed] (29.center) to (30.center);
  
		\draw [style=light-yellow line, in=-90, out=0, looseness=0.50] (18.center) to (23.center);
		\draw [style=light-yellow line, in=155, out=0, looseness=0.75] (19.center) to (20.center);
		\draw [style=light-yellow line, in=90, out=180, looseness=0.75] (26.center) to (28.center);
		\draw [style=light-yellow line, in=180, out=-90, looseness=0.75] (28.center) to (27.center);
		\draw [style=light-yellow line, in=110, out=-35, looseness=0.50] (20.center) to (24.center);
		\draw [style=light-yellow line,in=90, out=-110, looseness=0.50] (24.center) to (25.center);

        \draw [style=black line] (9.center) to (5.center);
		\draw [style=black line] (5.center) to (2.center);
		\draw [style=black line] (2.center) to (6.center);
		\draw [style=black line] (6.center) to (7.center);
		\draw [style=black line] (7.center) to (3.center);
		\draw [style=black line] (3.center) to (4.center);
		\draw [style=black line] (4.center) to (8.center);
		\draw [style=black line] (8.center) to (9.center);
		\draw [style=black line] (5.center) to (4.center);
		\draw [style=black line] (2.center) to (3.center);
    \end{tikzpicture}
    \caption{Restrictions $\Delta'_{0}$, $\Delta'_{1}$ and $\Delta'_{2}$ of the (modified) discriminant for \cref{example:overcounting-gauge-factors}, with the residual discriminant omitted for clarity. The printed vanishing orders correspond to the component vanishing orders in each component. Although the divisor $\mathcal{D}_{\mathrm{phys}}$ restricts to two irreducible curves both in $B^{0}$ and $B^{1}$, it restricts to a single irreducible curve in $B^{2}$, corresponding therefore to a single gauge factor.} 
    \label{fig:overcounting-example}
\end{figure}%
Performing the gauge assignment from the local information in the components $B^{0}$ or $B^{1}$ would lead to an overcounting of the gauge algebra factors. The global perspective clarifies that we have a single gauge algebra corresponding to vanishing orders $(1,2,3)$ over the divisor
\begin{equation}
    \mathcal{D}_{\mathrm{phys}} := \{e_0^2 e_1 t (v-2w) + s (v+w)(v-w) = 0\}_{B_{0}}
\end{equation}
of $B_{0}$, which is in the class
\begin{equation}
    \mathcal{D}_{\mathrm{phys}} = (\mathcal{T} + 2E_{0} + E_{1} + \mathcal{V})|_{\tilde{\mathcal{U}}} = (\mathcal{S} + 2\mathcal{V})|_{\tilde{\mathcal{U}}} = (2V_{0}) + (2V_{1}) + (T_{2} + V_{2})\,,
\label{eq:overcounting-example-decomposition-classes}
\end{equation}
and restricts in the components to the curves listed above and portrayed in \cref{fig:overcounting-example}.

In this example, it is particularly clear that a single gauge factor should be associated with the divisor $\mathcal{D}_{\mathrm{phys}}$ extending through the multiple components of the reducible variety $B_{0}$. The gauge algebra supported on it is already present in all models represented by the fibers at finite distance, since in any $Y_{\tilde{u} \neq 0}$ we have
\begin{equation}
    \ord{Y_{\tilde{u} \neq 0}}(f|_{\tilde{u} \neq 0}, g|_{\tilde{u} \neq 0}, \Delta|_{\tilde{u} \neq 0})_{\tilde{u}^{2} t (v-2w) + s (v+w)(v-w) = 0} = (1,2,3)\,.
\end{equation}
This enhancement is unaffected by the infinite-distance limit. Note that determining the gauge algebra from the central fiber $\hat{Y}_{0}$ of the unresolved degeneration $\hat{\rho}: \hat{\mathcal{Y}} \rightarrow D$ would lead to the same conclusions as the local analysis in the $Y^{0}$ component of the resolved degeneration, and therefore to an overcounting of the gauge factors; working with the global picture of $Y_{0}$ in $\rho: \mathcal{Y} \rightarrow D$ avoids this problem.
\end{example}

\subsection{Physical discriminant for the multi-component central fiber}
\label{sec:physical-discriminant}

Having illustrated through \cref{example:illustrative-example,example:overcounting-gauge-factors} the problems that can occur, let us now discuss how to extract in practice the codimension-one physics from the central fiber $Y_{0}$ of a resolved degeneration $\rho: \mathcal{Y} \rightarrow D$. We start with an informal, but operational explanation of the matter paired with an explicit example, concluding the section by concisely restating the information in a cleaner fashion.

As can be distilled from the examples above, we need to associate the gauge factors to divisors $\mathcal{D}_{\mathrm{phys}}$ of the multi-component base $B_{0}$ of the resolved degeneration, rather than to the irreducible components of the restrictions $\{ \Delta'_{p} \}_{0 \leq p \leq P}$ of the modified divisor $\Delta'$. These components, however, must consistently glue together between components, as shown in \cref{fig:illustrative-example,fig:overcounting-example}, to produce divisors defined in $B_{0}$. Since the surface $B_{0}$ is itself reducible, so will be the divisors extending between components. We then assign a single gauge algebra factor to each divisor $\mathcal{D}_{\mathrm{phys}}$ of $B_{0}$ obtained by consistently gluing together the irreducible components of $\{ \Delta'_{p} \}_{0 \leq p \leq P}$ and such that it restricts to a single irreducible divisor in at least one of the components $\{ B^{p} \}_{0 \leq p \leq P}$.

In practical terms, it would be useful to obtain these divisors by factorising a physical discriminant defined in $B_{0}$. In the same way that the restricted polynomials $\{ f_{p} \}_{0 \leq p \leq P}$, $\{ g_{p} \}_{0 \leq p \leq P}$ and $\{ \Delta'_{p} \}_{0 \leq p \leq P}$ associated to the individual components $\{ B^{p} \}_{0 \leq p \leq P}$ are obtained by restricting their counterparts in $\mathcal{B}$ to the vanishing locus of the appropriate $e_{p}$ coordinate, we can define similar quantities for the multi-component surface $B_{0} = \bigcup_{p=0}^{P} B^{p}$ as a whole. In light of the relation \eqref{eq:utilde-linear-equivalence}, we simply need to take the restriction of the defining polynomials to the vanishing locus of the coordinate $\tilde{u} = \prod_{p=0}^{P} e_{p}$. We will call these restrictions the physical defining polynomials, and denote them by $\fphys$, $\gphys$ and $\Dphys$. The divisors described in the previous paragraph, and to which the individual gauge algebra factors are associated, correspond to the vanishing loci of the factors of $\Dphys$. In terms of divisor classes, the physical defining polynomials correspond to the restrictions $\left. F \right|_{\tilde{U}}$, $\left. G \right|_{\tilde{U}}$ and $\left. \Delta' \right|_{\tilde{U}}$.

For the polynomial $\Delta'$ or its restrictions $\{ \Delta'_{p} \}_{0 \leq p \leq P}$, the factorization is performed in the standard way. When factorizing 
\begin{equation}
    \Dphys := \left.\Delta'\right|_{\tilde{u}=0}\,,\qquad \tilde{u} = \prod_{p=0}^{P} e_{p}\,,
\end{equation}
to obtain the loci supporting the gauge algebra, however, we need to keep in mind that the product $\tilde{u} = \prod_{p=0}^{P} e_{p}$ is zero when evaluated over $B_0$, and therefore the factorization is to be done up to terms proportional to $\tilde{u}$. That is, a divisor \mbox{$\mathcal{D}_{\mathrm{phys}} = \{p_{\mathcal{D}_{\mathrm{phys}}} = 0\}_{B_{0}}$} is a component of $\left. \Delta' \right|_{\tilde{\mathcal{U}}}$ if the remainder of the quotient $\Dphys$ by $p_{\mathcal{D}_{\mathrm{phys}}}$ is proportional to $\tilde{u}$, i.e.\ when 
\begin{equation}
    \Dphys = p_{\mathcal{D}_{\mathrm{phys}}} q + \tilde{u} r'
\label{eq:factorize-tildeu}
\end{equation}
for some polynomials $q$ and $r'$.

It may not always be needed to take this into account. For example, if a gauge enhancement is present at finite distance, i.e.\ for the generic fibers $Y_{\tilde{u} \neq 0}$ of the degeneration $\rho: \mathcal{Y} \rightarrow D$, the corresponding divisor will factorize in $\Delta'$ and hence appear factorized in $\Dphys$. This is, in fact, what occurs in \cref{example:illustrative-example,example:overcounting-gauge-factors}. To give one instance of this, consider $\mathcal{D}_\mathrm{phys}$ in \cref{example:overcounting-gauge-factors}. We have that the divisor
\begin{equation}
	\mathcal{D} = \{e_0^2 e_1 t (v-2w) + s (v+w)(v-w) = 0\}_{\mathcal{B}}
\end{equation}
in $\mathcal{B}$ restricts to the base $B_{\tilde{u}}$ of all elements of the degeneration such that $\left. \mathcal{D} \right|_{B_{\tilde{u}}} \subset \left. \Delta' \right|_{B_{\tilde{u}}}$, and in particular
\begin{equation}
\begin{split}
	\mathcal{D}_{\mathrm{phys}} = \mathcal{D}|_{\tilde{\mathcal{U}}} &= \{e_0^2 e_1 t (v-2w) + s (v+w)(v-w) = 0\}_{B_{0}}\\
	&= \left(\{v+w=0\}_{Y^{0}} \cup \{v-w=0\}_{Y^{0}}\right)\\
	&\quad \cup \left(\{v+w=0\}_{Y^{1}} \cup \{v-w=0\}_{Y^{1}}\right)\\
	&\quad \cup \left(\{e_1 (v-2 w) +s (v+w)(v-w) = 0\}_{Y^{2}}\right) = \mathcal{D}_{0} \cup \mathcal{D}_{1} \cup \mathcal{D}_{2}\,,
\end{split}
\end{equation}
with the factorization
\begin{equation}
	\Dphys = \Delta'|_{\tilde{\mathcal{U}}} = 3\mathcal{D}|_{\tilde{\mathcal{U}}} + \Delta''|_{\tilde{\mathcal{U}}} = \sum_{p=0}^{2} \left(3\mathcal{D}_{p} + \Delta''_{p}\right) =  \sum_{p=0}^{2} \Delta'_{p}\,.
\end{equation}

When the finite-distance tuning necessary to produce a certain gauge enhancement takes place at the same time as the infinite-distance limit is taken, it may occur that the factorization of the corresponding $\mathcal{D}_{\mathrm{phys}}$ divisor in $\Dphys$ is not immediately apparent unless we factorize up to terms proportional to $\tilde{u}$, see \eqref{eq:factorize-tildeu}, as we now show in a concrete example.
\begin{example}
\label{example:factorization-physical-polynomials}
	Consider the family variety $\hat{\mathcal{Y}}$ with base $\hat{\mathcal{B}} = \mathbb{F}_{1} \times D$ given by the Weierstrass model with defining polynomials
	\begin{subequations}
	\begin{align}
	    f &= s^3 t^3 v \left(s^2 u v^8+s t w^4 \left(u w^3+v^3+v^2 w+v w^2\right)+t^2 u w^4 \left(v^2+v w+w^2\right)\right)\,,\\
	    \begin{split}
	    g &= s^4 t^5 v^2 w^5 \left(s^3 v w^5+s^2 t v^2 (v+w) \left(v^2+w^2\right)+2 s t^2 u v (v+w) \left(v^2+w^2\right)\right.\\
	    &\quad+\left.t^3 u^2 (v+w) \left(v^2+w^2\right)\right)\,,
	    \end{split}\\
	    \Delta &= s^8 t^9 v^3 p_{7,24}([s:t],[v:w],u)\,.
	\end{align}
	\end{subequations}
	This model is only a minor modification of \cref{example:illustrative-example}, and we therefore do not analyse it fully. The difference between the two is that in the present model the finite-distance $(1,2,3)$ enhancement over $\mathcal{D}_{\mathrm{phys}} := \{ t e_{0} + s v = 0 \}_{B_{0}}$ is not present for the generic fibers of the degeneration; it only occurs at the endpoint of the limit. The open-chain resolution is given by the defining polynomials
	\begin{subequations}
	\begin{align}
	    f &= s^3 t^3 v \left(e_0 e_1^2 s^2 v^8+t w^4 \left(v^2+v w+w^2\right) \left(e_0 t+s v\right)+e_0 e_1 s t w^7\right)\,,\\
	    g &= s^4 t^5 v^2 w^5 \left(e_1 s^3 v w^5+t (v+w) \left(v^2+w^2\right) \left(e_0 t+s v\right)^2\right)\,,\\
	    \Delta &= s^8 t^9 v^3 p_{7,24,4}([s:t],[v:w],[s:e_{0}:e_{1}])\,.
	\end{align}
	\end{subequations}
	One thing that can be noted from this model and its cousin \cref{example:illustrative-example} is that their restrictions to the components $\{B^{p}\}_{0 \leq p \leq 3}$ coincide, showing that the finite-distance deviation from one another associated to the tuning over $\mathcal{D}_{\mathrm{phys}}$ does not alter the endpoint of the limit.
	
	Considering the local analysis performed in \cref{example:illustrative-example}, we expect $\mathcal{D}_{\mathrm{phys}}$ to factorize with multiplicities one and two in $\fphys$ and $\gphys$, respectively. Computing the restriction
	\begin{equation}
	    f|_{e_{0} e_{1}=0} = s^3 t^4 v w^4 \left(t e_{0}+s v\right) \left(v^2+v w+w^2\right)\,,
	\end{equation}
	the factorization of $p_{{D}_{\mathrm{phys}}}$ is indeed explicit. Computing the same restriction for $g$, however, leads to
	\begin{equation}
	    g|_{e_{0} e_{1}=0} = g\,,
	\end{equation}
	for which the expected factorization of $p_{{D}_{\mathrm{phys}}}^{2}$ is not apparent. Note that to obtain the restriction $g|_{e_{0}e_{1}=0}$ we only set to zero those monomials that contain powers of $\tilde{u} = e_{0} e_{1}$, rather than just individual powers of $e_{0}$ or $e_{1}$. The same is true for $\Delta'|_{e_{0} e_{1}=0}$, where we do not observe $p_{{D}_{\mathrm{phys}}}^{3}$ factorizing. However, recalling that the factorization needs to occur up to terms proportional to $\tilde{u} = e_{0}e_{1}$, we can perform the division of polynomials
	\begin{equation}
	    g|_{e_{0} e_{1}=0} = (t e_{0} + s v)^{2}q + r\,,
	\end{equation}
	finding that
	\begin{equation}
	    r = e_0^2 e_1 s^4 t^7 w^{10} \left(2 e_0 t+3 s v\right)\,,
	\end{equation}
	so that $r|_{e_{0} e_{1}} = 0$ and indeed the factorization goes through.
\end{example}

After these considerations, we now define the objects $\fphys$, $\gphys$ and $\Dphys$ more precisely. The reader only interested in the practical use of $\fphys$, $\gphys$ and $\Dphys$ can safely skip to \cref{sec:monodromy-cover}. Before we delve into the discussion, and since we are phrasing it in the context in which $\mathcal{B}$ has a global description as a toric variety in terms of the homogeneous coordinates, let us recall the definition of the homogeneous coordinate ring of a toric variety (see, e.g.,\ \cite{Cox2011}).
\begin{definition}
    The homogeneous coordinate ring of a toric variety $X$ defined by the toric fan $\Sigma_{X}$ is
    \begin{equation}
        S_{X} := \mathbb{C} \left[ x_{\rho} \mid \rho \in \Sigma_{X}(1) \right]\,.
    \end{equation}
    For each cone $\sigma \in \Sigma$, define the monomial
    \begin{equation}
        x^{\check{\sigma}} = \prod_{\rho \notin \sigma(1)} x_{\rho}\,.
    \end{equation}
    The irrelevant ideal of $X$ is defined to be
    \begin{equation}
        B_{X} := \langle x^{\check{\sigma}} \mid \sigma \in \Sigma_{X} \rangle = \langle x^{\check{\sigma}} \mid \sigma \in \Sigma_{X}^{\max} \rangle \subseteq S_{X}\,,
    \end{equation}
    where $\Sigma_{X}^{\max}$ is the set of maximal cones of $\Sigma_{X}$.
\end{definition}
The vanishing locus of $V(B_{X}) \subseteq \mathbb{C}^{\Sigma_{X}(1)}$ is the exceptional set of the quotient construction of the toric variety.

The homogeneous coordinate ring of a toric variety, introduced by Cox in \cite{Cox1992}, is the analogue in the toric context of the homogeneous coordinate ring of projective varieties, which is the object we would need to use if we phrased the discussion in that language. Cox's notion has been generalized to that of the total coordinate ring, which applies more generally.
\begin{definition}
    Let $X$ be a normal projective variety with divisor class group $\mathrm{Cl}(X)$, and assume that $\mathrm{Cl}(X)$ is a finitely generated free abelian group. The total coordinate ring or Cox ring of $X$ is
    \begin{equation}
        \mathrm{TC}(X) := \bigoplus_{D} H^{0} \left( X, \mathcal{O}_{X}(D) \right)\,,
    \end{equation}
    where the sum is over all Weil divisors contained in a fixed complete system of representatives of $\mathrm{Cl}(X)$.
\end{definition}
Note that nowadays, the homogeneous coordinate ring of a toric variety is usually called total coordinate ring, since it is a particular example of this notion.

The ideal-variety correspondence of affine and projective algebraic varieties goes through for toric varieties when we use the total coordinate ring.
\begin{proposition}
    Let $X$ be a simplicial toric variety associated with the fan $\Sigma_{X}$. Then there is a bijective correspondence \cite{Cox2011}
    \begin{equation}
        \{ \textrm{closed subvarieties of } X \} \longleftrightarrow \left\{ \textrm{radical homogeneous ideals } I \subseteq B_{X} \subseteq S_{X} \right\}\,.
    \end{equation}
\end{proposition}

With this in mind, we see that the physical defining polynomials $\fphys$, $\gphys$ and $\Dphys$ are defined in the following way.
\begin{definition}[Physical defining polynomials]
\label{def:physical-defining-polynomials}
    Let $B_{0}$ be the base central fiber of a resolved degeneration $\rho: \mathcal{Y} \rightarrow D$ in which $\mathcal{B}$ has homogeneous coordinate ring $S_{\mathcal{B}}$. Consider the homogeneous ideal
    \begin{equation}
        I_{\tilde{\mathcal{U}}} := \langle e_{0} \cdots e_{P} \rangle \trianglelefteq S_{\mathcal{B}}\,,
    \end{equation}
    whose vanishing locus corresponds to $B_{0}$. The physical defining polynomials $\fphys$, $\gphys$ and $\Dphys$, whose vanishing loci represent $\left. F \right|_{\tilde{\mathcal{U}}}$, $\left. G \right|_{\tilde{\mathcal{U}}}$ and $\left. \Delta' \right|_{\tilde{\mathcal{U}}}$, respectively, are
    \begin{equation}
        \fphys := f + I_{\tilde{\mathcal{U}}}\,,\qquad \gphys := g + I_{\tilde{\mathcal{U}}}\,, \qquad \Dphys := \Delta' + I_{\tilde{\mathcal{U}}}\,,
    \end{equation}
    where $\fphys, \gphys, \Dphys \in S_{\mathcal{B}}/I_{\tilde{\mathcal{U}}}$.
\end{definition}
From these, we obtain the physical vanishing orders used above as the appropriate vanishing orders in the global analysis of the central fiber $Y_{0}$.
\begin{definition}[Physical vanishing orders]
    Let $C = \{p_{C} = 0\} \subset B_{0}$ be a curve in the central fiber of a resolved degeneration $\rho: \mathcal{Y} \rightarrow D$ with base components $\{ B^{p} \}_{0 \leq p \leq P}$ and such that at least one of the restrictions $\{ \left. C \right|_{B^{p}} \}_{0 \leq p \leq P}$ is a single irreducible curve. Let $\alpha$, $\beta$ and $\gamma$ be
    \begin{subequations}
    \begin{align}
        \alpha &:= \max\left\{ i \in \mathbb{Z}_{\geq 0} \mid \langle \fphys \rangle \trianglelefteq \langle p_{C}^{i} \rangle \right\}\,,\\
        \beta &:= \max\left\{ j \in \mathbb{Z}_{\geq 0} \mid \langle \gphys \rangle \trianglelefteq \langle p_{C}^{j} \rangle \right\}\,,\\
        \gamma &:= \max\left\{ k \in \mathbb{Z}_{\geq 0} \mid \langle \Dphys \rangle \trianglelefteq \langle p_{C}^{k} \rangle \right\}\,.
    \end{align}
    \end{subequations}
    We define the physical vanishing orders over $C$, denoted by $\ord{Y_{0}}(\fphys,\gphys,\Dphys)_{C}$, to be
    \begin{equation}
        \ord{Y_{0}}(\fphys,\gphys,\Dphys)_{C} := (\alpha, \beta, \gamma)\,.
    \end{equation}
\end{definition}

We can now see how the need to be more cautious during the factorization process, as in \cref{example:factorization-physical-polynomials}, arises when we consider the multi-component central fiber $B_{0}$, but not when we consider the base family variety $\mathcal{B}$ or the individual base components $\{ B^{p} \}_{0 \leq p \leq P}$. Under the ideal-variety correspondence, irreducible subvarieties are associated with prime ideals. If the homogeneous coordinate ring of the variety is a GCD domain (hence, in particular, in a unique factorization domain), the notions of prime and irreducible element coincide. When we factorize the discriminant polynomial into irreducible polynomials to determine the prime divisors supporting the non-abelian gauge algebra, we are making use of this fact.\footnote{More precisely, we are performing the primary decomposition of the ideal generated by the discriminant polynomial. Primary ideals are powers of prime ideals, leading to the same vanishing locus but containing the multiplicity information necessary to determine the gauge algebra.} This works well for $\mathcal{B}$ and $\{ B^{p} \}_{0 \leq p \leq P}$ because their homogeneous coordinate rings are unique factorization domains, as can be deduced from the following result concerning the total coordinate ring \cite{Elizondo2003}.
\begin{theorem}
    The total coordinate ring of a connected normal Noetherian scheme whose divisor class group is a finitely generated free abelian group is a unique factorization domain.
\end{theorem}
For a smooth, and hence locally factorial, variety $X$ we have $\mathrm{Pic}(X) \cong \mathrm{Cl}(X)$, which applies to $\mathcal{B}$ and $\{ B^{p} \}_{0 \leq p \leq P}$. The $\{ B^{p} \}_{0 \leq p \leq P}$ surfaces are either $\mathbb{P}^{2}$, $\mathbb{F}_{n}$ or a blow-ups thereof, for which $\mathrm{Pic}(X)$ is a finitely generated free abelian group. $\mathcal{B}$ is the blow-up of $\hat{\mathcal{B}} = B \times D$, with $B$ one of the previously listed surfaces, and since $\mathrm{Cl}(B \times \mathbb{A}^{1}) \simeq \mathrm{Cl}(B)$ the result applies as well.

Hence, when we use the restrictions $\{ f_{p}, g_{p}, \Delta'_{p} \}_{0 \leq p \leq P}$ to compute the component vanishing orders, we are taking the image of the defining polynomials $f$, $g$ and $\Delta'$ of the elliptic fibration over $\mathcal{B}$ under the ring homomorphism
\begin{equation}
    \begin{aligned}
        \phi_{p}: S_{\mathcal{B}} &\longrightarrow S_{B^{p}} \cong S_{\mathcal{B}}/\langle e_{p} \rangle\\
        p &\longmapsto p + \langle e_{p} \rangle\,,
    \end{aligned}
\end{equation}
choosing $p|_{e_{p} = 0}$ as a concrete representative of $p + \langle e_{p} \rangle$. The homogeneous ideal $\langle e_{p} \rangle \trianglelefteq S_{\mathcal{B}}$, its vanishing locus corresponding to the irreducible subvariety $B^{p}$ of $\mathcal{B}$, is a prime ideal, and $S_{B^{p}}$ is a unique factorization domain. Hence, $S_{B^{p}}$ is in particular an integral domain, where the notion of irreducible element is well-defined; the representative $p|_{e_{p} = 0}$ is as good as any other to judge this property.

For the base central fiber $B_{0}$ the situation is different because it is a reducible surface, meaning that the homogeneous ideal $I_{\tilde{\mathcal{U}}} \trianglelefteq S_{\mathcal{B}}$ is not a prime ideal. In fact, its primary decomposition is \begin{equation}
    I_{\tilde{\mathcal{U}}} = \langle e_{0} \cdots e_{P} \rangle = \langle e_{0} \rangle \cap \cdots \cap \langle e_{P} \rangle\,.
\end{equation}
The physical polynomials $\fphys$, $\gphys$ and $\Dphys$ of \cref{def:physical-defining-polynomials} are the images of $f$, $g$ and $\Delta'$ under the ring homomorphism
\begin{equation}
    \begin{aligned}
        \phi: S_{\mathcal{B}} &\longrightarrow S_{\mathcal{B}}/I_{\tilde{\mathcal{U}}}\\
        p &\longmapsto p + I_{\tilde{\mathcal{U}}}\,.
    \end{aligned}
\end{equation}
Since $I_{\tilde{\mathcal{U}}}$ is not a prime ideal, the quotient $S_{\mathcal{B}}/I_{\tilde{\mathcal{U}}}$ is not an integral domain. It is this ring in which the factorisation of the physical defining polynomials takes place. However, the behaviour of polynomial factorisation in rings with zero divisors is vastly different from the one in integral domains. In fact, there are four different notions of irreducible element that one can define, which all coincide for integral domains. We do not delve into this topic further, providing only a small collection of relevant facts alongside useful references in \cref{sec:polynomial-factorization-zero-divisors}. For our purposes, it suffices to note that the image of $\phi: S_{\mathcal{B}} \longrightarrow S_{\mathcal{B}}/I_{\tilde{\mathcal{U}}}$ captures the same information as taking the restrictions into all individual components into account together, i.e.\ consistently gluing together the irreducible components of the $\{ \Delta'_{p} \}_{0 \leq p \leq P}$ as explained earlier in the section. This follows from the fact that the ring homomorphism
\begin{equation}
    \begin{aligned}
        \psi: S_{\mathcal{B}} &\longrightarrow S_{B^{0}} \times \cdots \times S_{B^{P}}\\
        p &\longmapsto \left(p + \langle e_{0} \rangle, \dotsc, p + \langle e_{P} \rangle\right)
    \end{aligned}
\end{equation}
is not surjective, since the ideals $\langle e_{p} \rangle$ and $\langle e_{q} \rangle$ are not coprime for all $p \neq q$. The First Isomorphism Theorem\footnote{We follow the numbering of \cite{hungerford2003algebra} for the isomorphism theorems.} then establishes, since $\mathrm{Ker}(\psi) = \bigcap_{p=0}^{P} \langle e_{p} \rangle$, that $\phi(S_{\mathcal{B}}) = S_{\mathcal{B}}/I_{\tilde{\mathcal{U}}} \cong \psi(S_{\mathcal{B}})$. One can, as a consequence, conclude that being able to factorise the defining polynomial of a curve $C = \{ p_{C} = 0 \}_{B_{0}}$ from $\Dphys$ is equivalent to being able to factorise its restrictions $\{ \left. p_{C} \right|_{e_{p}=0} \}_{0 \leq p \leq P}$ from all the $\{ \Delta'_{p} \}_{0 \leq p \leq P}$, as one would intuitively expect.
\begin{proposition}
    With the notation used above, we have that
    \begin{equation}
        p + I_{\tilde{\mathcal{U}}}\, |\, \Dphys \Leftrightarrow p + \langle e_{p} \rangle\, |\, \Delta' + \langle e_{p} \rangle\,,\quad \forall p = 0, \dotsc, P\,.
    \end{equation}
\end{proposition}
\begin{proof}
    Consider the ring homomorphism
    \begin{equation}
        \begin{aligned}
            \tilde{\phi}: S_{\mathcal{B}}/I_{\tilde{\mathcal{U}}} &\longrightarrow (S_{\mathcal{B}}/I_{\tilde{\mathcal{U}}})/(\langle e_{p} \rangle/I_{\tilde{\mathcal{U}}}) \cong S_{B^{p}}\\
            p + I_{\tilde{\mathcal{U}}} &\longmapsto p + \langle e_{p} \rangle\,,
        \end{aligned}
    \end{equation}
    where we have used the Third Isomorphism Theorem. It is then clear that
    \begin{equation}
        p + I_{\tilde{\mathcal{U}}}\, |\, \Dphys \Rightarrow \tilde{\phi}\left(p + I_{\tilde{\mathcal{U}}}\right)\, |\, \tilde{\phi}\left(\Dphys\right) \Rightarrow p + \langle e_{p} \rangle\, |\, \Delta' + \langle e_{p} \rangle\,.
    \end{equation}
    Conversely, consider that $p + \langle e_{p} \rangle\, |\, \Delta' + \langle e_{p} \rangle$, $\forall p = 0, \dotsc, P$. This implies that
    \begin{equation}
        \langle \Delta' + \langle e_{p} \rangle\rangle \trianglelefteq \langle p + \langle e_{p} \rangle \rangle\,,\quad \forall p = 0, \dotsc, P\,.
    \end{equation}
    Then, in the product ring $S_{B^{0}} \times \cdots \times S_{B^{P}}$ we have
    \begin{equation}
        \langle \Delta' + \langle e_{0} \rangle\rangle \times \cdots \times \langle \Delta' + \langle e_{P} \rangle\rangle \trianglelefteq \langle p + \langle e_{0} \rangle \rangle \times \cdots \times \langle p + \langle e_{P} \rangle \rangle\,.
    \end{equation}
    But note that these product ideals are in the image of $\psi: S_{\mathcal{B}} \longrightarrow S_{B^{0}} \times \cdots \times S_{B^{P}}$, and since $\phi(S_{\mathcal{B}}) = S_{\mathcal{B}}/I_{\tilde{\mathcal{U}}} \cong \psi(S_{\mathcal{B}})$ we have that
    \begin{equation}
        \psi(\Delta') \trianglelefteq \psi(p) \Leftrightarrow \langle \Dphys \rangle \trianglelefteq \langle p + I_{\tilde{\mathcal{U}}} \rangle \Leftrightarrow p + I_{\tilde{\mathcal{U}}}\, |\, \Dphys\,.
    \end{equation}
\end{proof}

This also shows that the subtleties in the factorization process highlighted in \cref{example:factorization-physical-polynomials} only arise for divisors that extend between components, while, e.g.\ the strict transform of the original $(-n)$-curve in a single infinite-distance limit horizontal model is not subject to them. In particular, this is true for all divisors appearing in the study of infinite-distance limits in the complex structure moduli space of eight-dimensional F-theory \cite{Lee:2021qkx,Lee:2021usk}, which are all points completely contained in a component.

\subsection{Monodromy cover}
\label{sec:monodromy-cover}

At the beginning of \cref{sec:extraction-codimension-one-information}, we recalled how non-abelian gauge algebras in F-theory arise from M2-branes wrapping the exceptional curves of the resolved fiber supported over the generic points of a given divisor, see \cref{sec:orders-of-vanishing} for their determination. The intersection matrix of the exceptional curves reproduces the Cartan matrix of the associated simply-laced Lie algebra, and the exceptional curves themselves correspond to the nodes of the appropriate ADE Dynkin diagram.\footnote{More precisely, the components of the fiber correspond to an affine ADE Dynkin diagram, with the additional node given by the fibral component intersecting the section of the fibration.} In eight-dimensional F-theory, this local analysis of the resolved elliptic fiber determines the gauge algebra associated to the divisor.

In F-theory models in six dimensions or fewer, global effects along the discriminant locus can modify this picture. Namely, the components of the resolved fiber may undergo monodromies that establish identifications among them; this corresponds to folding the ADE Dynkin diagram, meaning that non-simply-laced Lie algebras can arise as well. If this occurs or not can actually be determined directly in the singular Weierstrass model through Tate's algorithm \cite{Tate1975}, discussed in the context of F-theory in \cite{Bershadsky:1996nh,Katz:2011qp,Grassi:2011hq}.

Following the explanation of Tate's algorithm in \cite{Grassi:2011hq}, the monodromy can be described by means of a monodromy cover of the discriminant component under study. In practical terms, one studies a degree two or three polynomial involving an auxiliary variable, which is a meromorphic section of an appropriate line bundle over the discriminant component; the factorization properties of this polynomial inform us about the number of irreducible components of the monodromy cover, with the irreducible case indicating a folding of the ADE Dynkin diagram. Said polynomials are tabulated in \cite{Grassi:2011hq}.

As we have seen above, the gauge algebra information contained in the central fiber $Y_{0}$ of a resolved degeneration $\rho: \mathcal{Y} \rightarrow D$ is extracted by studying the irreducible components of $\{ \Delta' \}_{0 \leq p \leq P}$ consistently glued along the base components $\{ B^{p} \}_{0 \leq p \leq P}$. One such consistent gluing supporting a non-abelian gauge algebra factor corresponds to the factorizations of (the defining polynomial of) a divisor $\mathcal{D}_{\mathrm{phys}}$ in the physical discriminant $\Dphys$. Since the gauge algebra is associated to $\mathcal{D}_{\mathrm{phys}}$ taken as a whole, rather than to individual components that could suffer local gauge enhancements, a reduction of the gauge rank in a given component $B^{p}$ affects the gauge algebra globally read for $\mathcal{D}_{\mathrm{phys}}$. Hence, we study the monodromy cover in the conventional way for each of the irreducible components of $\{ \left. \mathcal{D}_{\mathrm{phys}} \right|_{B^{p}} \}_{0 \leq p \leq P}$. If we find that it is irreducible in one component, we have a monodromy action locally folding the Dynkin diagram and the gauge rank associated to $\mathcal{D}_{\mathrm{phys}}$ is reduced. If, on the contrary, the monodromy cover is split in all components, then it is split globally, and we assign to $\mathcal{D}_{\mathrm{phys}}$ the gauge algebra corresponding to the unfolded Dynkin diagram.

\subsection{Algorithm to read off the codimension-one gauge algebra}
\label{sec:algorithm-gauge-algebra}

Summarizing the discussion of this section, let us give a practical algorithm to read the gauge algebra associated to the central fiber of a degeneration $\hat{\rho}: \hat{\mathcal{Y}} \rightarrow D$.
\begin{enumerate}
    \item Follow the procedures described in \cref{sec:geometric-description-6D-F-theory-limits} to arrive at a resolved degeneration $\rho: \mathcal{Y} \rightarrow D$ free of obscured infinite-distance limits.

    \item Compute the restrictions $\left(f_{p},g_{p},\Delta'_{p}\right)$ of the defining polynomials and the modified discriminant. List the irreducible components $\{ {\Delta'}_{p,i_{p}} \}_{0 \leq i_{p} \leq I_{p}}$ for each of the component discriminants $\{ \Delta'_{p} \}_{0 \leq p \leq P}$.
    
    \item Consistently glue the $\{ {\Delta'}_{p,i_{p}} \}_{0 \leq p \leq P}^{0 \leq i_{p} \leq I_{p}}$ together into global divisors $\Delta_{\mathrm{phys}}^{i}$ of $B_{0}$ appearing as factors of $\Dphys$. If a local irreducible divisor has $\Delta'_{p,i_{p}} \cdot E_{q}|_{E_{p}} = 0$ for all $p \neq q \in \{ 0, \dotsc, P \}$, it does not extend to the adjacent components and directly captures the global information.

    \item Compute the physical vanishing orders $\ord{Y_{0}}(\fphys,\gphys,\Dphys)_{\Delta_{\mathrm{phys}}^{i}}$ to determine from the Kodaira-N\'eron classification the simply-laced covering gauge algebra $\tilde{\mathfrak{g}}_{i}$ supported on the global divisor $\Delta_{\mathrm{phys}}^{i}$.

    \item Determine the splitting properties of the monodromy covers associated to the restrictions $\{ \left. \Delta_{\mathrm{phys}}^{i} \right|_{B^{p}} \}_{0 \leq p \leq P}$, and assign to $\Delta_{\mathrm{phys}}^{i}$ the subalgebra $\mathfrak{g}_{i}$ of $\tilde{\mathfrak{g}}_{i}$ left invariant by the collection of inferred outer automorphisms.
\end{enumerate}

We stress again that some of the gauge algebra factors enhance into a higher algebra in the infinite-distance limits which correspond to decompactification limits. The above algebra is hence the algebra prior to taking this effect into account, which will be the subject of  \cite{ALWPart2}.

\subsection{Special fibers at the intersection of components}

The resolved degeneration $\rho: \mathcal{Y} \rightarrow D$ associated with a Class~1--4 degeneration $\hat{\rho}: \hat{\mathcal{Y}} \rightarrow D$ has a central fiber $Y_{0}$ in which the generic fibers over the base components $\{ B^{p} \}_{0 \leq p \leq P}$ can only be of Kodaira type $\mathrm{I}_{m}$. Moreover, two components $Y^{p}$ and $Y^{q}$ presenting codimension-zero $\mathrm{I}_{m}$ and $\mathrm{I}_{m'}$ fibers, respectively, intersect on an elliptically fibered surface $Y^{p} \cap Y^{q}$ with codimension-zero $\mathrm{I}_{m''}$ fibers, since intersections of other types would mean that $B^{p} \cap B^{q}$ is an obscured Class~5 curve, as discussed in \cref{sec:class-1-5-models}.

Still within Class~1--4 models, it may occur that $m'' > m+m'$. In this case, the component vanishing orders, computed in the $Y^{p}$ or the $Y^{q}$ component, indicate that the curve $B^{p} \cap B^{q}$ supports type $\mathrm{I}_{m''-m-m'}$ fibers. These special fibers located at the intersection of components make it slightly more ambiguous to determine the gauge algebra content, since they actually do not correspond to gauge enhancements. Let us argue why this is the case and how they can be removed through a series of transformations.

First, the $\mathrm{I}_{m''-m-m'}$ fibers over $B^{p} \cap B^{q}$ may be found at the level of the components $\{ Y^{p} \}_{0 \leq p \leq P}$, while absent for the family variety. This mismatch between component and family vanishing orders is analysed in our discussion of obscured infinite-distance limits in \cref{sec:obscured-infinite-distance-limits}; the two notions of vanishing orders can be made to agree by performing a base change with high enough branching degree.

Let us therefore assume that the family variety $\mathcal{Y}$ also presents $\mathrm{I}_{m''-m-m'}$ fibers over the curve $B^{p} \cap B^{q}$. By blowing the model down and performing a base change the resolution process demands additional base blow-ups. This gives rise to extra components in the central fiber $Y_{0}$ of the resolved degeneration, in which the former special fibers at the intersection of components are now the codimension-zero fibers. Hence, said special fibers did not correspond to information about gauge algebra enhancements, but about the background value of the axio-dilaton, which can be made explicit through an appropriate base change. This does not imply, however, that the new components may not present codimension-one enhancements. The analogous problem for degenerations of eight-dimensional F-theory models was analysed in \cite{Lee:2021qkx,Lee:2021usk}.

A base change with a high enough branching degree can both equate the component and family vanishing orders and remove all special fibers at the intersections of components, leading to a geometrical representative of the central fiber suitable to extract the physical information.
%auto-ignore

\section{Conclusions and future work}

In this first part of our analysis of non-minimal elliptic threefolds we have given a geometric account of the degenerations in which a family of Weierstrass fibrations specialises to a model exhibiting non-Kodaira singularities in codimension-one. Such degenerations are of interest because they encode a subclass of infinite-distance limits in the complex structure moduli space of the elliptic threefold. For instance, in F-theory they admit an interpretation as deformations in the non-perturbative open moduli space at infinite distance. The geometry of the degenerations studied in this work hence offers an entrance point to understanding the asymptotic physics along such trajectories in the moduli space. We will capitalise on this point of view in the second part \cite{ALWPart2} of our analysis.

Our goal has been to establish a concrete geometric picture of the described complex structure degenerations. To this end, we have studied the resolutions of the infinite-distance degenerations. They give rise, at the endpoint of the limit, to a reducible elliptic threefold free of non-minimalities that consists of intersecting log Calabi-Yau components. More precisely, the central element of the resolved degeneration is an elliptic fibration whose base space factors into, generally, a tree of intersecting surfaces. As we have discussed in detail, for a subclass of non-minimal configurations, the blow-up geometry forms an open chain: This is guaranteed to occur for so-called single infinite-distance limits. In these there are no intersections between different curves supporting non-minimal elliptic fibers, and there are no infinite-distance non-minimal singularities over points in the base. Such degenerations are equivalent, up to base changes and modifications, to configurations with non-minimal fibers over a single curve, which explains their name. 

The curves over which non-minimal elliptic fibers can occur are very constrained: They can either be of genus zero or of genus one, and in the latter case they must lie in the anti-canonical class of the base. We have focused in this work on the rich class of genus-zero non-minimal curves. For these, the blow-up of the base that removes the non-minimalities gives rise to Hirzebruch surfaces, whose types we have specified. 

At a slightly technical level, a special role is played by non-Kodaira singularities which are over-non-minimal, in the sense that the sections $f$ and $g$ of the Weierstrass model both vanish to orders strictly larger than the boundary values $4$ and $6$, respectively. The blow-ups required to cure these so-called Class~5 singularities (in the terminology of \cref{def:class-1-5}) give rise to degenerations that are not semi-stable. As we have discussed, the Semi-stable Reduction Theorem hence guarantees that degenerations presenting such singularities can be modified, possibly after a base change, into ones only exhibiting either minimal singularities or non-minimal singularities of Class~1--4. However, finding the required sequence of transformations may be very non-trivial in concrete cases. In \cite{ALWClass5} we explicitly determine these for non-minimal degenerations of elliptic fibrations over Hirzebruch base spaces. These results are also interesting for the analogous classification of degenerations of elliptic K3 surfaces into Kulikov models: Our analysis in \cite{ALWClass5} shows that, in such cases, the models presenting Class~5 non-minimalities can always be transformed to give rise to Kulikov models of Kulikov Type I, which lie at finite distance, or their infinite-distance counterparts of Kulikov Type II.a in the notation of \cite{Clingher:2003ui}, but not of Type III. In particular, it is not true that all non-minimal singularities lie at infinite distance\,---\,some of the Class~5 singularities are, in fact, equivalent to standard Kodaira singularities, after performing these transformations.

We have seen in this work that the non-minimal singularities of elliptic threefolds enjoy a very rich systematics. In fact, there are a number of important questions left for future work: First, it would be desirable to understand non-minimalities over curves in the anti-canonical class in a similarly detailed fashion, expanding on the observations and remarks in this direction made in \cref{sec:comments-genus-one-degenerations}. Second, one should also methodically study degenerations enhancing over the intersection locus of two or several curves. Most interestingly, infinite-distance non-minimal fibers can also occur entirely in such codimension-two loci, i.e.\ over isolated points on the base. Studying the systematics of their blow-up resolutions is an obvious direction that we are planning to return to in the future.

But even the single infinite-distance limits associated with codimension-one non-minimalities lead to a rather intricate structure, especially once viewed through the lens of F-theory. We have already highlighted subtleties in determining the components of the discriminant locus, which is a prerequisite to identify the gauge algebra of the effective theory. As stressed several times, however, this is only the beginning of a more involved analysis of the asymptotic physics. This analysis will appear as Part II of this work \cite{ALWPart2}. Inspired by the analogous problem in \cite{Lee:2021qkx,Lee:2021usk}, we will interpret the factorisation of the compactification space (the central fiber of the resolved degeneration) as indicating that the effective theory generically undergoes a partial decompactification (at least in a dual frame). On top of this, there can occur regions of weak coupling, associated with those components of the base over which the generic elliptic fiber degenerates to a Kodaira type $\mathrm{I}_{m>0}$ fiber. In combination, these two effects either result in decompactification or, possibly, global weak coupling limits. This intuition can be made particularly precise for the special subclass of models whose base geometries are Hirzebruch surfaces. The duality with the heterotic string then offers a welcome entrance point to the physics of the infinite-distance limits. Even in this class of models, however, we will find novel effects not present in eight dimensions \cite{Lee:2021qkx,Lee:2021usk}: The asymptotic theory generally contains defects which break the higher-dimensional Poincar\'e symmetry. With the exception of those factors localised in the defects, the naive gauge algebra undergoes certain gauge enhancements in the partial decompactification process. The appearance of defects is a slight twist on, but generally in agreement with the Emergent String Conjecture. Similar effects have been observed recently in different setups in \cite{Etheredge:2023odp}. The geometric analysis of non-minimal degenerations hence sheds interesting light on how string theory probes geometry near the boundaries of the moduli space.
%auto-ignore

\subsection*{\texorpdfstring{\textbf{Acknowledgements}}{Acknowledgements}}

We thank Hans-Christian von Bothmer, Vicente Cort\'es, Antonella Grassi, Martijn Kool and Helge Ruddat for useful discussions. R.\,A.-G.\ and T.\,W.\ are supported in part by Deutsche Forschungsgemeinschaft under Germany's Excellence Strategy EXC 2121  Quantum Universe 390833306 and by Deutsche Forschungsgemeinschaft through a German-Israeli Project Cooperation (DIP) grant ``Holography and the Swampland”. The work of S.-J.\,L.\ is supported by IBS under the project code IBS-R018-D1.

\appendix

%auto-ignore

\section{Six-dimensional F-theory bases}
\label{sec:six-dimensional-F-theory-bases}

In this appendix, we review the base spaces that can occur for elliptic Calabi-Yau threefolds, given by \cite{Grassi1991}
\begin{enumerate}
	\item $B = \mathbb{P}^{2}$, the complex projective plane;
	\item $B = \mathbb{F}_{n}$, the Hirzebruch surfaces with $0 \leq n \leq 12$; and
	\item $\mathrm{Bl}(\mathbb{P}^{2})$ and $\mathrm{Bl}(\mathbb{F}_{n})$, arbitrary blow-ups of the previous two.
\end{enumerate}

While the first two possibilities are very concrete, the third case encompasses a wealth of geometries, due to the many ways in which a surface can be blown up. The complex projective plane $\mathbb{P}^{2}$ and the Hirzebruch surfaces $\mathbb{F}_{n}$ with $n \neq 1$ are the minimal surfaces obtained by repeated application of Castelnuovo's contraction theorem in the class of surfaces that can be F-theory bases, hence their simplicity.

After first reviewing well-known properties of $\mathbb{P}^{2}$ and $\mathbb{F}_{n}$, we recall some basics of the blow-ups of algebraic surfaces at points and apply these facts to blow-ups of $\mathbb P^2$ and Hirzebruch surfaces.
The material of this appendix also serves as a preparation for \cref{sec:genus-restriction-proof}, where the genus-zero base curves that can support non-minimal singular elliptic fibers are analysed.

\subsection{\texorpdfstring{$\mathbb{P}^{2}$}{P2} and \texorpdfstring{$\mathbb{F}_{n}$}{Fn}}
\label{sec:P2-and-Fn}

The complex projective plane $\mathbb{P}^{2}$ can be described using the coordinates $[z_{1}:z_{2}:z_{3}]$ homogenous under the $\mathbb{C}^{*}$-action. $\mathbb{P}^{2}$ is a toric variety with fan $\Sigma_{\mathbb{P}^{2}}$ given by the edges
\begin{equation}
	z_{1} = (1,0)\,,\quad z_{2} = (0,1)\,,\quad z_{3} = (-1,-1)\,,
\end{equation}
in the lattice
\begin{equation}
    N := \langle (1,0),\, (0,1) \rangle_{\mathbb{Z}}\,.
\end{equation}
Its Picard group is
\begin{equation}
	\mathrm{Pic}\left( \mathbb{P}^{2} \right) = \langle H \rangle_{\mathbb{Z}}\,,
\end{equation}
where $H$ denotes the hyperplane class, with self-intersection
\begin{equation}
	H \cdot H = 1\,.
\end{equation}
The anticanonical class of $\mathbb{P}^{2}$ is given by
\begin{equation}
	\overline{K}_{\mathbb{P}^{2}} = 3H\,.
\end{equation}

Next, we centre our attention on $\mathbb{F}_{n}$, that is used as the base in most of our explicit examples. A Hirzebruch surface is a $\mathbb{P}^{1}$-bundle obtained from the projectivization of rank 2 vector bundles over $\mathbb{P}^{1}$, which can always be written as
\begin{equation}
    \mathbb{F}_{n} := \mathbb{P}(\pi: \mathcal{O}_{\mathbb{P}^{1}} \oplus \mathcal{O}_{\mathbb{P}^{1}}(n) \longrightarrow \mathbb{P}^{1})\,.
\label{eq:Hirzebruch-surface-definition}
\end{equation}
We can take $n \geq 0$, since $\mathbb{F}_{n} \simeq \mathbb{F}_{-n}$ due to the invariance of the projectivization of a bundle under twists by Abelian line bundles. The Picard group of $\mathbb{F}_{n}$ is
\begin{equation}
    \mathrm{Pic}(\mathbb{F}_{n}) = \langle h, f\rangle_{\mathbb{Z}}\,,
\end{equation}
where $h$ is the class of the zero section and $f$ is the class of a fiber. Their intersection products are given by
\begin{equation}
    h \cdot h = -n\,,\quad h \cdot f = 1\,, \quad f \cdot f = 0\,.
\end{equation}
Apart from the $(-n)$-curve $h$ coming from the sub-bundle $\mathcal{O}_{\mathbb{P}^{1}}$, there exists another independent section associated with the sub-bundle $\mathcal{O}_{\mathbb{P}^{1}}(n)$. This is the $(+n)$-curve, which we will denote by $C_{\infty}$ (using then also the notation $C_{0} := h$). Unlike the rigid curve $C_{0}$, the curve $C_{\infty}$ moves in an $n$-dimensional linear system, with any two representatives meeting in $n$ points. Using the linear equivalence
\begin{equation}
    C_{\infty} = h + nf\,,
\end{equation}
we obtain the intersection products
\begin{equation}
    C_{\infty} \cdot C_{\infty} = n\,,\quad C_{\infty} \cdot C_{0} = 0\,,\quad C_{\infty} \cdot f = 1\,.
\end{equation}
The anticanonical class of the Hirzebruch surface $\mathbb{F}_{n}$ is
\begin{equation}
    \overline{K}_{\mathbb{F}_{n}} = 2h + (2+n)f\,.
\end{equation}
Let us denote the fibral $\mathbb{P}^{1}$ by $\mathbb{P}^{1}_{f}$ and the base one by $\mathbb{P}^{1}_{b}$. We then introduce the homogeneous coordinates $[s:t]$ for $\mathbb{P}^{1}_{f}$ and $[v:w]$ for $\mathbb{P}^{1}_{b}$, with weights
\begin{equation}
    \begin{tblr}{cells={c},
    			hline{1}={2-5}{solid},
    			hline{2,3,4}={solid},
    		      vline{1}={2-3}{solid},
    			vline{2,3,4,5,6}={solid}
    			}
        & s & t & v & w\\
        \mathbb{C}^{*}_{\lambda_{1}} & 1 & 1 & 0 & 0\\
        \mathbb{C}^{*}_{\lambda_{2}} & 0 & n & 1 & 1
    \end{tblr}
\label{eq:Cstar-action-weights}
\end{equation}
under the two $\mathbb{C}^{*}$-actions of $\mathbb{F}_{n}$. Given a set of polynomials $\{f_{1}, \dotsc, f_{r}\}$, let us refer to the vanishing locus of (the ideal generated by) them simply by $\{ f_{1} = \cdots = f_{r} = 0\}$. With this notation, the coordinate divisors correspond to
\begin{equation}
    S := \{s = 0\} = C_{0}\,, \quad T := \{t = 0\} = C_{\infty}\,,\quad V := \{v=0\} = f = \{w=0\} =: W\,.
\label{eq:Hirzebruch-toric-divisors}
\end{equation}
$\mathbb{F}_{n}$ is a toric variety with fan $\Sigma_{\mathbb{F}_{n}}$ given by the edges
\begin{equation}
	v = (1,0)\,,\quad t = (0,1)\,,\quad w = (-1,-n)\,,\quad s = (0,-1)\,,
\end{equation}
in the lattice
\begin{equation}
    N := \langle (1,0),\, (0,1) \rangle_{\mathbb{Z}}\,.
\end{equation}

All other possible base surfaces are  blow-ups of $\mathbb{F}_{n}$. The ways in which we can blow-up a Hirzebruch surface are numerous, and we relegate a discussion of the resulting geometries and those properties of them relevant to our analysis to \cref{sec:list-arbitrary-blow-ups}. Note that, although $\mathbb{F}_{n}$ itself is toric, its blow-ups are not (in general) toric varieties. Additionally, some blow-ups may lead to surfaces with a non-effective anticanonical class, which would not constitute a valid six-dimensional F-theory base and should therefore be discarded.

Out of the non-trivial F-theory bases, the Hirzebruch surfaces $\mathbb{F}_{n}$ correspond to F-theory models with $n_{T} = 1$ tensors. The models over $\mathrm{Bl} \left( \mathbb{F}_{n} \right)$, containing $n_{T} > 1$ tensors, are closely related to those over $\mathbb{F}_{n}$. Namely, blowing down the exceptional divisors of $\mathrm{Bl} \left( \mathbb{F}_{n} \right)$ leads to a Weierstrass model over $\mathbb{F}_{n}$ with codimension-two finite-distance vanishing orders. This operation physically corresponds to going to the origin of the tensor branch, which takes us a finite distance away in moduli space from the original model. The presence of this type of singularities signals the existence of a strongly coupled six-dimensional SCFT sector, see \cite{Heckman:2018jxk} for a review.

Finally, the Weierstrass models over $\mathbb{P}^{2}$ correspond to F-theory models with $n_{T} = 0$ tensors. If they present at least one finite-distance non-minimal codimension-two singularity, we can blow it up in order to turn them into Weierstrass models over $\mathbb{F}_{1}$, by virtue of the isomorphism $\mathrm{Bl}_{1}(\mathbb{P}^{2}) = \mathrm{dP}_{1} \cong \mathbb{F}_{1}$. In the absence of such a singularity, tuning one would correspond to traversing a finite distance in moduli space. The geometrical connection between models over $\mathbb{F}_{1}$ and models over $\mathbb{P}^{2}$ means that their physics from the point of view of F-theory is also related, with an E-string wrapping the exceptional curve becoming light during the transition from the former set of models to the latter.

The connectedness of the six-dimensional F-theory moduli space under tensionless string transitions mirrors the mathematical minimal surface program. Obtaining a minimal model for a smooth surface by repeated application of Castelnuovo's contraction theorem corresponds in F-theory to moving to the origin of the tensor branch.

Returning to the degenerations of six-dimensional F-theory models discussed in \cref{sec:definition-of-degenerations}, assume that $\hat{B}$ in $\hat{\mathcal{B}} = \hat{B} \times D$ is not a minimal surface. If $C$ is a $(-1)$-curve in $\hat{B}$ such that $\pi^{*}_{\hat{B}}(C)$ exhibits infinite-distance non-minimal vanishing orders, the degeneration obtained by contracting $C$ to a point in $\hat{B}$ (hence $\pi^{*}_{\hat{B}}(C)$ to a curve in $\hat{\mathcal{B}}$) will present codimension-two infinite-distance non-minimal fibral singularities, beyond the codimension-two finite-distance non-minimal fibral singularities usually associated to the contraction of such a curve. This does not mean that codimension-one and codimension-two degenerations can always be connected in this way, since there exist models over the minimal model of the surface presenting the codimension-two non-minimal fibral singularities in the absence of the finite-distance ones.

In the explicit examples that we analyse both here and in \cite{ALWPart2} the birational transformations necessary to arrive at the resolved degeneration $\rho: \mathcal{Y} \rightarrow D$ clearly commute with the blow-ups needed to remove the finite-distance non-minimal fibral singularities; we therefore keep them unresolved in order to simplify the exposition. The discussion is nonetheless maintained general throughout, and the tools we provide apply to any of the allowed six-dimensional F-theory bases.

\subsection{Blow-ups of algebraic surfaces}

Given an algebraic surface $B$, we can blow it up with centre a point $p \in B$ by a local procedure to yield the blown up surface $\hat{B}$. This operation usually appears in the context of the resolution of singularities, but it can also be applied to smooth varieties. In this section, we recall how this blow-up process works in order to set the notation for the rest of the discussion. The properties of blow-ups, both for surfaces and varieties of other dimensionalities, are covered in most algebraic geometry textbooks, see e.g.\ \cite{hartshorne1977algebraic,griffiths2014principles,beauville1996,lazarsfeld2004positivity}. Before we start, let us already make some notational remarks.
\begin{remark}
	Let $B$ be an algebraic surface, that we will assume throughout to be smooth. When we speak of a curve $C \subset B$, we will always mean an effective divisor in $B$. The Picard group of $B$ will be referred to as $\mathrm{Pic}(B)$, while we will use the notation $\mathrm{NS}(B)$ for the N\'eron-Severi group. We will denote the effective cone of divisors by $\mathrm{Eff}(B)$, and its closure, the pseudoeffective cone, by $\overline{\mathrm{Eff}}(B)$. Since we are working with surfaces, the effective cone and the Mori cone, also known as the cone of curves, are coinciding notions.
\end{remark}

We will make extensive use of the following two results for irreducible curves on surfaces.
\begin{proposition}
\label{prop:negative-intersection}
    Let $B$ be an algebraic surface, and let $C \subset B$ be an irreducible curve. The intersection product $C \cdot C'$, where $C' \subset B$ is an arbitrary curve, is negative if and only if $C'$ contains $C$ as a component and $C \cdot C < 0$.
\end{proposition}
\begin{proposition}[Adjunction formula]
\label{prop:adjunction-genus}
    Let $C$ be a smooth, irreducible curve on an algebraic surface $B$, and denote its genus by $g(C)$. Then, we have the identity
    \begin{equation}
        C \cdot \left( \overline{K}_{M} - C \right) = 2-2g(C)\,.
    \end{equation}
\end{proposition}

Next, we review the definition of the blow-up of a surface with centre a point, and its associated exceptional divisor, as well as the concepts of strict and total transforms of the curves in the original surface.
\begin{definition}
    Let $B$ be an algebraic surface and $p \in B$ a point. Then, there exists a surface $\hat{B}$ and a morphism $\pi: \hat{B} \rightarrow B$, which are unique up to isomorphism, such that
    \begin{enumerate}
        \item the restriction of $\pi$ to $\pi^{-1}(M \setminus \{p\})$ is an isomorphism onto $M \setminus \{p\}$; and
        \item $\pi^{-1}(p) =: E$ is isomorphic to $\mathbb{P}^{1}$.
    \end{enumerate}
    We call $\pi$ the blow-up of $M$ with centre $p$, and $E$ the exceptional divisor of the blow-up.
\end{definition}
\begin{lemma}
\label{lemma:strict-proper-transform}
    Let $C \subset B$ be an irreducible curve that passes through $p$ with multiplicity $m$. The closure of $\pi^{-1}(C-{p})$ in $\hat{B}$ is an irreducible curve $\hat{C} \subset \hat{B}$ satisfying
    \begin{equation}
        \pi^{*}(C) = \hat{C} + mE\,.
    \end{equation}
    We call $\hat{C}$ the strict transform of $C$ and $\pi^{*}(C)$ the proper or total transform of $C$.
\end{lemma}
\begin{corollary}
\label{corollary:exceptional-intersection}
    With the same hypotheses, $\hat{C} \cdot_{\hat{B}} E = m$.
\end{corollary}
The next result characterizes many of the properties of the blown up surface $\hat{B}$ in terms of the analogous ones for the original surface $B$.
\begin{proposition}
\label{prop:blow-up-properties}
    Let $B$ be an algebraic surface, $\pi: \hat{B} \rightarrow B$ the blow-up of $B$ at a point $p \in B$, and $E \subset \hat{B}$ the exceptional divisor of the blow-up.
    \begin{enumerate}
        \item There is an isomorphism
        \begin{equation}
            \begin{aligned}
                    \sim: \mathrm{Pic}\, (B) \oplus \mathbb{Z} &\longrightarrow \rm{Pic}\, (\hat{B})\\
                    (D,n) &\longmapsto \pi^{*}(D) + nE\,.
            \end{aligned}
        \end{equation}
        
        \item Let $D$, $D'$ be divisors on $B$. Then
        \begin{equation}
            \pi^{*}(D) \cdot_{\hat{B}} \pi^{*}(D') = D \cdot_{B} D'\,, \quad E \cdot_{\hat{B}} \pi^{*}(D) = 0\,, \quad E \cdot_{\hat{B}} E = -1\,.
        \end{equation}
        
        \item $\rm{NS}(\hat{B}) \cong \rm{NS}(B) \oplus \langle E \rangle_{\mathbb{Z}}$.
        
        \item $K_{\hat{B}} = \pi^{*} \left( K_{B} \right) + E$.
    \end{enumerate}
\end{proposition}

\subsection{Arbitrary blow-ups of \texorpdfstring{$\mathbb{P}^{2}$}{P2} and \texorpdfstring{$\mathbb{F}_{n}$}{Fn}}
\label{sec:list-arbitrary-blow-ups}

By blowing up points in $\mathbb{P}^{2}$ and $\mathbb{F}_{n}$ (with $0 \leq n \leq 12$) we can produce a plethora of surfaces that can act as F-theory bases, due to the many ways in which we can choose the positions of the blow-up centres.

The simplest such blow-ups are the so-called del Pezzo surfaces $\mathrm{dP}_{k}$ (the two-dimensional Fano varieties), obtained by blowing up $\mathbb{P}^{2}$ in up to 8 points in general position, which we denote by $\mathrm{Bl}_{k}(\mathbb{P}^{2})$. If one blows up $\mathbb{P}^{2}$ at $k \geq 9$ points in general position, the Mori cone is no longer finitely generated, despite the finite generation of $\mathrm{NS} \left( \mathrm{Bl}_{k}(\mathbb{P}^{2}) \right) \cong \mathrm{Pic} \left( \mathrm{Bl}_{k}(\mathbb{P}^{2}) \right)$. A lot of information on blow-ups of $\mathbb{P}^{2}$ at generic points as well as the generators of the Mori cones of the del Pezzo surfaces can be found in \cite{Rosoff1980}.

But instead of choosing points in general position on the original surface, we can make other choices. For example, we may blow up a point in an exceptional divisor resulting from a previous blow-up. In what follows, we would like to analyse what these choices are and extract those features of these arbitrary blow-ups of $\mathbb{P}^{2}$ and $\mathbb{F}_{n}$ that are relevant for our analysis.

The first thing to be noted is that we can simply worry about the arbitrary blow-ups of $\mathbb{F}_{n}$, since this encompasses the arbitrary blow-ups of $\mathbb{P}^{2}$ as well. This is due to the well-known fact that the blow-up of $\mathbb{P}^{2}$ at one point is $\mathrm{dP}_{1} \cong \mathbb{F}_{1}$, and therefore any further blow-ups can be regarded as blow-ups of $\mathbb{F}_{1}$.

In order to analyse the possible blow-ups we can perform, let us blow up the surfaces one point at a time. That is, for a surface that we have obtained by blowing up $K$ points we have the maps
\begin{equation}
	\pi: \hat{B}_{K} \xrightarrow{\mathmakebox[2em]{\pi_{K}}} \hat{B}_{K-1} \xrightarrow{\mathmakebox[2em]{\pi_{K-1}}} \cdots \xrightarrow{\mathmakebox[2em]{\pi_{2}}} \hat{B}_{1} \xrightarrow{\mathmakebox[2em]{\pi_{1}}} B\,,\qquad \pi = \pi_{1} \circ \cdots \circ \pi_{K}\,.
\label{eq:blow-up-path}
\end{equation}
After $i$ blow-ups have been performed, we have the choice of where to locate the point $p_{i+1}$ that will be blown up, leading to different ``blow-up paths" that we can take.

In order to illustrate this, let us take the first few blow-ups of $\mathbb{P}^{2}$ as an example, summarizing the discussion in \cref{fig:blow-up-paths-P2}. For the first blow-up, we can only choose an arbitrary point. In the next step we can choose to blow-up a point in $\hat{B}_{1} \setminus E_{1}$, a generic point in the exceptional divisor $E_{1}$, or the point in which $E_{1}$ intersects the original surface. This leads to different surfaces, as can be seen from their anticanonical class, which we compute in \cref{sec:anticanonical-class-after-arbitrary-blow-up}. Let us choose the intermediate option. Now we have the choice of blowing up a point in $\hat{B}_{2} \setminus (E_{1} \cup E_{2})$, a generic point in the (strict transform of) the first exceptional divisor $E_{1}$, a generic point in the second exceptional divisor $E_{2}$, or one of the two points of intersection that exist. In this way, the blow-up paths quickly branch.
\begin{figure}[p!]
	\centering
	\[
    \rotatebox{90}{%
    \scalebox{0.6375}{
	\begin{tikzcd}[ampersand replacement=\&, row sep={1em}, column sep={10em}]
		\begin{tabular}{|c|} \hline\\[-1.5ex] $B = \mathbb{P}^{2}$\\[0.75ex] $\langle H \rangle_{\mathbb{Z}}$\\[0.75ex] $H \cdot H = 1$\\[0.75ex] $\overline{K}_{B} = 3H$\\[-1.5ex] \\ \hline \end{tabular} \arrow[d, "\text{(A) at $H$}"] \& \& \begin{tabular}{|c|} \hline\\[-1.5ex] $\hat{B}_{3}$\\[0.75ex] $\langle \pi^{*} \left( H \right), E_{(1)}, E_{((1),1)}, E_{((1),2)} \rangle_{\mathbb{Z}}$\\[0.75ex] $\pi^{*} \left( H \right) \cdot \pi^{*} \left( H \right) = 1$\\[0.75ex] $E_{(1)} \cdot E_{(1)} = -3$\\[0.75ex] $E_{((1),1)} \cdot E_{((1),1)} = -1$\\[0.75ex] $E_{((1),2)} \cdot E_{((1),2)} = -1$\\[0.75ex] $E_{(1)} \cdot E_{((1),1)} = 1$\\[0.75ex] $E_{(1)} \cdot E_{((1),2)} = 1$\\[0.75ex] $\overline{K}_{\hat{B}_{3}} = 3\pi^{*} \left( H \right) - E_{(1)} - 2E_{((1),1)} - 2E_{((1),2)}$\\[-1.5ex] \\ \hline \end{tabular}\\
		\begin{tabular}{|c|} \hline\\[-1.5ex] $\hat{B}_{1} = \mathrm{dP}_{1} \cong \mathbb{F}_{1}$\\[0.75ex] $\langle \pi^{*}\left( H \right), E_{(1)} \rangle_{\mathbb{Z}}$\\[0.75ex] $\pi^{*}\left( H \right) \cdot \pi^{*}\left( H \right) = 1$\\[0.75ex] $E_{(1)} \cdot E_{(1)} = -1$\\[0.75ex] $\overline{K}_{\hat{B}_{1}} = 3\pi^{*} \left( H \right) - E_{(1)}$\\[-1.5ex] \\ \hline \end{tabular} \arrow[d, "\text{(A) at $\hat{B}_{1} \setminus E_{(1)}$}"] \arrow[r, "\text{(B) at $E_{(1)}$}"] \& \begin{tabular}{|c|} \hline\\[-1.5ex] $\hat{B}_{2}$\\[0.75ex] $\langle \pi^{*}\left( H \right), E_{(1)}, E_{((1),1)} \rangle_{\mathbb{Z}}$\\[0.75ex] $\pi^{*} \left( H \right) \cdot \pi^{*} \left( H \right) = 1$\\[0.75ex] $E_{(1)} \cdot E_{(1)} = -2$\\[0.75ex] $E_{((1),1)} \cdot E_{((1),1)} = -1$\\[0.75ex] $E_{(1)} \cdot E_{((1),1)} = 1$\\[0.75ex] $\overline{K}_{\hat{B}_{2}} = 3\pi^{*} \left( H \right) - E_{(1)} - 2E_{((1),1)}$\\[-1.5ex] \\ \hline \end{tabular} \arrow[d ,"\text{(A) at $\hat{B}_{2} \setminus \bigcup_{\alpha} E_{\alpha}$}"] \arrow[ru, "\text{(B) at $E_{(1)}$}"] \arrow[r, "\text{(C) at $E_{(1)} \cap E_{((1),1)}$}"] \arrow[rd, swap, "\text{(B) at $E_{(((1),1),1)}$}"] \& \begin{tabular}{|c|} \hline\\[-1.5ex] $\hat{B}_{3}$\\[0.75ex] $\langle \pi^{*} \left( H \right), E_{(1)}, E_{((1),1)}, E_{[(1),((1),1)]} \rangle_{\mathbb{Z}}$\\[0.75ex] $\pi^{*} \left( H \right) \cdot \pi^{*} \left( H \right) = 1$\\[0.75ex] $E_{(1)} \cdot E_{(1)} = -3$\\[0.75ex] $E_{((1),1)} \cdot E_{((1),1)} = -2$\\[0.75ex] $E_{[(1),((1),1)]} \cdot E_{[(1),((1),1)]} = -1$\\[0.75ex] $E_{(1)} \cdot E_{[(1),((1),1)]} = 1$\\[0.75ex] $E_{((1),1)} \cdot E_{[(1),((1),1)]} = 1$\\[0.75ex] $\overline{K}_{\hat{B}_{3}} = 3\pi^{*} \left( H \right) - E_{(1)} - 2E_{((1),1)} - 4E_{[(1),((1),1)]}$\\[-1.5ex] \\ \hline \end{tabular}\\
		\begin{tabular}{|c|} \hline\\[-1.5ex] $\hat{B}_{2} = \mathrm{dP}_{2}$\\[0.75ex] $\langle \pi^{*}\left( H \right), E_{(1)}, E_{(2)} \rangle_{\mathbb{Z}}$\\[0.75ex] $\pi^{*} \left( H \right) \cdot \pi^{*}\left( H \right) = 1$\\[0.75ex] $E_{(1)} \cdot E_{(1)} = -1$\\[0.75ex] $E_{(2)} \cdot E_{(2)} = -1$\\[0.75ex] $\overline{K}_{\hat{B}_{2}} = 3\pi^{*}\left( H \right) - E_{(1)} - E_{(2)}$\\[-1.5ex] \\ \hline \end{tabular} \arrow[d, "\text{(A) at $\hat{B}_{2} \setminus \bigcup_{i=1}^{2} E_{(i)}$}"] \arrow[r, "\text{(B) at $E_{(1)}$}"] \& \begin{tabular}{|c|} \hline\\[-1.5ex] $\hat{B}_{3}$\\[0.75ex] $\langle \pi^{*} \left( H \right), E_{(1)}, E_{(2)}, E_{((1),1)} \rangle_{\mathbb{Z}}$\\[0.75ex] $\pi^{*} \left( H \right) \cdot \pi^{*} \left( H \right) = 1$\\[0.75ex] $E_{(1)} \cdot E_{(1)} = -2$\\[0.75ex] $E_{(2)} \cdot E_{(2)} = -1$\\[0.75ex] $E_{((1),1)} \cdot E_{((1),1)} = -1$\\[0.75ex] $E_{(1)} \cdot E_{((1),1)} = 1$\\[0.75ex] $\overline{K}_{\hat{B}_{3}} = 3\pi^{*} \left( H \right) - E_{(1)} - E_{(2)} - 2E_{((1),1)}$\\[-1.5ex] \\ \hline \end{tabular} \& \begin{tabular}{|c|} \hline\\[-1.5ex] $\hat{B}_{3}$\\[0.75ex] $\langle \pi^{*} \left( H \right), E_{(1)}, E_{((1),1)}, E_{(((1),1),1)} \rangle_{\mathbb{Z}}$\\[0.75ex] $\pi^{*} \left( H \right) \cdot \pi^{*} \left( H \right) = 1$\\[0.75ex] $E_{(1)} \cdot E_{(1)} = -2$\\[0.75ex] $E_{((1),1)} \cdot E_{((1),1)} = -2$\\[0.75ex] $E_{(((1),1),1)} \cdot E_{(((1),1),1)} = -1$\\[0.75ex] $E_{(1)} \cdot E_{((1),1)} = 1$\\[0.75ex] $E_{((1),1)} \cdot E_{(((1),1),1)} = 1$\\[0.75ex] $\overline{K}_{\hat{B}_{3}} = 3\pi^{*} \left( H \right) - E_{(1)} - 2E_{((1),1)} - 3E_{(((1),1),1)}$\\[-1.5ex] \\ \hline \end{tabular}\\
		\begin{tabular}{|c|} \hline\\[-1.5ex] $\hat{B}_{3} = \mathrm{dP}_{3}$\\[0.75ex] $\langle \pi^{*}\left( H \right), E_{(1)}, E_{(2)}, E_{(3)} \rangle_{\mathbb{Z}}$\\[0.75ex] $\pi^{*}\left( H \right) \cdot \pi^{*} \left( H \right) = 1$\\[0.75ex] $E_{(1)} \cdot E_{(1)} = -1$\\[0.75ex] $E_{(2)} \cdot E_{(2)} = -1$\\[0.75ex] $E_{(3)} \cdot E_{(3)} = -1$\\[0.75ex] $\overline{K}_{\hat{B}_{3}} = 3\pi^{*}\left( H \right) - E_{(1)} - E_{(2)} - E_{(3)}$\\[-1.5ex] \\ \hline \end{tabular} \& \& \vphantom{\begin{tabular}{|c|} \hline\\[-1.5ex] $\hat{B}_{3}$\\[0.75ex] $\langle \pi^{*} \left( H \right), E_{(1)}, E_{((1),1)}, E_{(((1),1),1)} \rangle_{\mathbb{Z}}$\\[0.75ex] $\pi^{*} \left( H \right) \cdot \pi^{*} \left( H \right) = 1$\\[0.75ex] $E_{(1)} \cdot E_{(1)} = -2$\\[0.75ex] $E_{((1),1)} \cdot E_{((1),1)} = -2$\\[0.75ex] $E_{(((1),1),1)} \cdot E_{(((1),1),1)} = -1$\\[0.75ex] $E_{(1)} \cdot E_{((1),1)} = 1$\\[0.75ex] $E_{((1),1)} \cdot E_{(((1),1),1)} = 1$\\[0.75ex] $\overline{K}_{\hat{B}_{3}} = 3\pi^{*} \left( H \right) - E_{(1)} - 2E_{((1),1)} - 3E_{(((1),1),1)}$\\[-1.5ex] \\ \hline \end{tabular}}
	\end{tikzcd}
	}
    }
    \]
	\caption{Arbitrary blow-ups of $\mathbb{P}^{2}$ in up to three points. The intersection products not explicitly printed are vanishing.}
	\label{fig:blow-up-paths-P2}
\end{figure}
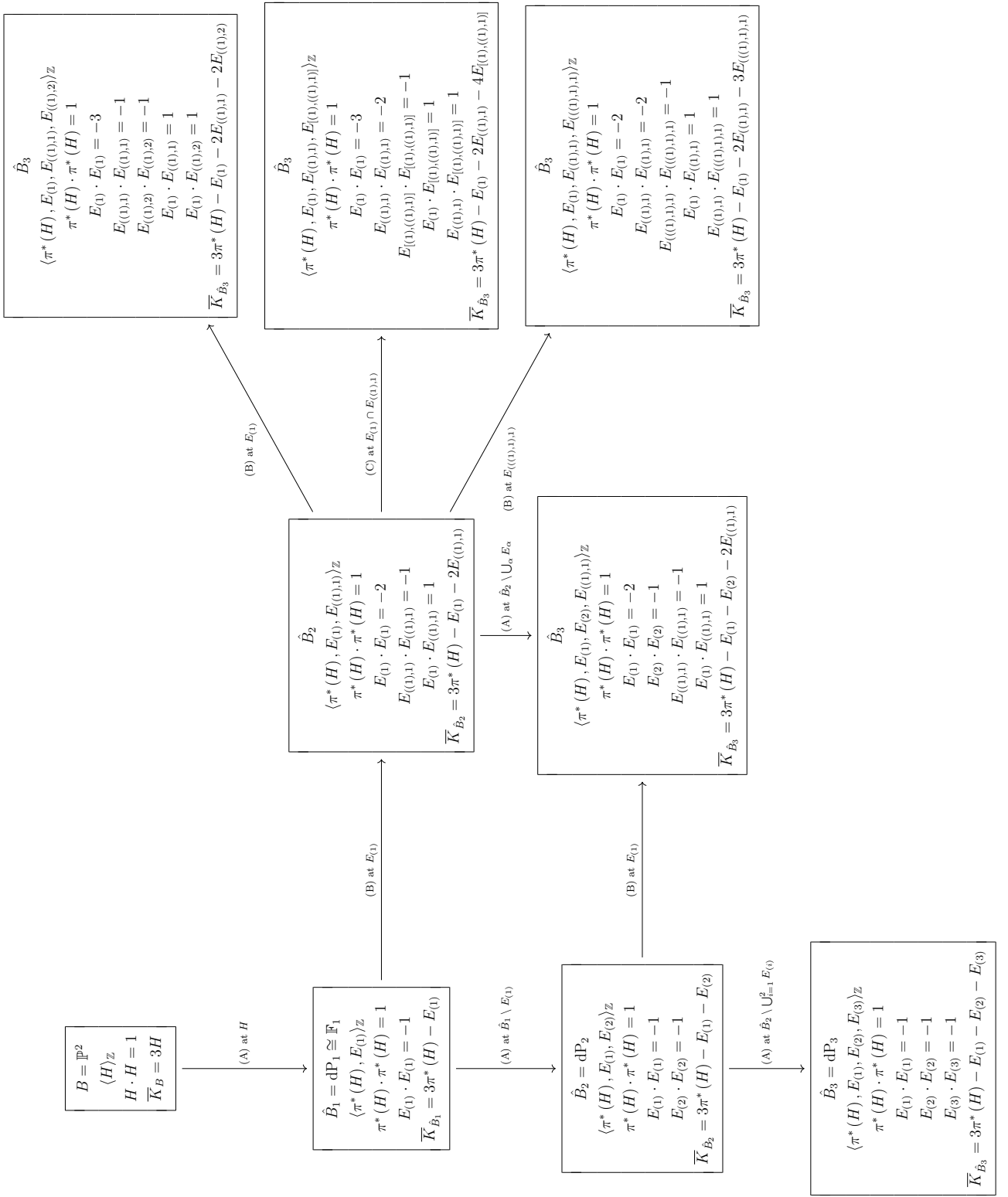
Note that, as shown in \cref{fig:blow-up-paths-P2}, two different blow-up paths can lead to the same surface; for example, the order in which we blow up general points in the original surface is not relevant.

The discussion above actually illustrates all the possible choices that we have at each step of the blow-up process, namely:
\begin{enumerate}[label=(\Alph*)]
	\item we can blow up a point in the original surface,
	
	\item we can blow up a generic point in the (strict transforms of) an exceptional divisor,
	
	\item we can blow up a point of intersection between the (strict transforms of) two exceptional divisors, or
	
	\item we can blow up a point of intersection between the original surface and (the strict transform of) an exceptional divisor.
\end{enumerate}
What we would like to do now is to analyse how, starting from $\mathbb{P}^{2}$ or $\mathbb{F}_{n}$, an arbitrary sequence of these individual blow-ups affects the anticanonical class of the surface and its intersection ring. This information will be used, in \cref{sec:genus-restriction-proof}, to restrict the genus of the smooth, irreducible curves over which non-minimal elliptic fibers of the type studied in the body of the work can be realized.

\subsubsection{Basis and notation for \texorpdfstring{$\mathrm{Bl}(\mathbb{P}^{2})$}{Bl(P2)} and \texorpdfstring{$\mathrm{Bl}(\mathbb{F}_{n})$}{Bl(Fn)}}

The surfaces from which we start the blow-up process are $\mathbb{P}^{2}$ and $\mathbb{F}_{n}$. For $\mathbb{P}^{2}$, we will denote the hyperplane class by $H$, which then forms a basis for both the Picard group and the effective cone
\begin{equation}
	\mathbb{P}^{2}:\qquad \mathrm{Pic} \left( \mathbb{P}^{2} \right) = \langle H \rangle_{\mathbb{Z}}\,,\qquad \overline{\mathrm{Eff}} \left( \mathbb{P}^{2} \right) = \langle H \rangle_{\mathbb{R}_{\geq 0}}\,.
\end{equation}
For $\mathbb{F}_{n}$, we use the notation introduced in \cref{sec:P2-and-Fn}, i.e.\ we denote by $h$ the section given by the $(-n)$-curve and by $f$ the fiber class, having then
\begin{equation}
	\mathbb{F}_{n}:\qquad \mathrm{Pic} \left( \mathbb{F}_{n} \right) = \langle h,f \rangle_{\mathbb{Z}}\,,\qquad \overline{\mathrm{Eff}} \left( \mathbb{F}_{n} \right) = \langle h,f \rangle_{\mathbb{R}_{\geq 0}}\,.
\end{equation}
Let us collectively denote the elements of these bases by $\{D_{i}\}_{i \in \mathcal{I}}$, with $\{ D_{1} \} = \{ H \}$ for $\mathbb{P}^{2}$ and $\{ D_{1}, D_{2} \} = \{ h,f \}$ for $\mathbb{F}_{n}$.

After each blow-up $\pi_{i}: \hat{B}_{i} \rightarrow \hat{B}_{i-1}$ in \eqref{eq:blow-up-path} we have a new exceptional irreducible divisor $E_{i}$. The collection of total transforms of the exceptional divisors stemming from the previous blow-ups $\{ \pi_{i-1}^{*} \circ \cdots \circ \pi_{1}^{*} \left( E_{j} \right) \}_{1 \leq j \leq i-1}$ may no longer be comprised of irreducible divisors. This will occur if some points in the exceptional divisors were blown up, in which case the total and strict transforms differ. To avoid confusion, we will always express the quantities at each step of the blow-up process in terms of the strict transforms of the exceptional divisors, which we will simply denote by $\{E_{1}, \dots, E_{i}\}$, dropping the hats.

Due to \cref{prop:blow-up-properties}, we know that at each step of the process a basis for the Picard group is given by
\begin{equation}
\begin{aligned}
	\mathrm{Pic} \big( \hat{B}_{i} \big) &= \langle \pi_{i-1}^{*} \circ \cdots \circ \pi_{1}^{*} \left( D_{j} \right), \pi_{i}^{*} \left( E_{1} \right), \dotsc, \pi_{i}^{*} \left( E_{i-1} \right), E_{i} \rangle_{\mathbb{Z}}\\
	&= \langle \pi_{i-1}^{*} \circ \cdots \circ \pi_{1}^{*} \left( D_{j} \right), E_{1}, \dotsc, E_{i} \rangle_{\mathbb{Z}}\,,
\end{aligned}
\end{equation}
where in the last line we have used the fact that $\pi_{i}^{*} \left( E_{j} \right)$ and $E_{j}$ differ by factors of $E_{i}$. Applying this inductively to the bases listed above, we can, after blowing-up $K$ points in $\mathbb{P}^{2}$ or $\mathbb{F}_{n}$, choose as basis for the Picard group
\begin{equation}
	\mathrm{Pic} \big( \hat{B}_{K} \big) = \langle \pi^{*} \left( D_{j} \right), E_{1}, \dotsc, E_{K} \rangle_{\mathbb{Z}}\,.
\label{eq:Picard-basis}
\end{equation}
These elements will not give a basis of the effective cone, the computation of which can be complicated \cite{Taylor:2015isa}. It will nonetheless be true that
\begin{equation}
	 \langle \pi^{*} \left( D_{j} \right), E_{1}, \dotsc, E_{K} \rangle_{\mathbb{Z}} \subset \overline{\mathrm{Eff}} \big( \hat{B}_{K} \big)\,.
\label{eq:Picard-basis-effective-cone}
\end{equation}
For our purposes using this basis is enough. We will express the quantities of interest to us using the total transforms of the original divisors $\{D_{i}\}_{i \in \mathcal{I}}$, which have trivial intersection with the $\{E_{j}\}_{1 \leq j \leq i}$ at each step. This makes the blow-up type (D) effectively type (B) from the computational point of view that matters in this section, since the properties of the classes $\{D_{i}\}_{i \in \mathcal{I}}$ that we employ in the Picard basis remain unaffected by it. We therefore omit this type of blow-up in the remainder of the section.

\subsubsection{Anticanonical class after an arbitrary blow-up}
\label{sec:anticanonical-class-after-arbitrary-blow-up}

The change in the anticanonical class of the surface after each blow-up $\pi_{i}: \hat{B}_{i} \rightarrow \hat{B}_{i-1}$ in \eqref{eq:blow-up-path} is given in \cref{prop:blow-up-properties}. We only need to compose these changes along the blow-up path. If at each step we have
\begin{equation}
	\overline{K}_{\hat{B}_{i}} = \pi_{i}^{*} \left( \overline{K}_{\hat{B}_{i-1}} \right) - E_{i}\,,
\end{equation}
the final anticanonical class after the $K$ blow-ups will be
\begin{equation}
	\overline{K}_{\hat{B}_{K}} = \pi^{*} \left( \overline{K}_{B} \right) - \sum_{i = 1}^{K-2} \pi_{K}^{*} \circ \cdots \circ \pi_{i+1}^{*} \left( E_{i} \right) - \pi_{K}^{*} \left( E_{K-1} \right) -  E_{K}\,.
\end{equation}
Expressing this in terms of the basis \eqref{eq:Picard-basis} we have
\begin{equation}
	\overline{K}_{\hat{B}_{K}} = \pi^{*} \left( \overline{K}_{B} \right) - \sum_{i = 1}^{K} d_{i} E_{i}\,,\quad d_{i} \in \mathbb{Z}_{\geq 0}\,,
\end{equation}
where the $d_{i}$ are known as the discrepancies.\footnote{Negative discrepancies do appear in the resolution process of log terminal and log canonical singularities, but here we are blowing up smooth points.}

Characterizing the value of the discrepancies $d_{i}$ is simple for the type of blow-ups that we are considering, i.e.\ those centred at smooth points of a surface. Let us do so by introducing a notation for the (strict transforms of) the exceptional divisors that we find useful, since it informs us about the ``blow-up history" of said divisor. In the following paragraph, we denote by $\pi: \hat{B} \rightarrow B$ the composition of all the blow-ups performed until that point in the discussion, and $\rho: \hat{B} \rightarrow \hat{B}_{\mathrm{old}}$ the last blow-up performed.

We start by considering all the blow-ups of type (A), i.e.\ the blow-ups of points in the original surface. Say that we perform $k$ said blow-ups. These will be characterized by specifying a collection $\{p_{i}\}_{1 \leq i \leq k}$ of points in $B$. Any order in which we perform these blow-ups leads to the same resulting surface, and we can therefore perform all such blow-ups in one step, leading to a collection $\{E_{i}\}_{1 \leq i \leq k}$ of exceptional divisors and a blown up surface $\hat{B}$ with anticanonical class
\begin{equation}
	\overline{K}_{\hat{B}} = \pi^{*} \left( \overline{K}_{B} \right) - \sum_{i=1}^{k} d_{i} E_{i}\,,\qquad d_{i} = 1\,,\quad \forall i = 1, \dotsc, k\,.
\label{eq:anticanonical-class-type-A-blow-up}
\end{equation}

We have performed all desired blow-ups of type (A). At this point of the blow-up process, there are no intersection points between the exceptional divisors, such that the only possi\-bility is to perform a blow-up of type (B), i.e.\ blowing up a point in one of the exceptional divisors. Say that we choose a collection of generic points $\{p_{(\alpha,i)}\}_{1 \leq i \leq k_{\alpha}}$ in the exceptional divisor \mbox{$E_{\alpha} \in \{E_{i}\}_{1 \leq i \leq k}$}. Performing the blow-up of these points, we have that
\begin{equation}
	\rho^{*} \left( E_{\alpha} \right) = E_{\alpha} + \sum_{i=1}^{k_{\alpha}} E_{(\alpha,i)}\,,
\end{equation}
leading to the anticanonical class
\begin{equation}
	\overline{K}_{\hat{B}} = \pi^{*} \left( \overline{K}_{B} \right) - \sum_{i=1}^{k} d_{i} E_{i} - \sum_{i=1}^{k_{\alpha}} d_{(\alpha,i)} E_{(\alpha,i)}\,,\qquad
	\begin{aligned}	
		d_{i} &= 1\,,\quad \forall i = 1, \dotsc, k\,, \\ d_{(\alpha,j)} &= d_{j} + 1\,,\quad \forall j = 1, \dotsc, k_{\alpha}\,.
	\end{aligned}
\end{equation}
The discrepancies $\{ d_{(\alpha,i)} \}_{1 \leq i \leq k_{\alpha}}$ of the exceptional divisors $\{ E_{\alpha,i} \}_{1 \leq i \leq k_{\alpha}}$ have increased by one with respect to the $d_{\alpha}$ of $E_{\alpha}$, owing to the fact that they are one level deeper in the blow-up chain. Let us continue by considering for now only blow-ups of type (B). We can choose now to blow-up a collection of generic points in $\{p_{(\beta,i)}\}_{1 \leq i \leq k_{\beta}}$ in an exceptional divisor $E_{\beta \neq \alpha} \in \{E_{i}\}_{1 \leq i \leq k}$. This would lead to a collection of exceptional divisors $\{E_{(\beta,i)}\}_{1 \leq i \leq k_{\beta}}$ at the second level in the blow-up chain, and appearing in the anticanonical class with discrepancies $d_{(\beta,i)} = 2, \forall i = 1, \dotsc, k_{\beta}$. Alternatively, and still within the blow-ups of type (B), we could blow-up a collection of generic points $\{p_{((\alpha,i),j)}\}_{1 \leq j \leq k_{(\alpha,i)}}$ in the exceptional divisor $E_{(\alpha,i)} \in \{E_{(\alpha,j)}\}_{0 \leq j \leq k_{\alpha}}$. This would lead to the exceptional divisors $\{E_{((\alpha,i),j)}\}_{1 \leq j \leq k_{(\alpha,i)}}$ at the third level in the blow-up chain, and appearing in the anticanonical class with discrepancies $d_{((\alpha,i),j)} = 3, \forall j = 1, \dotsc, k_{(\alpha,i)}$. Using this notation, in which the subindex of an exceptional divisor consists of parentheses with the left entry designating the divisor whose points are blown-up and the right entry listing the new exceptional divisors arising from said blow-up, and renaming the divisors in the first level to $\{E_{(i)}\}_{1 \leq i \leq k}$, the anticanonical class resulting from an arbitrary number of blow-ups of type (A) and (B) is
\begin{equation}
	\overline{K}_{\hat{B}} = \pi^{*} \left( K_{B} \right) - \sum_{\alpha} d_{\alpha} E_{\alpha}\,,\qquad d_{\alpha} = \text{level of the exceptional divisor} \geq 1\,.
\label{eq:anticanonical-class-type-B-blow-up}
\end{equation}
The subindex notation contains the ``blow-up history" of a given exceptional divisor, and the level can be computed by simply counting the number of parentheses pairs in the subindex.

Finally, we need to consider the possibility of performing blow-ups of type (C), i.e.\ blowing up the intersection point of two exceptional divisors. Let $E_{\alpha}$ and $E_{\beta}$ be two exceptional divisors with intersection product $E_{\alpha} \cdot E_{\beta} = 1$. Blowing up their intersection point, and denoting the resulting exceptional divisor by $E_{[\alpha,\beta]}$, we have that
\begin{subequations}
\begin{align}
	\rho^{*} \left( E_{\alpha} \right) = E_{\alpha} + E_{[\alpha,\beta]}\,,\\
	\rho^{*} \left( E_{\beta} \right) = E_{\beta} + E_{[\alpha,\beta]}\,,
\end{align}
\end{subequations}
leading to the anticanonical class
\begin{equation}
	\overline{K}_{\hat{B}} = \pi^{*} \left( K_{B} \right) - \sum_{\alpha} d_{\alpha} E_{\alpha} - d_{[\alpha,\beta]} E_{[\alpha,\beta]}\,,\qquad
	\begin{aligned}
		d_{\alpha} &= \text{level of the exceptional divisor}\,,\\
		d_{[\alpha,\beta]} &= 1 + d_{\alpha} + d_{\beta}\,.
	\end{aligned}
\end{equation}
On generic points of the resulting divisor $E_{[\alpha,\beta]}$ one can then perform blow-ups of type (B) or, alternatively, one can perform blow-ups of type (C) at the intersection points of $E_{[\alpha,\beta]}$ with $E_{\alpha}$ and $E_{\beta}$. The discrepancies of divisors arising from blow-ups of type (C) are one unit higher than the sum of the discrepancies of their parent divisors, as we see above.

One can then keep iterating blow-ups of type (B) and (C) until the desired arbitrary blow-up $\hat{B}$ of $B$ has been reached. Using the subindex notation that we have introduced, the anticanonical class is
\begin{equation}
	\overline{K}_{\hat{B}} = \pi^{*} \left( K_{B} \right) - \sum_{\alpha} d_{\alpha} E_{\alpha}\,, \qquad d_{\alpha} \geq 1\,,
\end{equation}
where the subindex $\alpha$ contains the ``blow-up history" of the exceptional divisor $E_{\alpha}$. To compute the value of the discrepancy, we keep adding one unit per pair of outer parentheses in $\alpha$ until we are done, or we encounter a pair of square brackets. These contribute by the discrepancy value assigned to their two entries plus one. Using this notation we can write branching blow-up diagrams like the one represented in \cref{fig:blow-up-paths-P2} and directly obtain the anticanonical class in our desired basis.

Let us give un example using $\mathbb{P}^{2}$ as the starting point. Blow it up at a point $p_{(1)}$, giving the exceptional divisor $E_{(1)}$. Continue by blowing up a generic point $p_{((1),1)}$ in $E_{(1)}$, producing the exceptional divisor $E_{((1),1)}$. Then, blow up the intersection point of the two exceptional divisors to produce $E_{[(1),((1),1)]}$. Finally, blow-up a generic point $p_{([(1),((1),1)],1)}$ in $E_{[(1),((1),1)]}$ to produce the exceptional divisor $E_{([(1),((1),1)],1)}$. Altogether, this leads to the exceptional divisors
\begin{subequations}
\begin{align}
	E_{(1)} &\longleftrightarrow d_{(1)} = 1\,,\\
	E_{((1),1)} &\longleftrightarrow d_{((1),1)} = 2\,,\\
	E_{[(1),((1),1)]} &\longleftrightarrow d_{[(1),((1),1)]} = 4\,,\\
	E_{([(1),((1),1)],1)} &\longleftrightarrow d_{([(1),((1),1)],1)} = 5\,,
\end{align}
\end{subequations}
meaning that the anticanonical class in our basis of choice is
\begin{equation}
	\overline{K}_{\hat{B}} = \pi^{*} \left( \overline{K}_{\mathbb{P}^{2}} \right) - E_{(1)} - 2E_{((1),1)} - 4E_{[(1),((1),1)]} - 5E_{([(1),((1),1)],1)}\,.
\end{equation}

An arbitrary blow-up of $B = \mathbb{P}^{2}$ or $B = \mathbb{F}_{n}$ may lead to a surface $\hat{B}$ with non-effective anticanonical class $\overline{K}_{\hat{B}}$, meaning that the global holomorphic sections necessary to construct the F-theory Weierstrass model are not available. These cases are therefore to be discarded in our analysis; in what follows, we implicitly assume that we are choosing blow-up paths that lead to surfaces with an effective anticanonical class.

\subsubsection{Intersection ring after an arbitrary blow-up}

The intersection product between the elements of the Picard basis \eqref{eq:Picard-basis} can be directly computed from \cref{lemma:strict-proper-transform}, \cref{corollary:exceptional-intersection} and \cref{prop:blow-up-properties}. At the start of the process, we have the Picard basis $\{D_{i}\}_{i \in \mathcal{I}}$, with known intersection products
\begin{subequations}
\begin{align}
	\mathbb{P}^{2}:&\qquad H \cdot H = 1\,,\\
	\mathbb{F}_{n}:&\qquad h \cdot h = -n\,,\quad f \cdot f = 0\,,\quad h \cdot f = 1\,.
\end{align}
\end{subequations}
Since we work with the total transforms of these divisors, their intersections remain the same after the blow-up process, and we therefore omit it in what follows.

Let us perform the blow-up process, following the same steps taken in the discussion of the anticanonical class. First, we perform all desired blow-ups of type (A). This leads to the Picard basis $\{D_{i}\}_{i \in \mathcal{I}} \cup \{E_{i}\}_{1 \leq i \leq k}$, with the intersection products
\begin{equation}
	\pi^{*} \left( D_{i} \right) \cdot E_{j} = 0\,,\quad E_{i} \cdot E_{j} = -\delta_{ij}\,.
\end{equation}
Performing now blow-ups of type (B) over points in a divisor $E_{\alpha} \in \{E_{i}\}_{0 \leq i \leq k}$ leads to the intersection products
\begin{equation}
	\begin{rcases}
	\begin{aligned}
		\pi^{*} \left( D_{i} \right) \cdot E_{j} &= 0\\
		E_{i} \cdot E_{j} &= -\delta_{ij}
	\end{aligned}
	\end{rcases}
	\longrightarrow
	\begin{cases}
	\begin{aligned}
		\pi^{*} \left( D_{i} \right) \cdot E_{\beta} &= 0\\
		E_{i} \cdot E_{j} &= -\delta_{ij} \left( 1 + \delta_{\alpha i}k_{\alpha} \right)\\
		E_{(\alpha,i)} \cdot E_{(\alpha,j)} &= -\delta_{ij}\\
		E_{\alpha} \cdot E_{(\alpha,i)} &= 1\\
		E_{i \neq \alpha} \cdot E_{(\alpha,j)} &= 0
	\end{aligned}
	\end{cases}
	\,,
\end{equation}
where $E_{\beta}$ stands for any of the $\{E_{i}\}_{1 \leq i \leq k}$ and $\{E_{(\alpha,i)}\}_{1 \leq i \leq k_{\alpha}}$. Of the old intersection products, only the self-intersection $E_{\alpha} \cdot E_{\alpha}$ is modified, owing to the local nature of the blow-up process. More generally, if we perform blow-ups of type (B) on generic points in a divisor $E_{\alpha}$, producing the exceptional divisors $\{E_{(\alpha,i)}\}_{0 \leq i \leq k_{\alpha}}$, the new intersection products are
\begin{equation}
	\begin{rcases}
	\begin{aligned}
		D_{\mathrm{old}} \cdot D'_{\mathrm{old}} &= \cdots\\
		E_{\alpha} \cdot E_{\alpha} &= -r_{\alpha}\\
		D_{\mathrm{old}} \cdot E_{\alpha} &= \cdots
	\end{aligned}
	\end{rcases}
	\longrightarrow
	\begin{cases}
	\begin{aligned}
		\rho^{*} \left( D_{\mathrm{old}} \right) \cdot \rho^{*} \left( D'_{\mathrm{old}} \right ) &= D_{\mathrm{old}} \cdot D'_{\mathrm{old}}\\
		E_{\alpha} \cdot E_{\alpha} &= -r_{\alpha} - k_{\alpha}\\
		E_{(\alpha,i)} \cdot E_{(\alpha,j)} &= -\delta_{ij}\\
		E_{\alpha} \cdot E_{(\alpha,i)} &= 1\\
		\rho^{*} \left( D_{\mathrm{old}} \right) \cdot E_{\alpha} &= D_{\mathrm{old}} \cdot E_{\alpha}\\
		\rho^{*} \left( D_{\mathrm{old}} \right) \cdot E_{(\alpha,i)} &= 0
	\end{aligned}
	\end{cases}
	\,,
\label{eq:blow-up-B-intersection-ring}
\end{equation}
where $D_{\mathrm{old}}$ stands for the divisors in $\hat{B}_{\mathrm{old}}$, with the exception of $E_{\alpha}$. The final type of blow-up we need to consider is that of type (C). Take two exceptional divisors $E_{\alpha}$ and $E_{\beta}$, with intersection product $E_{\alpha} \cdot E_{\beta} = 1$, and blow up their point of intersection to produce a new exceptional divisor $E_{[\alpha,\beta]}$. Then, the new intersection products are
\begin{equation}
	\begin{rcases}
	\begin{aligned}
		D_{\mathrm{old}} \cdot D'_{\mathrm{old}} &= \cdots\\
		E_{\alpha} \cdot E_{\alpha} &= -r_{\alpha}\\
		E_{\beta} \cdot E_{\beta} &= -r_{\beta}\\
		E_{\alpha} \cdot E_{\beta} &= 1\\
		D_{\mathrm{old}} \cdot E_{\alpha} &= \cdots\\
		D_{\mathrm{old}} \cdot E_{\beta} &= \cdots
	\end{aligned}
	\end{rcases}
	\longrightarrow
	\begin{cases}
	\begin{aligned}
		\rho^{*} \left( D_{\mathrm{old}} \right) \cdot \rho^{*} \left( D'_{\mathrm{old}} \right) &= D_{\mathrm{old}} \cdot D'_{\mathrm{old}}\\
		E_{\alpha} \cdot E_{\alpha} &= -r_{\alpha} - 1\\
		E_{\beta} \cdot E_{\beta} &= -r_{\beta} - 1\\
		E_{[\alpha,\beta]} \cdot E_{[\alpha,\beta]} &= -1\\
		E_{\alpha} \cdot E_{\beta} &= 0\\
		E_{\alpha} \cdot E_{[\alpha,\beta]} &= 1\\
		E_{\beta} \cdot E_{[\alpha,\beta]} &= 1\\
		\rho^{*} \left( D_{\mathrm{old}} \right) \cdot E_{\alpha} &= D_{\mathrm{old}} \cdot E_{\alpha}\\
		\rho^{*} \left( D_{\mathrm{old}} \right) \cdot E_{\beta} &= D_{\mathrm{old}} \cdot E_{\beta}\\
		\rho^{*} \left( D_{\mathrm{old}} \right) \cdot E_{[\alpha,\beta]} &= 0\\
	\end{aligned}
	\end{cases}
	\,,
\label{eq:blow-up-C-intersection-ring}
\end{equation}
where $D_{\mathrm{old}}$ stands for the divisors in $\hat{B}_{\mathrm{old}}$, with the exceptions of $E_{\alpha}$ and $E_{\beta}$.

Using these rules, it is easy to keep track of the intersection products of interest to us as we follow the branching blow-up paths leading from $B$ to the blown up surface $\hat{B}$, as exemplified in \cref{fig:blow-up-paths-P2}.
%auto-ignore

\section{Restricting the genus of non-minimal curves}
\label{sec:genus-restriction-proof}

In this appendix we prove \cref{prop:genus-restriction} of \cref{sec:curves-of-non-minimal-fibers}, which restricts the genus of the curves that can support non-minimal elliptic fibers. We will compute, via adjunction, the genus of such curves on general blow-ups of $\mathbb{P}^{2}$ and $\mathbb{F}_{n}$. The properties of the anticanonical class and the intersection pairing of divisors in arbitrary blow-ups of $\mathbb{P}^{2}$ and $\mathbb{F}_{n}$, as reviewed in \cref{sec:list-arbitrary-blow-ups}, then yield the restrictions on the genus as stated in \cref{prop:genus-restriction}.

As a final preparation, let us gather a couple of auxiliary results that we will invoke during the argument. First, we need a property of the effective divisors expressed in the Picard basis~\eqref{eq:Picard-basis}.
\begin{proposition}
\label{prop:total-transform-effective-positivity}
	Let $\hat{B}$ be an arbitrary blow-up of $B = \mathbb{P}^{2}$ or $B = \mathbb{F}_{n}$, and $D$ an effective divisor in $\hat{B}$. Expressing $D$ in terms of the Picard basis \eqref{eq:Picard-basis}, i.e.\ writing it as
	\begin{equation}
		D = \sum_{i \in \mathcal{I}} a_{i} \pi^{*} \left( D_{i} \right) + \sum_{\alpha} c_{\alpha} E_{\alpha}\,,
	\end{equation}
	the total transforms $\{\pi^{*} \left( D_{i} \right)\}_{i \in \mathcal{I}}$ always appear with non-negative coefficients $a_{i} \in \mathbb{Z}_{\geq 0}, \forall i \in \mathcal{I}$.
\end{proposition}

\begin{proof}
	Expressing $D$ in terms of the Picard basis \eqref{eq:Picard-basis} we have
	\begin{equation}
		D = \sum_{i \in \mathcal{I}} a_{i} \pi^{*}\left( D_{i} \right) + \sum_{\alpha} c_{\alpha}E_{\alpha}\,.
	\end{equation}
	If $D$ is a reducible effective divisor, it will be a positive linear combination of irreducible effective divisors. Hence, if we prove that $a_{i} \in \mathbb{Z}_{\geq 0}, \forall i \in \mathcal{I}$ for all irreducible effective divisors, the same result follows for all effective divisors. Moreover, the result is true for the collection of irreducible effective divisors $\{E_{\alpha}\}_{\alpha \in A}$. Let us therefore assume in what follows that $D$ is an irreducible effective divisor distinct from the $\{E_{\alpha}\}_{\alpha \in A}$. We will first prove the result assuming that only $k$ blow-ups of type (A) have been performed, and then generalize it to include arbitrary blow-ups.

	After $k$ blow-ups of type (A), we can write the divisor $D$ as
	\begin{equation}
		D = \sum_{i \in \mathcal{I}} a_{i} \pi^{*}\left( D_{i} \right) + \sum_{i=1}^{k} c_{i}E_{i}\,.
	\end{equation}
	Using \cref{prop:negative-intersection}, we see that, since $D$ and the $\{E_{i}\}_{0 \leq i \leq k}$ are irreducible
	\begin{equation}
		D \cdot E_{i} = -c_{i} \geq 0 \Rightarrow c_{i} \leq 0\,,\quad \forall i = 1, \dotsc, k\,.
	\end{equation}
	Let us now treat the blow-ups of $\mathbb{P}^{2}$ and $\mathbb{F}_{n}$ separately.
	\begin{itemize}
		\item Blow-ups $\mathrm{Bl}\left( \mathbb{P}^{2} \right)$: Assume that $a < 0$. Making the signs explicit, we have
		\begin{equation}
			D = -a^{+} \pi^{*} \left( H \right) - \sum_{i=1}^{k} c_{i}^{+} E_{i}\,,\quad a^{+}, c_{i}^{+} \geq 0\,.
		\end{equation}
		From the effective cone inclusion \eqref{eq:Picard-basis-effective-cone}, we see that
		\begin{equation}
			D \in \overline{\mathrm{Eff}}\left( \mathrm{Bl} \left( \mathbb{P}^{2} \right) \right) \Rightarrow -a^{+} \pi^{*} \left( H \right) \in \overline{\mathrm{Eff}}\left( \mathrm{Bl} \left( \mathbb{P}^{2} \right) \right)\,.
		\end{equation}
		But we know that $a^{+} \pi^{*} \left( H \right) \in \overline{\mathrm{Eff}}\left( \mathrm{Bl} \left( \mathbb{P}^{2} \right) \right)$. Since the effective cone is salient and $a^{+} \pi^{*} \left( H \right)$ is not trivial, $-a^{+} \pi^{*} \left( H \right)$ cannot be contained in $\overline{\mathrm{Eff}}\left( \mathrm{Bl} \left( \mathbb{P}^{2} \right) \right)$. Hence, $a \geq 0$.
		
		\item Blow-ups $\mathrm{Bl}\left( \mathbb{F}_{n} \right)$: We separate this case into various subcases.
		\begin{itemize}
			\item Assume that $a,b \leq 0$. Making the signs explicit, we have
			\begin{equation}
				D = -a^{+} \pi^{*} \left( h \right) - b^{+}\pi^{*} \left( f \right) - \sum_{i=1}^{k} c_{i}^{+} E_{i}\,,\quad a^{+}, b^{+}, c_{i}^{+} \geq 0\,.
			\end{equation}
			From the effective cone inclusion \eqref{eq:Picard-basis-effective-cone}, we obtain
			\begin{equation}
				D \in \overline{\mathrm{Eff}}\left( \mathrm{Bl} \left( \mathbb{P}^{2} \right) \right) \Rightarrow -a^{+} \pi^{*} \left( h \right) - b^{+}\pi^{*} \left( f \right) \in \overline{\mathrm{Eff}}\left( \mathrm{Bl} \left( \mathbb{P}^{2} \right) \right)\,.
			\end{equation}
			Since $a^{+} \pi^{*} \left( h \right) + b^{+}\pi^{*} \left( f \right) \in \overline{\mathrm{Eff}}\left( \mathrm{Bl} \left( \mathbb{P}^{2} \right) \right)$, either $a=b=0$, or we enter in contradiction with the fact that the effective cone is salient. But when $a=b=0,$ we have a negative linear combination of the $\{E_{i}\}_{0 \leq i \leq k}$. Using again the fact that the effective cone is salient and positive linear combinations of the $\{E_{i}\}_{0 \leq i \leq k}$ are contained in it, we conclude that $a,b \leq 0$ is not possible unless $D$ is trivial.
			
			\item Assume that $a \geq 0$ and $b < 0$. Making the signs explicit, we have
			\begin{equation}
				D = a^{+} \pi^{*} \left( h \right) - b^{+}\pi^{*} \left( f \right) - \sum_{i=1}^{k} c_{i}^{+} E_{i}\,,\quad a^{+}, b^{+}, c_{i}^{+} \geq 0\,.
			\end{equation}
			We now would like to exploit the negative self-intersection property of $\pi^{*} \left( h \right)$. Since $h \subset \mathbb{F}_{n}$ has a unique representative, no representative of $\pi^{*} \left( h \right)$ will be irreducible if a point in $h \subset \mathbb{F}_{n}$ has been blown up. We work instead with the strict transform $\hat{h}$ of $h$, which we know is irreducible and related to the total transform by a negative contribution of exceptional divisors
			\begin{equation}
				\hat{h} = \pi^{*} \left( h \right) - \sum_{i=1}^{k_{h}} c_{i}^{h} E_{i} \Rightarrow \hat{h} \cdot \hat{h} \leq h \cdot h = -n\,,\quad c_{i}^{h} \geq 0\,.
			\end{equation}
			Using then \cref{prop:negative-intersection}, we obtain the constraint
			\begin{equation}
				D \cdot \hat{h} = -a^{+}n - b^{+} - \sum_{i=1}^{k_{h}} c_{i}^{h} c_{i}^{+} \geq 0\,,\quad a^{+}, b^{+}, c_{i}^{h}, c_{i}^{+} \geq 0\,,
			\end{equation}
			since $D \neq \hat{h}$. Insisting on $a \geq 0$, this only has a chance of being positive if $b > 0$.
			
			\item Assume that $a < 0$ and $b \geq 0$. Making the signs explicit, we have
			\begin{equation}
				D = -a^{+} \pi^{*} \left( h \right) + b^{+}\pi^{*} \left( f \right) - \sum_{i=1}^{k} c_{i}^{+} E_{i}\,,\quad a^{+}, b^{+}, c_{i}^{+} \geq 0\,.
			\end{equation}
			We can always find a representative of $f \subset \mathbb{F}_{n}$ that is not affected by the blow-ups. Its total/strict transform $\pi^{*} \left( f \right)$ is then irreducible, and we have, using \cref{prop:negative-intersection}, that
			\begin{equation}
				D \cdot \pi^{*} \left( f \right) = -a^{+} \geq 0\,,\quad a^{+} > 0\,,
			\end{equation}
			since $D \neq \pi^{*} \left( f \right)$. We conclude that $a > 0$.
		\end{itemize}
		Taken together, we see that $a,b \geq 0$.
	\end{itemize}
	
	To obtain the desired result, we now need to show that the above arguments still hold when we allow for blow-ups of types (B) and (C). The only points that are slightly affected when we allow for further blow-ups are the proof of the negativity of the $c_{\alpha}$ coefficients and the constraint coming from the intersection product $D \cdot \hat{h}$.
	
	First, consider that we allow for arbitrary blow-ups of type (A) and (B). We then have a collection of exceptional divisors $\{E_{\alpha}\}_{\alpha \in A}$. Let us highlight among them those that arose from the type (A) blow-ups $\{E_{(i)}\}_{0 \leq i \leq k} \subset \{E_{\alpha}\}_{\alpha \in A}$. The intersection product with $D$ of one of these divisors is then
	\begin{align}
		D \cdot E_{i} &= -(1 + k_{i})c_{(i)} + \sum_{j=1}^{k_{i}} c_{((i),j)} \geq 0\,,\\
		D \cdot E_{(\alpha,i)} &= -(1+k_{(\alpha,i)}) + c_{\alpha} + \sum_{j=1}^{k_{(\alpha,i)}} c_{((\alpha,i),j)} \geq 0\,.
	\end{align}
	By summing all the inequalities obtained from the intersection products of $D$ with a given divisor $E_{\alpha}$ and all the other exceptional divisors that stem from the sequence of type (B) blow-ups of points in $E_{\alpha}$, we obtain the inequality
	\begin{equation}
		-c_{\alpha} + c_{\substack{\mathrm{parent}\\ \mathrm{divisor}}} \geq 0\,,
	\end{equation}
	where by parent divisor we mean the divisor where the point that was blown-up in order to generate $E_{\alpha}$ is located, assuming that for the $\{E_{(i)}\}_{0 \leq i \leq k}$ we take \smash{$c_{\substack{\mathrm{parent}\\ \mathrm{divisor}}} = 0$}. Using these inequalities iteratively level by level, we obtain
	\begin{equation}
		c_{\alpha} \leq 0\,,\quad \alpha \in A\,.
	\end{equation}
	Given two divisors $E_{\alpha}$ and $E_{\beta}$, we can allow now for a blow-up of type (C) over their intersection point $E_{\alpha} \cap E_{\beta}$. This modifies two of the already considered inequalities, leading to
	\begin{align}
		D \cdot E_{\alpha} \geq 0 \longrightarrow \text{old terms} - c_{\alpha} - c_{\beta} + c_{[\alpha, \beta]} \geq 0\,,\\
		D \cdot E_{\beta} \geq 0 \longrightarrow \text{old terms} - c_{\beta} - c_{\alpha} + c_{[\alpha, \beta]} \geq 0\,.
	\end{align}
	On top of this, we obtain the new inequality
	\begin{equation}
		D \cdot E_{[\alpha,\beta]} = -c_{[\alpha,\beta]} + c_{\alpha} + c_{\beta}\,.
	\end{equation}
	We see that we can use it to cancel the new terms in the old inequalities, in such a way that the arguments above apply and $c_{\gamma \neq [\alpha,\beta]} \leq 0, \forall \gamma \in A$, from where it then follows that also $c_{[\alpha,\beta]} \leq 0$. Allowing then for further blow-ups of type (B) over $E_{[\alpha,\beta]}$ reproduces the form of the inequalities studied above, and more blow-ups of type (C) only result in modifications like the one just discussed. Hence, we can conclude that
	\begin{equation}
		c_{\alpha} \leq 0\,,\quad \alpha \in A\,.
	\end{equation}
	for an arbitrary combination of types (A), (B) and (C) blow-ups.
	
	Let us now move to the analysis of how the intersection product $D \cdot \hat{h}$ changes. Let us recall that, after performing $k$ blow-ups of type (A), out of which $k_{h}$ were of points located in $h$, we have
	\begin{equation}
		D \cdot \hat{h} = -a^{+}n - b^{+} - \sum_{i=1}^{k_{h}} c_{i}^{h} c_{i}^{+} \geq 0\,,\quad a^{+}, b^{+}, c_{i}^{h}, c_{i}^{+} \geq 0\,.
	\end{equation}
	Although above we did not specify it, the $c_{i}^{h}$ are $c_{i}^{h} = 1$, since we are blowing up points with multiplicity one. Let us now perform arbitrary blow-ups of type (B) and (C), with the exception of those type (C) blow-ups involving the points $\{\hat{h} \cap E_{i}\}_{0 \leq i \leq k_{h}}$, and call the composition of these new blow-ups $\rho: \hat{B}_{2} \rightarrow \hat{B}_{1}$, with the original blow-up being $\pi: \hat{B}_{1} \rightarrow B$. We have that
	\begin{equation}
		\rho^{*} \left( \pi^{*} \left( h \right) - \sum_{i=1}^{k_{h}} E_{i} \right) = \hat{h}\,,
	\end{equation}
	where $\hat{h}$ now stands for the strict transform under the composition of all the blow-ups. Expressing $D$ in terms of the total transforms under the last set of blow-ups, we have
	\begin{equation}
	\begin{aligned}
		D &= a^{+} \rho^{*} \left( \pi^{*} \left( h \right) \right) - b^{+} \rho^{*} \left( \pi^{*} \left( f \right) \right) - \sum_{i=1}^{k} c_{i}E_{i} - \sum_{\rho \in P} c_{\rho}E_{\rho}\\
		&= \rho^{*} \left( a^{+} \pi^{*} \left( h \right) - b^{+} \pi^{*} \left( f \right) - \sum_{i=1}^{k} c_{i}E_{i} \right) - \sum_{\rho \in P} c'_{\rho}E_{\rho}\,,
	\end{aligned}
	\end{equation}
	where the $\{E_{\rho}\}_{\rho \in P}$ are the exceptional divisors arising from the last set of blow-ups. It is then clear that no change in the expression $D \cdot \hat{h}$ occurs. The situation is different when we blow-up one of the points $\big\{\hat{h} \cap E_{i}\big\}_{0 \leq i \leq k_{h}}$. Say that we blow-up the point $\hat{h} \cap E_{i}$, producing the exceptional divisor $E_{[\hat{h},i]}$. Then,
	\begin{equation}
		\rho^{*} \left( \pi^{*} \left( h \right) - \sum_{i=1}^{k_{h}} E_{i} \right) = \hat{h} + E_{[h,i]}\,.
	\end{equation}
	Expressing $D$ again in terms of the total transforms under the last set of blow-ups and separating the term containing $E_[\hat{h},i]$, we have
	\begin{equation}
	\begin{split}
		D &= a^{+} \rho^{*} \left( \pi^{*} \left( h \right) \right) - b^{+} \rho^{*} \left( \pi^{*} \left( f \right) \right) - \sum_{i=1}^{k} c_{i}E_{i} - \sum_{\rho \in P'} c_{\rho}E_{\rho} - c_{[\hat{h},i]}E_{[\hat{h},i]}\\
		&= \rho^{*} \left( a^{+} \pi^{*} \left( h \right) - b^{+} \pi^{*} \left( f \right) - \sum_{i=1}^{k} c_{i}E_{i} \right) - \sum_{\rho \in P} c'_{\rho}E_{\rho}- \left( c_{[\hat{h},i]} - c_{i} \right) E_{[\hat{h},i]}\,.
	\end{split}
	\end{equation}
	The intersection product is then modified to
	\begin{equation}
		D \cdot \hat{h} = -a^{+}n - b^{+} - c_{[\hat{h},i]} + c_{i} - \sum_{j=1}^{k_{h}} c_{j}E_{j}\,,
	\end{equation}
	but this does not affect the arguments used above. Further blow-ups of type (B) and (C) can be treated in the way just described, and therefore the discussion does generalize to an arbitrary combination of them.
\end{proof}
The classes of the irreducible curves $C$ over which we find Kodaira singularities in six-dimensional F-theory models are contained in the discriminant, meaning that they are effective divisors satisfying $C \leq \Delta = 12 \overline{K}_{B}$. It may, in general, be false that they satisfy $C \leq \overline{K}_{B}$. If $C$ supports non-minimal singularities, then it is further constrained.
\begin{proposition}
\label{prop:K-C-effectiveness}
	Let $Y$ be the Weierstrass model Calabi-Yau threefold with base $B$ having effective anticanonical class $\overline{K}_{B}$, and $C \subset B$ a smooth irreducible curve in the base supporting non-minimal singular fibers. Then, $C \leq \overline{K}_{B}$.
\end{proposition}
\begin{proof}
	If $C$ supports non-minimal fibers, the vanishing orders over it must satisfy
	\begin{equation}
		\ord{Y}(f,g)_{C} \geq (4,6)\,,
	\end{equation}
	which in turn implies
	\begin{equation}
		\begin{rcases}
		\begin{aligned}
			4C \leq F = 4\overline{K}_{B}\\
			6C \leq G = 6\overline{K}_{B}
		\end{aligned}
		\end{rcases}
		\Rightarrow
		C \leq \overline{K}_{B}\,.
	\end{equation}
\end{proof}

We can now restrict the genus of an irreducible curve supporting non-minimal fibers in a six-dimensional F-theory model, i.e.\ prove \cref{prop:genus-restriction}, whose precise statement we recall.
\genusrestriction*
\begin{proof}
	Let us consider a smooth, irreducible curve $C$ in an arbitrary blow-up of $\mathbb{P}^{2}$ or $\mathbb{F}_{n}$. Using the Picard basis \eqref{eq:Picard-basis}, we can express the class of $C$ as
	\begin{subequations}
	\begin{align}
		\hat{B} = \mathrm{Bl} \left( \mathbb{P}^{2} \right):&\qquad C = a \pi^{*} \left( H \right) + \sum_{\alpha} c_{\alpha}E_{\alpha}\,,\\
		\hat{B} = \mathrm{Bl} \left( \mathbb{F}_{n} \right):&\qquad C = a \pi^{*} \left( h \right) + b \pi^{*} \left( f \right) + \sum_{\alpha} c_{\alpha}E_{\alpha}\,.
	\end{align}
	\end{subequations}
	Using the same basis, the anticanonical class of the base $B$ is
	\begin{equation}
		\overline{K}_{\hat{B}} = \pi^{*} \left( K_{B} \right) - \sum_{\alpha} d_{\alpha} E_{\alpha}\,, \qquad d_{\alpha} \geq 0\,,
	\end{equation}
	see \cref{sec:anticanonical-class-after-arbitrary-blow-up}. The cases in which the base is $B = \mathbb{P}^{2}$ or $B = \mathbb{F}_{n}$ are recovered by setting $c_{\alpha} = d_{\alpha} = 0$ throughout the argument.
	
	Invoking \cref{prop:adjunction-genus}, we see that the genus of $C$ is given by
	\begin{equation}
		g(C) = \frac{C \cdot C - C \cdot \overline{K}_{\phat{B}}}{2} + 1\,.
	\end{equation}
	The implication $C = \overline{K}_{\phat{B}} \Rightarrow g(C) = 1$ is then immediate. An irreducible curve $C$ satisfies
	\begin{equation}
		g(C) \leq 1 \Leftrightarrow C \cdot C \leq C \cdot \overline{K}_{\phat{B}}\,.
	\end{equation}
	We would therefore like to prove that there exists no smooth irreducible curve $C$ supporting non-minimal fibers and satisfying $C \cdot C > C \cdot \overline{K}_{\phat{B}}$ at the same time.
	
	If we have only performed $k$ blow-ups of type (A), the inequality $C \cdot C \geq C \cdot \overline{K}_{\phat{B}}$ reads
	\begin{equation}
		\sum_{i \in \mathcal{I}} a_{i} \pi^{*} \left( D_{i} \right) \cdot  \sum_{j \in \mathcal{J}} a_{j} \pi^{*} \left( D_{j} \right) - \sum_{i = 1}^{k} c_{i}^{2} \geq \sum_{i \in \mathcal{I}} a_{i} \pi^{*} \left( D_{i} \right) \cdot \pi^{*} \left( \overline{K}_{\phat{B}} \right) + \sum_{i = 1}^{k} c_{i}d_{i}\,.
	\end{equation}
	We can define the functions
	\begin{align}
		h\left( c_{i} \right)_{\text{(A)-blow-up}} &:= \sum_{i = 1}^{k} c_{i}^{2} + c_{i}\,,\\
		f(a_{i})_{B} &:= \sum_{i \in \mathcal{I}} a_{i} \pi^{*} \left( D_{i} \right) \cdot  \sum_{j \in \mathcal{J}} a_{j} \pi^{*} \left( D_{j} \right) - \sum_{i \in \mathcal{I}} a_{i} \pi^{*} \left( D_{i} \right) \cdot \pi^{*} \left( \overline{K}_{\phat{B}} \right)\,,
	\end{align}
	where we have used that $d_{i} = 1, \forall i = 1, \dotsc, k$, to express it as
	\begin{equation}
		f(a_{i})_{B} \geq h\left( c_{i} \right)_{\text{(A)-blow-up}}\,.
	\end{equation}
	The function $f(a_{i})_{B}$ only involves total transforms of the divisors $\{D_{i}\}_{i \in \mathcal{I}}$ of $B$, which have trivial intersection with the rest of divisors in the Picard basis \eqref{eq:Picard-basis}. This implies that, even after allowing for an arbitrary blow-up, the function $f(a_{i})_{B}$ will not change. The same is not true for $h\left( c_{i} \right)_{\text{(A)-blow-up}}$, and we need to determine its most general form. Taking into account the form of the intersection ring \eqref{eq:blow-up-B-intersection-ring}, we can see that the additional terms contributing to $C \cdot C \geq C \cdot \overline{K}_{\phat{B}}$ after blowing up $k_{\alpha}$ generic points in an exceptional divisor $E_{\alpha}$, i.e.\ by a blow-up of type (B), are
	\begin{equation}
		\cdots - k_{\alpha} c_{\alpha}^{2} + 2\sum_{i=1}^{k_{\alpha}} c_{\alpha} c_{(\alpha,i)} - \sum_{i=1}^{k_{\alpha}} c^{2}_{(\alpha,i)} \geq \cdots - k_{\alpha}c_{\alpha}d_{\alpha} + \sum_{i=1}^{k_{\alpha}} c_{(\alpha,i)}d_{(\alpha,i)}\,.
	\end{equation}
	Hence, allowing for blow-ups both of type (A) and (B) the r.h.s.\ of the inequality needs to be
	\begin{equation}
		h\left( c_{i} \right)_{\text{(AB)-blow-up}} = h\left( c_{i} \right)_{\text{(A)-blow-up}} + \sum_{\substack{\textrm{divisors $E_{\alpha}$}\\ \textrm{where $k_{\alpha}$}\\ \textrm{generic points}\\ \textrm{were blown up}}} \sum_{i=1}^{k_{\alpha}} \left[ \left( c_{\alpha}^{2} - c_{\alpha} \right) - 2c_{\alpha} c_{(\alpha,i)} + \left( c_{(\alpha,i)}^{2} + c_{(\alpha,i)} \right) \right]\,,
	\end{equation}
	where we have used that $d_{(\alpha,i)} = d_{\alpha} + 1$ to simplify the expression. The terms added to $C \cdot C \geq C \cdot \overline{K}_{\phat{B}}$ by a blow-up of type (C) can be computed to be
	\begin{equation}
	\begin{aligned}
		\cdots -c_{\alpha}^{2} - c_{\beta}^{2} - 2c_{\alpha}c_{\beta} + 2c_{\alpha}c_{[\alpha,\beta]} + 2c_{\beta}c_{[\alpha,\beta]} - c_{[\alpha,\beta]}^{2} \geq &\cdots + c_{\alpha}d_{\alpha} + c_{\beta} d_{\beta} + c_{\alpha}d_{\beta} + c_{\beta}d_{\alpha}\\
		&- c_{\alpha}d_{[\alpha,\beta]} - c_{\beta}d_{[\alpha,\beta]} - d_{\alpha}c_{[\alpha,\beta]} - d_{\beta}c_{[\alpha,\beta]}\\
		&+ c_{[\alpha,\beta]}d_{[\alpha,\beta]}
	\end{aligned}
	\end{equation}
	using the intersection ring \eqref{eq:blow-up-C-intersection-ring}. This leads to the r.h.s.\ of the inequality being
	\begin{equation}
		h\left( c_{i} \right)_{\text{(ABC)-blow-up}} = h\left( c_{i} \right)_{\text{(AB)-blow-up}} +  \sum_{\substack{\textrm{pairs $E_{\alpha}$, $E_{\beta}$}\\ \textrm{whose intersection}\\ \textrm{point was}\\ \textrm{blown up}}} \left(c_{\alpha} + c_{\beta} - c_{[\alpha,\beta]}\right)\left(c_{\alpha} + c_{\beta} - c_{[\alpha,\beta]} - 1\right)\,,
	\end{equation}
	where we have used $d_{[\alpha,\beta]} = d_{\alpha} + d_{\beta} + 1$ to cancel terms. Further blow-ups of type (B) or (C) just add more terms of the form computed above. Therefore, for a general blow-up of $B$ we have
	\begin{equation}
		C \cdot C \geq C \cdot \overline{K}_{\phat{B}} \Leftrightarrow f(a_{i})_{B} \geq h\left( c_{i} \right)_{\text{blow-up}}\,,
	\label{eq:genus-restriction-inequality}
	\end{equation}
	with the definitions
	\begin{align}
		h\left( c_{i} \right)_{\text{blow-up}} &:= \sum_{i=1}^{k} p_{1} \left( c_{i} \right) + \sum_{\substack{\textrm{divisors $E_{\alpha}$}\\ \textrm{where $k_{\alpha}$}\\ \textrm{generic points}\\ \textrm{were blown up}}} \sum_{i=1}^{k_{\alpha}} p_{2}\left(c_{\alpha},c_{(\alpha,i)}\right)_{i} + \sum_{\substack{\textrm{pairs $E_{\alpha}$, $E_{\beta}$}\\ \textrm{whose intersection}\\ \textrm{point was}\\ \textrm{blown up}}} p_{3} \left( c_{\alpha},c_{\beta},c_{[\alpha,\beta]} \right)\,,\\
		f(a_{i})_{B} &:= \sum_{i \in \mathcal{I}} a_{i} \pi^{*} \left( D_{i} \right) \cdot  \sum_{j \in \mathcal{J}} a_{j} \pi^{*} \left( D_{j} \right) - \sum_{i \in \mathcal{I}} a_{i} \pi^{*} \left( D_{i} \right) \cdot \pi^{*} \left( \overline{K}_{\phat{B}} \right)\,,
	\end{align}
	where
    \begin{align}
        p_{1}(c_{i}) &:= (c_{i}^{2} + c_{i})\,,\\
        p_{2}\left(c_{\alpha},c_{(\alpha,i)}\right)_{i} &:= \left( c_{\alpha}^{2} - c_{\alpha} \right) - 2c_{\alpha} c_{(\alpha,i)} + \left(c_{(\alpha,i)}^{2} + c_{(\alpha,i)} \right)\,,\\
        p_{3}(c_{\alpha},c_{\beta},c_{[\alpha,\beta]}) &:= \left(c_{\alpha} + c_{\beta} - c_{[\alpha,\beta]}\right)\left(c_{\alpha} + c_{\beta} - c_{[\alpha,\beta]} - 1\right)\,.
    \end{align}
    
    Let us first use these expressions to prove that, with the hypotheses of the proposition, the strict version of the inequality \eqref{eq:genus-restriction-inequality} cannot be satisfied. The anticanonical classes of $\mathbb{P}^{2}$ and $\mathbb{F}_{n}$ are
    \begin{equation}
    		\overline{K}_{\mathbb{P}^{2}} = 3H\,,\qquad \overline{K}_{\mathbb{F}_{n}} = 2h + (2+n)f\,.
    \end{equation}
    Computing the form of $f(a_{i})_{B}$ for blow-ups of these two surfaces we obtain
    \begin{align}
    		f(a)_{\mathbb{P}^{2}} &= a^{2} - 3a\,,\\
    		f(a,b)_{\mathbb{F}_{n}} &= 2ab - a^{2}n - 2a -2b + an\,.
    \end{align}
    From \cref{prop:total-transform-effective-positivity} and \cref{prop:K-C-effectiveness} we obtain the constraints
    \begin{align}
        \mathbb{P}^{2}:\qquad &0 \leq a \leq 3\,,\\
        \mathbb{F}_{n}:\qquad &0 \leq a \leq 2\,,\quad 0 \leq b \leq 2+n\,.
    \end{align}
    This leads to the values
    \begin{subequations}
    \begin{align}
        f(0)_{\mathbb{P}^{2}} &= 0 \leq 0\,,\\
        f(1)_{\mathbb{P}^{2}} &= -2 \leq 0\,,\\
        f(2)_{\mathbb{P}^{2}} &= -2 \leq 0\,,\\
        f(3)_{\mathbb{P}^{2}} &= 0 \leq 0\,,
    \end{align}
    \label{eq:fB-values-P2}%
    \end{subequations}
    and
    \begin{subequations}
    \begin{align}
        f(0,b)_{\mathbb{F}_{n}} &= -2b \leq 0\,,\\
        f(1,b)_{\mathbb{F}_{n}} &= -2 \leq 0\,,\\
        f(2,b)_{\mathbb{F}_{n}} &= -4 + 2(b - n) \leq 0\,.
    \end{align}
    \label{eq:fB-values-Fn}%
    \end{subequations}
    Therefore, for $B = \mathbb{P}^{2}$ and $B = \mathbb{F}_{n}$ we have that $f(a_{i})_{B} \leq 0$ for the allowed values of the $a_{i}$ coefficients. On the other hand,
    \begin{equation}
        p_{1}(c_{i}) \geq 0\,,\quad p_{2}\left(c_{\alpha},c_{(\alpha,i)}\right)_{i} \geq 0\,,\quad p_{3}\left( c_{\alpha},c_{\beta},c_{[\alpha,\beta]} \right) \geq 0\,,\quad c_{i}, c_{\alpha}, c_{\beta}, c_{(\alpha,i)}, c_{[\alpha,\beta]} \in \mathbb{Z}\,.
    \end{equation}
    We hence conclude that
    \begin{equation}
    		C \cdot C > C \cdot \overline{K}_{\phat{B}} \Leftrightarrow 0 \geq f(a_{i})_{B} > h\left( c_{i} \right)_{\mathrm{blow-up}} \geq 0\,,
    \end{equation}
    which is not possible, meaning that $g(C) \leq 1$.
    
    Above we have seen that $C = \overline{K}_{\phat{B}}$ immediately implies $g(C) = 1$. To conclude, let us prove the converse statement. We have that
    \begin{equation}
    		g(C) = 1 \Leftrightarrow C \cdot C = C \cdot \overline{K}_{\phat{B}} \Leftrightarrow f(a_{i})_{B} = h\left( c_{i} \right)_{\mathrm{blow-up}} = 0\,.
    \end{equation}
    From \eqref{eq:fB-values-P2} and \eqref{eq:fB-values-Fn} we see that this can occur, regarding  $f(a_{i})_{B}$, only for
    \begin{align}
    		\mathbb{P}^{2}:&\qquad a = 0\quad \mathrm{or}\quad a = 3\,,\\
    		\mathbb{F}_{n}:&\qquad a = 0\,, b = 0\quad \mathrm{or}\quad a = 2\,,\; b = 2 + n\,.
    \end{align}
    Since all the terms in $h\left( c_{i} \right)_{\mathrm{blow-up}}$ are positive, they must vanish separately, leading to
    \begin{align}
    		\left( c_{i}^{2} + c_{i} \right) = 0 &\Leftrightarrow c_{i} = -1\quad \textrm{or}\quad c_{i} = 0\,,\label{eq:h-blow-up-values-1}\\
    		\left( c_{\alpha}^{2} - c_{\alpha} \right) - 2c_{\alpha} c_{(\alpha,i)} + \left(c_{(\alpha,i)}^{2} + c_{(\alpha,i)} \right) = 0 &\Leftrightarrow c_{(\alpha,i)} = c_{\alpha}\quad \textrm{or}\quad c_{(\alpha,i)} = -1 + c_{\alpha}\,,\label{eq:h-blow-up-values-2}\\
    		\left(c_{\alpha} + c_{\beta} - c_{[\alpha,\beta]}\right)\left(c_{\alpha} + c_{\beta} - c_{[\alpha,\beta]} - 1\right) = 0 &\Leftrightarrow c_{[\alpha,\beta]} = c_{\alpha} + c_{\beta}\quad \textrm{or}\quad c_{[\alpha,\beta]} = -1 + c_{\alpha} + c_{\beta}\,.\label{eq:h-blow-up-values-3}
    \end{align}
    The solutions \eqref{eq:h-blow-up-values-1} are non-positive, which makes the solutions \eqref{eq:h-blow-up-values-2} non-positive as well. This then implies that the solutions \eqref{eq:h-blow-up-values-3} are also non-negative. Assume now that we choose as solution of $f(a_{i})_{B} = 0$ either
    \begin{equation}
    		\mathbb{P}^{2}:\qquad a = 0\,,\qquad \mathrm{or}\qquad \mathbb{F}_{n}:\qquad a = 0\,,\; b = 0\,.
    \end{equation}
    Then $C$ would be a non-negative sum of the $\{E_{\alpha}\}_{\alpha \in A}$. Since positive sums of the $\{E_{\alpha}\}_{\alpha \in A}$ are in the effective cone, and said cone is salient, a negative linear combination of the $\{E_{\alpha}\}_{\alpha \in A}$ must be either trivial or not effective. Hence, we must discard these solutions and are restricted to
    \begin{equation}
    		\mathbb{P}^{2}:\qquad a = 3\,,\qquad \mathrm{or}\qquad \mathbb{F}_{n}:\qquad a = 2\,,\; b = 2+n\,.
    \end{equation}
    \cref{prop:K-C-effectiveness} implies that
    \begin{equation}
    		c_{\alpha} \leq -d_{\alpha} \leq -1\,,\quad \forall \alpha \in A\,.
    \end{equation}
    Using this constraint successively for \eqref{eq:h-blow-up-values-1}, \eqref{eq:h-blow-up-values-2} and \eqref{eq:h-blow-up-values-3} leads to
    \begin{equation}
    		c_{\alpha} = -d_{\alpha}\,,\quad \forall \alpha \in A\,,
    \end{equation}
    which in turn implies $C = \overline{K}_{\phat{B}}$.
\end{proof}
\begin{remark}
	For $\mathbb{F}_{n}$ this can be seen directly by examining the list of curves in \cref{prop:non-minimal-curves}.
\end{remark}
\begin{remark}
\label{rem:genus-one-degeneration-over-blow-ups}
	The case $g(C) = 1 \Leftrightarrow C = \overline{K}_{\phat{B}}$ can only be realized if the anticanonical class does not have reducible generic representatives. This is what limits it to the complex projective plane $\mathbb{P}^{2}$, the Hirzebruch surfaces without non-Higgsable clusters (i.e.\ $\mathbb{F}_{n}$ with $0 \leq n \leq2$), and blow-ups of these.
	
	Note, however, that not every blow-up will preserve the property of the generic representative of $\overline{K}_{B}$ being irreducible. To exemplify this, take the blow-up of $\mathbb{P}^{2}$ at four points. If the points are in general position, we obtain $\mathrm{Bl}_{4} \left( \mathbb{P}^{2} \right) = \mathrm{dP}_{4} \cong \mathrm{Bl}_{3} \left( \mathbb{F}_{1} \right)$. Since the points are in general position, a representative of the hyperplane class $H$ of $\mathbb{P}^{2}$ can pass through at most two of them, and the effective cone is generated by
	\begin{equation}
		\overline{\mathrm{Eff}} \left( \mathrm{dP}_{4} \right) = \langle \hat{H}_{ij} := \pi^{*} \left( H \right) - E_{i} - E_{j}, E_{k} \rangle_{\mathbb{Z}_{\geq 0}}\,, \quad 1 \leq i < j \leq 4\,,\quad 1 \leq k \leq 4\,.
	\end{equation}
	The anticanonical class
	\begin{equation}
		\overline{K}_{\mathrm{dP}_{4}} = 3\pi^{*}\left( H \right) - \sum_{k=1}^{4}E_{k}
	\end{equation}
	is generically irreducible, with intersection products
	\begin{equation}
		\hat{H}_{ij} \cdot \overline{K}_{\mathrm{dP}_{4}} = 1\,,\quad E_{k} \cdot \overline{K}_{\mathrm{dP}_{4}} = 1\,,\quad 1 \leq i < j \leq 4\,,\quad 1 \leq k \leq 4\,.
	\end{equation}
	If instead we choose the points to be in the special position in which they are aligned, there exists a representative of the hyperplane class $H$ of $\mathbb{P}^{2}$ that passes through all of them, which means that its strict transform is the irreducible effective divisor
	\begin{equation}
		\hat{H} = \pi^{*}\left( H \right) - \sum_{k=1}^{4} E_{k}\,,
	\end{equation}
	whose intersection product with the anticanonical class is
	\begin{equation}
		\hat{H} \cdot \overline{K}_{\mathrm{Bl}\left( \mathbb{P}^{2} \right)} = -1\,.
	\end{equation}
	From \cref{prop:negative-intersection} we see that $\overline{K}_{\mathrm{Bl}\left( \mathbb{P}^{2} \right)}$ will be reducible, containing an $\hat{H}$ component, and therefore the $g(C) = 1$ case cannot be realized using this base.
\end{remark}
%auto-ignore

\section{Obscured infinite-distance limits}
\label{sec:obscured-infinite-distance-limits}

In this appendix we address the phenomenon of obscured infinite-distance limits. Recall from \cref{sec:modifications-of-degenerations} that there can arise situations in which the family vanishing orders of the defining polynomials of the Weierstrass model are minimal over a codimension-one curve in the base of a component of the central fiber, while the component vanishing orders are non-minimal. This can occur over a curve $C \subset \hat{B}_{0} = \{u=0\}_{\hat{\mathcal{B}}}$ in the central fiber $\pi: \hat{Y}_{0} \rightarrow \hat{B}_{0}$ of a degeneration $\hat{\rho}: \hat{\mathcal{Y}} \rightarrow D$ in which $\hat{\mathcal{Y}}$ presents no non-minimal singular elliptic fibers,
\begin{equation}
	\text{minimal} \sim \ord{\hat{\mathcal{Y}}}(f,g)_{C} \leq \ord{\hat{Y}_{0}}(\left. f \right|_{u=0}, \left. g \right|_{u=0})_{C} \sim \text{non-minimal}\,,
\end{equation}
or over a curve $C \subset B^{p} = \{e_{p} = 0\}_{\mathcal{B}}$ in a component $\pi_{p}: Y^{p} \rightarrow B^{p}$ of the multi-component central fiber of a resolved degeneration $\rho: \mathcal{Y} \rightarrow \mathcal{B}$ obtained as explained in \cref{sec:base-blow-ups},
\begin{equation}
	\text{minimal} \sim \ord{\mathcal{Y}}(f_{b},g_{b})_{C} \leq \ord{Y^{p}}(\left. f_{b} \right|_{e_{p}=0}, \left. g_{b} \right|_{e_{p}=0})_{C} \sim \text{non-minimal}\,.
\end{equation}
This phenomenon arises because, for such a model, the slice $\pi: \hat{Y}_{0} \rightarrow \hat{B}_{0}$ (or $\pi_{p}: Y^{p} \rightarrow B^{p}$) of the model $\Pi_{\mathrm{ell}}: \hat{\mathcal{Y}} \rightarrow \hat{\mathcal{B}}$ (or $\Pi_{\mathrm{ell}}: \mathcal{Y} \rightarrow \mathcal{B}$) used to compute the component vanishing orders is non-generic. However, because of the interpretation of the degeneration as a family of six-dimensional F-theory models limiting to the one described by the central fiber, we are forced to give special consideration to these non-generic slices of the family variety. We will refer to degenerations presenting this feature as having an obscured infinite-distance limit, since (at least part of) their infinite-distance non-minimal nature is not directly apparent when looking at the elliptic fibers of the family variety.

The problem with such a degeneration is that the F-theory model given by the geometrical representative $\phat{Y}_{0}$ of the central fiber of the degeneration presents non-minimal singular elliptic fibers over a curve, that need to be removed. However, since the family variety $\phat{\Pi}_{\mathrm{ell}}: \phat{\mathcal{Y}} \rightarrow \phat{\mathcal{B}}$ is minimal over the same curve, the total transform divisors $\tilde{F}$, $\tilde{G}$ and $\tilde{\Delta}$ under a base blow-up map do not contain enough components of the exceptional divisor $E$ to allow for a line bundle shift. Insisting on such a line bundle shift would make the shifted divisors $F$, $G$ and $\Delta$ not effective. Hence, the resolution process described in \cref{sec:base-blow-ups} and consisting of iterative base blow-ups followed by line bundle shifts in order to preserve the Calabi-Yau condition is not possible.

As mentioned above, the mismatch between the family and component vanishing orders occurs because the slice of the family variety containing the elliptic fibration over the component is not generic, i.e.\ it intersects the discriminant of the family variety with higher multiplicity than the generic slice does, yielding a subvariety with worse singular elliptic fibers. In order to equate the two notions of vanishing orders, we need to find an equivalent degeneration in which the slice giving the component is the generic slice. This can be achieved by incorporating base changes into the resolution process, as we now describe.

If, of the two cases described above, we find the obscured infinite-distance limit over a curve $C \subset B^{p}$ in a component of the central fiber $Y_{0}$ of the resolved degeneration $\rho: \mathcal{Y} \rightarrow D$, we start by blowing down to said component to obtain a degeneration $\hat{\rho}: \hat{\mathcal{Y}} \rightarrow D$ in which $\hat{\mathcal{B}} = B \times D$. In the other case, this is already the starting point. The triviality of the divisor class $D$ means that we can add copies of it to the defining holomorphic line bundle $\hat{\mathcal{L}}$ of the Weierstrass model without affecting the Calabi-Yau condition or, in more practical terms, that the coordinate $u$ can appear with arbitrary degree in the monomials of $f$, $g$ and $\Delta$. We are therefore allowed to perform the base change
\begin{equation}
	\begin{aligned}
		\delta_{k}: D &\longrightarrow D\\
		u &\longmapsto u^{k}\,,
	\end{aligned}
\end{equation}
which for high enough $k$ makes the $\{u=0\}_{\hat{\mathcal{B}}}$ slice no longer tangent to the discriminant of the family variety. Then, $\{u=0\}_{\hat{\mathcal{B}}}$ becomes the generic slice of $\Pi_{\mathrm{ell}}: \hat{\mathcal{Y}} \rightarrow \hat{\mathcal{B}}$ and we have
\begin{equation}
	\ord{\hat{\mathcal{Y}}}(f,g)_{C} = \ord{Y_{0}}(\left. f \right|_{u=0}, \left. g \right|_{u=0})_{C} \sim \text{non-minimal}\,,
\end{equation}
allowing us to blow-up the family variety along $C$ and carry out the line bundle shift to recover the Calabi-Yau condition. If we obtained $\hat{\rho}: \hat{\mathcal{Y}} \rightarrow D$ by blowing down $\rho: \mathcal{Y} \rightarrow D$, we will have additional curves of non-minimal fibers that need to be resolved; due to the base change this will require more blow-ups than it originally did, increasing the number of components of the central fiber of the resolved degeneration not only along the new blow-up centre, but also along the original ones. This process can be repeated until all obscure infinite-distance limits have been removed.

This type of base change was already important in the analysis of infinite-distance limits in the complex structure moduli space of eight-dimensional F-theory in \cite{Lee:2021qkx,Lee:2021usk}. Indeed, although not discussed as explicitly, the vanishing orders of the defining polynomials of a Weierstrass model in \cite{Lee:2021qkx,Lee:2021usk} over a point in the family base were understood as the maximal vanishing orders that could be obtained after a base change over said point. In the more general context of the semi-stable reduction theorem \cite{Mumford1973}, a base change may need to be performed before a given degeneration admits a semi-stable modification, as we further discuss in \cref{sec:class-1-5-models}. In our more restrictive setup, in which we consider degenerations of elliptic fibrations with a section, the obscure infinite-distance limits are one avatar of this need for base changes, that manifests itself very explicitly through the non-genericity of the $\{u=0\}_{\hat{\mathcal{B}}}$ slice.

Before we conclude the section, let explore a couple of illustrative examples of obscured infinite-distance limits, the first one for a degeneration of elliptic K3 surfaces (for which the above discussion proceeds analogously but in one dimension lower), and the second one for a degeneration of elliptic fibrations over a Hirzebruch surface. In the former we present a family variety in which all the fibers are minimal but an obscured infinite-distance limit is present, while in the latter the obscured infinite-distance limit appears in the exceptional component arising from the blow-up along a curve of non-minimal singularities in the original component.
\begin{example}
\label{example:obscured-infinite-distance-limit-K3}
	We can already give a simple example of an obscured infinite-distance limit in the context of \cite{Lee:2021qkx,Lee:2021usk}. Consider a degeneration of elliptic K3 surfaces given by a Weierstrass model over $\mathcal{B} = \mathbb{P}^{1} \times D$ with defining polynomials of the form
	\begin{subequations}
	\begin{align}
		f &= s^{4} p_{4}^{0}([s:t]) + \sum_{i=1}^{4} u^{i} p_{8}^{i}([s:t],u^{i})\,,\\
		g &= s^{6} p_{6}^{0}([s:t]) + \sum_{i=1}^{6} u^{i} p_{8}^{i}([s:t],u^{i})\,,\\
		\Delta &= s^{12} p_{12}^{0}([s:t]) + \sum_{i=1}^{12} u^{i} p_{24}^{i}([s:t],u^{i}) \label{eq:obscure-infinite-distance-limit-K3-Delta}\,,
	\end{align}
	\end{subequations}
	where the $p_{\smallbullet}^{i}([s:t],u^{i})$ polynomials in $f$ and $g$ are generic. The family and component vanishing orders over $\{ s=u=0 \}_{\mathcal{B}}$ are
	\begin{equation}
		(1,1,2) = \ord{\hat{\mathcal{Y}}}(f,g,\Delta)_{s=u=0} \leq \ord{\hat{Y}_{0}}\left( \left. f \right|_{u=0}, \left. g \right|_{u=0}, \left. \Delta \right|_{u=0} \right)_{s=u=0} = (4,6,12)\,.
	\end{equation}
	Working in the patch $([s:t],u)$, the family vanishing orders can be computed by restricting $\hat{\mathcal{B}}$ to the line $\mathcal{W} := \{ (s,u) = (\mu a, \mu b) \}_{\hat{\mathcal{B}}}$ with a generic choice for the coefficients $a$ and $b$. The component vanishing orders correspond to the choice $b=0$, which can be seen to be the only choice intersecting $\Delta$ at $\{ s=u=0 \}_{\mathcal{B}} = \{ \mu = 0 \}_{\mathcal{W}}$ with higher multiplicity than the rest of the lines. This is represented in \cref{fig:Deltaaffplotu1}.
	\begin{figure}[t!]
	     \centering
	     \begin{subfigure}[b]{0.45\textwidth}
	         \centering
	         \includegraphics[width=\textwidth]{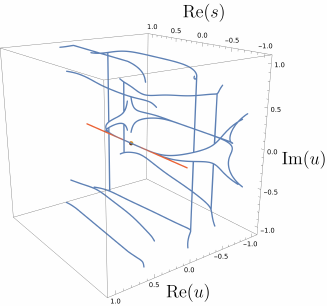}
	         \caption{Base change $u \mapsto u$.}
	         \label{fig:Deltaaffplotu1}
	     \end{subfigure}
	     \hfill
	     \begin{subfigure}[b]{0.45\textwidth}
	         \centering
	         \includegraphics[width=\textwidth]{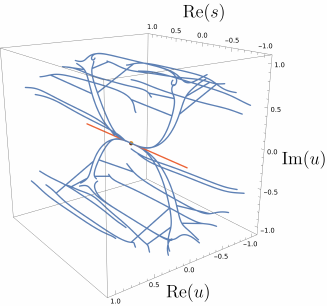}
	         \caption{Base change $u \mapsto u^{2}$.}
	         \label{fig:Deltaaffplotu2}
	     \end{subfigure}
	     \par\bigskip
	     \begin{subfigure}[b]{0.45\textwidth}
	         \centering
	         \includegraphics[width=\textwidth]{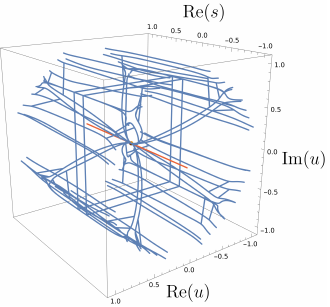}
	         \caption{Base change $u \mapsto u^{4}$.}
	         \label{fig:Deltaaffplotu4}
	     \end{subfigure}
	     \hfill
	     \begin{subfigure}[b]{0.45\textwidth}
	         \centering
	         \includegraphics[width=\textwidth]{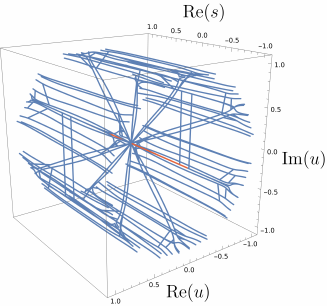}
	         \caption{Base change $u \mapsto u^{6}$.}
	         \label{fig:Deltaaffplotu6}
	     \end{subfigure}
	     \caption{(Real) three-dimensional cut $\left. \Delta \right|_{\mathrm{Im}(s) = 0}$ of the discriminant \eqref{eq:obscure-infinite-distance-limit-K3-Delta} for a particular, but generic, choice of coefficients in the patch $([s:1],u)$. The locus $\{u=0\}_{\mathcal{B}}$ is shown in red. We observe that its tangent intersection with the discriminant becomes transverse for a base change with high enough branching degree.}
	     \label{fig:DeltaaffplotK3}
	\end{figure}
	After a base change
	\begin{equation}
		\begin{aligned}
			\delta_{6}: D &\longrightarrow D\\
			u &\longmapsto u^{6}\,,
		\end{aligned}
	\end{equation}
	we obtain
	\begin{equation}
		(4,6,12) = \ord{\hat{\mathcal{Y}}}(f,g,\Delta)_{s=u=0} \leq \ord{\hat{Y}_{0}}\left( \left. f \right|_{u=0}, \left. g \right|_{u=0}, \left. \Delta \right|_{u=0} \right)_{s=u=0} = (4,6,12)\,,
	\end{equation}
	and the choice $b=0$ becomes generic, which can be seen by computing the intersection multiplicity or pictorially by looking at the progression of the intersection shown in \cref{fig:DeltaaffplotK3} as we increase the branching degree of the base change. Hence, the model can be resolved by blowing up $\hat{\mathcal{B}}$ along $\{ s=u=0 \}_{\mathcal{B}}$ and shifting the line bundle afterwards, leading to a two-component central fiber.
\end{example}

\begin{example}
\label{example:obscured-infinite-distance-limit-F0}
	Consider the Weierstrass model
	\begin{subequations}
	\begin{align}
		f &= s^2 t^4 v^4 \left(s^2 v^4+s^2 v^2 w^2+s^2 w^4+t^2 u^2 v^4\right)\,,\\
		g &= s^4 t^5 v^4 \left(s^3 w^8+s^2 t v^4 w^4+s^2 t v^2 w^6+t^3 u^2 v^8+t^3 u^2 v^2 w^6\right)\,,\\
		\Delta &= s^6 t^{10} v^8 p_{8,16}({[s:t],[v:w],u})\,,
	\end{align}
    \label{eq:example-obscured-infinite-distance-limit-F0}%
	\end{subequations}
	defining an elliptically fibered variety $\hat{\mathcal{Y}}$ over the base $\hat{\mathcal{B}} = \mathbb{F}_{0} \times D$ and supporting non-minimal singular fibers over the curve $\mathcal{S} \cap \mathcal{U}$ with
	\begin{equation}
		\ord{\hat{\mathcal{Y}}}(f,g,\Delta)_{s=u=0} = (4,6,12)\,.
	\end{equation}
	The family vanishing orders over all the other codimension-one loci are minimal, and over the codimension-two loci are either minimal or finite-distance non-minimal. The same is true for the component vanishing orders. Hence, at first sight, we seem not to have any obscured infinite-distance limits lurking in this model. Performing a (toric) blow-up of $\hat{\mathcal{B}}$ along $\mathcal{S} \cap \mathcal{U}$ we obtain the defining polynomials
	\begin{subequations}
	\begin{align}
		f_{b} &= s^2 t^4 v^4 \left(e_0^2 t^2 v^4+s^2 v^4+s^2 v^2 w^2+s^2 w^4\right)\,,\\
		g_{b} &= s^4 t^5 v^4 \left(e_1 s^3 w^8+e_0^2 t^3 v^8+e_0^2 t^3 v^2 w^6+s^2 t v^4 w^4+s^2 t v^2 w^6\right)\,,\\
		\Delta_{b} &= s^6 t^{10} v^8 p_{8,16,6}([s:t],[v:w],[s:e_{0}:e_{1}])  \label{eq:obscure-infinite-distance-limit-F0-Delta}\,,
	\end{align}
	\label{eq:example-obscured-infinite-distance-limit-F0-blow-up}%
	\end{subequations}
	with Stanley-Resiner ideal
	\begin{equation}
		\mathscr{I}_{\mathcal{B}} = \langle st, vw, se_{0}, te_{1} \rangle
	\end{equation}
	giving the Weierstrass model $\Pi_{\mathrm{ell}}: \mathcal{Y} \rightarrow \mathcal{B}$ of the resolved degeneration $\rho: \mathcal{Y} \rightarrow D$.
	\begin{figure}[t!]
	    \centering
	    \begin{tikzpicture}
			\node [] (0) at (0, 2.5) {};
			\node [] (1) at (0, -2.5) {};
			\node [] (2) at (6, 2.5) {};
			\node [] (3) at (6, -2.5) {};
			\node [] (4) at (-6, -2.5) {};
			\node [] (5) at (-6, 2.5) {};
			\node [] (6) at (-6, 1.5) {};
			\node [] (7) at (6, 1.5) {};
			\node [] (8) at (0, 1.5) {};
			\node [] (9) at (-6, -1.5) {};
			\node [] (10) at (0, -1.5) {};
			\node [] (11) at (6, -1.5) {};
			\node [] (12) at (5, 2.5) {};
			\node [label={[align=center, yshift=-1.5cm]\textcolor{diagLightBlue}{$s=0$}\\ \textcolor{diagLightBlue}{$(2,4,6)$}}] (13) at (5, -2.5) {};
			\node [label={[align=center, yshift=-1.5cm]\textcolor{diagLightPurple}{$t=0$}\\ \textcolor{diagLightPurple}{$(4,5,10)$}}] (14) at (-5, -2.5) {};
			\node [] (15) at (-5, 2.5) {};
			\node [] (16) at (4, 2.5) {};
	        \node [] (17) at (4, 2.5) {};
	        \node [label={[align=center, xshift=-6.5cm, yshift=-3.95cm]\textcolor{diagMediumRed}{$v=0$}\\ \textcolor{diagMediumRed}{$(4,4,8)$}}] (18) at (4, 2.5) {};
	        \node [label={[align=center, xshift=6.5cm, yshift=-3.95cm]\textcolor{diagMediumRed}{$v=0$}\\ \textcolor{diagMediumRed}{$(4,6,12)$}}] (19) at (-4, 2.5) {};
	        \node [label={[yshift=0cm]$\{e_{0} = 0\}_{\makebox[0pt]{$\scriptstyle \;\;\mathcal{B}$}}$}] (20) at (-3, 2.5) {};
	        \node [label={[yshift=0cm]$\{e_{1} = 0\}_{\makebox[0pt]{$\scriptstyle \;\;\mathcal{B}$}}$}] (21) at (3, 2.5) {};
	
			\draw [style=medium-red line] (9.center) to (10.center);
			\draw [style=medium-red zigzag] (10.center) to (11.center);
			\draw [style=light-purple line] (15.center) to (14.center);
			\draw [style=light-blue line] (12.center) to (13.center);
	
	        \draw [style=black line] (0.center) to (1.center);
			\draw [style=black line] (5.center) to (0.center);
			\draw [style=black line] (0.center) to (2.center);
			\draw [style=black line] (2.center) to (3.center);
			\draw [style=black line] (3.center) to (1.center);
			\draw [style=black line] (1.center) to (4.center);
			\draw [style=black line] (4.center) to (5.center);
		\end{tikzpicture}
	    \caption{Restrictions $\Delta'_{0}$ and $\Delta'_{1}$ of the (modified) discriminant for \cref{example:obscured-infinite-distance-limit-F0}, with the residual discriminant omitted for clarity. The printed vanishing orders correspond to the component vanishing orders in each component. We observe an obscured infinite-distance limit in the $B^{1}$ component.}
	    \label{fig:obscured-example}
	\end{figure}
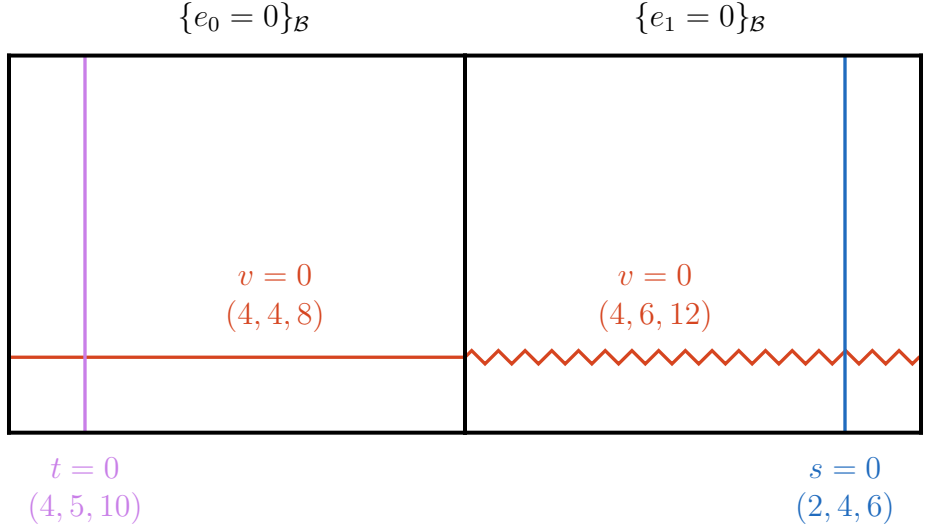	
	While all the family vanishing orders of the resolved model can be seen to be minimal in codimension-one, and minimal or finite-distance non-minimal in codimension-two, we have an obscured infinite-distance limit in the exceptional component arising from the blow-up, as can be seen from
	\begin{equation}
		(4,5,10) = \ord{\mathcal{Y}}(f_{b},g_{b},\Delta_{b})_{v=u=0} \leq \ord{Y^{1}}\left(\left. f \right|_{e_{1}=0}, \left. g \right|_{e_{1}=0},\left. \Delta \right|_{e_{1}=0}\right)_{v=u=0} (4,6,12)\,.
	\end{equation}

    Applying the procedure described above, we blow the model down to the $Y^{1}$ component and perform the base change
    \begin{equation}
		\begin{aligned}
			\delta_{6}: D &\longrightarrow D\\
			u &\longmapsto u^{6}\,,
		\end{aligned}
	\end{equation}
    after which we read over the curves $\mathcal{T} \cap \mathcal{U}$ and $\mathcal{V} \cap \mathcal{U}$ the non-minimal family vanishing orders
    \begin{subequations}
    \begin{align}
		\ord{\hat{\mathcal{Y}}}(f,g,\Delta)_{t=u=0} = (4,6,12)\,,\\
        \ord{\hat{\mathcal{Y}}}(f,g,\Delta)_{v=u=0} = (4,6,12)\,,
    \end{align}
	\end{subequations}
    i.e.\ the obscured infinite-distance limit is now apparent at the level of the family variety. Blowing up and performing the necessary line bundle shifts successively along the curves $\mathcal{V} \cap \mathcal{U}$, $\mathcal{T} \cap E_{0}$ and $\mathcal{T} \cap E_{1}$ leads to the Weierstrass model $\Pi_{\mathrm{ell}}: \mathcal{Y} \rightarrow \mathcal{B}$ given by the defining polynomials\footnote{Blowing up in a different order gives the same defining polynomials $f_{b}$, $g_{b}$ and $\Delta_{b}$, but a different Stanley-Reisner ideal. For example, we could have performed the blow-ups along the curves $\mathcal{T} \cap \mathcal{U}$, $\mathcal{T} \cap E_{1}$ and $\mathcal{V} \cap E_{3}$, obtaining the Stanley-Reisner ideal
    \begin{equation}
    		\mathscr{I}_{\mathcal{B}} = \langle s t,e_1 s,e_2 s,e_0 t,e_1 t,e_3 t,v w,e_0 v,e_3 w,e_0 e_2,e_2 e_3 \rangle\,.
    \end{equation}
    While with this resolution no infinite-distance non-minimal family vanishing orders are found, we still have the infinite-distance non-minimal component vanishing orders
    \begin{equation}
    		\ord{Y^{1}}(\left. f_{b} \right|_{e_{1}=0}, \left. g_{b} \right|_{e_{1}=0}, \left. \Delta_{b} \right|_{e_{1}=0})_{e_{1}=v=0} = (4,6,12)\,.
    \end{equation}
    These are not problematic, since they can be removed by performing the flop $\{ e_{1} = v = 0 \}_{\mathcal{B}} \leftrightarrow \{ e_{2} = e_{3} = 0 \}_{\mathcal{B}}$, which constitutes a valid modification of the degeneration. The resulting $\mathcal{B}$ is, however, singular. Its singularity can be removed by performing the flop $\{ e_{2} = v = 0 \}_{\mathcal{B}} \leftrightarrow \{ e_{3} = t = 0 \}_{\mathcal{B}}$, at which point we obtain the resolution given by \eqref{eq:obscure-infinite-distance-limit-F0-sr-ideal} up to a relabelling of the exceptional components. The fact that the two resolutions are related by some flops is expected, see \cref{rem:blow-up-order-flops}. In the body of the text we have chosen the resolution process that directly gives the geometrical representative for the central fiber of the degeneration with the most favourable properties, in order to simplify the discussion.}
    \begin{subequations}
		\begin{align}
			f_{b} &= s^2 t^4 v^4 \left(e_1^4 s^2 v^4+e_1^2 s^2 v^2 w^2+e_1^4 e_2^2 e_3^4 t^2 v^4+s^2 w^4\right)\,,\\
			g_{b} &= s^4 t^5 v^4 \left(e_0^2 e_2 s^3 w^8+e_1^2 s^2 t v^4 w^4+e_1^6 e_2^2 e_3^4 t^3 v^8+e_2^2 e_3^4 t^3 v^2 w^6+s^2 t v^2 w^6\right)\,,\\
			\Delta_{b} &= s^6 t^{10} v^8 p_{8,16,4,2,2}([s:t],[v:w],[v:e_{0}:e_{1}],[t:e_{0}:e_{2}],[t:e_{2}:e_{3}])\,,
		\end{align}
	\end{subequations}
    together with the Stanley-Reisner ideal
    \begin{equation}
        \mathscr{I}_{\mathcal{B}} = \langle s t,e_2 s,e_3 s,e_0 t,e_2 t,v w,e_0 v,e_2 v,e_3 v,e_1 w,e_0 e_3 \rangle\,.
    \label{eq:obscure-infinite-distance-limit-F0-sr-ideal}
    \end{equation}
    One can check that no infinite-distance non-minimal singularities are found at the level of the family variety, its restriction to the components or its restriction to the intersections of components; the base change has indeed allowed us to obtain a resolved degeneration $\rho: \mathcal{Y} \rightarrow \mathcal{B}$ free of obscured infinite-distance limits.

    While the family and component vanishing orders did not allow us to detect the presence of the obscured infinite-distance limit, the unresolved degeneration \eqref{eq:example-obscured-infinite-distance-limit-F0} already hints at its existence. To see how, let us plot the restriction $\left. \Delta_{b} \right|_{B_{0}}$ of \eqref{eq:obscure-infinite-distance-limit-F0-Delta}, which we do in \cref{fig:DeltaaffplotF0}.
	\begin{figure}[t!]
	     \centering
	     \begin{subfigure}[b]{0.45\textwidth}
	         \centering
	         \includegraphics[width=\textwidth]{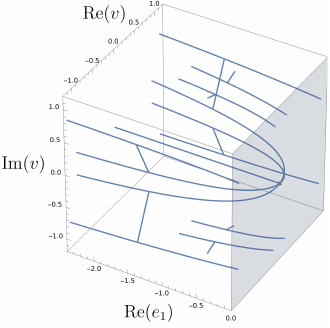}
	         \caption{Patch $\left. ([1:1],[v:1],[1:0:e_{1}]) \right|_{\mathrm{Im}(e_{1})=0}$.}
	         \label{fig:Deltaaffplote0}
	     \end{subfigure}
	     \hfill
	     \begin{subfigure}[b]{0.45\textwidth}
	         \centering
	         \includegraphics[width=\textwidth]{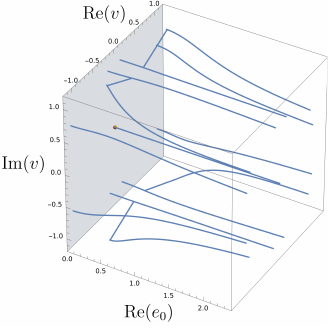}
	         \caption{Patch $\left. ([1:1],[v:1],[1:e_{0}:0]) \right|_{\mathrm{Im}(e_{0})=0}$.}
	         \label{fig:Deltaaffplote1}
	     \end{subfigure}
	     \caption{Plot of the restriction $\left. \Delta \right|_{B_{0}}$ of the discriminant \eqref{eq:obscure-infinite-distance-limit-F0-Delta}. The intersections of the discriminant with the interface curve $E_{0} \cap E_{1} \subset B_{0} \subset \mathcal{B}$ need to agree from both sides. This implies that the point $\{e_{0} = e_{1} = v = 0\}_{\mathcal{B}}$ is to be non-minimal in the $B^{0}$ component if we compute the vanishing orders of the defining polynomials along the (in this case non-generic) curve $E_{0} \cap E_{1}$, i.e.\ if we compute the interface vanishing orders.}
	     \label{fig:DeltaaffplotF0}
	\end{figure}
    The component of $\left. \Delta_{b} \right|_{E_{1}}$ responsible for the obscured infinite-distance limit intersects the curve over which the base components intersect at the point $\{e_{0} = e_{1} = v = 0\}_{\mathcal{B}}$. The component vanishing orders for this point computed from both sides disagree, having
    \begin{equation}
        (4,5,10) = \ord{Y^{0}}(\left. f_{b} \right|_{e_{0}}, \left. g_{b} \right|_{e_{0}}, \left. \Delta_{b} \right|_{e_{0}})_{e_{0}=e_{1}=v=0} \neq \ord{Y^{1}}(\left. f_{b} \right|_{e_{1}}, \left. g_{b} \right|_{e_{1}}, \left. \Delta_{b} \right|_{e_{1}})_{e_{0}=e_{1}=v=0} = (4,6,12)\,.
    \end{equation}
    What agrees are the vanishing orders found when we restrict the model to the intersection of the two components from either side, i.e.\ the interface vanishing orders
    \begin{equation}
        \ord{Y^{0} \cap Y^{1}}(\left. f_{b} \right|_{e_{0} = e_{1} = 0}, \left. g_{b} \right|_{e_{0} = e_{1} = 0}, \left. \Delta_{b} \right|_{e_{0} = e_{1} = 0})_{e_{0} = e_{1} = v = 0} = (4,6,12)\,.
    \label{eq:example-obscured-infinite-distance-limit-F0-interface-vanishing-orders}
    \end{equation}
    We see that the curve $E_{0} \cap E_{1}$ shaded in grey in \cref{fig:DeltaaffplotF0} is the generic slice from the point of view of the $B^{1}$ component, while it is non-generic from the $B^{0}$ side, explaining the discrepancy.
    
    At the end of \cref{sec:orders-of-vanishing}, we mentioned that these interface vanishing orders can be directly computed in the unresolved degeneration. Since the point on the interface curve with non-minimal interface vanishing orders is so tightly related to the obscured infinite-distance limit, this gives us a way to detect it even before starting to blow-up the base. Namely, a point with non-minimal interface vanishing orders, but minimal family vanishing orders, on top of a curve supporting non-minimal elliptic fibers at the level of the family variety signals the presence of an obscured infinite-distance limit in the exceptional components arising from the base blow-ups centred at the aforementioned non-minimal curve.
    
    In the example under scrutiny, we can indeed obtain the interface vanishing orders \eqref{eq:example-obscured-infinite-distance-limit-F0-interface-vanishing-orders} directly from \eqref{eq:example-obscured-infinite-distance-limit-F0}. To achieve this, recognize that, if we denote the blow-up map leading to \eqref{eq:example-obscured-infinite-distance-limit-F0-blow-up} by $\pi: \mathcal{B} \rightarrow \hat{\mathcal{B}}$ we have for this model that
    \begin{subequations}
    \begin{align}
    		\pi_{*}(F_{0}) &= \left. F \right|_{\mathcal{U}} - 4 \left( \mathcal{S} \cap \mathcal{U} \right)\,,\\
    		\pi_{*}(G_{0}) &= \left. G \right|_{\mathcal{U}} - 6 \left( \mathcal{S} \cap \mathcal{U} \right)\,,\\
    		\pi_{*}(\Delta_{0}) &= \left. \Delta \right|_{\mathcal{U}} - 12 \left( \mathcal{S} \cap \mathcal{U} \right)\,,
    \end{align}
    \end{subequations}
    see \eqref{eq:F0-F-G0-G-relations}. We obtain then the interface vanishing orders \eqref{eq:example-obscured-infinite-distance-limit-F0-interface-vanishing-orders} by computing
    \begin{equation}
        \ord{\pi^{*}(\mathcal{S} \cap\;\! \mathcal{U})} \left( \left. \frac{f}{s^{4}} \right|_{u=s=0}, \left. \frac{g}{s^{6}} \right|_{u=s=0}, \left. \frac{\Delta}{s^{12}} \right|_{u=s=0} \right)_{v=0} = (4,6,12)\,,
    \end{equation}
    and therefore the obscured infinite-distance limit could have indeed been detected\footnote{Performing the base change $u \mapsto u^{2}$ directly in the original model \eqref{eq:example-obscured-infinite-distance-limit-F0} allows us to blow-up and line bundle shift along the curves $\mathcal{S} \cap \mathcal{U}$, $\mathcal{S} \cap E_{1}$ and $\mathcal{V} \cap E_{2}$. After this resolution, no infinite-distance non-minimal family vanishing orders are found, but the infinite-distance non-minimal component vanishing orders
    \begin{equation}
    		\ord{Y^{1}}(\left. f_{b} \right|_{e_{1}=0}, \left. g_{b} \right|_{e_{1}=0}, \left. \Delta_{b} \right|_{e_{1}=0})_{e_{1} = v = 0} = (4,6,12)
    \end{equation}
    are still present. These can be removed through the flop $\{ e_{1} = v = 0 \}_{\mathcal{B}} \leftrightarrow \{ e_{0} = e_{3} = 0 \}_{\mathcal{B}}$, that leads, however, to a singular $\mathcal{B}$.} at the start of the discussion.
\end{example}

A variation of what we have just observed in \cref{example:obscured-infinite-distance-limit-F0} would be given by a model in which the interface curve between the two components presents a point with non-minimal interface vanishing orders, but no codimension-one (obscured or conventional) infinite-distance limits are observed in either of the two components. This also corresponds, in fact, to a form of obscured infinite-distance limit in codimension-one. Upon performing a base change, blowing up the base to arrive at the resolve degeneration leads to more than two-components for its central fiber, and in the ones arising from the intermediate blow-ups we will exhibit an obscured infinite-distance limit of the form described earlier in the section. This special type of obscured infinite-distance limits and their interpretation in the heterotic dual models (whenever these are available) are discussed in greater length in \cite{ALWPart2}.

The signature of obscured infinite-distance limits as points with non-minimal interface vanishing orders on top of the blow-up centres of the original degeneration will be relevant later on in \cref{sec:single-infinite-distance-limits-and-open-chain-resolutions}.
%auto-ignore

\section{Resolution trees}
\label{sec:res-trees}

In this appendix, we generalise the discussion of \cref{sec:geometry-components-single-limit} beyond the class of single infinite-distance limits and their open-chain resolutions. First, we will prove and exemplify \cref{prop:component-geometry-general}, which identifies the components of more general infinite-distance limits as Hirzebruch surfaces or suitable blow-ups thereof, which form a resolution tree rather than an open chain. Then we characterise the Weierstrass models over the components of these resolutions.

\subsection{Geometry of the components}
\label{sec:geometry-components-arbitrary-limit}

Let us use the same notation as in \cref{sec:geometry-components-single-limit}. In order to determine the geometry of the $\{B^{p}\}_{0 \leq p \leq P}$ components in the resolution of a general degeneration, we need to produce the analogue of \cref{prop:component-geometry-single-bis} after dropping the assumption of vanishing intersection among the $\{C_{p}\}_{1 \leq p \leq P}$ curves. This results in the following proposition.
\componentgeometrygeneral*
\begin{proof}
	The proof of \cref{prop:component-geometry-single} applies until we reach the computation of $\mathcal{N}_{E_{i}/\mathrm{Bl}_{p-1}(\hat{B})}$, which needs to be modified. In the more general situation, we have that
	\begin{equation}
		\mathcal{N}_{E_{i}/\mathrm{Bl}_{p-1}(\hat{B})} = \left. \mathcal{O}_{\mathrm{Bl}_{p-1}(\hat{B})}\left( E_{i} \right) \right|_{E_{i}} = E_{i} \cdot_{\mathrm{Bl}_{p-1}(\hat{B})} E_{i} = - \sum_{\substack{q = 0\\ q \neq i}}^{p-1} E_{i} \cdot_{\mathrm{Bl}_{p-1}(\hat{B})} E_{q} = - \sum_{\substack{q = 0\\ q \neq i}}^{p-1} \left. E_{q} \right|_{E_{i}}\,,
	\end{equation}
	leading to
	\begin{equation}
		\left. \mathcal{N}_{B^{i}/\mathrm{Bl}_{p-1}(\hat{B})} \right|_{C_{p}} = \mathcal{O}_{\mathbb{P}^{1}} \left( -m_{p} \right)\,,\qquad m_p:=\sum_{\substack{q = 0\\q \neq i}}^{p-1} \left. E_{q} \right|_{E_{i}} \cdot_{B^{i}} C_{p}\,.
	\end{equation}
	Altogether, we find
	\begin{equation}
	\begin{aligned}
		E_{p} = \mathbb{P} \left( \mathcal{O}_{\mathbb{P}^{1}} \oplus \mathcal{O}_{\mathbb{P}^{1}} \left( \left| n_{p} \right| \right) \right) = \mathbb{F}_{\left| n_{p} \right|}\,,\qquad n_{p} := C_{p} \cdot_{B^{i}} C_{p} + \sum_{\substack{q = 0\\q \neq i}}^{p-1} \left. E_{q} \right|_{E_{i}} \cdot_{B^{i}} C_{p} \,. 
	\end{aligned}
	\end{equation}
	Consider now a fixed component $B^{q}$. Since $C_{p}$ is an irreducible curve, if
	\begin{equation}
		\mathrm{codim}_{B^{i}} \left( \left. E_{q} \right|_{E_{i}} \cdot_{B^{i}} C_{p} \right) = 1\,,
	\end{equation}
	we have that $C_{p} \subset \left. E_{q} \right|_{E_{i}}$, and therefore the blow-up centre $C_{p} \subset B^{q}$. Then
	\begin{equation}
		\pi^{*} \left( E_{q} \right) = E_{q} + E_{p}\,,
	\end{equation}
	and with our notation $B^{q}$ represents the strict transform of the former $B^{p}$ after the blow-up $\pi_{p}: \mathrm{Bl}_{p}(\hat{\mathcal{B}}) \rightarrow \mathrm{Bl}_{p-1}(\hat{\mathcal{B}})$. Hence, the surface to which $B^{q}$ refers has not changed. If instead
	\begin{equation}
		\mathrm{codim}_{B^{i}} \left( \left. E_{q} \right|_{E_{i}} \cdot_{B^{i}} C_{p} \right) = 1\,,
	\end{equation}
	only a set of points of the blow-up centre sits in $B^{q}$. The blow-up $\pi_{p}: \mathrm{Bl}_{p}(\hat{\mathcal{B}}) \rightarrow \mathrm{Bl}_{p-1}(\hat{\mathcal{B}})$ then induces a surface blow-up $\pi_{p,q}: \mathrm{Bl}_{\left. E_{q} \right|_{E_{i}} \cdot_{B^{i}} C_{p}}(B^{q}) \rightarrow B^{q}$ along the points $\left. E_{q} \right|_{E_{i}} \cdot_{B^{i}} C_{p} \subset B^{q}$, and we relabel the irreducible component to be $B^{q} := \mathrm{Bl}_{\left. E_{q} \right|_{E_{i}} \cdot_{B^{i}} C_{p}}(B^{q})$.
\end{proof}
\begin{remark}
\label{rem:blow-up-order-flops}
	Consider a set of intersecting curves of non-minimal singular fibers contained in a given base component. The order in which we blow-up the base family variety along them matters for the resulting geometry, as the last exceptional component will be a Hirzebruch surface, while the rest will be blow-ups thereof. Since any order in which we perform the blow-ups is a valid modification of the original degeneration, all the resulting geometrical representatives of the central fiber correspond to the same limit. In fact, they are related to each other by flopping the exceptional curves arising over the intersection points of the blow-up centres.
\end{remark}

Let us see how this geometry is realized in some concrete examples. We start with a model in which the base of the central fiber of the degeneration contains two curves of non-minimal singular fibers intersecting at a point and leading to a three-component central fiber for the resolved family variety, focusing on how the two possible blow-up orders are related by a flop.
\begin{example}
\label{example:geometry-components-general}
	Consider the Weierstrass model describing an elliptically fibered variety $\hat{\mathcal{Y}}$ over the base $\hat{\mathcal{B}} = \mathbb{F}_{3} \times D$ given by
	\begin{subequations}
	\begin{align}
		f &= s^4 t^4 v^4 \left(u v^4+u w^4+v^4+v^2 w^2+w^4\right)\,,\\
		g &= s^5 t^5 v^5 \left(s^2 u v^{10}+s^2 u w^{10}+s^2 v^{10}+s t v w^6+t^2 u w^4\right)\,,\\
		\Delta &= s^{10} t^{10} v^{10} p_{4,20}([s:t],[v:w:t],u)\,.
	\end{align}
	\end{subequations}
	It supports non-minimal singular fibers over the curves $\mathcal{S} \cap \mathcal{U}$ and $\mathcal{V} \cap \mathcal{U}$, as can be seen from
	\begin{subequations}
	\begin{align}
	    \ord{\hat{\mathcal{Y}}}(f,g,\Delta)_{s=u=0} = (4, 6, 12)\,,\\
	    \ord{\hat{\mathcal{Y}}}(f,g,\Delta)_{v=u=0} = (4, 6, 12)\,.
	\end{align}
	\end{subequations}
	Since the example is amenable to a toric treatment, we will display the information about the geometry of the different components in terms of their toric fan, in order to be more concise. The starting toric fan describing	 $\hat{\mathcal{B}} = \mathbb{F}_{3} \times D$ is
	\begin{equation}
		v = (1,0,0)\,,\quad t = (0,1,0)\,,\quad w = (-1,-n,0)\,,\quad s = (0,-1,0)\,,\quad u = (0,0,1)\,,
	\end{equation}
	in the lattice
	\begin{equation}
	    N := \mathbb{Z} \langle (1,0,0),\, (0,1,0),\, (0,0,1) \rangle\,.
	\end{equation}
	Performing the two (toric) blow-ups along these curves together with the appropriate line bundle shifts, we arrive at
	\begin{subequations}
	\begin{align}
		f_{b} &= s^4 t^4 v^4 \left(e_0 e_{s} e_{v}^5 v^4+e_0 e_{s} e_{v} w^4+e_{v}^4 v^4+e_{v}^2 v^2 w^2+w^4\right)\,,\\
		g_{b} &= s^5 t^5 v^5 \left(e_0 e_{s}^2 e_{v}^{10} s^2 v^{10}+e_0 e_{s}^2 s^2 w^{10}+e_0 t^2 w^4+e_{s} e_{v}^9 s^2 v^{10}+s t v w^6\right)\,,\\
		\Delta_{b} &= s^{10} t^{10} v^{10} p_{4,20,2,2}([s:t],[v:w:t],[s:e_{0},e_{s}],[v:e_{0},e_{v}])\,,
	\end{align}
	\label{eq:example2-defining-polynomials-blow-up}%
	\end{subequations}
	where we have denoted the exceptional coordinates by $e_{s}$ and $e_{v}$ to keep track of their origin. Blowing up first along $\mathcal{S} \cap \mathcal{U}$ and then along $\mathcal{V} \cap E_{0}$ or vice versa only affects the resulting Stanley-Reisner ideal, yielding
	\begin{align}
		\mathscr{I}_{\mathcal{B}}^{sv} &= \langle st, vw, se_{0}, te_{s}, ve_{0}, we_{v}, se_{v} \rangle\,,\\
		\mathscr{I}_{\mathcal{B}}^{vs} &= \langle st, vw, ve_{0}, we_{v}, se_{0}, te_{s}, ve_{s} \rangle\,,
	\end{align}
	respectively, meaning that the resulting toric fans $\Sigma_{\mathcal{B}}^{sv}$ and $\Sigma_{\mathcal{B}}^{vs}$ only differ by some 2-cones. $\mathscr{I}_{\mathcal{B}}^{sv}$ and $\mathscr{I}_{\mathcal{B}}^{vs}$ are therefore related to each other, as can be seen by comparing the fans in \cref{fig:geometry-components-blow-up-order}, by flopping the curves $\{v = e_{s} = 0\} \leftrightarrow \{s = e_{v} = 0\}$ in the base family variety $\mathcal{B}$, as we commented on in \cref{rem:blow-up-order-flops}. In both cases, we observe that there no longer are infinite-distance non-minimal singularities present.
	\begin{figure}[t!]
	     \centering
	     \begin{subfigure}[b]{0.45\textwidth}
	         \centering
	         \includegraphics[width=\textwidth]{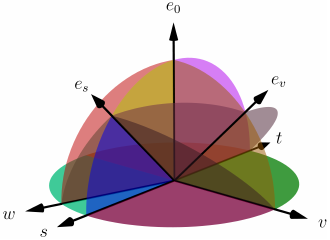}
	         \caption{Toric fan $\Sigma_{\mathcal{B}}^{sv}$ of $\mathbb{F}_{3} \times \mathbb{C}$ blown up along $\mathcal{S} \cap \mathcal{U}$ and then along $\mathcal{V} \cap E_{0}$.}
	         \label{fig:F3xDbhorverfan}
	     \end{subfigure}
	     \hfill
	     \begin{subfigure}[b]{0.45\textwidth}
	         \centering
	         \includegraphics[width=\textwidth]{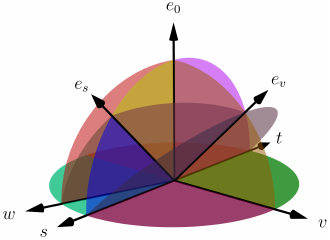}
	         \caption{Toric fan $\Sigma_{\mathcal{B}}^{vs}$ of $\mathbb{F}_{3} \times \mathbb{C}$ blown up along $\mathcal{V} \cap \mathcal{U}$ and then along $\mathcal{S} \cap E_{0}$.}
	         \label{fig:F3xDbverhorfan}
	     \end{subfigure}
	     \caption{Toric fans associated to the two possible blow-up orders. They are related to each other by exchanging the 2-cones $(v,e_{s})$ and $(s,e_{v})$, i.e.\ flopping the corresponding curves.}
	     \label{fig:geometry-components-blow-up-order}
	\end{figure}
	
	Let us now focus on the geometry of the components obtained by first blowing up along $\mathcal{S} \cap \mathcal{U}$ and then along $\mathcal{V} \cap E_{0}$. The three components $\{B^{0}, B^{s}, B^{v}\}$ of $B_{0}$ are also toric varieties, and their toric fans can be obtained from $\Sigma_{\mathcal{B}}^{sv}$ by computing the orbit closure of the edges $\{e_{0}, e_{s}, e_{v}\}$. The results are shown in \cref{fig:geometry-components-sv-ev}. The toric computation tells us that the base components of the central fiber correspond to the surfaces
	\begin{equation}
		B^{0} = \mathbb{F}_{3}\,,\qquad B^{s} = \mathrm{Bl}_{1}(\mathbb{F}_{3})\,,\qquad B^{v} = \mathbb{F}_{1}\,.
	\end{equation}
	We now compare these with the expectations from \cref{prop:component-geometry-general}. The two blow-up centres are contained in $\hat{B}_{0} = \mathbb{F}_{3}$, and therefore the component $B^{0}$ is just its strict transform. To obtain the other two components, we first label the blow-up centre in order, i.e.\ $C_{1} = \mathcal{S} \cap \mathcal{U}$ and $C_{2} = \mathcal{V} \cap E_{0}$. The two curves intersect at one point, namely at
	\begin{equation}
		C_{1} \cap C_{2} = \mathcal{S} \cap \mathcal{V} \cap \mathcal{U}' = \{s=v=u=0\}\,.
	\end{equation}
	The first blow-up produces
	\begin{equation}
		B^{s} = \mathbb{F}_{\left| C_{1} \cdot_{\hat{B}_{0}} C_{1} \right|} = \mathbb{F}_{3}\,.
	\end{equation}	
	The second blow-up produces
	\begin{equation}
		B^{v} = \mathbb{F}_{\left| C_{2} \cdot_{B^{0}} C_{2} + \left. E_{s} \right|_{E_{0}} \cdot_{B^{0}} C_{2} \right|} = \mathbb{F}_{\left| C_{2} \cdot_{B^{0}} C_{2} + C_{1} \cdot_{B^{0}} C_{2} \right|} = \mathbb{F}_{1}\,,
	\end{equation}
	and the former component $B^{s}$ must be blown-up at the point $\{s=v=0\}_{E_{s}}$, yielding the surface $B^{s} = \mathrm{Bl}_{1}(\mathbb{F}_{3})$. This agrees with the toric computation.
    \begin{figure}[t!]
	     \centering
	     \hspace{0.0cm}
	     \begin{subfigure}[b]{0.31\textwidth}
	         \centering
	         \includegraphics[width=\textwidth]{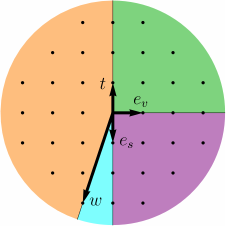}
	         \caption{Orbit closure of $e_{0}$ in $\Sigma_{\mathcal{B}}^{sv}$.}
	         \label{fig:geometry-components-sv-e0}
	     \end{subfigure}
	     \hfill
	     \begin{subfigure}[b]{0.31\textwidth}
	         \centering
	         \includegraphics[width=\textwidth]{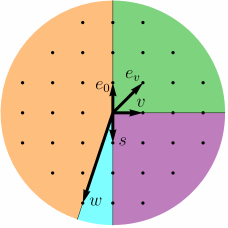}
	         \caption{Orbit closure of $e_{s}$ in $\Sigma_{\mathcal{B}}^{sv}$.}
	         \label{geometry-components-sv-es}
	     \end{subfigure}
	     \hfill
	     \begin{subfigure}[b]{0.31\textwidth}
	         \centering
	         \includegraphics[width=\textwidth]{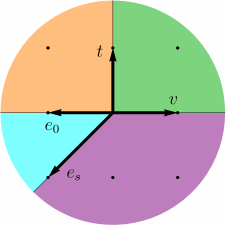}
	         \caption{Orbit closure of $e_{v}$ in $\Sigma_{\mathcal{B}}^{sv}$.}
	         \label{geometry-components-sv-ev}
	     \end{subfigure}
	     \hspace{0.0cm}
	     \caption{Geometry of the components of $\mathcal{B}$ obtained from $\hat{B}$ by first blowing up along $\mathcal{S} \cap \mathcal{U}$ and then along $\mathcal{V} \cap E_{0}$.}
	     \label{fig:geometry-components-sv-ev}
	\end{figure}

	Consider now the case in which we blow up first along $\mathcal{V} \cap \mathcal{U}$ and the along $\mathcal{S} \cap E_{0}$. The toric fans obtained from the orbit closure of the edges $\{e_{0}, e_{s}, e_{v}\}$ in $\Sigma_{\mathcal{B}}^{vs}$ are depicted in \cref{fig:geometry-components-vs-ev}. We identify the surfaces corresponding to the base components of the central fiber to be
	\begin{equation}
		B^{0} = \mathbb{F}_{3}\,,\qquad B^{s} = \mathbb{F}_{2}\,,\qquad B^{v} = \mathrm{Bl}_{1}(\mathbb{F}_{1})\,.
	\end{equation}
	Again, $B^{0}$ is simply the strict transform of $\hat{B}_{0} = \mathbb{F}_{3}$, since the latter contains both blow-up centres. We compare with \cref{prop:component-geometry-general}, now labelling the curves $C_{1} = \mathcal{V} \cap \mathcal{U}$ and $C_{2} = \mathcal{S} \cap E_{0}$, in agreement with the new blow-up order. From the first blow-up we obtain
	\begin{equation}
		B^{v} = \mathbb{F}_{\left| C_{1} \cdot_{\hat{B}_{0}} C_{1} \right|} = \mathbb{F}_{0}\,.
	\end{equation}
	The second blow-up produces
	\begin{equation}
		B^{s} = \mathbb{F}_{\left|  C_{2} \cdot_{B^{0}} C_{2} + \left. E_{v} \right|_{E_{0}} \cdot_{B^{0}} C_{2} \right|} = \mathbb{F}_{\left|  C_{2} \cdot_{B^{0}} C_{2} + C_{1} \cdot_{B^{0}} C_{2} \right|} = \mathbb{F}_{2}\,,
	\end{equation}
	and also blows-up the former component $B^{s}$ at the point $\{s=v=0\}_{E_{v}}$, leading to $B^{v} = \mathrm{Bl}_{1}(\mathbb{F}_{0})$, in agreement with the toric computation.
	
	The effect of the flop $\{v = e_{s} = 0\} \leftrightarrow \{s = e_{v} = 0\}$ connecting the two base family varieties of the resolved degeneration can be seen at the level of components in \cref{fig:geometry-components-sv-ev,fig:geometry-components-vs-ev}, which only differ by the addition or subtraction of certain edges to the toric fans of $B^{s}$ and $B^{v}$.

    \begin{figure}[t!]
	     \centering
	     \hspace{0.0cm}
	     \begin{subfigure}[b]{0.31\textwidth}
	         \centering
	         \includegraphics[width=\textwidth]{figures/geometry-components-e0.pdf}
	         \caption{Orbit closure of $e_{0}$ in $\Sigma_{\mathcal{B}}^{vs}$.}
	         \label{fig:geometry-components-vs-e0}
	     \end{subfigure}
	     \hfill
	     \begin{subfigure}[b]{0.31\textwidth}
	         \centering
	         \includegraphics[width=\textwidth]{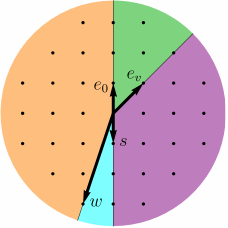}
	         \caption{Orbit closure of $e_{s}$ in $\Sigma_{\mathcal{B}}^{vs}$.}
	         \label{geometry-components-vs-es}
	     \end{subfigure}
	     \hfill
	     \begin{subfigure}[b]{0.31\textwidth}
	         \centering
	         \includegraphics[width=\textwidth]{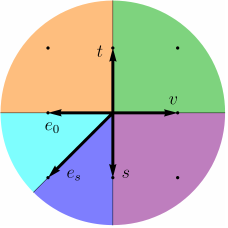}
	         \caption{Orbit closure of $e_{v}$ in $\Sigma_{\mathcal{B}}^{vs}$.}
	         \label{geometry-components-vs-ev}
	     \end{subfigure}
	     \hspace{0.0cm}
	     \caption{Geometry of the components of $\mathcal{B}$, obtained from $\hat{B}$ by first blowing up along $\mathcal{V} \cap \mathcal{U}$ and then along $\mathcal{S} \cap E_{0}$.}
	     \label{fig:geometry-components-vs-ev}
	\end{figure}
\end{example}

At each step in the previous example, only the centres of the already performed blow-ups were relevant in order to determine the geometry of the new components. More generally, two components need not intersect along a blow-up centre, which is why \eqref{eq:general-blow-up-Fn} cannot be expressed only in terms of the curves $\{C_{p}\}_{1 \leq p \leq P}$. This is showcased in the next example, in which we also see that the geometry of the original component $\hat{B}_{0}$ can be affected by the blow-up process as well, the end product not always simply corresponding to the original type of surface.
\begin{example}
\label{example:geometry-components-general-additional}
	Consider the Weierstrass model
	\begin{subequations}
	\begin{align}
		f &= t^3 \left(s^5 u^3 v^9+s^4 t u^4 w^8+s^4 t v^8+s^4 t v^4 w^4+t^5 w^4\right)\,,\\
		g &= t^4 \left(s^8 u^6 v^{14}+s^8 u^6 v^4 w^{10}+s^6 t^2 u^6 w^{12}+s^6 t^2 v^{12}+s^6 t^2 v^6 w^6+t^8 w^6\right)\,,\\
		\Delta &= t^8 p_{16,28}([s:t],[v:w:t],u)\,,
	\end{align}
	\end{subequations}
	defining an elliptically fibered variety $\hat{\mathcal{Y}}$ over the base $\hat{\mathcal{B}} = \mathbb{F}_{1} \times D$ and supporting non-minimal singular fibers over the curve $\mathcal{T} \cap \mathcal{U}$, as can be seen from
	\begin{equation}
	    \ord{\hat{\mathcal{Y}}}(f,g,\Delta)_{t=u=0} = (4, 6, 12)\,.
	\end{equation}
	One possible sequence of blow-ups leading to the resolved degeneration $\rho: \mathcal{Y} \rightarrow D$ is
	\begin{subequations}
	\begin{alignat}{3}
		&\pi_{1}: \mathrm{Bl}_{1}(\hat{\mathcal{B}}) \rightarrow \hat{\mathcal{B}}\,,\qquad &\mathrm{along}\qquad &C_{1} = \mathcal{T} \cap \mathcal{U}\,,\\
		&\pi_{2}: \mathrm{Bl}_{2}(\hat{\mathcal{B}}) \rightarrow \mathrm{Bl}_{1}(\hat{\mathcal{B}})\,,\qquad &\mathrm{along}\qquad &C_{2} = \mathcal{V} \cap E_{1}\,,\\
		&\pi_{3}: \mathrm{Bl}_{3}(\hat{\mathcal{B}}) \rightarrow \mathrm{Bl}_{2}(\hat{\mathcal{B}})\,,\qquad &\mathrm{along}\qquad &C_{3} = \mathcal{T} \cap E_{1}\,,\\
		&\pi_{4}: \mathrm{Bl}_{4}(\hat{\mathcal{B}}) \rightarrow \mathrm{Bl}_{3}(\hat{\mathcal{B}})\,,\qquad &\mathrm{along}\qquad &C_{4} = \mathcal{T} \cap E_{3}\,,\\
		&\pi_{5}: \mathcal{B} \rightarrow \mathrm{Bl}_{4}(\hat{\mathcal{B}})\,,\qquad &\mathrm{along}\qquad &C_{5} = \mathcal{T} \cap E_{2}\,.
	\end{alignat}
	\label{eq:geometry-components-general-additional-blow-up-1}%
	\end{subequations}
	This produces the defining polynomials of the blown up Weierstrass model
	\begin{subequations}
	\begin{align}
		f_{b} &= t^3 \left(e_0^3 e_1^2 e_2^7 e_3 e_5^6 s^5 v^9+e_2^4 e_5^4 s^4 t v^8+e_0^4 e_1^4 e_3^4 e_4^4 s^4 t w^8+e_1^4 e_3^8 e_5^4 e_4^{12} t^5 w^4+s^4 t v^4 w^4\right)\,,\\
        \begin{split}
		g_{b} &= t^4 \left(e_0^6 e_1^4 e_2^{12} e_3^2 e_5^{10} s^8 v^{14}+e_0^6 e_1^4 e_2^2 e_3^2 s^8 v^4 w^{10}+e_2^6 e_5^6 s^6 t^2 v^{12}+e_0^6 e_1^6 e_3^6 e_4^6 s^6 t^2 w^{12}\right.\\
		&\quad \left.+e_1^6 e_3^{12} e_5^6 e_4^{18} t^8 w^6+s^6 t^2 v^6 w^6\right)\,,
        \end{split}\\
		\Delta_{b} &= t^8 p_{16,28,4,12,4,4,4}(s,t,v,w,e_{0},e_{1},e_{2},e_{3},e_{4},e_{5})\,,
	\end{align}
	\end{subequations}
	with the subscripts in $p_{16,28,4,12,4,4,4}(s,t,v,w,e_{0},e_{1},e_{2},e_{3},e_{4},e_{5})$ referring to the homogeneous degrees in the coordinates $[s:t]$, $[v:w:t]$, $[t:e_{0}:e_{1}]$, $[v:e_{1}:e_{2}]$, $[t:e_{1}:e_{3}]$, $[t:e_{3}:e_{4}]$ and $[t:e_{2}:e_{5}]$, respectively. The resulting Stanley-Reisner ideal is
	\begin{equation}
	\begin{aligned}
		\mathscr{I} = \langle &s t,e_1 s,e_2 s,e_3 s,e_4 s,e_5 s,e_0 t,e_1 t,e_2 t,e_3 t,v w,\\
		&e_1 v,e_3 v,e_4 v,e_2 w,e_5 w,e_0 e_3,e_0 e_4,e_0 e_5,e_1 e_4,e_1 e_5,e_3 e_5 \rangle\,.
	\end{aligned}
	\end{equation}
	
	By repeated application of \cref{prop:component-geometry-general}, we can compute the surfaces that correspond to the $\{B^{p}\}_{0 \leq p \leq 5}$ base components of the central fiber $Y^{0}$ of $\rho: \mathcal{Y} \rightarrow D$. We summarize the result of the blow-up process, step by step, in \cref{fig:geometry-components-general-additional-blow-up-1}.
	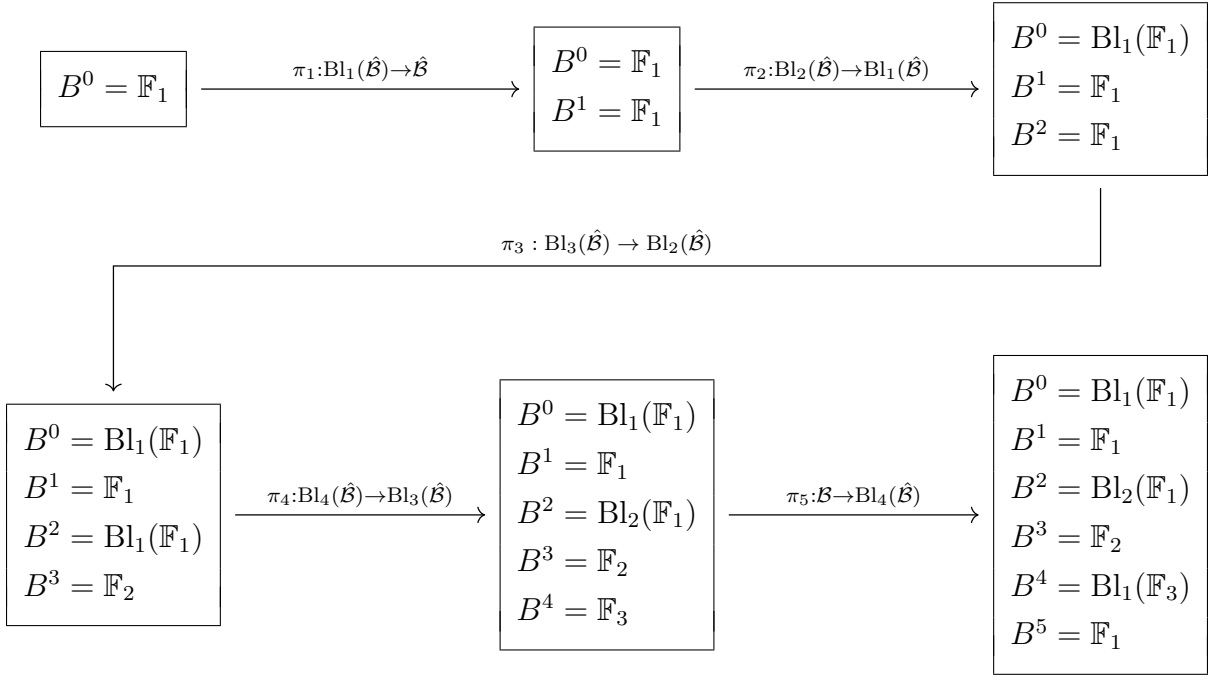
\begin{figure}[t!]
		\centering
		\begin{tikzcd}[row sep={5em}, column sep={8em}]
			\begin{tabular}{|l|} \hline\\[-1.5ex] $B^{0} = \mathbb{F}_{1}$\\[-1.5ex] \\ \hline \end{tabular} \arrow[r, "\pi_{1}: \mathrm{Bl}_{1}(\hat{\mathcal{B}}) \rightarrow \hat{\mathcal{B}}"] & \begin{tabular}{|l|} \hline\\[-1.5ex] $B^{0} = \mathbb{F}_{1}$\\[0.75ex] $B^{1} = \mathbb{F}_{1}$\\[-1.5ex] \\ \hline \end{tabular} \arrow[r, "\pi_{2}: \mathrm{Bl}_{2}(\hat{\mathcal{B}}) \rightarrow \mathrm{Bl}_{1}(\hat{\mathcal{B}})"] & \begin{tabular}{|l|} \hline\\[-1.5ex] $B^{0} = \mathrm{Bl}_{1}(\mathbb{F}_{1})$\\[0.75ex] $B^{1} = \mathbb{F}_{1}$\\[0.75ex] $B^{2} = \mathbb{F}_{1}$\\[-1.5ex] \\ \hline \end{tabular} \arrow[dll, to path={-- ([yshift=-2.5em]\tikztostart.south)  -| node[near start, above]{\scriptsize $\pi_{3}: \mathrm{Bl}_{3}(\hat{\mathcal{B}}) \rightarrow \mathrm{Bl}_{2}(\hat{\mathcal{B}})$} (\tikztotarget) }]\\
			\begin{tabular}{|l|} \hline\\[-1.5ex] $B^{0} = \mathrm{Bl}_{1}(\mathbb{F}_{1})$\\[0.75ex] $B^{1} = \mathbb{F}_{1}$\\[0.75ex] $B^{2} = \mathrm{Bl}_{1}(\mathbb{F}_{1})$\\[0.75ex] $B^{3} = \mathbb{F}_{2}$\\[-1.5ex] \\ \hline \end{tabular} \arrow[r, "\pi_{4}: \mathrm{Bl}_{4}(\hat{\mathcal{B}}) \rightarrow \mathrm{Bl}_{3}(\hat{\mathcal{B}})"] & \begin{tabular}{|l|} \hline\\[-1.5ex] $B^{0} = \mathrm{Bl}_{1}(\mathbb{F}_{1})$\\[0.75ex] $B^{1} = \mathbb{F}_{1}$\\[0.75ex] $B^{2} = \mathrm{Bl}_{2}(\mathbb{F}_{1})$\\[0.75ex] $B^{3} = \mathbb{F}_{2}$\\[0.75ex] $B^{4} = \mathbb{F}_{3}$\\[-1.5ex] \\ \hline \end{tabular} \arrow[r, "\pi_{5}: \mathcal{B} \rightarrow \mathrm{Bl}_{4}(\hat{\mathcal{B}})"] & \begin{tabular}{|l|} \hline\\[-1.5ex] $B^{0} = \mathrm{Bl}_{1}(\mathbb{F}_{1})$\\[0.75ex] $B^{1} = \mathbb{F}_{1}$\\[0.75ex] $B^{2} = \mathrm{Bl}_{2}(\mathbb{F}_{1})$\\[0.75ex] $B^{3} = \mathbb{F}_{2}$\\[0.75ex] $B^{4} = \mathrm{Bl}_{1}(\mathbb{F}_{3})$\\[0.75ex] $B^{5} = \mathbb{F}_{1}$\\[-1.5ex] \\ \hline \end{tabular}
		\end{tikzcd}
		\caption{Components $\{B^{p}\}_{0 \leq p \leq 5}$ of $B_{0}$ arising from the blow-up sequence \eqref{eq:geometry-components-general-additional-blow-up-1}.}
		\label{fig:geometry-components-general-additional-blow-up-1}
	\end{figure}
	The starting point is $B^{0} = \mathbb{F}_{1}$, from which we obtain the surface $B^{1} = \mathbb{F}_{1}$ by blowing up along $C_{1}$, which is the class of the $(+1)$-curve of $B^{0}$. The next blow-up is along $C_{2}$, in the fiber class of $B^{1}$. Due to its intersection point with $\left. E_{0} \right|_{E_{1}} = C_{1} \subset B^{1}$, we obtain $B^{2} = \mathbb{F}_{1}$ and the zeroth component must be blown-up once to become $B^{0} = \mathrm{Bl}_{1}(\mathbb{F}_{1})$. We continue by blowing up along $C_{3}$, in the class of the $(+1)$-curve of $B^{1}$ and with an intersection point with $\left. E_{2} \right|_{E_{0}} = C_{2} \subset B^{1}$, leading to $B^{3} = \mathbb{F}_{2}$ and prompting us to substitute the second component by $B^{2} = \mathrm{Bl}_{1}(\mathbb{F}_{1})$. The components $B^{2}$ and $B^{3}$ meet along a curve that is in the fiber class in $B^{3}$ and the exceptional curve of the surface blow-up in $B^{2}$, which is compatible since
	\begin{equation}
	\begin{aligned}
		\left. E_{2} \right|_{E_{3}} \cdot_{B^{3}} \left. E_{2} \right|_{E_{3}} &= E_{2} \cdot_{\mathrm{Bl}_{3}(\hat{\mathcal{B}})} E_{2} \cdot_{\mathrm{Bl}_{3}(\hat{\mathcal{B}})} E_{3} = E_{2} \cdot_{\mathrm{Bl}_{3}(\hat{\mathcal{B}})} \left( - \sum_{\substack{q = 0\\ q \neq 2}}^{3} E_{q} \right) \cdot_{\mathrm{Bl}_{3}(\hat{\mathcal{B}})} E_{3}\\
		&= - \left. E_{3} \right|_{E_{2}} \cdot_{B^{2}} \left. E_{3} \right|_{E_{2}} - 1\,.
	\end{aligned}
	\end{equation}
	This curve is not one of the blow-up centres, i.e.\ it is not in the set $\{C_{p}\}_{1 \leq p \leq 5}$. It affects the next blow-up, however, since it intersects $C_{4}$, which is in the class of the $(+2)$-curve of $B^{3}$, leading to the component $B^{4} = \mathbb{F}_{3}$ and requiring the substitution of the second component by $B^{2} = \mathrm{Bl}_{2}(\mathbb{F}_{2})$. The final blow-up is along the curve $C_{5}$, which is the strict transform of the representative of the fiber class of $B^{2}$ that has been blown-up twice, and therefore $C_{5} \cdot_{B^{2}} C_{5} = -2$. Together with the intersection point with $\left. E_{4} \right|_{E_{2}} = C_{4} \subset B^{2}$, this yields the component $B^{5} = \mathbb{F}_{1}$ and the surface blow-up $B^{4} = \mathrm{Bl}_{1}(\mathbb{F}_{3})$.
	
	An alternative modification of the degeneration $\hat{\rho}: \hat{\mathcal{Y}} \rightarrow D$ is given by the sequence of blow-ups
	\begin{subequations}
	\begin{alignat}{3}
		&\breve{\pi}_{1}: \mathrm{Bl}_{1}(\hat{\mathcal{B}}) \rightarrow \hat{\mathcal{B}}\,,\qquad &\mathrm{along}\qquad &\breve{C}_{1} = \mathcal{T} \cap \mathcal{U}\,,\\
		&\breve{\pi}_{2}: \mathrm{Bl}_{2}(\hat{\mathcal{B}}) \rightarrow \mathrm{Bl}_{1}(\hat{\mathcal{B}})\,,\qquad &\mathrm{along}\qquad &\breve{C}_{2} = \mathcal{T} \cap \breve{E}_{1}\,,\\
		&\breve{\pi}_{3}: \mathrm{Bl}_{3}(\hat{\mathcal{B}}) \rightarrow \mathrm{Bl}_{2}(\hat{\mathcal{B}})\,,\qquad &\mathrm{along}\qquad &\breve{C}_{3} = \mathcal{T} \cap \breve{E}_{2}\,,\\
		&\breve{\pi}_{4}: \mathrm{Bl}_{4}(\hat{\mathcal{B}}) \rightarrow \mathrm{Bl}_{3}(\hat{\mathcal{B}})\,,\qquad &\mathrm{along}\qquad &\breve{C}_{4} = \mathcal{V} \cap \breve{E}_{1}\,,\\
		&\breve{\pi}_{5}: \breve{\mathcal{B}} \rightarrow \mathrm{Bl}_{4}(\hat{\mathcal{B}})\,,\qquad &\mathrm{along}\qquad &\breve{C}_{5} = \mathcal{V} \cap \breve{E}_{2}\,.
	\end{alignat}
	\label{eq:geometry-components-general-additional-blow-up-2}%
	\end{subequations}
	The resulting Stanley-Reisner ideal is
	\begin{equation}
	\begin{aligned}
		\breve{\mathscr{I}} = \langle &s t,\breve{e}_1 s,\breve{e}_2 s,\breve{e}_3 s,\breve{e}_4 s,\breve{e}_5 s,\breve{e}_0 t,\breve{e}_1 t,\breve{e}_2 t,\breve{e}_4 t,\breve{e}_5 t,\\
		&v w,\breve{e}_1 v,\breve{e}_2 v,\breve{e}_4 w,\breve{e}_5 w,\breve{e}_0 \breve{e}_2,\breve{e}_0 \breve{e}_3,\breve{e}_0 \breve{e}_5,\breve{e}_1 \breve{e}_3,\breve{e}_1 \breve{e}_5,\breve{e}_3 \breve{e}_4 \rangle\,.
	\end{aligned}
	\end{equation}
	The components resulting from the two blow-up sequences considered can be related to each other by identifying the homogeneous coordinates
	\begin{equation}
		\{ e_{0}, e_{1}, e_{2}, e_{3}, e_{4}, e_{5} \} \longleftrightarrow \{ \breve{e}_0, \breve{e}_1, \breve{e}_3, \breve{e}_4, \breve{e}_2, \breve{e}_5 \}\,,
	\end{equation}
	from which we can see that $\mathcal{B}$ and $\breve{B}$ are connected by performing the flops
	\begin{subequations}
	\begin{align}
		\{ t = e_{5} = 0 \}_{\mathcal{B}} &\longleftrightarrow \{ v = \breve{e}_{3} = 0 \}_{\breve{\mathcal{B}}}\,,\\
		\{ e_{2} = e_{4} = 0 \}_{\mathcal{B}} &\longleftrightarrow \{ \breve{e}_{2} = \breve{e}_{5} = 0 \}_{\breve{\mathcal{B}}}
	\end{align}
	\end{subequations}
	The base components for the central fiber of the resolved degeneration $\breve{\rho}: \breve{\mathcal{Y}} \rightarrow D$ are collected, step by step, in \cref{fig:geometry-components-general-additional-blow-up-2}. We do not detail this blow-up sequence further.
	\begin{figure}[t!]
		\centering
		\begin{tikzcd}[row sep={5em}, column sep={8em}]
			\begin{tabular}{|l|} \hline\\[-1.5ex] $\breve{B}^{0} = \mathbb{F}_{1}$\\[-1.5ex] \\ \hline \end{tabular} \arrow[r, "\breve{\pi}_{1}: \mathrm{Bl}_{1}(\hat{\mathcal{B}}) \rightarrow \hat{\mathcal{B}}"] & \begin{tabular}{|l|} \hline\\[-1.5ex] $\breve{B}^{0} = \mathbb{F}_{1}$\\[0.75ex] $\breve{B}^{1} = \mathbb{F}_{1}$\\[-1.5ex] \\ \hline \end{tabular} \arrow[r, "\breve{\pi}_{2}: \mathrm{Bl}_{2}(\hat{\mathcal{B}}) \rightarrow \mathrm{Bl}_{1}(\hat{\mathcal{B}})"] & \begin{tabular}{|l|} \hline\\[-1.5ex] $\breve{B}^{0} = \mathbb{F}_{1}$\\[0.75ex] $\breve{B}^{1} = \mathbb{F}_{1}$\\[0.75ex] $\breve{B}^{2} = \mathbb{F}_{1}$\\[-1.5ex] \\ \hline \end{tabular} \arrow[dll, to path={-- ([yshift=-2.5em]\tikztostart.south)  -| node[near start, above]{\scriptsize $\breve{\pi}_{3}: \mathrm{Bl}_{3}(\hat{\mathcal{B}}) \rightarrow \mathrm{Bl}_{2}(\hat{\mathcal{B}})$} (\tikztotarget) }]\\
			\begin{tabular}{|l|} \hline\\[-1.5ex] $\breve{B}^{0} = \mathbb{F}_{1}$\\[0.75ex] $\breve{B}^{1} = \mathbb{F}_{1}$\\[0.75ex] $\breve{B}^{2} = \mathbb{F}_{1}$\\[0.75ex] $\breve{B}^{3} = \mathbb{F}_{1}$\\[-1.5ex] \\ \hline \end{tabular} \arrow[r, "\breve{\pi}_{4}: \mathrm{Bl}_{4}(\hat{\mathcal{B}}) \rightarrow \mathrm{Bl}_{3}(\hat{\mathcal{B}})"] & \begin{tabular}{|l|} \hline\\[-1.5ex] $\breve{B}^{0} = \mathrm{Bl}_{1}(\mathbb{F}_{1})$\\[0.75ex] $\breve{B}^{1} = \mathbb{F}_{1}$\\[0.75ex] $\breve{B}^{2} = \mathrm{Bl}_{1}(\mathbb{F}_{1})$\\[0.75ex] $\breve{B}^{3} = \mathbb{F}_{1}$\\[0.75ex] $\breve{B}^{4} = \mathbb{F}_{2}$\\[-1.5ex] \\ \hline \end{tabular} \arrow[r, "\breve{\pi}_{5}: \breve{\mathcal{B}} \rightarrow \mathrm{Bl}_{4}(\hat{\mathcal{B}})"] & \begin{tabular}{|l|} \hline\\[-1.5ex] $\breve{B}^{0} = \mathrm{Bl}_{1}(\mathbb{F}_{1})$\\[0.75ex] $\breve{B}^{1} = \mathbb{F}_{1}$\\[0.75ex] $\breve{B}^{2} = \mathrm{Bl}_{1}(\mathbb{F}_{1})$\\[0.75ex] $\breve{B}^{3} = \mathrm{Bl}_{1}(\mathbb{F}_{1})$\\[0.75ex] $\breve{B}^{4} = \mathrm{Bl}_{1}(\mathbb{F}_{1})$\\[0.75ex] $\breve{B}^{5} = \mathbb{F}_{1}$\\[-1.5ex] \\ \hline \end{tabular}
		\end{tikzcd}
		\caption{Components $\{\breve{B}^{p}\}_{0 \leq p \leq 5}$ of $\breve{B}_{0}$ arising from the blow-up sequence \eqref{eq:geometry-components-general-additional-blow-up-2}.}
		\label{fig:geometry-components-general-additional-blow-up-2}
	\end{figure}
\end{example}

\subsection{Line bundles}
\label{sec:line-bundles-components-arbitrary-limit}

Adapting now the discussion of \cref{sec:line-bundles-components-single-limit} to the general case, let us compute the holomorphic line bundles associated to the Weierstrass models describing the components of the central fiber of a resolved degeneration.
\begin{proposition}
\label{prop:component-line-bundle-general-limit}
	Let $\{B^{p}\}_{0 \leq p \leq P}$ be the base components of the central fiber $Y_{0}$ of the modification $\rho: \mathcal{Y} \rightarrow D$ giving the resolution of a degeneration $\hat{\rho}: \hat{\mathcal{Y}} \rightarrow D$. Then, the holomorphic line bundles $\{\mathcal{L}_{p}\}_{0 \leq p \leq P} := \{\mathcal{L}_{B^{p}}\}_{0 \leq p \leq P}$ defining the Weierstrass models over the $\{B^{p}\}_{0 \leq p \leq P}$ are
	\begin{equation}
		\mathcal{L}_{p} = \overline{K}_{B^{p}} - \sum_{\substack{q=0\\q \neq p}}^{P} \left. E_{q} \right|_{E_{p}}\,,\qquad 0 \leq p \leq P\,.
	\label{eq:general-component-line-bundle}
	\end{equation}
\end{proposition}
\begin{proof}
	Since the $\{ E_{p} \}_{0 \leq p \leq P}$ and $\mathcal{B}$ are smooth, we obtain from the adjunction formula
	\begin{equation}
		\mathcal{L}_{p} = \overline{K}_{B^{p}} + \left. E_{p} \right|_{E_{p}} = \overline{K}_{B^{p}} - \sum_{\substack{q=0\\q \neq p}}^{P} \left. E_{q} \right|_{E_{p}}\,,\qquad 0 \leq p \leq P\,.
	\end{equation}
\end{proof}
We see that, as occurred for single infinite-distance limits, the component Weierstrass models \mbox{$\pi_{p}: Y^{p} \rightarrow B^{p}$} are not describing Calabi-Yau varieties, since $\mathcal{L}_{p} \neq \overline{K}_{B^{p}}$. Rather, the pairs
\begin{equation}
	\Biggl( Y^{p}, \pi^{*}\Biggl( \sum_{\substack{q=0\\ q \neq p}}^{P} \left. E_{q} \right|_{E_{p}} \Biggl) \Biggl)\,,\qquad 0 \leq p \leq P
\end{equation}
are log Calabi-Yau spaces, with their union $Y_{0} = \bigcup_{p=0}^{P} Y^{p}$ along the boundaries yielding a Calabi-Yau variety.

Using the modified discriminant of \cref{def:modified-discriminant}, the component by component analysis of a model is performed working with the polynomials $\{f_{p}, g_{p}, \Delta'_{p}\}_{0 \leq p \leq P}$. Let us collect their associated divisor classes.
\begin{proposition}
\label{prop:component-divisor-classes-general-limit}
	Let $\{B^{p}\}_{0 \leq p \leq P}$ be the base components of the central fiber $Y_{0}$ of the modification $\rho: \mathcal{Y} \rightarrow D$ giving the resolution of a degeneration $\hat{\rho}: \hat{\mathcal{Y}} \rightarrow D$, and let
	\begin{equation}
		 \ord{\mathcal{Y}}(f_{b},g_{b},\Delta_{b})_{E_{p}} = (0,0,n_{p})\,,\qquad 0 \leq p \leq P\,,
	\end{equation}
	be the vanishing orders associated to the codimension-zero singular fibers in said components. The (modified) divisor classes associated to the Weierstrass models in the components are
	\begin{subequations}
	\begin{align}
		F_{p} &= 4\overline{K}_{B^{p}} - \sum_{\substack{q \neq p}} \left. 4E_{q} \right|_{E_{p}}\,,
		\qquad \quad  G_{p} = 6\overline{K}_{B^{p}} - \sum_{\substack{q \neq p}} \left. 6E_{q} \right|_{E_{p}}\,,\\
		\Delta'_{p} &= 12\overline{K}_{B^{p}} + \sum_{\substack{q \neq p}} \left. (n_{q}-12)E_{q} \right|_{E_{p}}\,,
	\end{align}
	\end{subequations}
	for $0 \leq p \leq P$.
\end{proposition}
\begin{proof}
	It follows from \cref{prop:component-line-bundle-general-limit} and \cref{def:modified-discriminant}.
\end{proof}

We complete the examples of \cref{sec:geometry-components-arbitrary-limit} in order to illustrate the above discussion.
\begin{example}
	Continuing with \cref{example:geometry-components-general}, we see from \eqref{eq:example2-defining-polynomials-blow-up} that the components of the central fiber of the resolved degeneration $\{Y^{0}, Y^{s}, Y^{v} \}$ are smooth in codimension-zero, i.e.\
	\begin{subequations}
	\begin{align}
	    \ord{\mathcal{Y}}(f_{b},g_{b},\Delta_{b})_{E_{0}} &= (0,0,0)\,,\\
	     \ord{\mathcal{Y}}(f_{b},g_{b},\Delta_{b})_{E_{s}} &= (0,0,0)\,,\\
	      \ord{\mathcal{Y}}(f_{b},g_{b},\Delta_{b})_{E_{v}} &= (0,0,0)\,.
	\end{align}
	\end{subequations}
	Using \cref{prop:component-divisor-classes-general-limit}, we determine for the blow-up order $C_{1} = \mathcal{S} \cap \mathcal{U}$ and $C_{2} = \mathcal{V} \cap E_{0}$ the holomorphic line bundles associated to the Weierstrass models in the components to be
	\begin{subequations}
	\begin{align}
		\mathcal{L}_{0} &= S_{0} + 4V_{0}\,,\\
		\mathcal{L}_{s} &= S_{s} + 2V_{s} + C_{E}^{1}\,,\\
		\mathcal{L}_{v} &= S_{v} + 2V_{v}\,,
	\end{align}
	\end{subequations}
	where we have denoted by $C_{E}^{1}$ the exceptional curve in $B^{s} = \mathrm{Bl}_{1} (\mathbb{F}_{3})$. If, instead, we consider the blow-up order $C_{1} = \mathcal{V} \cap \mathcal{U}$ and $C_{2} = \mathcal{S} \cap E_{0}$, we have the holomorphic line bundles
	\begin{subequations}
	\begin{align}
		\mathcal{L}_{0} &= S_{0} + 4V_{0}\,,\\
		\mathcal{L}_{s} &= S_{s} + V_{s}\,,\\
		\mathcal{L}_{v} &= 2S_{v} + V_{v}\,.
	\end{align}
	\end{subequations}
	Although the line bundle defining the Weierstrass model $\rho: \mathcal{Y} \rightarrow D$ is the same independently of which of the two blow-up orders is chosen, the same is not true for the component line bundles, since they are affected by higher codimension effects in $\mathcal{B}$ like the flop of curves. In the absence of codimension-zero singularities, the divisor classes associated to the defining polynomials and the modified discriminant of the component Weierstrass models are simply $F_{p} = 4\mathcal{L}_{p}$, $G_{p} = 6\mathcal{L}_{p}$ and $\Delta'_{p} = 12\mathcal{L}_{p}$ for $p \in \{0,s,v\}$. Indeed, the restrictions of the polynomials $\{f_{b},g_{b},\Delta'_{b}\}$ to the $B^{0}$ component,
	\begin{subequations}
	\begin{align}
		f_{0} &= s^4 t^4 v^4 \left(e_{v}^2 v^2-e_{v} v w+w^2\right) \left(e_{v}^2 v^2+e_{v} v w+w^2\right)\,,\\
		g_{0} &= s^6 t^5 v^6 \left(e_{s} e_{v}^9 s v^9+t w^6\right)\,,\\
        \begin{split}
		\Delta'_{0} &= s^{12} t^{10} v^{12} \left(27 e_{s}^2 e_{v}^{18} s^2 v^{18}+54 e_{s} e_{v}^9 s t v^9 w^6+4 e_{v}^{12} t^2 v^{12}+12 e_{v}^{10} t^2 v^{10} w^2 \right.\\
		&\quad \left.+24 e_{v}^8 t^2 v^8 w^4+28 e_{v}^6 t^2 v^6 w^6+24 e_{v}^4 t^2 v^4 w^8+12 e_{v}^2 t^2 v^2 w^{10}+31 t^2 w^{12}\right)\,,
        \end{split}
	\end{align}
	\end{subequations}
	to the $B^{s}$ component,
	\begin{subequations}
	\begin{align}
		f_{s} &= s^4 t^4 v^4 \left(e_{v}^2 v^2-e_{v} v w+w^2\right) \left(e_{v}^2 v^2+e_{v} v w+w^2\right)\,,\\
		g_{s} &= s^5 t^6 v^5 w^4 \left(e_0 t+s v w^2\right)\,,\\
        \begin{split}
		\Delta'_{s} &= s^{10} t^{12} v^{10} \left(54 e_0 s t v w^{10}+27 e_0^2 t^2 w^8+4 e_{v}^{12} s^2 v^{14}+12 e_{v}^{10} s^2 v^{12} w^2+24 e_{v}^8 s^2 v^{10} w^4 \right.\\
		&\quad \left.+28 e_{v}^6 s^2 v^8 w^6+24 e_{v}^4 s^2 v^6 w^8+12 e_{v}^2 s^2 v^4 w^{10}+31 s^2 v^2 w^{12}\right)\,,
        \end{split}
	\end{align}
	\end{subequations}
	 and to the $B^{v}$ component
	\begin{subequations}
	\begin{align}
		f_{v} &= s^4 t^4 v^4 w^4\,,\\
		g_{v} &= s^5 t^5 v^5 w^4 \left(e_0 e_{s}^2 s^2 w^6+e_0 t^2+s t v w^2\right)\,,\\
        \begin{split}
		\Delta'_{v} &= s^{10} t^{10} v^{10} w^8 \left(27 e_0^2 e_{s}^4 s^4 w^{12}+54 e_0 e_{s}^2 s^3 t v w^8+54 e_0^2 e_{s}^2 s^2 t^2 w^6+54 e_0 s t^3 v w^2 \right.\\
		&\quad \left.+27 e_0^2 t^4+31 s^2 t^2 v^2 w^4\right)
        \end{split}
	\end{align}
	\end{subequations}
	are sections of the appropriate line bundles. Here, we have not used the available $\mathbb{C}^{*}$-actions to fix the redundant coordinates to one in order to provide expressions valid for both blow-up orders.
\end{example}

\begin{example}
	For completeness, let us also compute the line bundles over the base components in \cref{example:geometry-components-general-additional}. We observe from \eqref{eq:example2-defining-polynomials-blow-up} that the $\{Y^{p}\}_{0 \leq p \leq 5}$ are smooth in codimension-zero, since
	\begin{equation}
		 \ord{\mathcal{Y}}(f_{b},g_{b},\Delta_{b})_{E_{0}} = (0,0,0)\,,\qquad 0 \leq p \leq 5\,.
	\end{equation}
	This means that, for both sequences of blow-ups, the divisor classes associated to the defining polynomials and the modified discriminant on the component Weierstrass models are just appropriate multiples of $\mathcal{L}_{p}$ (or $\breve{\mathcal{L}}_{p}$) for $0 \leq p \leq 5$. We compute said line bundles using \eqref{eq:general-component-line-bundle}. For the sequence of blow-ups \eqref{eq:geometry-components-general-additional-blow-up-1} we find
	\begin{subequations}
	\begin{align}
		\mathcal{L}_{0} &= \overline{K}_{\mathrm{Bl}_{1}(\mathbb{F}_{1})} - \left. E_{1} \right|_{E_{0}} - \left. E_{2} \right|_{E_{0}} = S_{0} + 2V_{0} + C_{E}^{1}\,,\\
		\mathcal{L}_{1} &=  \overline{K}_{\mathbb{F}_{1}} - \left. E_{0} \right|_{E_{1}} - \left. E_{2} \right|_{E_{1}} - \left. E_{3} \right|_{E_{1}} = V_{0}\,,\\
		\mathcal{L}_{2} &= \overline{K}_{\mathrm{Bl}_{2}(\mathbb{F}_{1})} - \left. E_{0} \right|_{E_{2}} - \left. E_{1} \right|_{E_{2}} - \left. E_{3} \right|_{E_{2}} - \left. E_{4} \right|_{E_{2}} - \left. E_{5} \right|_{E_{2}} = S_{2} + V_{2} + 2C_{E}^{1} + 3C_{E}^{2}\,,\\
		\mathcal{L}_{3} &=  \overline{K}_{\mathbb{F}_{2}} - \left. E_{1} \right|_{E_{3}} - \left. E_{2} \right|_{E_{3}} - \left. E_{4} \right|_{E_{3}} = V_{3}\,,\\
		\mathcal{L}_{4} &= \overline{K}_{\mathrm{Bl}_{1}(\mathbb{F}_{3})} - \left. E_{2} \right|_{E_{4}} - \left. E_{3} \right|_{E_{4}} - \left. E_{5} \right|_{E_{4}} = S_{4} + 4V_{4} + 3C_{E}^{1}\,,\\
		\mathcal{L}_{5} &=  \overline{K}_{\mathbb{F}_{1}} - \left. E_{2} \right|_{E_{5}} - \left. E_{4} \right|_{E_{5}} = S_{5} + 2V_{5}\,,
	\end{align}
	\end{subequations}
	while for the sequence of blow-ups \eqref{eq:geometry-components-general-additional-blow-up-2}  we obtain instead
	\begin{subequations}
	\begin{align}
		\breve{\mathcal{L}}_{0} &= \overline{K}_{\mathrm{Bl}_{1}(\mathbb{F}_{1})} - \left. E_{1} \right|_{E_{0}} - \left. E_{4} \right|_{E_{0}} = S_{0} + 2V_{0} + C_{E}^{1}\,,\\
		\breve{\mathcal{L}}_{1} &=  \overline{K}_{\mathbb{F}_{1}} - \left. E_{0} \right|_{E_{1}} - \left. E_{2} \right|_{E_{1}} - \left. E_{4} \right|_{E_{1}} = V_{1}\,,\\
		\breve{\mathcal{L}}_{2} &= \overline{K}_{\mathrm{Bl}_{1}(\mathbb{F}_{1})} - \left. E_{1} \right|_{E_{2}} - \left. E_{3} \right|_{E_{2}} - \left. E_{4} \right|_{E_{2}} - \left. E_{5} \right|_{E_{2}} = V_{2} + C_{E}^{1}\,,\\
		\breve{\mathcal{L}}_{3} &= \overline{K}_{\mathrm{Bl}_{1}(\mathbb{F}_{1})} - \left. E_{2} \right|_{E_{3}} - \left. E_{5} \right|_{E_{3}} = S_{3} + 3V_{3} + 3C_{E}^{1}\,,\\
		\breve{\mathcal{L}}_{4} &= \overline{K}_{\mathrm{Bl}_{1}(\mathbb{F}_{1})} - \left. E_{0} \right|_{E_{4}} - \left. E_{1} \right|_{E_{4}} - \left. E_{2} \right|_{E_{4}} - \left. E_{5} \right|_{E_{4}} = S_{4} + V_{4} + 2C_{E}^{1}\,,\\
		\breve{\mathcal{L}}_{5} &= \overline{K}_{\mathbb{F}_{1}} - \left. E_{2} \right|_{E_{5}} - \left. E_{3} \right|_{E_{5}} - \left. E_{4} \right|_{E_{5}} = V_{5}\,.
	\end{align}
	\end{subequations}
	Above, we have denoted the exceptional curves of the surface blow-ups by $C_{E}^{\bullet}$. One can easily check that the restrictions to the components of the polynomials \eqref{eq:example2-defining-polynomials-blow-up} are indeed sections of powers of the listed line bundles.
\end{example}
%auto-ignore

\section{Single infinite-distance limits and their resolutions}
\label{sec:single-infinite-distance-limits-and-open-chain-resolutions}

In \cref{sec:curves-of-non-minimal-fibers} we stated that single infinite-distance limit degenerations in the sense of \cref{def:single-infinite-distance-limit-original} lead to open-chain resolutions. In this appendix we provide the proof for this central result.

We start by recalling the notion of a single infinite-distance limit degeneration.
\singleinfinitedistancelimit*
The role of each of these conditions in enforcing the open-chain resolution structure is intuitively clear. Let us analyse them in turn.

First, suppose that two curves supporting non-minimal elliptic fibers in the family variety intersect each other. Then, the exceptional components arising from blowing up along these curves will also intersect each other, on top of intersecting the strict transform of the component that contained these curves of non-minimal elliptic fibers. Hence, the components cannot intersect as in an open-chain resolution. This is prevented by Condition \ref{item:single-infinite-distance-limit-original-1} in \cref{def:single-infinite-distance-limit-original}. Note that here we have referred to curves with non-minimal family vanishing orders, while \cref{def:single-infinite-distance-limit-original} is concerned with curves presenting non-minimal component vanishing orders; since the former implies the latter, Condition \ref{item:single-infinite-distance-limit-original-1} is still in effect.

However, we could conceive of having a third curve of non-minimal singular elliptic fibers in $\hat{B}_{0}$ that does not intersect the others. This would mean that $B^{0}$ would intersect more than two components, which cannot happen for a component in an open-chain resolution. The situation just described is not prevented, a priori, by Condition \ref{item:single-infinite-distance-limit-original-1} in \cref{def:single-infinite-distance-limit-original}. Nonetheless, we argue in \cref{sec:restricting-star-degenerations} that this cannot occur in a Calabi-Yau Weierstrass model whose base is one of the allowed six-dimensional F-theory bases. Namely, the following happens if we attempt to tune such a model.
\begin{restatable}{proposition}{restrictingstardegenerations}
\label{prop:restricting-star-degenerations}
	Let $\hat{\rho}: \hat{\mathcal{Y}} \rightarrow D$ be a degeneration of the type described in \cref{sec:base-blow-ups} such that there is a collection of curves $\hat{\mathscr{C}}_{r} := \{\mathcal{C}_{i} \cap \mathcal{U}\}_{1 \leq i \leq r}$ in $\hat{\mathcal{B}} = B \times D$ with non-minimal component vanishing orders. If $r \geq 3$, then at least two of these curves intersect.
\end{restatable}

Next, we may worry that, starting from curves of non-minimal fibers in $\hat{B}_{0}$ that seem like they would lead to an open-chain resolution, we might encounter a curve with non-minimal (possibly only component) vanishing orders in an exceptional component arising in the blow-up sequence, and whose resolution (possibly after base change) would destroy the open-chain structure. This is what can be observed in \cref{example:obscured-infinite-distance-limit-F0}, where we also see that Condition \ref{item:single-infinite-distance-limit-original-2} in \cref{def:single-infinite-distance-limit-original} allows us to detect this. Indeed, this is true in general.
\begin{lemma}
\label{lemma:obscured-non-minimal-curves-produce-non-minimal-interface-points}
	Let $\rho: \mathcal{Y} \rightarrow D$ be an open-chain resolution of a degeneration $\hat{\rho}: \hat{\mathcal{Y}} \rightarrow D$. If a component $Y^{p}$, with $p \neq 0$, of the open-chain $Y_{0} = \bigcup_{p=0}^{P} Y^{p}$ contains a curve $C_{p} \subset B^{p}$ presenting non-minimal component vanishing orders, and this curve is not an end-curve of the open-chain, then one of the blow-up centres in $B^{0}$ contains a point with non-minimal interface vanishing orders.
\end{lemma}
\begin{proof}
	First, let us note that because we are dealing with an open-chain resolution, all components $B^{p}$, with $p \neq 0$, are Hirzebruch surfaces $\mathbb{F}_{n_{p}}$ due to \cref{prop:component-geometry-single}. $C_{p}$ cannot be one of the curves over which the components intersect, see the comments in \cref{sec:class-1-5-models}. We can distinguish two cases depending on the component in which $C_{p}$ is contained, which we discuss in turn.
	\begin{enumerate}[label=(\arabic*)]
		\item If $C_{p}$ is contained in an intermediate component $B^{p}$ of the base central fiber, we know from \eqref{eq:line-bundles-Fn-components-open-chain} that $F_{p}$ and $G_{p}$ only contain fiber classes. Hence, $C_{p} \sim V_{p}$, and since $S_{p} \cdot V_{p} = T_{p} \cdot V_{p} = 1$, this makes a point in each of the component interfaces have non-minimal interface vanishing orders. This needs to be realized from the other side of the interface as well. If $B^{0} \cap B^{p} \neq \emptyset$, we have obtained the desired result. If, on the other hand, $B^{0} \cap B^{p} = \emptyset$, take the chain of components connecting $B^{0}$ and $B^{p}$, which (possibly after relabelling) we can take to be $B^{0} - B^{1} - \cdots - B^{p-1} - B^{p}$. The intermediate components $\{B^{q}\}_{1 \leq q \leq p-1}$ also have $F_{q}$ and $G_{q}$ consisting only of fiber classes, and therefore the only way to have consistent interface vanishing orders throughout is for the non-minimal curve $C_{p}$ to extend through all $B^{0} - B^{1} - \cdots - B^{p-1} - B^{p}$, at which point we can take $B^{p}$ to be $B^{1}$ and reduce to the previous situation.
		
		\item If $C_{p}$ is contained in an end-component $B^{p}$, we are assuming that it is distinct from the end-curve of the component. Since the end-curve of $B^{p}$ is the only curve class in $\mathcal{L}_{p}$ that does not intersect the interface with the adjacent component, $C_{p}$ must intersect it. This makes a point at the interface curve have non-minimal vanishing orders. The adjacent component is either $B^{0}$ or an intermediate component; either way, the arguments given earlier apply, and we reach the desired result.
	\end{enumerate}
\end{proof}

Condition \ref{item:single-infinite-distance-limit-original-3} in \cref{def:single-infinite-distance-limit-original} rules out the presence of codimension-two points in $B^{0}$ with infinite-distance non-minimal component vanishing orders. These points cannot appear in the other components either, unless they are located over a non-minimal curve, as we now argue.
\begin{lemma}
\label{lemma:obscured-codimension-two-in-B0-only}
	Let $\rho: \mathcal{Y} \rightarrow D$ be an open-chain resolution of a degeneration $\hat{\rho}: \hat{\mathcal{Y}} \rightarrow D$. No component $Y^{p}$, with $p \neq 0$, of the open-chain $Y_{0} = \bigcup_{p=0}^{P} Y^{p}$ can contain a point in $B^{p}$ with infinite-distance non-minimal component vanishing orders that is not located over a curve with non-minimal component vanishing orders.
\end{lemma}
\begin{proof}
	Over the base components apart from $B^{0}$, we recall from \eqref{eq:line-bundles-Fn-components-open-chain} that the defining holomorphic line bundle of the Weierstrass model is
	\begin{subequations}
	\begin{align}
		\mathcal{L}_{p} = 2V_{p}\,,\qquad &\mathrm{if }B^{p}\text{ is an intermediate component,}\\
		\mathcal{L}_{p} = S_{p} + 2V_{p}\quad \mathrm{or}\quad \mathcal{L}_{p} = T_{p} + 2V_{p}\,,\qquad &\mathrm{if }B^{p}\text{ is an end-component.}
	\end{align}
	\end{subequations}
	When $B^{p}$ is an intermediate component, this implies that $F_{p} = 8V_{p}$ and $G_{p} = 12V_{p}$. Since $V_{p} \cdot V_{p} = 0$, there can be no codimension-two enhancements. When $B^{p}$ is an end-component, we can try to tune high codimension-two vanishing orders by forcing $F_{p}$, $G_{p}$ and $\Delta_{p}$ to self-intersect many times over a point. Once we fix said point, there is a unique representative of $V_{p}$ that passes through it, let us call it $\breve{V}_{p}$. If said point lies over the unique representative of $S_{p}$, both curves would need to support non-minimal vanishing orders in order for the point to support infinite-distance non-minimal vanishing orders. Assume instead that it is over a representative of $T_{p}$, and that we have tuned a factor of $\alpha \breve{V}_{p}$ and of $\beta \breve{V}_{p}$ over $F_{p}$ and $G_{p}$ respectively. Since $(F_{p} - \alpha \breve{V}_{p}) \cdot \breve{V}_{p} = 4$ and $(G_{p} - \beta \breve{V}_{p}) \cdot \breve{V}_{p} = 6$, the best component vanishing orders that we can tune over the point are $(4+\alpha,6+\beta) < (8,12)$ unless the tuning of $\breve{V}_{p}$ is non-minimal.
\end{proof}

Finally, after resolving a single infinite-distance limit degeneration $\hat{\rho}: \hat{\mathcal{Y}} \rightarrow D$ and obtaining an open-chain resolution $\rho: \mathcal{Y} \rightarrow D$ free of obscured infinite-distance limits, we may fear that a different modification of $\breve{\rho}: \breve{\mathcal{Y}} \rightarrow D$ could lead, upon applying the procedures explained in \cref{sec:modifications-of-degenerations}, to a resolution that does not have the open-chain structure. This can also be discarded.
\begin{lemma}
\label{lemma:open-chain-birationally-stable-unless-obscured}
	Let $\rho: \mathcal{Y} \rightarrow D$ be an open-chain resolution of a degeneration $\hat{\rho}: \hat{\mathcal{Y}} \rightarrow D$. If a combination of base changes and modifications allows us to obtain a resolution $\breve{\rho}: \breve{\mathcal{Y}} \rightarrow D$ that does not have an open-chain structure, then $\rho: \mathcal{Y} \rightarrow D$ contained an obscured infinite-distance limit that did not occur over one of the end-curves of the open-chain.
\end{lemma}
\begin{proof}
	The modifications respecting the elliptic fibration of the family variety are of the type explained in \cref{sec:base-blow-ups}, i.e.\ a combination of base blow-ups and blow-downs followed by line bundle shifts in order to restore the Calabi-Yau condition, see the comments in \cref{sec:class-1-5-models}. By themselves, these do not change the open-chain nature of the obtained resolutions. The base change can lead to the need for blow-ups along additional curves, potentially spoiling the open-chain structure if they are not end-curves of the open-chain. Saying that this occurs is equivalent to saying that an obscured infinite-distance limit exists over a curve besides the end-curves of the open-chain, see \cref{sec:obscured-infinite-distance-limits}.
\end{proof}

We can now integrate the results given above into the statement that single infinite-distance limit degenerations lead to open-chain resolutions.
\begin{proposition}
\label{prop:single-infinite-distance-limits-and-open-resolutions-app}
	Let $\hat{\rho}: \hat{\mathcal{Y}} \rightarrow D$ be a single infinite-distance limit degeneration. Its resolved modifications $\rho: \mathcal{Y} \rightarrow \mathcal{B}$, obtained as explained in \cref{sec:modifications-of-degenerations}, are open-chain resolutions.
\end{proposition}
\begin{proof}
	We have three possibilities for $\rho: \mathcal{Y} \rightarrow D$. Either the resolution is
	\begin{enumerate}[label=(\arabic*)]
		\item an open-chain resolution free of obscured infinite-distance limits, or \label{item:proposition-single-infinite-distance-limits-1}
		\item an open-chain resolution containing obscured infinite-distance limits, or \label{item:proposition-single-infinite-distance-limits-2}
		\item is not an open-chain resolution. \label{item:proposition-single-infinite-distance-limits-3}
	\end{enumerate}
	
	If we are in Case \ref{item:proposition-single-infinite-distance-limits-1}, we are done, since by \cref{lemma:open-chain-birationally-stable-unless-obscured} any sequence of base changes and modifications of $\hat{\rho}: \hat{\mathcal{Y}} \rightarrow D$ preserving the elliptic fibration will also lead to resolutions falling under Case \ref{item:proposition-single-infinite-distance-limits-1}.
	
	Let us now consider Case \ref{item:proposition-single-infinite-distance-limits-2}. This case can be subdivided into the following subcases, depending on how the obscured infinite-distance limit arises.
	\begin{enumerate}[label=(2.\alph*)]
		\item\label{item:proposition-single-infinite-distance-limits-2a} There is an obscured infinite-distance limit over a curve $C$ in $B^{0}$: Denoting the composition of base blow-ups by $\pi: \mathcal{B} \rightarrow \hat{\mathcal{B}}$, one can see using the relations \eqref{eq:utilde-linear-equivalence} and \eqref{eq:definition-restricted-shifted-defining-divisors} that for an open-chain resolution
		\begin{subequations}
		\begin{align}
			\pi_{*} \left( F_{0} \right) &= \left. F \right|_{\mathcal{U}} - \sum_{i=1}^{r} 4C_{r}\,,\\
			\pi_{*} \left( G_{0} \right) &= \left. G \right|_{\mathcal{U}} - \sum_{i=1}^{r} 6C_{r}\,,
		\end{align}
		\label{eq:F0-F-G0-G-relations}%
		\end{subequations}
		where $\{ C_{i} \}_{0 \leq i \leq r}$ is the collection of curves over which $B^{0}$ intersects other components. This relation is still true if a base change is performed, since we are considering the restrictions $\left. F \right|_{\mathcal{U}}$ and $\left. G \right|_{\mathcal{U}}$. As a consequence, the component vanishing orders over $C$ in $\hat{B}_{0}$ are also non-minimal. By assumption, $C$ will not intersect any of the other curves with non-minimal component vanishing orders. A high enough base change will make the fibers over $C$ non-minimal singular elliptic fibers of $\hat{\mathcal{Y}}$, at which point we can resolve over the curve $C$ as well to obtain a modification $\breve{\rho}: \breve{\mathcal{Y}} \rightarrow D$, which we categorize again. We can repeat this process until Case \ref{item:proposition-single-infinite-distance-limits-2a} is no longer realized.
		
		\item\label{item:proposition-single-infinite-distance-limits-2b} There is an obscured infinite-distance limit over a curve in $B^{p}$, with $p \neq 0$: Invoking \cref{lemma:obscured-non-minimal-curves-produce-non-minimal-interface-points}, we see that either Condition \ref{item:single-infinite-distance-limit-original-2} of \cref{def:single-infinite-distance-limit-original} is violated, leading to a contradiction, or the obscured infinite-distance limit would not destroy the open-chain structure if manifest at the family level (i.e.\ it is found over an end-curve of the open-chain). In the latter case, we make the obscured infinite-distance limit apparent at the family level as explained in \cref{sec:obscured-infinite-distance-limits}, and resolve to obtain the modification $\breve{\rho}: \breve{\mathcal{Y}} \rightarrow D$, to which we apply the proposition again.
		
		\item\label{item:proposition-single-infinite-distance-limits-2c} There is a point with non-minimal interface vanishing orders in one of the intersection curves $B^{p} \cap B^{q}$: Given the fact that the intermediate components of an open-chain reso\-lu\-tion are Hirzebruch surfaces in which $F_{p}$ and $G_{p}$ consist only of fiber classes, see \eqref{eq:line-bundles-Fn-components-open-chain}, this case reduces to Case \ref{item:proposition-single-infinite-distance-limits-2a} unless we are dealing with a two-component resolution. But due to \eqref{eq:F0-F-G0-G-relations} that would imply that Condition \ref{item:single-infinite-distance-limit-original-2} of \cref{def:single-infinite-distance-limit-original} is once again violated, leading to a contradiction.
		
		\item\label{item:proposition-single-infinite-distance-limits-2d} There is a point with infinite-distance non-minimal component vanishing orders in one of the components $B^{p}$: From \cref{lemma:obscured-codimension-two-in-B0-only}, we know that this can only occur in $B^{0}$, but due to \eqref{eq:F0-F-G0-G-relations} this implies that Condition \ref{item:single-infinite-distance-limit-original-3} of \cref{def:single-infinite-distance-limit-original} is violated.
	\end{enumerate}
	
	Finally, consider Case \ref{item:proposition-single-infinite-distance-limits-3}. Start by partially blowing down $\rho: \mathcal{Y} \rightarrow D$ until an open-chain structure for the (now partial) resolution $\breve{\rho}: \breve{\mathcal{Y}} \rightarrow D$ is obtained. The blow-down leads to non-minimal family vanishing orders along the former blow-up centre, which in particular means non-minimal component vanishing orders along the same loci. The components that we have to blow-down in order to reach the open-chain structure can be of the following types.
	\begin{enumerate}[label=(3.\alph*)]
		\item A component $B^{p}$ arising from a blow-up along a curve $C$: Note that this curve cannot be one of the  end-curves of an end-component, since their blow-up would not destroy the open-chain structure, and we would therefore not have blown the associated component down. This means that the curve can be either in
		\begin{itemize}
			\item the strict transform $B^{0}$ of the original component $\hat{B}_{0}$, in which case the fact that blowing it up destroys the open-chain structure means that it intersects one of the other blow-up centres in $\breve{B}^{0}$, violating Condition \ref{item:single-infinite-distance-limit-original-1} of \cref{def:single-infinite-distance-limit-original}, and therefore leading to a contradiction; or in
			
			\item an intermediate or end-component of the chain, in which case \cref{lemma:obscured-non-minimal-curves-produce-non-minimal-interface-points} implies a violation of Condition \ref{item:single-infinite-distance-limit-original-2} of \cref{def:single-infinite-distance-limit-original} and, hence, a contradiction again.
		\end{itemize}
		
		\item A component arising from a blow-up along an isolated\footnote{Meaning that it does not sit on top of a curve with non-minimal family vanishing orders.} codimension-two infinite-distance non-minimal point: The blow-down of this component leads to a point with non-minimal component vanishing orders that, due to \cref{lemma:obscured-codimension-two-in-B0-only}, must be located in the strict transform of the original component. Then, due to \eqref{eq:F0-F-G0-G-relations}, this implies that Condition \ref{item:single-infinite-distance-limit-original-3} of \cref{def:single-infinite-distance-limit-original} is violated. 
	\end{enumerate}
\end{proof}

Hence, as we had claimed, single infinite-distance limit degenerations will indeed lead to open-chain resolutions, and the results of \cref{sec:geometry-components-single-limit} and \cref{sec:line-bundles-components-single-limit} for the latter apply, in particular, to the former.
%auto-ignore

\section{Restricting star degenerations}
\label{sec:restricting-star-degenerations}

Degenerations with an open-chain resolution, see \cref{def:single-infinite-distance-limit}, have a multi-component central fiber for their base family manifold whose structure consists of a distinguished component (the strict transform of the original base component of the central fiber of the unresolved degeneration) to which one or two strings of Hirzebruch surface components intersecting in a chain are attached. This is represented in \cref{fig:single-infinite-distance-limit-open-chain}.

Suppose now that there exists a degeneration in which the original base component contains more than two mutually non-intersecting curves of non-minimal elliptic fibers. Given \cref{prop:component-geometry-single-bis},  this would lead to a resolution with a similar structure, but in which more than two strings of Hirzebruch surfaces are attached to the strict transform of the original base component, resembling a star-shaped resolution, rather than an open-chain one.

In this appendix, we argue that such star-shaped resolutions cannot occur for degenerations of  Calabi-Yau Weierstrass models constructed over one of the allowed six-dimensional F-theory bases (as reviewed in \cref{sec:definition-of-degenerations}) if their star resolution is to be free of obscured infinite-distance limits. While one can try to tune more than two mutually non-intersecting curves of non-minimal elliptic fibers, the structure of the Weierstrass model always forces some additional curves that intersect the originally tuned non-minimal curves to also factorize non-minimally. If the non-minimal nature of these additional curves is apparent at the level of the family vanishing orders, we will not obtain a star resolution; in the cases in which the non-minimal nature of these curves only manifests itself at the level of the component vanishing orders, the obtained star resolution will contain obscured infinite-distance limits.

We proceed with this discussion as part of the characterization of single infinite-distance limits done in \cref{sec:single-infinite-distance-limits-and-open-chain-resolutions}. Since non-minimal vanishing orders along a curve imply, in particular, non-minimal component vanishing orders along the same curve, we perform the study at the level of the central fiber of the starting degeneration. The analysis is carried out inductively, proving it first for Calabi-Yau Weierstrass models constructed over $\hat{B}_{0} = \mathbb{P}^{2}$ or $\hat{B}_{0} = \mathbb{F}_{n}$, and arguing in steps that it generalizes to models constructed over the arbitrary blow-ups $\hat{B}_{0} = \mathrm{Bl}(\mathbb{F}_{n})$, following the same path as in \cref{sec:list-arbitrary-blow-ups}. The final claim is the following.
\restrictingstardegenerations*

\subsection{Models constructed over \texorpdfstring{$\mathbb{P}^{2}$}{P2} or \texorpdfstring{$\mathbb{F}_{n}$}{Fn}}

The claim of \cref{prop:restricting-star-degenerations} can readily be proven for models constructed over $\hat{B}_{0} = \mathbb{P}^{2}$ or $\hat{B}_{0} = \mathbb{F}_{n}$ by directly solving for all the curves that can be simultaneously tuned to be non-minimal without mutually intersecting each other.

\begin{proposition}
\label{prop:star-degenerations-P2-Fn}
	Let $\pi: Y \rightarrow B$ be a Calabi-Yau Weierstrass model over $B = \mathbb{P}^{2}$ or $B = \mathbb{F}_{n}$ in which a set of curves $\mathscr{C}_{r} := \{C_{i}\}_{1 \leq i \leq r}$ in $B$ supports non-minimal elliptic fibers. If $r \geq 3$, then $C_{i} \cdot C_{j} \neq 0$ for some $1 \leq i < j \leq r$. In fact, the only collections of mutually non-intersecting curves on non-minimal elliptic fibers over these surfaces are
	\begin{itemize}
		\item $\mathscr{C}_{1} = \{H\}$, $\mathscr{C}_{1} = \{2H\}$ or $\mathscr{C}_{1} = \{3H\}$ for $B = \mathbb{P}^{2}$; or
		
		\item $\mathscr{C}_{1} = \{C\}$, with $C$ one of the curves listed in \cref{tab:genus-zero-Hirzebruch-summary}, or $\mathscr{C}_{2} = \{C_{0}, C_{\infty}\}$ for $\mathbb{F}_{n}$.
	\end{itemize}
\end{proposition}
\begin{proof}
	When $B = \mathbb{P}^{2}$, it is clear that this is the case, since all curves in $\mathbb{P}^{2}$ intersect each other and only the curves listed fulfil $C \leq \overline{K}_{\mathbb{P}^{2}}$, see \cref{prop:K-C-effectiveness}.
	
	For $B = \mathbb{F}_{n}$, \cref{prop:non-minimal-curves} (with the refinements of \cref{sec:restriction-cases-C-D}) provides the list of all non-minimal curves $C$ that can be tuned without forcing a second non-minimal curve to factorize, which therefore are the only valid elements of $\mathscr{C}_{r}$ when $r = 1$. Assume now that we have at least two curves of non-minimal elliptic fibers. Two irreducible curves
	\begin{subequations}
	\begin{align}
		C_{1} &= ah + bf\,,\qquad a \geq 0\,,\quad b \geq 0\,,\\
		C_{2} &= ch + df\,,\qquad c \geq 0\,,\quad d \geq 0\,,
	\end{align}
	\end{subequations}
	fulfil the conditions
	\begin{equation}
		\begin{rcases}
			C_{1} \cdot C_{2} = 0\\
			C_{1} + C_{2} \leq \overline{K}_{\mathbb{F}_{n}}
		\end{rcases}
		\Leftrightarrow
		\begin{cases}
			C_{1} = C_{2} = f\,,\\
			C_{1} = h\,,\quad C_{2} = h+nf\,,
		\end{cases}
	\end{equation}
	where we have used \cref{prop:K-C-effectiveness}. The first case leads to
	\begin{subequations}
	\begin{align}
		F - 8f &= 4C_{0} + 4C_{\infty}\,,\\
		G - 12f &= 6C_{0} + 6C_{\infty}\,,
	\end{align}
	\end{subequations}
	where we see that no curves with trivial intersection with $C_{1}$ and $C_{2}$ can be tuned to be non-minimal. Moreover, due to \cref{prop:negative-intersection}, we need to demand that $n = 0$ to avoid non-minimal elliptic fibers over $C_{0}$, at which point saying $C_{1} = C_{2} = f$ or $C_{1} = C_{0}$ and $C_{2} = C_{\infty}$ is merely a matter of convention. The second case leads to
	\begin{subequations}
	\begin{align}
		F - 8f &= 8f\,,\\
		G - 12f &= 12f\,,
	\end{align}
	\end{subequations}
	where $n$ can be $0 \leq n \leq 12$ and we see that no further non-minimal curves with trivial intersection with $C_{1}$ and $C_{2}$ can be tuned.
\end{proof}

\subsection{Models constructed over \texorpdfstring{$\mathrm{Bl}(\mathbb{F}_{n})$}{Bl(Fn)} of type (A)}

Beyond $\hat{B}_{0} = \mathbb{P}^{2}$ and $\hat{B}_{0} = \mathbb{F}_{n}$, six-dimensional F-theory also allows us to construct models over arbitrary blow-ups $\hat{B}_{0} = \mathrm{Bl}(\mathbb{F}_{n})$, of which we recall that $\hat{B}_{0} = \mathrm{Bl}(\mathbb{P}^{2})$ is a particular case. As we saw in \cref{sec:list-arbitrary-blow-ups}, the collection of possible base surfaces is huge, and we therefore cannot tackle them individually.

Instead, let us work inductively by exploiting the fact that we know the result to be true for $\hat{B}_{0} = \mathbb{P}^{2}$ and $\hat{B}_{0} = \mathbb{F}_{n}$ from \cref{prop:star-degenerations-P2-Fn}, and that the remaining candidate surfaces are constructed by successive blow-ups of these. To this end, we need to analyse how the blow-up of a point may increase our prospects of violating the claim of \cref{prop:restricting-star-degenerations}. Since the blow-up operation is a local one, the candidate set of curves $\hat{\mathscr{C}}_{r}$ must at least contain one curve affected by the blow-up in some way, as otherwise its pushforward would be a valid set of mutually non-intersecting non-minimal curves in the blown down surface.
\begin{lemma}
\label{lemma:triplet-candidates}
	Let $\pi_{\mathrm{ell}}: Y \rightarrow B$ be a Calabi-Yau Weierstrass model over a smooth surface $B$ in which a set of smooth irreducible curves $\mathscr{C}_{r} := \{C_{i}\}_{1 \leq i \leq r}$ in $B$ supports non-minimal elliptic fibers and such that if $r \geq 3$, then $C_{i} \cdot C_{j} \neq 0$ for some $1 \leq i < j \leq r$. Let $\pi: \hat{B} \rightarrow B$ be the blow-up of $B$ at a point $p \in B$ and $\hat{\pi}_{\mathrm{ell}}: \hat{Y} \rightarrow \hat{B}$ a Calabi-Yau Weierstrass model constructed over it. A collection $\hat{\mathscr{C}}_{r} = \{C_{r}\}_{1 \leq i \leq r}$ of irreducible curves in $\hat{B}$ supporting non-minimal elliptic fibers and with $r \geq 3$ has $C_{i} \cdot C_{j} \neq 0$ for some $1 \leq i < j \leq r$ unless $E \in \hat{\mathscr{C}}_{r}$ or $C'_{p} \in \hat{\mathscr{C}}_{r}$, where $E$ is the exceptional divisor of the blow-up and $C'_{p}$ is the strict transform of a curve $C_{p} \subset B$ passing through $p \in B$.
\end{lemma}
\begin{proof}
	The defining holomorphic line bundles $\hat{\mathcal{L}}$ and $\mathcal{L}$ of the Calabi-Yau Weierstrass models $\hat{\pi}_{\mathrm{ell}}: \hat{Y} \rightarrow \hat{B}$ and $\pi_{\mathrm{ell}}: Y \rightarrow B$ are related by
	\begin{equation}
		\hat{\mathcal{L}} = \overline{K}_{\hat{B}} = \pi^{*} \overline{K}_{B} - E = \pi^{*}\mathcal{L} - E\,.
	\label{eq:line-bundle-after-the-blow-up}
	\end{equation}
	Consider a collection $\hat{\mathscr{C}}_{r} = \{C_{i}\}_{1 \leq i \leq r}$ of mutually non-intersecting irreducible curves in $\hat{B}$. Assume that the elements in $\hat{\mathscr{C}}$ are all total/strict transforms of curves in $B$ not passing through the blow-up centre $p \in B$. If $\sum_{i=1}^{r}C_{i} \leq \overline{K}_{\hat{B}}$, the collection of curves can simultaneously support non-minimal elliptic fibers, according to \cref{prop:K-C-effectiveness}. Since they are chosen to be mutually non-intersecting, this is a valid set $\hat{\mathscr{C}}_{r}$ unless
\begin{equation}
	\hat{F}_{\mathrm{res}} := \hat{F} - \sum_{i=1}^{r}4C_{i}\,,\qquad \Ghres := \hat{G} - \sum_{i=1}^{r} 6C_{i}\,,
\end{equation}
are reducible and forcing an $(r+1)$-th curve $C_{r+1}$ to factorize non-minimally, i.e.\ the divisors $\hat{F}_{\mathrm{res}}$ and $\hat{G}_{\mathrm{red}}$ contain components $4C_{r+1}$ and $6C_{r+1}$ respectively. Then, either $C_{r+1}$ intersects one of the curves in $\hat{\mathscr{C}}_{r}$, in which case the original set was not valid, or we are forced to have a bigger set $\hat{\mathscr{C}}_{r+1}$ of mutually non-intersecting smooth irreducible curves of non-minimal fibers. Continuing in this way, we end up either discovering that we started with a bad candidate set, or with a collection $\hat{\mathscr{C}}_{r'}$ of mutually non-intersecting curves in $\hat{B}$ supporting non-minimal fibers and with $r' \geq 3$. The non-minimal factorization process stops when
\begin{equation}
	\left( \Ghres - 5C\right) \cdot C \geq 0\,,\qquad \forall C \subset \hat{B}\quad \mathrm{with}\quad C \cdot C < 0\,,
\label{eq:non-minimal-factorization-bound}
\end{equation}
see \cref{prop:negative-intersection}. If due to these forced factorizations the set $\hat{\mathscr{C}}_{r'}$ is now such that $E \in \hat{\mathscr{C}}_{r'}$ or $C'_{p}\in \hat{\mathscr{C}}_{r'}$, we are done. Otherwise, consider the set of curves $\mathscr{C}_{r'}$ in $B$ given by the pushforward of the elements in $\hat{\mathscr{C}}_{r'}$. The set $\mathscr{C}_{r'}$ still consists of $r'$ distinct mutually non-intersecting curves in $B$, given the assumptions on $\hat{\mathscr{C}}_{r'}$. Moreover, they can simultaneously support non-minimal elliptic fibers, since our assumptions on $\hat{\mathscr{C}}_{r}$ in conjunction with \eqref{eq:line-bundle-after-the-blow-up} imply
\begin{equation}
	\sum_{i=1}^{r} C_{i} \leq \overline{K}_{\hat{B}} \Rightarrow \sum_{i=1}^{r} \pi_{*}C_{i} \leq \overline{K}_{B}\,.
\end{equation}
In addition, if $C$ is the strict transform of a curve in $B$ passing through $p \in B$ with multiplicity $m$, we have with our assumptions for $\hat{\mathscr{C}}_{r'}$ that
\begin{equation}
	\left( \hat{G} - \sum_{i=1}^{r} 6C_{i} - 5C \right) \cdot_{\hat{B}} C = \left( G - \sum_{i=1}^{r}6 \pi_{*}C_{i} - 5\pi_{*}C \right) \cdot_{B} \pi_{*}C - m\,,
\end{equation}
meaning that \eqref{eq:non-minimal-factorization-bound} implies that no forced non-minimal factorizations occur in $B$ upon tuning $\mathscr{C}_{r'}$ to support non-minimal elliptic fibers. Hence, $\mathscr{C}_{r'}$ is a set of mutually non-intersecting curves of non-minimal elliptic fibers in $B$ with $r' \geq 3$, leading to a contradiction.
\end{proof}

In \cref{sec:list-arbitrary-blow-ups} we listed a possible way to classify the types of blow-ups that we can take of a surface. We commence our iterative study of \cref{prop:restricting-star-degenerations} by considering first models constructed over the $\hat{B}_{0} = \mathrm{Bl}(\mathbb{F}_{n})$ obtained by a succession of type (A) blow-ups.

Type (A) blow-ups of a surface $B$ are those in which we choose a collection of points $\{p_{i}\}_{1 \leq i \leq n_{p}}$ in $B$ and blow them up, rather than allowing blow-ups at points of the exceptional divisors as well. This means that the blow-up maps commute, and the order in which we take the points is not relevant. 

Since a Hirzebruch surface $\mathbb{F}_{n}$ is a $\mathbb{P}^{1}$-bundle over $\mathbb{P}^{1}$, we can subdivide type (A) blow-ups into those in which the blow-up centre touches the base, and those in which it does not. In order to introduce the notation that we will use below, let us be rather explicit about the types of divisors fitting into $\overline{K}_{\mathrm{Bl}(\mathbb{F}_{n})}$ that we can encounter in $\hat{B}_{0} = \mathrm{Bl}(\mathbb{F}_{n})$ after the composition of a series of type (A) blow-ups.

First, note that the case $\hat{B}_{0} = \mathrm{Bl}(\mathbb{F}_{0})$ is slightly different in this regard, since the starting geometry is $\mathbb{F}_{0} = \mathbb{P}^{1} \times \mathbb{P}^{1}$. This means that what we call fiber and section is arbitrary, and that $C_{0} = C_{\infty}$; we have no rigid curve with negative self-intersection in the starting surface. The curves in the original surface $\mathbb{F}_{0}$ will be denoted using the notation introduced in \cref{sec:P2-and-Fn}. To refer to the strict transforms of curves in $\mathbb{F}_{n}$ under the composition of all blow-up maps we will use primes, when the curve passes through a blow-up centre and hence its total and strict transform differ, and tildes, when the curve does not pass through a blow-up centre and its total and strict transform coincide. Exceptional divisors will always be denoted without a prime, as in \cref{sec:list-arbitrary-blow-ups}, referring to their strict transform under the composition of all posterior blow-up maps. Occasionally, we will indicate in square brackets some of the blow-up centres associated with the strict transform or exceptional divisor, to avoid ambiguities. With this notation, the curves that we need to consider are:
\begin{itemize}
	\item $\tilde{C}_{0}$, the strict transform of a representative of $C_{0}$ not passing through any $p \in \{ p_{i} \}_{0 \leq i \leq n_{p}}$;
	
	\item $\tilde{f}$, the strict transform of a representative of $f$ not passing through any $p \in \{ p_{i} \}_{0 \leq i \leq n_{p}}$;
	
	\item $C'_{0,i}[p_{j}, \neg p_{k}]$, the strict transform of a representative of $C_{0}$ passing through the blow-up centre $p_{j} \in \{ p_{i} \}_{1 \leq i \leq n_{p}}$, not passing through the blow-up centre $p_{k} \in \{ p_{i} \}_{1 \leq i \leq n_{p}}$, and in which a total of $i$ points have been blown up;
	
	\item $f'_{i}[p_{j}, \neg p_{k}]$, the strict transform of a representative of $f$ passing through the blow-up centre $p_{j} \in \{ p_{i} \}_{1 \leq i \leq n_{p}}$, not passing through the blow-up centre $p_{k} \in \{ p_{i} \}_{1 \leq i \leq n_{p}}$, and in which a total of $i$ points have been blown up; and
	
	\item $E_{i}$, the exceptional divisor associated to the blow-up with centre $p_{i} \in \{ p_{i} \}_{1 \leq i \leq n_{p}}$.
\end{itemize}
Not every type (A) blow-up of $\mathbb{F}_{0}$ leads to a base $\hat{B}_{0} = \mathrm{Bl}(\mathbb{F}_{0})$ with effective anticanonical class $\overline{K}_{\mathrm{Bl}(\mathbb{F}_{0})}$. Recalling from \cref{sec:anticanonical-class-after-arbitrary-blow-up} that after a composition of type (A) blow-ups we have the anticanonical class \eqref{eq:anticanonical-class-type-A-blow-up}, or written with our current notation
\begin{equation}
	\overline{K}_{\mathrm{Bl}(\mathbb{F}_{0})} = 2\tilde{C}_{0} + 2\tilde{f} - \sum_{i=1}^{n_{p}} E_{i}\,,
\end{equation}
we see that blowing up more than four points in general position leads to a non-effective anticanonical class. It is possible to blow up more than four points, but some of them must lie on the same representative of $C_{0}$ of $f$. Due to \cref{prop:K-C-effectiveness}, the situation is even more stringent when it comes to tuning exceptional divisors to be non-minimal, since, e.g.,\ tuning $E_{i}$ and $E_{j}$ to be non-minimal means that if a third $E_{k}$ is tuned to be non-minimal it must stem from the blow-up of a point in the same representative of $C_{0}$ or $f$ that gave rise to either $E_{i}$ or $E_{j}$.

Moving now to the cases $\hat{B}_{0} = \mathrm{Bl}(\mathbb{F}_{n})$, with $1 \leq n \leq 12$, we can distinguish those blow-ups associated to points $p \in \{p_{i}\}_{1 \leq i \leq n_{p}}$ such that $p \in C_{0}$ from those in which $p \notin C_{0}$. Let us denote the total number of the former type of blow-ups by $n_{0}$ and the associated blow-up centres by $\{ p_{i} \}_{i \in \mathcal{N}_{0}} \subset \{ p_{i} \}_{1 \leq i \leq n_{p}}$, and the total number of the latter by $\ninftot$ and the associated blow-up centres by $\{ p_{i} \}_{i \in \mathcal{N}_{\infty}} \subset \{ p_{i} \}_{1 \leq i \leq n_{p}}$, such that $n_{p} = n_{0} + \ninftot$. Each point $p \notin C_{0}$ sits in a unique representative of $f$. It will be convenient to define a quantity $n_{\infty} \leq \ninftot$ counting the number of distinct $f$ representatives affected by the $\ninftot$ blow-ups at points $p \notin C_{0}$. Occasionally, we will need to refer to all the blow-up centres associated with a strict transform $C'$, which we will do by $p\{C'\}$. We will use a notation similar to the one employed in the $\hat{B}_{0} = \mathrm{Bl}(\mathbb{F}_{0})$ case. With this notation, the curves that we need to consider are:
\begin{itemize}
	\item $\tilde{C}_{\infty}$, the strict transform of a representative of $C_{\infty}$ not passing through any $p \in \{ p_{i} \}_{0 \leq i \leq n_{p}}$;
	
	\item $\tilde{f}$, the strict transform of a representative of $f$ not passing through any $p \in \{ p_{i} \}_{0 \leq i \leq n_{p}}$;
	
	\item $C'_{0}$, the strict transform of the unique representative of $C'_{0}$, in which a total of $n_{0}$ points have been blown up;
	
	\item $C'_{\infty,i}[p_{j},\neg p_{k}]$: the strict transform of a representative of $C_{\infty}$ passing through the blow-up centre $p_{j} \in \{ p_{i} \}_{1 \leq i \leq n_{p}}$, not passing through the blow-up centre $p_{k} \in \{ p_{i} \}_{1 \leq i \leq n_{p}}$, and in which a total of $i$ points have been blown up;
	
	\item $f'_{0,i}$ the strict transform of a representative of $f$ passing through a blow-up centre $p \in C_{0}$, and in which a total of $i$ points have been blown up;
	
	\item $f'_{\infty,i}[p_{j}, \neg p_{k}]$ the strict transform of a representative of $f$ passing through the blow-up centre $p_{j} \in \{ p_{i} \}_{1 \leq i \leq n_{p}}$, not passing through the blow-up centre $p_{k} \in \{ p_{i} \}_{1 \leq i \leq n_{p}}$, and in which a total of $i$ points $p \notin C_{0}$ have been blown up;
	
	\item $f'_{0/\infty,i}[p_{j}, \neg p_{k}]$ the strict transform of a representative of $f$ passing through the blow-up centre $p_{j} \in \{ p_{i} \}_{1 \leq i \leq n_{p}}$, not passing through the blow-up centre $p_{k} \in \{ p_{i} \}_{1 \leq i \leq n_{p}}$, and in which a total of $i-1$ points $p \notin C_{0}$ and a point $p \in C_{0}$ have been blown up;
	
	\item $E_{i}^{0}$: the exceptional divisor associated to a blow-up with centre $p \in C_{0}$; and
	
	\item $E_{i}^{\infty}[C'_{j}, \neg C'_{k}]$: the exceptional divisor associated to a blow-up with centre $p \notin C_{0}$, with the pushforwards of $C'_{j}$ and $C'_{k}$ passing and not passing through $p$, respectively.
\end{itemize}
The cases $f'_{0,i}$ and $f'_{0/\infty,i}[p_{j}, \neg p_{k}]$ are essentially the same, but it is contextually useful to use this notation, as the latter expression will mean that the point of intersection of the fiber class representative from which $f'_{0/\infty,i}[p_{j}, \neg p_{k}]$ stems with the particular representative of $C_{\infty}$ whose strict transform is relevant at that point in the discussion has been blown up. Every so often we will write $C'_{1}[p\{C'_{2}\}]$; this does not mean that the pushforward of $C'_{1}[p\{C'_{2}\}]$ passes through all points $p\{C'_{2}\}$, but rather that any points of intersection between the pushforwards of the two curves have been blown up. Sporadically, we will drop the square brackets after first introducing a curve, if the subindices are enough to distinguish it in the subsequent context.

As occurred for the $\hat{B}_{0} = \mathrm{Bl}(\mathbb{F}_{0})$ case, not every type (A) blow-up of $\mathbb{F}_{n}$, with $1 \leq n \leq 12$, leads to a base $\hat{B}_{0} = \mathrm{Bl}(\mathbb{F}_{n})$ with effective anticanonical class $\overline{K}_{\mathrm{Bl}(\mathbb{F}_{n})}$. Writing \eqref{eq:anticanonical-class-type-A-blow-up} with our current notation we have
\begin{equation}
	\overline{K}_{\mathrm{Bl}(\mathbb{F}_{n})} = 2\tilde{C}_{0} + (2+n)\tilde{f} - \sum_{i=1}^{n_{p}} E_{i} = 2C'_{0} + (2+n)\tilde{f} + \sum_{i \in \mathcal{N}_{0}}E_{i}^{0} - \sum_{i \in \mathcal{N}_{\infty}}E_{i}^{\infty}\,.
\end{equation}
The $\{E_{i}^{0}\}_{i \in \mathcal{N}_{0}}$ do not pose a threat to the effectiveness of the anticanonical class $\overline{K}_{\mathrm{Bl}(\mathbb{F}_{n})}$. Since $C_{0}$ has a unique representative, whose total transform gives
\begin{equation}
	\tilde{C}_{0} = C'_{0} + \sum_{i \in \mathcal{N}_{0}} E_{i}^{0}\,,
\label{eq:linear-equivalence-blown-up-C0}
\end{equation}
and this class appears twice in $\overline{K}_{\mathrm{Bl}(\mathbb{F}_{n})}$, we have that $E_{i}^{0} \leq \overline{K}_{B}$ for all $i \in \mathcal{N}_{0}$, and can tune them to support non-minimal fibers according to \cref{prop:K-C-effectiveness}. The $\{E_{i}^{\infty}\}_{i \in \mathcal{N}_{\infty}}$, on the other hand, are related to particular representatives of $f$, meaning that we have
\begin{equation}
	\tilde{f} = f'_{\infty,i} + \sum_{\substack{i \in \mathcal{N}_{0}\; \mathrm{s.t.}\\p_{i} \in p\{f'_{\infty,i}\}}} E_{i}^{\infty} = f'_{\infty,i'} + \sum_{\substack{i' \in \mathcal{N}_{0}\; \mathrm{s.t.}\\p_{i'} \in p\{f'_{\infty,i'}\}}} E_{i'}^{\infty}\,.
\label{eq:linear-equivalence-blown-up-f}
\end{equation}
Hence, we need to have at least $n_{\infty}$ representatives of $\tilde{f}$ available in $\overline{K}_{\mathrm{Bl}(\mathbb{F}_{n})}$ for it to still be effective after the blow-up process, from where we obtain the effectiveness bound
\begin{equation}
	n_{\infty} \leq n + 2\,.
\end{equation}
This bound is just to have $\overline{K}_{\mathrm{Bl}(\mathbb{F}_{n})}$ be effective, but if we want to tune an exceptional divisor $E_{i}^{\infty}$ to be non-minimal we need in addition $E_{i}^{\infty} \leq \overline{K}_{\mathrm{Bl}(\mathbb{F}_{n})}$ to comply with \cref{prop:K-C-effectiveness}. Denoting by $m$ the number of elements in $\{E_{i}^{\infty}\}_{i \in \mathcal{N}_{\infty}}$ stemming from different fiber class representatives that we want to tune non-minimally, the effectiveness bound becomes
\begin{equation}
	n_{\infty} + m \leq n + 2\,.
\label{eq:effectiveness-bound-type-A-shifted}
\end{equation}

Although we have just seen that the number of blow-ups $n_{0}$ over points $p \in C_{0}$ is not restricted by requiring the effectiveness of $\overline{K}_{\mathrm{Bl}(\mathbb{F}_{n})}$, if we blow $C_{0}$ up at too many points the self-intersection of $C'_{0}$ becomes so negative that a non-minimal non-Higgsable cluster appears for $\mathrm{Bl}(\mathbb{F}_{n})$. This would mean that any degeneration of elliptic Calabi-Yau threefolds fibered over $\hat{B}_{0} = \mathrm{Bl}(\mathbb{F}_{n})$ would present, at least at the level of the component vanishing orders, non-minimal singularities over all $u \in D$, and not only over the central fiber. In order to avoid this, we also need to demand
\begin{equation}
	n + n_{0} \leq 12\,.
\end{equation}

The divisor classes listed above are not the only ones over which we can tune non-minimal elliptic fibers in $\mathrm{Bl}(\mathbb{F}_{n})$. The rest, however, are obtained as combinations of these, and are a less efficient use of the divisor classes available in $\overline{K}_{\mathrm{Bl}(\mathbb{F}_{n})}$ (if our aim is to find a candidate triplet of curves violating \cref{prop:restricting-star-degenerations}), meaning that they will intersect more curves and force more factorizations of the residual defining polynomials. Hence, we need to search for candidate triplets among the divisors listed earlier, which we now do.

\begin{proposition}
\label{prop:star-degenerations-type-A}
	Let $\pi: Y \rightarrow B$ be a Calabi-Yau Weierstrass model over $B$, where $B = \mathrm{Bl}(\mathbb{F}_{n})$ is the surface obtained by choosing a collection of points in $\mathbb{F}_{n}$ and blowing them up. Let $\mathscr{C}_{r} := \{ C_{i} \}_{1 \leq i \leq r}$ be a set of curves in $B$ that support non-minimal elliptic fibers. If $r \geq 3$, then $C_{i} \cdot C_{j} \neq 0$ for some $1 \leq i < j \leq r$.
\end{proposition}
\begin{proof}
	Let us work by induction, with the base case provided by \cref{prop:star-degenerations-P2-Fn}. Assume that the result holds for Calabi-Yau Weierstrass models constructed over $B = \mathrm{Bl}_{k}(\mathbb{F}_{n})$, obtained as the blow-up of $\mathbb{F}_{n}$ with centre at the points $\{p_{i}\}_{1 \leq i \leq k}$, where $k := n_{p} = n_{0} + \ninftot$. Consider now the Calabi-Yau Weierstrass models constructed over $\hat{B} = \mathrm{Bl}_{k+1}(\mathbb{F}_{n})$, obtained by a further blow-up with a $(k+1)$-th point $p_{k+1}$ in $\mathbb{F}_{n}$ as centre, such that $\hat{n}_{p} = n_{0} + \ninftot + 1$. We can distinguish the three cases $n=0$, $n \geq 1$ with $p_{k+1} \in C_{0}$, and $n \geq 1$ with $p_{k+1} \notin C_{0}$, that we treat separately.

To avoid any ambiguity, let us clarify that below $n_{0}$, $n_{\infty}$ and $\ninftot$ refer to the $B$ surface, and do not count the $(k+1)$-th	blow-up. When we write a strict transform in $\hat{B}$ like, e.g.,\ $f'_{\infty,i}$, the subindex $i$ refers to all the blow-ups affecting the representative of $f$ that $f'_{\infty,i}$ stems from, including the $(k+1)$-th blow-up if appropriate.

	\begin{enumerate}[label=(\arabic*)]
		\item $n=0$: According to \cref{lemma:triplet-candidates}, we need to consider the candidate triplets $\{ E_{k+1}, \smallbullet, \smallbullet \}$, $\{ C'_{0,i}(p_{k+1}), \smallbullet, \smallbullet \}$ and $\{ f'_{i}(p_{k+1}), \smallbullet, \smallbullet \}$. In fact, when $n=0$, the choice of what we call section and fiber is arbitrary, and therefore the second and the third candidate triplets are analogous. Hence, we only address the first.
		\begin{enumerate}[label=(\arabic{enumi}.\alph*)]
			\item $\{ E_{k+1}, \smallbullet, \smallbullet \}$: $E_{k+1}$ intersects the total transforms $C'_{0,i}[p_{k+1}]$ and $f'_{i}[p_{k+1}]$, which discards them as candidates to complete the triplet. We need to analyse the following candidates for triplet completion: $\tilde{C}_{0}$, $C'_{0,i}[\neg p_{k+1}]$, $\tilde{f}$, $f'_{i}[\neg p_{k+1}]$ and $E_{i}$.
			\begin{enumerate}[label=(\arabic{enumi}.\alph{enumii}.\roman*)]
				\item\label{item:neq0-E-C0} $\{ E_{k+1}, \tilde{C}_{0}, \smallbullet \}$: After tuning these two divisors to be non-minimal, the residual $\hat{G}$ divisor is
				\begin{equation}
					\Ghres = 6\tilde{C}_{0} + 12\tilde{f} - \sum_{i=1}^{n_{p}}6 E_{i} - 12 E_{k+1}\,.
				\end{equation}
				Since
				\begin{equation}
					\left( \Ghres - \alpha f'_{i}[p_{k+1}] \right) \cdot f'_{i}[p_{k+1}] = -(6-\alpha)i
				\end{equation}
				and $i \geq 1$, we see that $f'_{i}[p_{k+1}]$ factorizes non-minimally, with $f'_{i}[p_{k+1}] \cdot E_{k+1} = 1$.
				
				\item\label{item:neq0-E-C0prime} $\{ E_{k+1}, C'_{0,i}[\neg p_{k+1}], \smallbullet \}$: In this case the residual $\hat{G}$ divisor after the tuning is
				\begin{equation}
					\hat{G} = 12\tilde{C}_{0,i} + 12\tilde{f} - \sum_{i=1}^{n_{p}} 6E_{i} - 6E_{k+1} - 6C'_{0,i}[\neg p_{k+1}]\,.
				\end{equation}
				We can distinguish two subcases:
				\begin{itemize}
					\item if the intersection point $f'_{j}[p_{k+1}] \cap C'_{0,i}[\neg p_{k+1}] =: p$ in $\mathbb{F}_{0}$ has not been blown-up, we have
					\begin{equation}
						\left( G_{\mathrm{res}} - \alpha f'_{j}[p_{k+1}] \right) \cdot f'_{j}[p_{k+1}] = -(6-\alpha)j\,,
					\end{equation}
					with $j \geq 1$, such that $f'_{j}[p_{k+1}]$ factorizes non-minimally; and
					
					\item if the intersection point of $f'_{j}[p_{k+1}]$ and $C'_{0,i}[\neg p_{k+1}]$ in $\mathbb{F}_{0}$ has been blown-up, we have instead
					\begin{equation}
						\left( G_{\mathrm{res}} - \alpha f'_{j}[p_{k+1}] \right) \cdot f'_{j}[p_{k+1}] = 6-(6-\alpha)j\,,
					\end{equation}
					with $j \geq 2$. This means that at least $3f'_{j}[p_{k+1}]$ will factorize, implying, in turn, that due to
					\begin{equation}
						\left( G_{\mathrm{res}} - 3f'_{j}[p_{k+1}] - \alpha E_{l}\left[ p \right] \right) = -3+\alpha\,,
					\end{equation}
					at least $3E_{l}[p]$ will factorize. Then
					\begin{equation}
						\left( G_{\mathrm{res}} - 3 f'_{j}[p_{k+1}] - 3E_{l}[p] -\alpha f'_{j}[p_{k+1}] \right) \cdot f'_{j}[p_{k+1}] = 3-(3-\alpha)i\,,
					\end{equation}
					with $j \geq 2$, leading to an additional factorization of at least $3f'_{j}[p_{k+1}]$. Since
					\begin{equation}
						\left( G_{\mathrm{res}} - 5f'_{j}[p_{k+1}] - 3E_{l}[p] - \alpha E_{l} \right) = -2+\alpha\,,
					\end{equation}
					at least $2E_{l}[p]$ further factorize, which finally leads to
					\begin{equation}
						\left( G_{\mathrm{res}} - 5 f'_{j}[p_{k+1}] - 5E_{l}[p] \right) \cdot f'_{j}[p_{k+1}] = 1-j\,,
					\end{equation}
					with $j \geq 2$, meaning that $f'_{j}[p_{k+1}]$ factorizes non-minimally.
				\end{itemize}
				In both cases $f'_{j}[p_{k+1}] \cdot E_{k+1} = 1$.
				
				\item $\{ E_{k+1}, \tilde{f}, \smallbullet \}$: This case is analogous to Case~\ref{item:neq0-E-C0}, with the roles of $\{\tilde{C}_{0}, f'_{i}[p_{k+1}]\}$ and $\{ \tilde{f}, C'_{0,i}[p_{k+1}] \}$ exchanged.
				
				\item $\{ E_{k+1}, f'_{i}[p_{k+1}], \smallbullet \}$: This case is analogous to Case~\ref{item:neq0-E-C0prime}, after exchanging the roles of $\{ C'_{0,i}[\neg p_{k+1}], f'_{j}[p_{k+1}] \}$ and $\{ f'_{i}[\neg p_{k+1}], C'_{0,j}[p_{k+1}] \}$.
				
				\item $\{ E_{k+1}, E_{j}, \smallbullet \}$: The residual $\hat{G}$ divisor after tuning these two divisors to be non-minimal is
				\begin{equation}
					\Ghres = 12\tilde{C}_{0} + 12\tilde{f} - \sum_{i=1}^{n_{p}} 6E_{i} - 6E_{j} - 12E_{k+1}\,.
				\end{equation}
				We can distinguish two subcases:
				\begin{itemize}
					\item if $E_{k+1}$ and $E_{j}$ are associated to the blow-ups of the same $C_{0}$ representative (or $f$, for which the same argument would apply), then
					\begin{equation}
						\left( \Ghres -  \alpha C'_{0,i}[p_{k+1},p_{j}] \right) \cdot C'_{0,i}[p_{k+1},p_{j}] = -(6-\alpha)i\,,
					\end{equation}
					with $i \geq 2$, yielding a non-minimal factorization of $C'_{0,i}[p_{k+1},p_{j}]$; and
					
					\item if $E_{k+1}$ and $E_{j}$ are not associated to the blow-ups of a common $C_{0}$ or $f$ representative, in which case the residual $\Ghres$ discriminant can be rewritten as
					\begin{equation}
						\Ghres = 12C'_{0,i}[p_{k+1}] + 12f'_{l}[p_{j}] + \sum_{r} 6E_{r}[C'_{0,i}[p_{k+1}]] + \sum_{s} 6E_{s}[f'_{l}[p_{j}]]\,,
					\end{equation}
					from where we see that the only possible non-minimal effective tunings that can be performed of divisors not intersecting $E_{k+1}$ and $E_{j}$ are those involving another exceptional divisor $E_{m}$. It must be, however, related to either the blow-ups of the pushforward of $C'_{0,i}[p_{k+1}]$, or to the blow-ups of the pushforward of $f'_{l}[p_{j}]$, which amounts again to the previous case.
				\end{itemize}
			\end{enumerate}		
			
			\item $\{ C'_{0,i}[p_{k+1}], \smallbullet, \smallbullet \}$: $C'_{0,i}[p_{k+1}]$ intersects $E_{k+1}$ and the other exceptional divisors associated with the blow-ups of the pushforward of $C'_{0,j}[p_{k+1}]$, as well as the strict transforms $f'_{j}[\neg p_{k+1}]$. We need to consider the following candidates for triplet completion: $\tilde{C}_{0}$, $C'_{0,j}[\neg p_{k+1}]$, $f'_{j}[p_{k+1}]$ and $E_{j}[C'_{0,i}[p_{k+1}]]$.
			\begin{enumerate}[label=(\arabic{enumi}.\alph{enumii}.\roman*)]
				\item $\{ C'_{0,i}[p_{k+1}], \tilde{C}_{0}, \smallbullet \}$: Tuning these two divisors to be non-minimal leaves us with
				\begin{equation}
					\Ghres = 6\tilde{C}_{0} + 12\tilde{f} - \sum_{i=1}^{n_{p}} 6E_{i} - 6E_{k+1} - 6C'_{0,i}\,.
				\end{equation}
				Since
				\begin{equation}
					\left( \Ghres - \alpha f'_{l}[p_{k+1}] \right) \cdot f'_{l}[p_{k+1}] = -(6-\alpha)l\,,
				\end{equation}
				with $l \geq 1$, $f'_{l}[p_{k+1}]$ factorizes non-minimally, with $f'_{l}[p_{k+1}] \cdot \tilde{C}_{0,i} = 1$.
				
				\item $\{ C'_{0,i}[p_{k+1}], C'_{0,j}[\neg p_{k+1}], \smallbullet \}$: After tuning these two divisors to be non-minimal, we have
				\begin{equation}
					\Ghres = 12\tilde{f} - \sum_{r} 6E_{r}[\neg C'_{0,i}[p_{k+1}], \neg C'_{0,j}[\neg p_{k+1}]]\,.
				\end{equation}
				There are two types of divisors that do not intersect $C'_{0,i}[p_{k+1}]$ and $C'_{0,j}[\neg p_{k+1}]$ that can be tuned effectively:
				\begin{itemize}
					\item we can tune an exceptional divisor $E_{r}[\neg C'_{0,i}[p_{k+1}], \neg C'_{0,j}[\neg p_{k+1}]]$, which would lead us to Case~\ref{item:neq0-E-C0prime}; or
					
					\item assuming that the pushforward of $f'_{l}[p_{k+1}]$ intersects the pushforward of $C'_{0,j}[\neg p_{k+1}]$ at a point $p_{j}$ that has been blown-up, we can tune the strict transform $f'_{l}[p_{k+1},p_{j}]$. However, this leads to the residual $\Ghres$ divisor
					\begin{equation}
					\begin{split}
						\Ghres &= 6\tilde{f} - \sum_{r'} 6E_{r'}[\neg C'_{0,i}[p_{k+1}], \neg C'_{0,j}[\neg p_{k+1}], \neg f'_{l}[p_{k+1},p_{j}]]\\
						&\quad + 6E_{k+1}[p_{k+1}] + 6E_{j}[p_{j}]\,.
					\end{split}
					\end{equation}
					Then $E_{k+1}[p_{k+1}]$ and $E_{j}[p_{j}]$ factorize non-minimally, with the intersections $C'_{0,i}[p_{k+1}] \cdot E_{k+1}[p_{k+1}] = 1$ and $C'_{0,j}[\neg p_{k+1}] \cdot E_{j}[p_{j}] = 1$.
				\end{itemize}
				
				\item $\{ C'_{0,i}[p_{k+1}], f'_{0,j}[p_{k+1}], \smallbullet \}$: After tuning these two divisors to be non-minimal, the residual $\Ghres$ divisor is
				\begin{equation}
					\Ghres = 6\tilde{C}_{0} + 6\tilde{f} - \sum_{r} 6E_{r}[\neg C'_{0,i}[p_{k+1}], \neg f'_{0,j}[p_{k+1}]] + 6E_{k+1}\,,
				\end{equation}
				from which we see that $E_{k+1}$ factorizes non-minimally, with $C'_{0,i}[p_{k+1}] \cdot E_{k+1} = 1$ and $f'_{0,j}[p_{k+1}] \cdot E_{k+1} = 1$.
				
				\item $\{ C'_{0,i}[p_{k+1}], E_{j}[\neg p_{k+1}], \smallbullet \}$: This is analogous to Case~\ref{item:neq0-E-C0prime}.
			\end{enumerate}
		\end{enumerate}		
		
		\item $n \geq 1$ with $p_{k+1} \in C_{0}$: According to \cref{lemma:triplet-candidates}, we need to consider the candidate triplets $\{ E^{0}_{k+1} , \smallbullet, \smallbullet \}$, $\{ C'_{0}[p_{k+1}], \smallbullet, \smallbullet \}$ and $\{ f'_{0,i}[p_{k+1}] , \smallbullet, \smallbullet \}$.
		\begin{enumerate}[label=(\arabic{enumi}.\alph*)]
			\item $\{ E^{0}_{k+1} , \smallbullet, \smallbullet \}$: After tuning $E_{k+1}^{0}$ to be non-minimal, the residual $\hat{G}$ divisor is
		\begin{equation}
			\Ghres = 12\tilde{C}_{0} + (12+6n)\tilde{f} - \sum_{i \in \mathcal{N}_{0}} 6E_{i}^{0} - \sum_{i\in \mathcal{N}_{\infty}} 6E_{i}^{\infty} - 12 E_{k+1}^{0}\,,
		\end{equation}
		for which we have
		\begin{equation}
			\left( \Ghres - \alpha C'_{0} \right) \cdot C'_{0} = -(6-\alpha)(n - n_{0})\,.
		\end{equation}
		Since $n \geq 1$, this implies that that $C'_{0}$ factorizes non-minimally, with $C'_{0} \cdot E_{k+1}^{0} = 1$.
		
		\item $\{ C'_{0}, \smallbullet, \smallbullet \}$: $C'_{0}$ intersects the strict transforms $\tilde{f}$, $f'_{\infty,i}$ and the exceptional divisors $E^{0}_{i}$. The possible triplet completions to be considered are: $\tilde{C}_{\infty}$, $C'_{\infty,j}$, $f'_{0,j}$ and $E^{\infty}_{j}$.
		\begin{enumerate}[label=(\arabic{enumi}.\alph{enumii}.\roman*)]
			\item $\{ C'_{0,i}, \tilde{C}_{\infty}, \smallbullet \}$: After tuning these two divisors to be non-minimal, the residual $\hat{G}$ divisor is
			\begin{equation}
				\Ghres = (12+6n)\tilde{f} - \sum_{i \in \mathcal{N}_{\infty}} 6E^{\infty}_{i}\,,
			\end{equation}
			from which we see that the strict transforms $f'_{\infty,j}$ factorize non-minimally, with $f'_{\infty,j} \cdot C'_{0} = 1$.
			
			\item $\{ C'_{0}, C'_{\infty,i}, \smallbullet \}$: The residual $\hat{G}$ divisor is in this occasion
			\begin{equation}
				\Ghres = 12\tilde{C}_{0} + (12+6n)\tilde{f} - \sum_{i=1}^{n_{p}} E_{i} - 6E^{0}_{k+1} - 6C'_{0} - 6C'_{\infty,i}\,.
			\end{equation}
			Given the divisors that intersect $C'_{0}$, listed above, and the fact that $f'_{0,j}[\neg C'_{\infty,i}]$ and $E^{\infty}_{l}[C'_{\infty,i}]$ intersect $C'_{\infty,i}$, the only two types of divisors that we can effectively tune non-minimally are the strict transform of a vertical line $f'_{0}[C'_{\infty,i}]$, or an exceptional divisor $E^{\infty}_{l}[\neg C'_{\infty,i}]$. However, we see that
			\begin{itemize}
				\item due to the intersection
				\begin{equation}
					\left( \Ghres - 6f'_{0,j}[\neg C'_{\infty,i}] - \alpha E^{0}_{l}[f'_{0,j}[\neg C'_{\infty,i}]] \right) \cdot E^{0}_{l}[f'_{0,j}[\neg C'_{\infty,i}]] = -6+\alpha\,,
				\end{equation}
				$E^{0}_{l}[f'_{0,j}[\neg C'_{\infty,i}]]$ factorizes non-minimally, with $E^{0}_{l}[f'_{0,j}[\neg C'_{\infty,i}]] \cdot C'_{0,i} = 1$; and
				
				\item due to the intersection
				\begin{equation}
					\left( \Ghres - 6E^{\infty}_{l}[\neg C'_{\infty,i}, f'_{(0/)\infty,m}] - \alpha f'_{(0/)\infty,m} \right) \cdot f'_{(0/)\infty,m} = -(6-\alpha)m\,,
				\end{equation}
				with $m \geq 1$, $f'_{(0/)\infty,m}$ factorizes non-minimally, and $f'_{(0/)\infty,m} \cdot E^{\infty}_{l} = 1$.
			\end{itemize}
		\end{enumerate}
		
		\item $\{ f'_{0,i}[p_{k+1}], \smallbullet, \smallbullet \}$: $f'_{0,i}$ intersects $\tilde{C}_{\infty}$, $C'_{\infty,j}$ and $E^{0}_{k+1}$. The possible triplet completions that we need to consider are: $C'_{0}$, $\tilde{f}$, $f'_{0,j}[\neg p_{k+1}]$, $f'_{\infty,j}[\neg p_{k+1}]$, $E^{0}_{i}$ and $E^{\infty}_{i}$.
		\begin{enumerate}[label=(\arabic{enumi}.\alph{enumii}.\roman*)]
			\item $\{ f'_{0,i}[p_{k+1}], C'_{0}, \smallbullet \}$: After tuning these two divisors to be non-minimal, the residual $\hat{G}$ divisor is
			\begin{equation}
				\Ghres = 6C'_{0} + (12+6n)\tilde{f} - \sum_{i \in \mathcal{N}_{\infty}}6E^{\infty}_{i} - 6f'_{0,i}[p_{k+1}]\,.
			\end{equation}
			Given the intersection
			\begin{equation}
				\left( \Ghres - \alpha E^{0}_{k+1} \right) \cdot E^{0}_{k+1} = -6+\alpha\,,
			\end{equation}
			we see that $E^{0}_{k+1}$ factorizes non-minimally, with $E^{0}_{k+1} \cdot f'_{0,i}[p_{k+1}] = 1$.
			
			\item\label{item:ngeq1C0-f0prime-f} $\{ f'_{0,i}[p_{k+1}], \tilde{f}, \smallbullet \}$: The residual $\hat{G}$ divisor is
			\begin{equation}
				\Ghres = 12\tilde{C}_{0} + (12+6n)\tilde{f} - \sum_{i=1}^{n_{p}}6E_{i} - 6E_{k+1}^{0} - 6f'_{0,i}[p_{k+1}] - 6\tilde{f}\,.
			\end{equation}
			This leads to
			\begin{equation}
				\left( \Ghres - \alpha C'_{0} \right) \cdot C'_{0} = -6(n+n_{0})\,,
			\end{equation}
			with $n \geq 1$ and $n_{0} \geq 1$, meaning that $C'_{0}$ factorizes non-minimally. This implies
			\begin{equation}
				\left( \Ghres - 6C'_{0} - \alpha E^{0}_{k+1} \right) \cdot E^{0}_{k+1} = -6 + \alpha\,,
			\end{equation}
			yielding a non-minimal factorization of $E^{0}_{k+1}$, with $E^{0}_{k+1} \cdot f'_{0,i}[p_{k+1}] = 1$.
			
			\item\label{item:ngeq1C0-f0prime-fprime} $\{ f'_{0,i}[p_{k+1}], f'_{0,j}[\neg p_{k+1}], \smallbullet \}$: The residual $\hat{G}$ divisor is
			\begin{equation}
				\Ghres = 12\tilde{C}_{0} + (12+6n)\tilde{f} - \sum_{i=1}^{n_{p}}6E_{i} - 6E_{k+1}^{0} - 6f'_{0,i}[p_{k+1}] - 6f'_{0,j}[\neg p_{k+1}]\,.
			\end{equation}
			The intersection
			\begin{equation}
				\left( \Ghres - \alpha C'_{0} \right) \cdot C'_{0} = 6-(6-\alpha)(n+n_{0})\,,
			\end{equation}
			with $n \geq 1$ and $n_{0} \geq 1$, means that at least $3C'_{0}$ factorize. Then
			\begin{equation}
				\left( \Ghres - 3C'_{0} - \alpha E^{0}_{k+1} \right) \cdot E^{0}_{k+1} = -3+\alpha\,,
			\end{equation}
			leading to at least $3E^{0}_{k+1}$ factorizing. This, in turn, implies that
			\begin{equation}
				\left( \Ghres - 3C'_{0} - 3E^{0}_{k+1} - \alpha C'_{0} \right) \cdot C'_{0} = 3-(3-\alpha)(n+n_{0})\,,
			\end{equation}
			yielding an additional factorization of at least $2C'_{0}$. Then
			\begin{equation}
				\left( \Ghres - 5C'_{0} - 3E^{0}_{k+1} - \alpha E^{0}_{k+1} \right) \cdot E^{0}_{k+1} = -2+\alpha\,,
			\end{equation}
			forces an additional factorization of at least $2E^{0}_{k+1}$. From these, we obtain
			\begin{equation}
				\left( \Ghres - 5C'_{0} - 5E^{0}_{k+1} - \alpha C'_{0} \right) \cdot C'_{0} = -(1-\alpha)(n+n_{0}) + 1\,,
			\end{equation}
			leading to an additional factorization of at least $C'_{0}$. Finally, this yields
			\begin{equation}
				\left( \Ghres - 6C'_{0} - 5E^{0}_{k+1} - \alpha E^{0}_{k+1} \right) \cdot E^{0}_{k+1} = -1+\alpha\,,
			\end{equation}
			meaning that $E^{0}_{k+1}$ factorizes non-minimally, with $E^{0}_{k+1} \cdot f'_{0,i}[p_{k+1}] = 1$.
			
			\item $\{ f'_{0,i}[p_{k+1}], f'_{\infty,j}[\neg p_{k+1}], \smallbullet \}$: This case is analogous to Case~\ref{item:ngeq1C0-f0prime-f}.
			
			\item $\{ f'_{0,i}[p_{k+1}], E^{0}_{j}, \smallbullet \}$: This case is analogous to Case~\ref{item:ngeq1C0-f0prime-f}.
			
			\item $\{ f'_{0,i}[p_{k+1}], E^{\infty}_{j}, \smallbullet \}$: Since we have discarded all the other cases, the only possible remaining triplet completion is $\{ f'_{0,i}[p_{k+1}], E^{\infty}_{j}, E^{\infty}_{l} \}$, where $E^{\infty}_{l}$ may be associated to a blow-up of the same representative of $f$ as $E^{\infty}_{j}$, or a different one. The residual $\hat{G}$ divisor after tuning these three divisors to be non-minimal is
			\begin{equation}
				\Ghres = 12\tilde{C}_{0} + (12+6n)\tilde{f} - \sum_{i=1}^{n_{p}}E_{i} - 6E^{0}_{k+1} - 6f'_{0,i}[p_{k+1}] - 6E^{\infty}_{j} - 6E^{\infty}_{l}\,.
			\end{equation}
			Then, we have that
			\begin{itemize}
				\item if $E^{\infty}_{j}$ and $E^{\infty}_{l}$ stem from the blow-up of two different representatives of $f$, the factorization process is analogous to that of Case~\ref{item:ngeq1C0-f0prime-fprime}; while
				
				\item if $E^{\infty}_{j}$ and $E^{\infty}_{l}$ stem from the blow-up of the same representative of $f$, call its strict transform $f'_{0/\infty,m}$, we have instead
				\begin{equation}
					\left( \Ghres - \alpha f'_{0/\infty,m} \right) \cdot f'_{0/\infty,m} = -(6-\alpha)m\,,
				\end{equation}
				leading to a non-minimal factorization of $f'_{0/\infty,m}$, with $f'_{0/\infty,m} \cdot E^{\infty}_{j} = 1$ and $f'_{0/\infty,m} \cdot E^{\infty}_{l} = 1$.
			\end{itemize}
		\end{enumerate}
		\end{enumerate}
		
		\item $n \geq 1$ with $p \notin C_{0}$: According to \cref{lemma:triplet-candidates}, we need to consider the candidate triplets $\{ E^{\infty}_{k+1} , \smallbullet, \smallbullet \}$, $\{ C'_{\infty,i}, \smallbullet, \smallbullet \}$, $\{ f'_{\infty,i}[p_{k+1}] , \smallbullet, \smallbullet \}$ and $\{ f'_{0/\infty,i}[p_{k+1}] , \smallbullet, \smallbullet \}$.
		\begin{enumerate}[label=(\arabic{enumi}.\alph*)]
			\item $\{ E^{\infty}_{k+1} , \smallbullet, \smallbullet \}$: $E^{\infty}_{k+1}$ intersects the strict transforms of $\tilde{C}_{\infty}[p_{k+1}]$ and $f[p_{k+1}]$ represen\-tatives. We need to consider the following candidates for triplet completion: $C'_{0}$, $\tilde{C}_{\infty}$, $C'_{\infty,i}[\neg p_{k+1}]$, $\tilde{f}$, $f'_{\infty,i}[\neg p_{k+1}]$, $f'_{0,i}$, $f'_{0/\infty,i}[\neg p_{k+1}]$, $E^{0}_{i}$ and $E^{\infty}_{i}$.
			\begin{enumerate}[label=(\arabic{enumi}.\alph{enumii}.\roman*)]
				\item\label{item:ngeq1Cinf-E-C0} $\{ E^{\infty}_{k+1}, C'_{0}, \smallbullet \}$: After tuning these two divisors to be non-minimal, the residual $\hat{G}$ divisor is
				\begin{equation}
					\Ghres = 6\tilde{C}_{0} + (12+6n)\tilde{f} - \sum_{i \in \mathcal{N}_{\infty}} E_{i}^{\infty} - 12E^{\infty}_{k+1}\,.
				\end{equation}
				Let us denote the intersection point of the representative $f[p_{k+1}]$ of $f$ with $C_{0}$ as $f[p_{k+1}] \cap C_{0} =: p_{j}$. We need to distinguish two subcases:
				\begin{itemize}
					\item if $p_{j}$ has not been blown up, we have the intersection
					\begin{equation}
						\left( \Ghres - \alpha f'_{\infty,i}[p_{k+1}] \right) \cdot f'_{\infty,i}[p_{k+1}] = -(6-\alpha)i\,,
					\end{equation}
					where $i \geq 1$, and as a consequence $f'_{\infty,i}[p_{k+1}]$ factorizes non-minimally, with $f'_{\infty,i}[p_{k+1}] \cdot E^{\infty}_{k+1} = 1$;
					
					\item if $p_{j}$ has been blown up, we have instead
					\begin{equation}
						\left( \Ghres - \alpha f'_{0/\infty,i}[p_{k+1}] \right) \cdot f'_{0/\infty,i}[p_{k+1}] = 6-(6-\alpha)i\,,
					\end{equation}
					with $i \geq 2$, leading to the cascade of factorizations
					\begin{equation}
					\begin{aligned}
						\Ghres &\longrightarrow \Ghres - 3f'_{0/\infty,i}[p_{k+1}]\\
						&\longrightarrow \Ghres - 3f'_{0/\infty,i}[p_{k+1}] - 3E_{l}^{0}[p_{j}]\\
						&\longrightarrow \Ghres - 5f'_{0/\infty,i}[p_{k+1}] - 3E_{l}^{0}[p_{j}]\\
						&\longrightarrow \Ghres - 5f'_{0/\infty,i}[p_{k+1}] - 5E_{l}^{0}[p_{j}]\\
						&\longrightarrow \Ghres - 6f'_{0/\infty,i}[p_{k+1}] - 5E_{l}^{0}[p_{j}]\,,
					\end{aligned}
					\end{equation}
					such that $f'_{0/\infty,i}[p_{k+1}]$ factorizes non-minimally, with $f'_{0/\infty,i}[p_{k+1}] \cdot E_{k+1}^{\infty} = 1$.
				\end{itemize}
				After the various explicit examples given above, we have been more succinct here, only indicating the cascade of factorizations that occurs, without printing all the relevant intersection products. We will often do this in what follows.
				
				\item $\{ E^{\infty}_{k+1}, \tilde{C}_{\infty}, \smallbullet \}$: This case is analogous to Case~\ref{item:ngeq1Cinf-E-C0}.
				
				\item\label{item:ngeq1Cinf-E-Cinfprime} $\{ E^{\infty}_{k+1}, C'_{\infty,i}[\neg p_{k+1}], \smallbullet \}$: This case is analogous to Case~\ref{item:ngeq1Cinf-E-C0}.
				
				\item\label{item:ngeq1Cinf-E-f} $\{ E^{\infty}_{k+1}, \tilde{f}, \smallbullet \}$: The following candidate divisors are still possible triplet completions: $\tilde{f}$, $f'_{\infty,i}[\neg p_{k+1}]$, $f'_{0,i}$, $f'_{0/\infty,i}[\neg p_{k+1}]$, $E^{0}_{i}$ and $E^{\infty}_{i}$.
				\begin{itemize}
					\item $\{ E^{\infty}_{k+1}, \tilde{f}, \tilde{f} \}$ and $\{ E^{\infty}_{k+1}, \tilde{f}, f'_{\infty,i}[\neg p_{k+1}] \}$ lead to the intersection
					\begin{equation}
						\left( \Ghres - \alpha C'_{0} \right) \cdot C'_{0} = -(6-\alpha)(n+n_{0})\,,
					\end{equation}
					with $n \geq 1$. This means that $C'_{0}$ factorizes non-minimally, with $C'_{0} \cdot \tilde{f} = 1$ and $C'_{0} \cdot f'_{\infty,i}[\neg p_{k+1}] = 1$.
					
					\item $\{ E^{\infty}_{k+1}, \tilde{f}, f'_{0,i} \}$, $\{ E^{\infty}_{k+1}, \tilde{f}, f'_{0/\infty,i}[\neg p_{k+1}] \}$ and $\{ E^{\infty}_{k+1}, \tilde{f}, E_{i}^{0} \}$ all imply that $n \geq 1$ and $n_{0} \geq 1$.
					
					Assume first that the intersection point of the representative $f[p_{k+1}]$ of $f$ and $C_{0}$ was not blown-up, i.e.\ that we have $f'_{\infty,j}[p_{k+1}]$. The intersection
					\begin{equation}
						\left( \Ghres - \alpha C'_{0} \right) \cdot C'_{0} = 6 - (6-\alpha)(n+n_{0})
					\end{equation}
					leads to a factorization of at least $3C'_{0}$. Then,
					\begin{equation}
						\left( \Ghres - 3C'_{0} - \alpha f'_{\infty,j}[p_{k+1}] \right) \cdot f'_{\infty,j}[p_{k+1}] = 3 - (6-\alpha)j\,,
					\end{equation}
					with $j \geq 1$. This leads to a factorization of at least $3f'_{\infty,j}[p_{k+1}]$. This means, in turn, that
					\begin{equation}
						\left( \Ghres - 3 C'_{0} - 3 f'_{\infty,j}[p_{k+1}] - \alpha C'_{0} \right) \cdot C'_{0} = 3 - (3-\alpha)(n+n_{0})\,,
					\end{equation}
					forcing an additional factorization of at least $2C'_{0}$. Then
					\begin{equation}
						\left( \Ghres - 5C'_{0} - 3f'_{\infty,j}[p_{k+1}] - \alpha f'_{\infty,j}[p_{k+1}] \right) \cdot f'_{\infty,j}[p_{k+1}] = 1 - (3-\alpha)j\,,
					\end{equation}
					leading to a further factorization of at least $2f'_{\infty,j}[p_{k+1}]$. Finally, since
					\begin{equation}
						\left( \Ghres - 5 C'_{0} - 5 f'_{\infty,i}[p_{k+1}] - \alpha C'_{0} \right) \cdot C'_{0} = 1 - (1-\alpha)(n+n_{0})\,,
					\end{equation}
					$C'_{0}$ factorizes non-minimally, with $C'_{0} \cdot \tilde{f} = 1$.
					
					Assume instead that the intersection point of the representative $f[p_{k+1}]$ of $f$ and $C_{0}$ was blown-up, i.e.\ that we have $f'_{0/\infty,j}[p_{k+1}]$. In the case at hand, we observe that $n \geq 1$, $n_{0} \geq 2$ and $j \geq 2$. Then,
					\begin{equation}
					\begin{aligned}
						\Ghres &\longrightarrow \Ghres - 4C'_{0} - 3f'_{0/\infty,j}[p_{k+1}]\\
						&\longrightarrow \Ghres - 4C'_{0} - 3f'_{0/\infty,j}[p_{k+1}] - E^{0}_{l}[p_{k+1}]\\
						&\longrightarrow \Ghres - 5C'_{0} - 4f'_{0/\infty,j}[p_{k+1}] - E^{0}_{l}[p_{k+1}]\\
						&\longrightarrow \Ghres - 5C'_{0} - 4f'_{0/\infty,j}[p_{k+1}] - 3E^{0}_{l}[p_{k+1}]\\
						&\longrightarrow \Ghres - 6C'_{0} - 5f'_{0/\infty,j}[p_{k+1}] - 3E^{0}_{l}[p_{k+1}]\,,
					\end{aligned}
					\end{equation}
					with $C'_{0} \cdot \tilde{f} = 1$.
					
					\item $\{E^{\infty}_{k+1}, \tilde{f}, E^{\infty}_{i} \}$: This case works analogously to the previous one, but starting from $n \geq 2$ and $n_{0} \geq 0$.
				\end{itemize}
				
				\item $\{ E^{\infty}_{k+1}, f'_{\infty,i}[\neg p_{k+1}], \smallbullet \}$: This case is analogous to Case~\ref{item:ngeq1Cinf-E-f}.
				
				\item\label{item:ngeq1Cinf-E-f0prime} $\{ E^{\infty}_{k+1}, f'_{0,i}, \smallbullet \}$: The following divisors still need to be considered as valid triplet completions: $f'_{0,j}$, $f'_{0/\infty,j}[\neg p_{k+1}]$, $E^{0}_{i}$ and $E^{\infty}_{i}$.
				\begin{itemize}
					\item $\{ E^{\infty}_{k+1}, f'_{0,i}, f'_{0,j} \}$, $\{ E^{\infty}_{k+1}, f'_{0,i}, f'_{0/\infty,j}[\neg p_{k+1}] \}$ and $\{ E^{\infty}_{k+1}, f'_{0,i}, E^{0}_{j} \}$ all imply that $n \geq 1$ and $n_{0} \geq 2$.
					
					Assume first that the intersection point of the representative $f[p_{k+1}]$ of $f$ and $C_{0}$ was not blown-up, i.e.\ that we have $f'_{\infty,l}[p_{k+1}]$. Under this assumption, we have the cascade of factorizations
                    {\allowdisplaybreaks
					\begin{align}
						\Ghres &\longrightarrow \Ghres - 2C'_{0}\nonumber\\
						&\longrightarrow \Ghres - 2C'_{0} - 2f'_{\infty,l}[p_{k+1}]\nonumber\\
						&\longrightarrow \Ghres - 3C'_{0} - 2f'_{\infty,l}[p_{k+1}]\nonumber\\
						&\longrightarrow \Ghres - 3C'_{0} - 3f'_{\infty,l}[p_{k+1}]\nonumber\\
						&\longrightarrow \Ghres - 3C'_{0} - 3f'_{\infty,l}[p_{k+1}] - 3E^{0}_{m}[C'_{0},f'_{0,i}]\nonumber\\
						&\longrightarrow \Ghres - 4C'_{0} - 3f'_{\infty,l}[p_{k+1}] - 3E^{0}_{m}[C'_{0},f'_{0,i}]\\
						&\longrightarrow \Ghres - 4C'_{0} - 4f'_{\infty,l}[p_{k+1}] - 4E^{0}_{m}[C'_{0},f'_{0,i}]\nonumber\\
						&\longrightarrow \Ghres - 5C'_{0} - 4f'_{\infty,l}[p_{k+1}] - 4E^{0}_{m}[C'_{0},f'_{0,i}]\nonumber\\
						&\longrightarrow \Ghres - 5C'_{0} - 5f'_{\infty,l}[p_{k+1}] - 5E^{0}_{m}[C'_{0},f'_{0,i}]\nonumber\\
						&\longrightarrow \Ghres - 6C'_{0} - 5f'_{\infty,l}[p_{k+1}] - 5E^{0}_{m}[C'_{0},f'_{0,i}]\nonumber\\
						&\longrightarrow \Ghres - 6C'_{0} - 6f'_{\infty,l}[p_{k+1}] - 6E^{0}_{m}[C'_{0},f'_{0,i}]\,,\nonumber
					\end{align}
                    }
					with $f'_{\infty,l}[p_{k+1}] \cdot E^{\infty}_{k+1} = 1$ and $E^{0}_{m}[C'_{0},f'_{0,i}] \cdot f'_{0,i} = 1$.
					
					Assume first that the intersection point of the representative $f[p_{k+1}]$ of $f$ and $C_{0}$ was blown-up, i.e.\ that we have $f'_{0/\infty,l}[p_{k+1}]$. This implies that $n \geq 1$, $n_{0} \geq 3$ and $l \geq 2$. This leads to the cascade of factorizations
					\begin{equation}
					\begin{aligned}
						\Ghres &\longrightarrow \Ghres - 3C'_{0} - 3f'_{0/\infty}[p_{k+1}]\\
						&\longrightarrow \Ghres - 3C'_{0} - 3f'_{0/\infty,l}[p_{k+1}] - 3E^{0}_{m}[C'_{0},f'_{0,i}]\\
						&\longrightarrow \Ghres - 4C'_{0} - 3f'_{0/\infty,l}[p_{k+1}] - 3E^{0}_{m}[C'_{0},f'_{0,i}]\\
						&\longrightarrow \Ghres - 4C'_{0} - 3f'_{0/\infty,l}[p_{k+1}] - 4E^{0}_{m}[C'_{0},f'_{0,i}]\\
						&\longrightarrow \Ghres - 4C'_{0} - 3f'_{0/\infty,l}[p_{k+1}] - 4E^{0}_{m}[C'_{0},f'_{0,i}] - E^{0}_{n}[C'_{0},f'_{0/\infty,l}[p_{k+1}]]\\
						&\longrightarrow \Ghres - 5C'_{0} - 4f'_{0/\infty,l}[p_{k+1}] - 4E^{0}_{m}[C'_{0},f'_{0,i}] - E^{0}_{n}[C'_{0},f'_{0/\infty,l}[p_{k+1}]]\\
						&\longrightarrow \Ghres - 5C'_{0} - 4f'_{0/\infty,l}[p_{k+1}] - 5E^{0}_{m}[C'_{0},f'_{0,i}] - 3E^{0}_{n}[C'_{0},f'_{0/\infty,l}[p_{k+1}]]\\
						&\longrightarrow \Ghres - 5C'_{0} - 5f'_{0/\infty,l}[p_{k+1}] - 5E^{0}_{m}[C'_{0},f'_{0,i}] - 3E^{0}_{n}[C'_{0},f'_{0/\infty,l}[p_{k+1}]]\\
						&\longrightarrow \Ghres - 5C'_{0} - 5f'_{0/\infty,l}[p_{k+1}] - 5E^{0}_{m}[C'_{0},f'_{0,i}] - 4E^{0}_{n}[C'_{0},f'_{0/\infty,l}[p_{k+1}]]\\
						&\longrightarrow \Ghres - 6C'_{0} - 6f'_{0/\infty,l}[p_{k+1}] - 5E^{0}_{m}[C'_{0},f'_{0,i}] - 4E^{0}_{n}[C'_{0},f'_{0/\infty,l}[p_{k+1}]]\,,
					\end{aligned}
					\end{equation}
					with $f'_{0/\infty,l}[p_{k+1}] \cdot E^{\infty}_{k+1} = 1$.
				\end{itemize}
				
				\item $\{ E^{\infty}_{k+1}, f'_{0/\infty,j}[\neg p_{k+1}], \smallbullet \}$: This case is analogous to Case~\ref{item:ngeq1Cinf-E-f0prime}.
				
				\item $\{ E^{\infty}_{k+1}, E^{0}_{i}, \smallbullet \}$: Tuning these two divisors to be non-minimal leaves us with the residual $\hat{G}$ divisor
				\begin{equation}
					\Ghres = 12\tilde{C}_{0} + (12+6n)\tilde{f} - \sum_{i=1}^{n_{p}}E_{i} - 12E^{\infty}_{k+1} - 6E^{0}_{i}\,,
				\end{equation}
				which yields the intersection
				\begin{equation}
					\left( \Ghres - \alpha C'_{0} \right) \cdot C'_{0} = -(6-\alpha)(n + n_{0})\,.
				\end{equation}
				This leads to a non-minimal factorization of $C'_{0}$, with $C'_{0} \cdot E^{0}_{i} = 1$
				
				\item $\{ E^{\infty}_{k+1}, E^{\infty}_{i}, \smallbullet \}$: The only candidate for triplet completion left is a third $E^{\infty}_{j}$. Tuning $\{ E^{\infty}_{k+1}, E^{\infty}_{i}, E^{\infty}_{j} \}$ gives the residual $\hat{G}$ divisor
				\begin{equation}
					\Ghres = 12\tilde{C}_{0} + (12+6n)\tilde{f} - \sum_{i=1}^{n_{p}}E_{i} - 12E^{\infty}_{k+1} - 6E^{\infty}_{i} - 6E^{\infty}_{j}\,.
				\end{equation}
				We now need to distinguish two subcases depending on the existing relation among these exceptional divisors.
				\begin{itemize}
					\item If at least two of $E^{\infty}_{k+1}$, $E^{\infty}_{i}$ and $E^{\infty}_{j}$ stem from blow-ups of the same repre\-sentative of $f$ (say $f[p_{k+1},p_{i}]$, for concreteness), we have
					\begin{equation}
						\left( \Ghres - \alpha f[p_{k+1},p_{i}] \right) \cdot f[p_{k+1},p_{i}] \leq -(6-\alpha)j\,,
					\end{equation}
					yielding a non-minimal factorization of $f[p_{k+1},p_{i}]$, with $f[p_{k+1},p_{i}] \cdot E_{k+1} = 1$.
					
					\item If $E^{\infty}_{k+1}$, $E^{\infty}_{i}$ and $E^{\infty}_{j}$ each stem from the blow-up of a different representative of $f$, we need to have $n \geq 4$ to satisfy the effectiveness bounds. We now need to distinguish the cases in which the intersection points of said representatives of $f$ with $C_{0}$ have not been blown-up, and those in which they have. We consider the two extreme cases, in which none or all of them have been blown-up, with the intermediate cases leading to hybrids of the factorization cascades presented below.
					
					Assume first that none of the intersection points of the three relevant representatives of $f$ with $C_{0}$ have been blown-up. We then have
					\begin{equation}
					\begin{aligned}
						\Ghres &\longrightarrow \Ghres - 3C'_{0}\\
						&\longrightarrow \Ghres - 3C'_{0} - 3f'_{\infty,l}[p_{k+1}] - 3f'_{\infty,m}[p_{i}] - 3f'_{\infty,n}[p_{j}]\\
						&\longrightarrow \Ghres - 6C'_{0} - 3f'_{\infty,l}[p_{k+1}] - 3f'_{\infty,m}[p_{i}] - 3f'_{\infty,n}[p_{j}]\\
		                  &\longrightarrow \Ghres - 6C'_{0} - 6f'_{\infty,l}[p_{k+1}] - 6f'_{\infty,m}[p_{i}] - 6f'_{\infty,n}[p_{j}]\,,
					\end{aligned}
					\end{equation}
					with $f'_{\infty,l}[p_{k+1}] \cdot E^{\infty}_{k+1} = 1$, etc.
					
					Assume now that the intersection points of the three relevant representatives of $f$ with $C_{0}$ have been blown-up. This not only entails $n \geq 4$, but also $n_{0} \geq 3$. We then have
					{\allowdisplaybreaks
					\begin{align}
						\Ghres &\longrightarrow \Ghres - 5C'_{0} - 3f'_{0/\infty,l}[p_{k+1}] - 3f'_{0/\infty,m}[p_{i}] - 3f'_{0/\infty,n}[p_{j}]\nonumber\\
						&\longrightarrow \Ghres - 5C'_{0} - 3f'_{0/\infty,l}[p_{k+1}] - 3f'_{0/\infty,m}[p_{i}] - 3f'_{0/\infty,n}[p_{j}]\nonumber\\*
						&\phantom{\longrightarrow} - 2E^{0}_{k+1} - 2E^{0}_{i} - 2E^{0}_{j}\nonumber\\
						&\longrightarrow \Ghres - 7C'_{0} - 4f'_{0/\infty,l}[p_{k+1}] - 4f'_{0/\infty,m}[p_{i}] - 4f'_{0/\infty,n}[p_{j}]\nonumber\\*
						&\phantom{\longrightarrow} - 2E^{0}_{k+1} - 2E^{0}_{i} - 2E^{0}_{j}\\
						&\longrightarrow \Ghres - 7C'_{0} - 4f'_{0/\infty,l}[p_{k+1}] - 4f'_{0/\infty,m}[p_{i}] - 4f'_{0/\infty,n}[p_{j}]\nonumber\\*
						&\phantom{\longrightarrow} - 5E^{0}_{k+1} - 5E^{0}_{i} - 5E^{0}_{j}\nonumber\\
						&\longrightarrow \Ghres - 9C'_{0} - 7f'_{0/\infty,l}[p_{k+1}] - 7f'_{0/\infty,m}[p_{i}] - 7f'_{0/\infty,n}[p_{j}]\nonumber\\*
						&\phantom{\longrightarrow} - 5E^{0}_{k+1} - 5E^{0}_{i} - 5E^{0}_{j}\,,\nonumber
					\end{align}
					}
					with $f'_{0/\infty,l}[p_{k+1}] \cdot E^{\infty}_{k+1} = 1$, etc.
				\end{itemize}
			\end{enumerate}
			
			\item $\{ C'_{\infty,i}[p_{k+1}] , \smallbullet, \smallbullet \}$: $C'_{\infty,i}[p_{k+1}]$ intersects the strict transforms $\tilde{C}_{\infty}$, $\tilde{C}_{\infty}[\neg p\{C'_{\infty,i}[p_{k+1}]\}]$, $\tilde{f}$, $f'_{\infty,j}[\neg p\{C'_{\infty,i}[p_{k+1}]\}]$, $f'_{0,j}[\neg p\{C'_{\infty,i}[p_{k+1}]\}]$, $f'_{0/\infty,j}[\neg p\{C'_{\infty,i}[p_{k+1}]\}]$, and the exceptional divisors $E^{\infty}_{j}[C'_{\infty,i}[p_{k+1}]]$. We need to consider as candidates for triplet completion: $C'_{0}$, $f'_{\infty,j}[p\{C'_{\infty,i}[p_{k+1}]\}]$, $f'_{0/\infty,j}[p\{C'_{\infty,i}[p_{k+1}]\}]$, $E^{0}_{j}$ and $E^{\infty}_{j}[\neg C'_{\infty,i}[p_{k+1}]]$.
			\begin{enumerate}[label=(\arabic{enumi}.\alph{enumii}.\roman*)]
				\item $\{ C'_{\infty,i}[p_{k+1}], C'_{0}, \smallbullet \}$: The only possible triplet completions not intersecting the curves $C'_{\infty,i}[p_{k+1}]$ or $C'_{0}$ are $f'_{0/\infty,j}[p_{k+1}]$ and $E^{\infty}_{j}[\neg C'_{\infty,i}[p_{k+1}]]$.
				\begin{itemize}
					\item Tuning $\{ C'_{\infty,i}[p_{k+1}], C'_{0}, f'_{0/\infty,j}[p\{C'_{\infty,i}[p_{k+1}]\}] \}$ leads to
					\begin{equation}
						\left( \Ghres - \alpha E^{0}_{l}[f'_{0/\infty,j}[p\{C'_{\infty,i}[p_{k+1}]\}]] \right) \cdot E^{0}_{l}[f'_{0/\infty,j}[p\{C'_{\infty,i}[p_{k+1}]\}]] = -6+\alpha\,,
					\end{equation}
					meaning that $E^{0}_{l}[f'_{0/\infty,j}[p\{C'_{\infty,i}[p_{k+1}]\}]]$ factorizes non-minimally with the intersection $E^{0}_{l}[f'_{0/\infty,j}[p\{C'_{\infty,i}[p_{k+1}]\}]] \cdot C'_{0} = 1$.
					
					\item Tuning $\{ C'_{\infty,i}[p_{k+1}], C'_{0}, E^{\infty}_{j}[\neg C'_{\infty,i}[p_{k+1}]] \}$ is analogous to Case~\ref{item:ngeq1Cinf-E-C0}.
				\end{itemize}
				
				\item\label{item:ngeq1Cinf-Cinf-finfprime} $\{ C'_{\infty,i}[p_{k+1}], f'_{\infty,j}[p_{k+1}], \smallbullet \}$: This leads to the intersection
				\begin{equation}
					\left( \Ghres - \alpha E^{\infty}_{k+1} \right) \cdot E^{\infty}_{k+1} = -6+\alpha\,,
				\end{equation}
				making $E^{\infty}_{k+1}$ factorize non-minimally, with $E^{\infty}_{k+1} \cdot f'_{\infty,j}[p_{k+1}] = 1$.
				
				\item $\{ C'_{\infty,i}[p_{k+1}], f'_{0/\infty,j}[p_{k+1}], \smallbullet \}$: This case is analogous to Case~\ref{item:ngeq1Cinf-Cinf-finfprime}.
				
				\item $\{ C'_{\infty,i}[p_{k+1}], E^{0}_{j}, \smallbullet \}$: Consider the representative $f[p_{j}]$ of $f$ passing through the blow-up centre from which $E^{0}_{j}$ stems. Depending on if the intersection point of $f[p_{j}]$ with $C'_{\infty,i}[p_{k+1}]$ has been blown-up or not, we distinguish two cases.
				
				Assume first that it has not been blown-up. We then have
				\begin{equation}
					\left( \Ghres - \alpha f'_{0,l}[p_{j}] \right) \cdot f'_{0,l}[p_{j}] = -(6-\alpha)l
				\end{equation}
				with $l \geq 1$, leading to $f'_{0,l}[p_{j}]$ factorizing non-minimally with $f'_{0,l}[p_{j}] \cdot E^{0}_{j} = 1$.
				
				Assume now that it has been blown-up. Then, the cascade of factorizations
				\begin{equation}
				\begin{aligned}
					\Ghres &\longrightarrow \Ghres - 3f'_{0/\infty,l}[p_{j}]\\
					&\longrightarrow \Ghres - 3f'_{0/\infty,l}[p_{j}] - 3E^{\infty}_{m}[f'_{0/\infty,l}[p_{j}]]\\
					&\longrightarrow \Ghres - 5f'_{0/\infty,l}[p_{j}] - 3E^{\infty}_{m}[f'_{0/\infty,l}[p_{j}]]\\
					&\longrightarrow \Ghres - 5f'_{0/\infty,l}[p_{j}] - 5E^{\infty}_{m}[f'_{0/\infty,l}[p_{j}]]\\
					&\longrightarrow \Ghres - 6f'_{0/\infty,l}[p_{j}] - 5E^{\infty}_{m}[f'_{0/\infty,l}[p_{j}]]\,
				\end{aligned}
				\end{equation}
				makes $f'_{0/\infty,l}[p_{j}]$ factorize non-minimally, with $f'_{0/\infty,l}[p_{j}] \cdot E^{0}_{j} = 1$.
				
				\item $\{ C'_{\infty,i}[p_{k+1}], E^{\infty}_{i}[\neg C'_{\infty,i}[p_{k+1}]], \smallbullet \}$: This case is analogous to Case~\ref{item:ngeq1Cinf-E-Cinfprime}.
			\end{enumerate}
			
			\item $\{ f'_{\infty,i}[p_{k+1}] , \smallbullet, \smallbullet \}$: $f'_{\infty,i}[p_{k+1}]$ intersects $C'_{0}$, $\tilde{C}_{\infty}$, $C'_{\infty,j}[\neg p\{ f'_{\infty,i}[p_{k+1}] \}]$ and the exceptional divisor $E^{\infty}_{k+1}$. We need to study the factorizations forced by the following triplet completion candidates: $C'_{\infty,j}[p\{ f'_{\infty,i}[p_{k+1}] \}]$, $\tilde{f}$, $f'_{0,j}[\neg p_{k+1}]$, $f'_{\infty,j}[\neg p_{k+1}]$, $f'_{0/\infty,j}[\neg p_{k+1}]$, $E^{0}_{j}$ and $E^{\infty}_{j}[\neg p\{ f'_{\infty,i}[p_{k+1}] \}]$.
			\begin{enumerate}[label=(\arabic{enumi}.\alph{enumii}.\roman*)]
				\item\label{item:ngeq1Cinf-fprime-Cinfprime} $\{ f'_{\infty,i}[p_{k+1}] , C'_{\infty,j}[p\{ f'_{\infty,i}[p_{k+1}] \}], \smallbullet \}$: This leads to the intersection
				\begin{equation}
					\left( \Ghres - \alpha E^{\infty}_{l}[f'_{\infty,i}, C'_{\infty,j}] \right) \cdot E^{\infty}_{l}[f'_{\infty,i}, C'_{\infty,j}] = -6 + \alpha\,,
				\end{equation}
				meaning that $E^{\infty}_{l}[f'_{\infty,i}, C'_{\infty,j}]$ factorizes non-minimally, while having the intersections $E^{\infty}_{l}[f'_{\infty,i}, C'_{\infty,j}] \cdot f'_{\infty,i}[p_{k+1}] = 1$ and $E^{\infty}_{l}[f'_{\infty,i}, C'_{\infty,j}] \cdot C'_{\infty,j}[p\{ f'_{\infty,i}[p_{k+1}] \}] = 1$.
				
				\item\label{item:ngeq1Cinf-fprime-f} $\{ f'_{\infty,i}[p_{k+1}] , \tilde{f}, \smallbullet \}$: Due to the intersection
				\begin{equation}
					\left( \Ghres - \alpha C'_{0} \right) \cdot C'_{0} = -(6-\alpha)(n+n_{0})\,,
				\end{equation}
				we have a non-minimal factorization of $C'_{0}$, with $C'_{0} \cdot f'_{\infty,i}[p_{k+1}] = 1$.
				
				\item\label{item:ngeq1Cinf-fprime-fprime} $\{ f'_{\infty,i}[p_{k+1}] , f'_{0,j}[\neg p_{k+1}], \smallbullet \}$: This tuning leads to the cascade of factorizations
				\begin{equation}
				\begin{aligned}
					\Ghres &\longrightarrow \Ghres - 3C'_{0}\\
					&\longrightarrow \Ghres - 3C'_{0} - 3E^{0}_{l}[f'_{0,j}]\\
					&\longrightarrow \Ghres - 5C'_{0} - 3E^{0}_{l}[f'_{0,j}]\\
					&\longrightarrow \Ghres - 5C'_{0} - 5E^{0}_{l}[f'_{0,j}]\\
					&\longrightarrow \Ghres - 6C'_{0} - 5E^{0}_{l}[f'_{0,j}]\,,
				\end{aligned}
				\end{equation}
				with $C'_{0}$ factorizing non-minimally and $C'_{0} \cdot f'_{\infty,i}[p_{k+1}] = 1$.
				
				\item $\{ f'_{\infty,i}[p_{k+1}] , f'_{\infty,j}[\neg p_{k+1}], \smallbullet \}$: This case is analogous to Case~\ref{item:ngeq1Cinf-fprime-f}.
				
				\item $\{ f'_{\infty,i}[p_{k+1}] , f'_{0/\infty,j}[\neg p_{k+1}], \smallbullet \}$: This case is analogous to Case~\ref{item:ngeq1Cinf-fprime-fprime}.
				
				\item $\{ f'_{\infty,i}[p_{k+1}] , E^{0}_{i}, \smallbullet \}$: This leads to the intersection
				\begin{equation}
					\left( \Ghres - \alpha C'_{0} \right) \cdot C'_{0} = -(6-\alpha)(n+n_{0})\,,
				\end{equation}
				from which we see that $C'_{0}$ factorizes non-minimally, with $C'_{0} \cdot E^{0}_{i} = 1$.
				
				\item $\{ f'_{\infty,i}[p_{k+1}], E^{\infty}_{j}[\neg p\{ f'_{\infty,i}[p_{k+1}] \}], \smallbullet \}$: This case is analogous to Case~\ref{item:ngeq1Cinf-E-f}.
			\end{enumerate}
			
			\item $\{ f'_{0/\infty,i}[p_{k+1}] , \smallbullet, \smallbullet \}$: $f'_{0/\infty,i}[p_{k+1}]$ intersects $\tilde{C}_{\infty}$, $C'_{\infty,j}[\neg p\{f'_{0/\infty,i}[p_{k+1}]\}]$ and the exceptional divisor $E^{\infty}_{k+1}$. We need to study the factorizations forced the following candidates for triplet completion: $C'_{0}$, $C'_{\infty,j}[p\{f'_{0/\infty,i}[p_{k+1}]\}]$, $\tilde{f}$, $f'_{0,j}[\neg p_{k+1}]$, $f'_{\infty,j}[\neg p_{k+1}]$, $f'_{0/\infty,j}[\neg p_{k+1}]$, $E^{0}_{j}[\neg f'_{0/\infty,i}[p_{k+1}]]$ and $E^{\infty}_{j}[\neg p\{f'_{0/\infty,i}[p_{k+1}]\}]$.
			\begin{enumerate}[label=(\arabic{enumi}.\alph{enumii}.\roman*)]
				\item $\{ f'_{0/\infty,i}[p_{k+1}] , C'_{0}, \smallbullet \}$: Let us denote the intersection points of the pushforwards $f[p_{k+1}] \cap C_{0} =: p_{j}$. Tuning these divisors to be non-minimal leads
				\begin{equation}
					\left( \Ghres - \alpha E^{0}_{l}[p_{j}] \right) \cdot E^{0}_{l}[p_{j}] = -6 + \alpha\,,
				\end{equation}
				from which we see that $E^{0}_{l}[p_{j}]$ factorizes non-minimally, with the intersections $E^{0}_{l}[p_{j}] \cdot f'_{0/\infty,i}[p_{k+1}] = 1$ and $E^{0}_{l}[p_{j}] \cdot C'_{0} = 1$.
				
				\item $\{ f'_{0/\infty,i}[p_{k+1}] , C'_{\infty,j}[p\{f'_{0/\infty,i}[p_{k+1}]\}], \smallbullet \}$: This case is analogous to Case~\ref{item:ngeq1Cinf-fprime-Cinfprime}.
				
				\item\label{item:ngeq1Cinf-f0infprime-f} $\{ f'_{0/\infty,i}[p_{k+1}] , \tilde{f}, \smallbullet \}$: This leads to the cascade of factorizations
				\begin{equation}
				\begin{aligned}
					\Ghres &\longrightarrow \Ghres - 3C'_{0}\\
					&\longrightarrow \Ghres - 3C'_{0} - 3E^{0}_{j}[f'_{0/\infty,i}[p_{k+1}]]\\
					&\longrightarrow \Ghres - 5C'_{0} - 3E^{0}_{j}[f'_{0/\infty,i}[p_{k+1}]]\\
					&\longrightarrow \Ghres - 5C'_{0} - 5E^{0}_{j}[f'_{0/\infty,i}[p_{k+1}]]\\
					&\longrightarrow \Ghres - 6C'_{0} - 5E^{0}_{j}[f'_{0/\infty,i}[p_{k+1}]]\\
					&\longrightarrow \Ghres - 6C'_{0} - 6E^{0}_{j}[f'_{0/\infty,i}[p_{k+1}]]\,,
				\end{aligned}
				\end{equation}
				such that $E^{0}_{j}[f'_{0/\infty,i}[p_{k+1}]]$ factorizes non-minimally with the intersection product $E^{0}_{j}[f'_{0/\infty,i}[p_{k+1}]] \cdot f'_{0/\infty,i}[p_{k+1}] = 1$.
				
				\item\label{item:ngeq1Cinf-f0infprime-f0prime} $\{ f'_{0/\infty,i}[p_{k+1}], f'_{0,j}[\neg p_{k+1}], \smallbullet \}$: This leads to the cascade of factorizations
				\begin{equation}
				\begin{aligned}
					\Ghres &\longrightarrow \Ghres - 2C'_{0}\\
					&\longrightarrow \Ghres - 2C'_{0} - 2E^{0}_{l}[f'_{0/\infty,i}[p_{k+1}]] - 2E^{0}_{m}[f'_{0,j}[\neg p_{k+1}]]\\
					&\longrightarrow \Ghres - 4C'_{0} - 2E^{0}_{l}[f'_{0/\infty,i}[p_{k+1}]] - 2E^{0}_{m}[f'_{0,j}[\neg p_{k+1}]]\\
					&\longrightarrow \Ghres - 4C'_{0} - 4E^{0}_{l}[f'_{0/\infty,i}[p_{k+1}]] - 4E^{0}_{m}[f'_{0,j}[\neg p_{k+1}]]\\
					&\longrightarrow \Ghres - 5C'_{0} - 4E^{0}_{l}[f'_{0/\infty,i}[p_{k+1}]] - 4E^{0}_{m}[f'_{0,j}[\neg p_{k+1}]]\\
					&\longrightarrow \Ghres - 5C'_{0} - 5E^{0}_{l}[f'_{0/\infty,i}[p_{k+1}]] - 5E^{0}_{m}[f'_{0,j}[\neg p_{k+1}]]\\
					&\longrightarrow \Ghres - 6C'_{0} - 5E^{0}_{l}[f'_{0/\infty,i}[p_{k+1}]] - 5E^{0}_{m}[f'_{0,j}[\neg p_{k+1}]]\\
					&\longrightarrow \Ghres - 6C'_{0} - 6E^{0}_{l}[f'_{0/\infty,i}[p_{k+1}]] - 6E^{0}_{m}[f'_{0,j}[\neg p_{k+1}]]\,,
				\end{aligned}
				\end{equation}
				from where we see that $E^{0}_{l}[f'_{0/\infty,i}[p_{k+1}]]$ and $E^{0}_{m}[f'_{0,j}[\neg p_{k+1}]]$ non-minimally factorize, with $E^{0}_{l}[f'_{0/\infty,i}[p_{k+1}]] \cdot f'_{0/\infty,i}[p_{k+1}] = 1$ and $E^{0}_{m}[f'_{0,j}[\neg p_{k+1}]] \cdot f'_{0,j}[\neg p_{k+1}] = 1$.
				
				\item $\{ f'_{0/\infty,i}[p_{k+1}], f'_{\infty,j}[\neg p_{k+1}], \smallbullet \}$: This case is analogous to Case~\ref{item:ngeq1Cinf-f0infprime-f}.
				
				\item $\{ f'_{0/\infty,i}[p_{k+1}], f'_{0/\infty,j}[\neg p_{k+1}], \smallbullet \}$: This case is analogous to Case~\ref{item:ngeq1Cinf-f0infprime-f0prime}.
				
				\item\label{item:ngeq1Cinf-f0infprime-E0} $\{ f'_{0/\infty,i}[p_{k+1}], E^{0}_{j}[\neg f'_{0/\infty,i}[p_{k+1}]], \smallbullet \}$: Tuning these divisors to be non-minimal leads to the cascade of factorizations
				\begin{equation}
				\begin{aligned}
					\Ghres &\longrightarrow \Ghres - 4C'_{0}\\
					&\longrightarrow \Ghres - 4C'_{0} - 4f'_{0(/\infty),l}[E^{0}_{j}[\neg f'_{0/\infty,i}[p_{k+1}]]]\\
					&\longrightarrow \Ghres - 5C'_{0} - 4f'_{0(/\infty),l}[E^{0}_{j}[\neg f'_{0/\infty,i}[p_{k+1}]]]\\
					&\longrightarrow \Ghres - 5C'_{0} - 5f'_{0(/\infty),l}[E^{0}_{j}[\neg f'_{0/\infty,i}[p_{k+1}]]]\\
					&\longrightarrow \Ghres - 6C'_{0} - 5f'_{0(/\infty),l}[E^{0}_{j}[\neg f'_{0/\infty,i}[p_{k+1}]]]\,,
				\end{aligned}
				\end{equation}
				from where we see that $C'_{0}$ factorizes non-minimally, with the intersection product $C'_{0} \cdot E^{0}_{j}[\neg f'_{0/\infty,i}[p_{k+1}]] = 1$.
				
				\item $\{ f'_{0/\infty,i}[p_{k+1}], E^{\infty}_{j}[\neg p\{f'_{0/\infty,i}[p_{k+1}]\}], \smallbullet \}$: This case is analogous to Case~\ref{item:ngeq1Cinf-f0infprime-E0}. 
			\end{enumerate}
		\end{enumerate}
	\end{enumerate}
\end{proof}

\subsection{Models constructed over the remaining \texorpdfstring{$\mathrm{Bl}(\mathbb{F}_{n})$}{Bl(Fn)}}

With \cref{prop:star-degenerations-P2-Fn} and \cref{prop:star-degenerations-type-A} at hand, we have established the result of \cref{prop:restricting-star-degenerations} for Calabi-Yau Weierstrass models over $\hat{B}_{0} = \mathbb{P}^{2}$, $\hat{B}_{0} = \mathbb{F}_{n}$, and their type (A) blow-ups. Proceeding in the same fashion, one can extend the result to the $\hat{B}_{0} = \mathrm{Bl}(\mathbb{F}_{n})$ surfaces obtained by also including type (B), (C) and (D) blow-ups. Since this is not a very enlightening discussion, we choose instead to comment on some features of these additional cases and give two detailed examples showing that the prospects of finding a triplet of curves violating \cref{prop:restricting-star-degenerations} are not improved by performing these types of blow-ups.

\subsubsection{Type (B) blow-ups}

Let us start by considering also type (B) blow-ups. i.e.\ the ones performed over generic points of exceptional divisors, away from the intersections of these with other divisors.

We recall from \cref{sec:anticanonical-class-after-arbitrary-blow-up}, that after performing type (A) and (B) blow-ups the anticanonical class of the surface is given by \eqref{eq:anticanonical-class-type-B-blow-up}, which we can write
\begin{equation}
	\overline{K}_{\hat{B}} = 2\tilde{C}_{0} + (2+n)\tilde{f} - \sum_{\alpha} d_{\alpha} E_{\alpha}\,,\qquad d_{\alpha} = \text{level of the exceptional divisor} \geq 1\,.
\end{equation}
The discrepancies of the exceptional divisors grow the deeper they are located within a chain of type (B) blow-ups. However, in the analogues of the linear equivalences \eqref{eq:linear-equivalence-blown-up-C0} and \eqref{eq:linear-equivalence-blown-up-f} all exceptional divisors appear with coefficient one, given that type (B) divisors only affect a single exceptional divisor at a time. As a consequence, the effectiveness bounds become more stringent as we perform further type (B) blow-ups.

To give a first example of this, consider the surfaces $\hat{B}_{0} = \mathrm{Bl}(\mathbb{F}_{0})$ in which we have performed a single type (A) blow-up, giving rise to an exceptional divisor on which we then perform type~(B) blow-ups. It is then clear that for the anticanonical class $\overline{K}_{\hat{B}_{0}}$ to be effective we need
\begin{equation}
	\max_{\alpha}(d_{\alpha}) \leq 4\,.
\end{equation}
Both this bound and the ones we will discuss below are not modified by the number of points blown up in a given exceptional divisor; the only relevant aspect is the mismatch between the discrepancies and the appearance of the exceptional divisors in the strict transforms of the $\tilde{C}_{0}$ and $\tilde{f}$ representatives. If we have a second chain of type (B) blow-ups on the exceptional divisor of a type (A) blow-up of a generic point in the original surface, we obtain
\begin{equation}
	\max_{\alpha_{1}} \left( d_{\alpha_{1}}^{1} \right) + \max_{\alpha_{2}} \left( d_{\alpha_{2}}^{2} \right) \leq 4\,.
\end{equation}
Since we are assuming these to be chains of type (B) blow-ups we have $d_{\alpha}^{i} \geq 2$ for $i = 1,2$, concluding that two type (B) chains in general position are the best we can do. The situation can be improved by performing the type (B) blow-ups over exceptional divisor arising from type (A) blow-ups with centre at points on the same representatives of $C_{0}$ of $f$. In this way, we can achieve at most four type (B) chains going up to level two.

Moving to the surfaces $\hat{B}_{0} = \mathrm{Bl}(\mathbb{F}_{n})$, with $1 \leq n \leq 12$, we can argue in the same fashion that
\begin{equation}
	\sum_{i=1}^{n_{0}} \left( \max_{\alpha_{i}} d_{\alpha_{i}}^{i} - 2 \right) \theta \left( \max_{\alpha_{i}} d_{\alpha_{i}}^{i} - 2 \right) + \sum_{i=1}^{n_{\infty}} \max_{\alpha_{i}} d_{\alpha_{i}}^{i} \leq n+2
\label{eq:effectiveness-bound-type-B}
\end{equation}
for $\overline{K}_{\hat{B}_{0}}$ to be effective. The first term is shifted to account for the fact that $C_{0}$ has a unique representative, whose total transform $\tilde{C}_{0}$ can account for discrepancies of up to two units. The Heaviside function prevents this term from giving negative contributions.

The bounds given above ensure that $\overline{K}_{\hat{B}_{0}}$ is effective after the blow-up, but tuning an exceptional divisor $E_{i}$ to be non-minimal requires $E_{i} \leq \overline{K}_{\hat{B}_{0}}$, according to \cref{prop:K-C-effectiveness}. Hence, the l.h.s.\ of \eqref{eq:effectiveness-bound-type-B} must be increased by one unit if we want to tune any exceptional divisor related to a type (A) blow-up at $p \in C_{0}$, and one unit for each exceptional divisor related to a type (A) blow-up at $p \notin C_{0}$ on distinct representatives of $f$, cf.\ \eqref{eq:effectiveness-bound-type-A-shifted}.

Carefully taking these bounds into account is relevant in order to discard some candidate triplets when generalizing the results of \cref{prop:star-degenerations-type-A} to include type (B) blow-ups. In this regard, we can also note that, due to \cref{lemma:triplet-candidates}, we only need to take into account at each step candidate triplets $\{ E_{p+1}, \smallbullet, \smallbullet \}$, where $E_{p+1}$ is the exceptional divisor arising from the last type (B) blow-up. The candidate triplets including the strict transform of the exceptional divisor $E_{p}$ that $E_{p+1}$ is stemming from behave in the same way as they did in the blown down surface and do not need to be checked again. One can see that tuning $E_{p+1}$ to be non-minimal can trigger a non-minimal factorization of $E_{p}$ if the level of $E_{p+1}$ is too high. This does not occur in the example that we now analyse, but a factorization of various copies of $E_{p}$ can already be observed for a rather short sequence of type (B) blow-ups.

\begin{example}
\label{example:restriction-star-degenerations-example-1}
Consider a Hirzebruch surface $\mathbb{F}_{n}$ in which three type (A) blow-ups have been performed, $n_{0} = 1$ at a point $p \in C_{0}$ with exceptional divisor $E_{1}$, and $n_{\infty} = \ninftot = 2$ at two points $p \notin C_{0}$ in distinct representatives of the fiber class $f$ with exceptional divisors $E_{2}$ and $E_{3}$. Perform then a further blow-up of type (B) on the exceptional divisor $E_{3}$ to arrive at the six-dimensional F-theory base $\hat{B}_{0} = \mathrm{Bl}_{4}(\mathbb{F}_{n})$. The intersection properties of the strict transforms of the curves affected by the successive blow-ups are represented in \cref{fig:restriction-star-degenerations-example-1}.
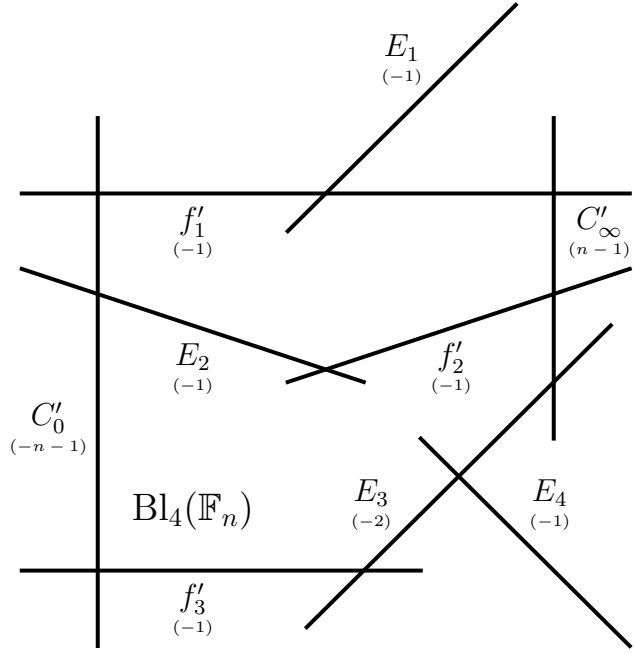
\begin{figure}[t!]
    \centering
    \begin{tikzpicture}
		\node (0) at (-3.5, 3.5) {};
		\node (1) at (-4.5, 2.5) {};
		\node (2) at (-3.5, -3.5) {};
		\node (3) at (-4.5, -2.5) {};
		\node (5) at (3.5, 2.5) {};
		\node (6) at (2.5, 3.5) {};
		\node (7) at (0.75, -2.5) {};
		\node (9) at (2.5, -0.75) {};
		\node (10) at (3.25, 0.75) {};
		\node (11) at (-0.75, -3.25) {};
		\node (12) at (0.75, -0.75) {};
		\node (13) at (3.5, -3.5) {};
		\node (16) at (-1, 2) {};
		\node (17) at (2, 5) {};
		\node (18) at (-4.5, 1.5) {};
		\node (19) at (0, 0) {};
		\node (20) at (-1, 0) {};
		\node (21) at (3.5, 1.5) {};

		\draw [style=black line] (0.center) to (2.center);
		\draw [style=black line] (1.center) to (5.center);
		\draw [style=black line] (3.center) to (7.center);
		\draw [style=black line] (6.center) to (9.center);
		\draw [style=black line] (11.center) to (10.center);
		\draw [style=black line] (12.center) to (13.center);
		\draw [style=black line] (16.center) to (17.center);
		\draw [style=black line] (18.center) to (19.center);
		\draw [style=black line] (20.center) to (21.center);
		
		\node [label={[align=center]$C'_{0}$\\[-5pt] {\tiny $(-n-1)$}}] (22) at (-4.15, -1.25) {};
		\node [label={[align=center]$C'_{\infty}$\\[-5pt] {\tiny $(n-1)$}}] (23) at (3.1, 1.35) {};
		\node [label={[align=center]$f'_{1}$\\[-5pt] {\tiny $(-1)$}}] (24) at (-2.25, 1.35) {};
		\node [label={[align=center]$E_{1}$\\[-5pt] {\tiny $(-1)$}}] (25) at (0.5, 3.65) {};
		\node [label={[align=center]$f'_{2}$\\[-5pt] {\tiny $(-1)$}}] (26) at (1.15, -0.45) {};
		\node [label={[align=center]$E_{2}$\\[-5pt] {\tiny $(-1)$}}] (27) at (-2.25, -0.45) {};
		\node [label={[align=center]$f'_{3}$\\[-5pt] {\tiny $(-1)$}}] (28) at (-2.25, -3.65) {};
		\node [label={[align=center]$E_{3}$\\[-5pt] {\tiny $(-2)$}}] (29) at (0.1, -2.25) {};
		\node [label={[align=center]$E_{4}$\\[-5pt] {\tiny $(-1)$}}] (30) at (2.45, -2.25) {};
		\node [label={[align=center]{\large $\mathrm{Bl}_{4}(\mathbb{F}_{n})$}}] (31) at (-2.25, -2.25) {};
	\end{tikzpicture}
    \caption{Surface $\hat{B}_{0} = \mathrm{Bl}_{4}(\mathbb{F}_{n})$ obtained by performing three type (A) and one type (B) blow-ups on the Hirzebruch surface $\mathbb{F}_{n}$. The self-intersection of the depicted curves is shown below their names in parentheses.}
    \label{fig:restriction-star-degenerations-example-1}
\end{figure}

The anticanonical class $\overline{K}_{\hat{B}_{0}}$ after the blow-ups can be written as
\begin{equation}
	\overline{K}_{\hat{B}_{0}} = 2\tilde{C}_{0} + (2+n)\tilde{f} - E_{1} - E_{2} - E_{3} - 2E_{4}\,.
\end{equation}
It is useful to recall that we have the linear equivalences
\begin{subequations}
\begin{align}
	\tilde{C}_{0} &= C'_{0} + E_{2}\,,\\
	\tilde{f} &= f'_{1} + E_{1}\,,\\
	\tilde{f} &= f'_{2} + E_{2}\,,\\
	\tilde{f} &= f'_{3} + E_{3} + E_{4}\,,
\end{align}
\end{subequations}
from where we see that we need $n \geq 1$ for $\overline{K}_{\hat{B}_{0}}$ to be effective, as can be read in \eqref{eq:effectiveness-bound-type-B}. The type (B) blow-up opens the possibility of a $\{E_{4}, \smallbullet, \smallbullet \}$ triplet perhaps violating \cref{prop:restricting-star-degenerations}. Let us explore this possibility by first tuning $E_{4}$ to be non-minimal, which leads to the $\hat{G}$ divisor
\begin{equation}
	\hat{G} = (6E_{4}) + \left[ 12\tilde{C}_{0} + (12+6n)\tilde{f} - 6E_{1} - 6E_{2} - 6E_{3} - 18E_{4} \right]\,,
\end{equation}
where in parentheses we are showing the factorized curves and in square brackets the class of $\Ghres$. This tuning increases the effectiveness bound to $n \geq 2$. Tuning $E_{4}$ to be non-minimal forces additional factorizations. We can see from
\begin{subequations}
\begin{align}
	\left( \Ghres - \alpha E_{3} \right) \cdot E_{3} &= -6 + 2\alpha\,,\\
	\left( \Ghres - \alpha C'_{0} \right) \cdot C'_{0} &= 12 - (6-\alpha)(n+n_{0})\,,
\end{align}
\end{subequations}
where $n_{0} = 1$, that we have a factorization of at least $3E_{3}$ and $2C'_{0}$, leading to
\begin{equation}
	\hat{G} = (6E_{4}) + (3E_{3}) + (2C'_{0}) + \left[ 6\tilde{C}_{0} + 4C'_{0} + (12+6n)\tilde{f} - 6E_{1} - 9E_{3} - 18E_{4} \right]\,.
\end{equation}
We can now consider the candidates for triplet completion, which in this case are $C'_{0}$, $\tilde{C}_{\infty}$, $C'_{\infty}$, $\tilde{f}$, $f'_{1}$, $f'_{2}$, $f'_{3}$, $E_{1}$ and $E_{2}$. Let us list how all of them lead to $E_{3}$ factorizing non-minimally, with $E_{3} \cdot E_{4} = 1$.
\begin{enumerate}[label=(\arabic*)]
	\item\label{item:star-degenerations-example-1-E4-C0} $\{E_{4}, C'_{0}, \smallbullet \}$: A cascade of factorizations occurs between $f'_{3}$ and $E_{3}$, ultimately leading to $E_{3}$ factorizing non-minimally.
	
	\item $\{ E_{4}, \tilde{C}_{\infty}, \smallbullet \}$: This case is analogous to Case~\ref{item:star-degenerations-example-1-E4-C0}.
	
	\item\label{item:star-degenerations-example-1-E4-Cinfprime} $\{ E_{4}, C'_{\infty}, \smallbullet \}$: This directly leads to $E_{3}$ factorizing non-minimally.
	
	\item\label{item:star-degenerations-example-1-E4-f} $\{ E_{4}, \tilde{f}, \smallbullet \}$: This triggers a cascade of $C'_{0}$, $f'_{3}$ and $E_{3}$ factorizations, leading in the end to $E_{3}$ factorizing non-minimally.
	
	\item $\{ E_{4}, f'_{1}, \smallbullet \}$: This case is analogous to Case~\ref{item:star-degenerations-example-1-E4-f}.
	
	\item $\{ E_{4}, f'_{3}, \smallbullet \}$: This case is analogous to Case~\ref{item:star-degenerations-example-1-E4-Cinfprime}.
	
	\item $\{ E_{4}, E_{2}, \smallbullet \}$: This case is analogous to Case~\ref{item:star-degenerations-example-1-E4-f}.
	
	\item\label{item:star-degenerations-example-1-E4-E1} $\{ E_{4}, E_{1}, \smallbullet \}$: This tuning modifies the effectiveness bound to $n \geq 3$ and triggers a cascade of $C'_{0}$, $f'_{1}$, $f'_{3}$ and $E_{3}$ factorizations making $E_{3}$ non-minimal.
	
	\item $\{ E_{4}, f'_{2}, \smallbullet \}$: Tuning these two divisors to be non-minimal is possible without any non-minimal factorizations being forced, but the previous analysis leaves no candidates left to complete the triplet.
\end{enumerate}
\end{example}

\subsubsection{Type (C) and (D) blow-ups}

The discussion proceeds similarly for surfaces $\hat{B}_{0} = \mathrm{Bl}(\mathbb{F}_{n})$ obtained by also allowing for type~(C) and (D) blow-ups, with the effectiveness constraints becoming even more stringent for these cases. We only provide an illustrative example for this class of models, to avoid repeating a discussion analogous to the one above.

\begin{example}
\label{example:restriction-star-degenerations-example-2}
Consider the blown up Hirzebruch surface $\mathrm{Bl}_{4}(\mathbb{F}_{n})$ of \cref{example:restriction-star-degenerations-example-1} and perform an additional type (B) blow-up on $E_{4}$, with exceptional divisor $E_{5}$, and a type (C) blow-up at the intersection point $E_{4} \cap E_{5}$. This leads to the six-dimensional F-theory base $\hat{B}_{0} = \mathrm{Bl}_{6}(\mathbb{F}_{n})$, that we represent in \cref{fig:restriction-star-degenerations-example-2}.
\begin{figure}[t!]
    \centering
    \begin{tikzpicture}
		\node (0) at (-3.5, 3.5) {};
		\node (1) at (-4.5, 2.5) {};
		\node (2) at (-3.5, -3.5) {};
		\node (3) at (-4.5, -2.5) {};
		\node (5) at (3.5, 2.5) {};
		\node (6) at (2.5, 3.5) {};
		\node (7) at (0.75, -2.5) {};
		\node (9) at (2.5, -0.75) {};
		\node (10) at (3.25, 0.75) {};
		\node (11) at (-0.75, -3.25) {};
		\node (12) at (0.75, -0.75) {};
		\node (13) at (3.5, -3.5) {};
		\node (14) at (2.25, -3) {};
		\node (15) at (6.5, -3) {};
		\node (16) at (-1, 2) {};
		\node (17) at (2, 5) {};
		\node (18) at (-4.5, 1.5) {};
		\node (19) at (0, 0) {};
		\node (20) at (-1, 0) {};
		\node (21) at (3.5, 1.5) {};
		\node (22) at (5.25, -3.5) {};
		\node (23) at (8, -0.75) {};

		\draw [style=black line] (0.center) to (2.center);
		\draw [style=black line] (1.center) to (5.center);
		\draw [style=black line] (3.center) to (7.center);
		\draw [style=black line] (6.center) to (9.center);
		\draw [style=black line] (11.center) to (10.center);
		\draw [style=black line] (12.center) to (13.center);
		\draw [style=black line] (14.center) to (15.center);
		\draw [style=black line] (16.center) to (17.center);
		\draw [style=black line] (18.center) to (19.center);
		\draw [style=black line] (20.center) to (21.center);
		\draw [style=black line] (22.center) to (23.center);
		
		\node [label={[align=center]$C'_{0}$\\[-5pt] {\tiny $(-n-1)$}}] (23) at (-4.15, -1.25) {};
		\node [label={[align=center]$C'_{\infty}$\\[-5pt] {\tiny $(n-1)$}}] (24) at (3.1, 1.35) {};
		\node [label={[align=center]$f'_{1}$\\[-5pt] {\tiny $(-1)$}}] (25) at (-2.25, 1.35) {};
		\node [label={[align=center]$E_{1}$\\[-5pt] {\tiny $(-1)$}}] (26) at (0.5, 3.65) {};
		\node [label={[align=center]$f'_{2}$\\[-5pt] {\tiny $(-1)$}}] (27) at (1.15, -0.45) {};
		\node [label={[align=center]$E_{2}$\\[-5pt] {\tiny $(-1)$}}] (28) at (-2.25, -0.45) {};
		\node [label={[align=center]$f'_{3}$\\[-5pt] {\tiny $(-1)$}}] (29) at (-2.25, -3.65) {};
		\node [label={[align=center]$E_{3}$\\[-5pt] {\tiny $(-2)$}}] (30) at (0.1, -2.25) {};
		\node [label={[align=center]$E_{4}$\\[-5pt] {\tiny $(-3)$}}] (31) at (2.45, -2.25) {};
		\node [label={[align=center]$E_{5}$\\[-5pt] {\tiny $(-2)$}}] (32) at (6.35, -2.25) {};
		\node [label={[align=center]$E_{6}$\\[-5pt] {\tiny $(-1)$}}] (33) at (4.35, -3.1) {};
		\node [label={[align=center]{\large $\mathrm{Bl}_{6}(\mathbb{F}_{n})$}}] (34) at (-2.25, -2.25) {};
	\end{tikzpicture}
    \caption{Surface $\hat{B}_{0} = \mathrm{Bl}_{6}(\mathbb{F}_{n})$ obtained by performing three type (A), two type (B) and one type (C) blow-ups on the Hirzebruch surface $\mathbb{F}_{n}$. The self-intersection of the depicted curves is shown below their names in parentheses.}
    \label{fig:restriction-star-degenerations-example-2}
\end{figure}
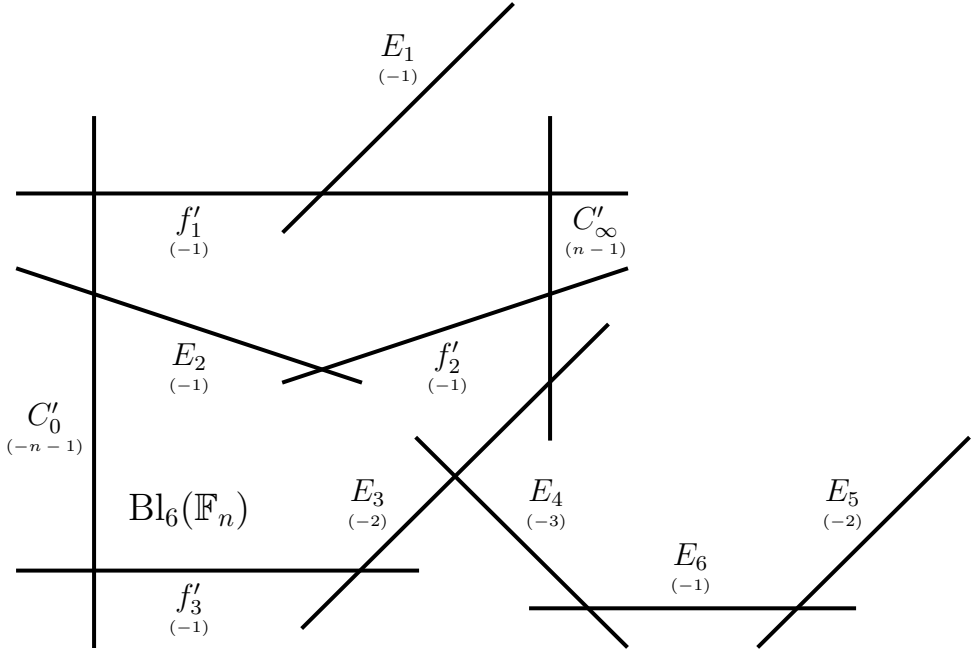

The anticanonical class $\overline{K}_{\hat{B}_{0}}$ after the blow-ups can be written as
\begin{equation}
	\overline{K}_{\hat{B}_{0}} = 2\tilde{C}_{0} + (2+n)\tilde{f} - E_{1} - E_{2} - E_{3} - 2E_{4} - 3E_{5} - 6E_{6}\,.
\end{equation}
Keeping in mind the linear equivalences 
\begin{subequations}
\begin{align}
	\tilde{C}_{0} &= C'_{0} + E_{2}\,,\\
	\tilde{f} &= f'_{1} + E_{1}\,,\\
	\tilde{f} &= f'_{2} + E_{2}\,,\\
	\tilde{f} &= f'_{3} + E_{3} + E_{4} + E_{5} + 2E_{6}
\end{align}
\end{subequations}
we see that the effectiveness bound is $n \geq 2$. The surface has non-Higgsable clusters, such that without any further tuning we already have
\begin{equation}
	\hat{G} = (2C'_{0}) + (2E_{3}) + (3E_{4}) + \left[ 6\tilde{C}_{0} + 4C'_{0} + (12+6n)\tilde{f} - 6E_{1} - 8E_{3} - 15E_{4} - 18E_{5} - 36E_{6} \right]\,.
\end{equation}
After the type (C) blow-up, \cref{lemma:triplet-candidates} tells us that we need to consider the candidate triplets $\{ E_{4}, \smallbullet, \smallbullet \}$, $\{ E_{5}, \smallbullet, \smallbullet \}$ and $\{ E_{6}, \smallbullet, \smallbullet \}$, which we do in turn.
\begin{enumerate}[label=(\arabic*)]
	\item\label{item:star-degenerations-example-2-E4-E5} $\{ E_{4}, \smallbullet, \smallbullet \}$: The analysis of these types of triplets is the same that was performed during the study of \cref{example:restriction-star-degenerations-example-1}, with the exception of the new candidate triplet $\{ E_{4}, E_{5}, \smallbullet \}$. Tuning these two divisors directly forces $E_{6}$ to factorize non-minimally, with $E_{4} \cdot E_{6} = E_{5} \cdot E_{6} = 1$.
	
	\item $\{ E_{5}, \smallbullet, \smallbullet \}$: Tuning $E_{5}$ to be non-minimal forces some additional factorizations in $\hat{G}$, leading to the residual $\Ghres$ divisor
	\begin{equation}
		\begin{split}
		\hat{G} &= (3C'_{0}) + (2E_{3}) + (4E_{4}) + (6E_{5}) + (4E_{6})\\
		&\quad+ \left[ 6\tilde{C}_{0} + 3C'_{0} + (12+6n)\tilde{f} - 6E_{1} - 8E_{3} - 16E_{4} - 24E_{5} - 40E_{6} \right]\,.
		\end{split}
	\end{equation}
	We can now consider the candidates for triplet completion, which in this case are $C'_{0}$, $\tilde{C}_{\infty}$, $C'_{\infty}$, $\tilde{f}$, $f'_{1}$, $f'_{2}$, $f'_{3}$, $E_{1}$, $E_{2}$ and $E_{3}$, since we have already considered $\{E_{5}, E_{4}, \smallbullet\}$ above in Case~\ref{item:star-degenerations-example-2-E4-E5}. With the exception of $\{ E_{5}, f'_{2}, \smallbullet \}$, the tuning of all these pairs leads to a cascade of factorizations ultimately making $E_{6}$ factorize non-minimally, with $E_{6} \cdot E_{5} = 1$. In the case of the pair $\{ E_{5}, E_{1}, \smallbullet \}$ we need to take into account during the analysis that the effectiveness bound is modified to $n \geq 4$ to allow for the tuning. Although the pair $\{ E_{5}, f'_{2}, \smallbullet \}$ can be tuned to be non-minimal without additional non-minimal factorizations, there is no candidate left in order to complete the triplet.
	
	\item $\{ E_{6}, \smallbullet, \smallbullet \}$: Tuning $E_{6}$ to be non-minimal forces some additional factorizations in $\hat{G}$, leading to the residual $\Ghres$ divisor
	\begin{equation}
		\begin{split}
		\hat{G} &= (3C'_{0}) + (3E_{3}) + (5E_{4}) + (3E_{5}) + (6E_{6})\\
		&\quad + \left[ 6\tilde{C}_{0} + 3C'_{0} + (12+6n)\tilde{f} - 6E_{1} - 9E_{3} - 17E_{4} - 21E_{5} - 42E_{6} \right]\,.
		\end{split}
	\end{equation}
	We can now consider the candidates for triplet completion, which are still $C'_{0}$, $\tilde{C}_{\infty}$, $C'_{\infty}$, $\tilde{f}$, $f'_{1}$, $f'_{2}$, $f'_{3}$, $E_{1}$, $E_{2}$ and $E_{3}$. Similarly to the previous case, tuning these pairs to be non-minimal forces a cascade of factorizations leading to $E_{4}$ factorizing non-minimally, with $E_{4} \cdot E_{6} = 1$. In the analysis of the pair $\{ E_{6}, E_{1}, \smallbullet \}$ we need to take into account that the effectiveness bound becomes $n \geq 4$ for the tuning to be possible.
\end{enumerate}
\end{example}
%auto-ignore

\section{Blowing down vertical components}
\label{sec:blowing-down-vertical-components}

While in horizontal models the open-chain resolution can be directly blown down to a component different from $B^{0}$, this is not possible in vertical models; for such models one needs to flop some curves before carrying out the blow-down.\footnote{See, e.g., \cite{Artin1970FormalModuli} for some classical results regarding the issue of contracting divisors.} This affects, in particular, the explicit method to remove horizontal and vertical Class~5 models discussed in \cite{ALWClass5} and some of the manipulations entering the discussion of obscured infinite-distance limits in \cref{sec:obscured-infinite-distance-limits}, since these may entail blowing down vertical components.

Let us consider a single infinite-distance limit degeneration $\hat{\rho}: \hat{\mathcal{Y}} \rightarrow D$ of vertical type and assume that $\hat{B}_{0} = \mathbb{F}_{n}$ with $n > 0$, to avoid the horizontal case for which the problem does not arise. If we were able to perform a blow-down directly to a component $B^{p}$, with $p = 1, \dotsc, P$, of the open-chain resolution, we would lose information about the degree $n$ of the line bundle $\mathcal{O}_{\mathbb{P}^{1}}(n)$ intervening in the construction of the original base component $\hat{B}_{0} = \mathbb{F}_{n}$. Instead, the geometry forces us to flop curves until the component to which we blow down has also the Hirzebruch surface $\mathbb{F}_{n}$ as its base.

This can be seen very directly from the toric diagrams used in \cref{sec:degenerations-Hirzebruch-models} to describe $\hat{\mathcal{B}}$ and the result of blowing it up $\mathcal{B}$. Consider first a horizontal model whose base $\hat{\mathcal{B}}$ has been blown-up leading to the toric fans described in \cref{sec:geometry-components-horizontal} and portrayed, in a two-component example, in \cref{fig:F1xDbhorfan}. Blowing down back to $B^{0}$ is clearly possible, since this just gives us the original model back. Consider instead that we want to blow down to an $B^{p}$ component besides $B^{0}$. For concreteness, let us focus on the two-component example and consider, hence, the blow-down to the component $B^{1}$ given by the map
\begin{equation}
	\sigma: \mathcal{B} \longrightarrow \mathrm{Bd}_{E_{0}}(\mathcal{B})\,,
\end{equation}
that contracts the $E_{0}$ divisor in $\mathcal{B}$. Let us denote the 1-skeleton of the fan $\Sigma_{X}$ of a toric variety $X$ by $\Sigma(1)_{X}$. We observe that the 1-skeleton $\Sigma(1)_{\mathrm{Bd}_{E_{0}}(\mathcal{B})} := \Sigma(1)_{\mathcal{B}} \setminus \{e_{0}\}$ can be completed into a fan in a unique way, the resulting variety being $\mathrm{Bd}_{E_{0}}(\mathcal{B})$, whose fan we denote $\Sigma_{\mathrm{Bd}_{E_{0}}(\mathcal{B})}$. The 1-skeleton $\Sigma(1)_{\mathcal{B}}$ can also only be completed into a fan in a unique way, yielding indeed $\Sigma_{\mathcal{B}}$. The~fan $\Sigma_{\mathcal{B}}$ is a refinement of $\Sigma_{\mathrm{Bd}_{E_{0}}(\mathcal{B})}$, in which the 3-cone spanned by $(t,w,e_{1})$ is subdivided by introducing the edge $e_{0}$ and the appropriate 2-cones.

Going through the same procedure for a vertical model, we see that the situation is different. Let us consider a vertical model whose base $\hat{\mathcal{B}}$ has been blown-up until the resolved base $\mathcal{B}$ with the toric fan described in \cref{sec:geometry-components-vertical} was obtained, that we plot for a two-component example in \cref{fig:F1xDbverfan}. As above, we discuss the geometry for the concrete case of a two-component model, but the results hold in general for a model with $P+1$ components.
\begin{figure}[t!]
     \centering
     \begin{subfigure}[b]{0.45\textwidth}
         \centering
         \includegraphics[width=\textwidth]{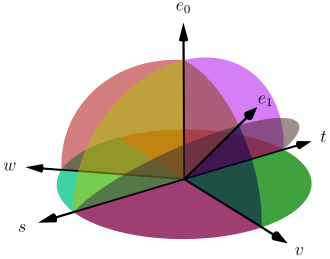}
         \caption{Toric fan of $\mathbb{F}_{n} \times \mathbb{C}$ blown up along $\mathcal{V} \cap \mathcal{U}$.}
         \label{fig:F1xDbverfan}
     \end{subfigure}
     \hfill
     \begin{subfigure}[b]{0.45\textwidth}
         \centering
         \includegraphics[width=\textwidth]{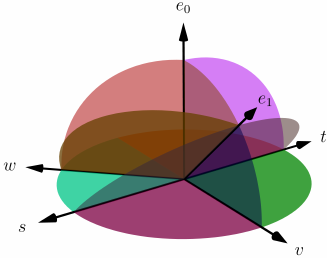}
         \caption{Toric fan of $\mathbb{F}_{n} \times \mathbb{C}$ blown up along $\mathcal{V} \cap \mathcal{U}$ after flopping the curve $\mathcal{S} \cap E_{0}$ to $\mathcal{W} \cap E_{1}$.}
         \label{fig:F1xDbverflopfan}
     \end{subfigure}
     \caption{Toric fans associated with a vertical model.}
     \label{fig:toric-fans-vertical}
\end{figure}
The 1-skeleton $\Sigma(1)_{\mathrm{Bd}_{E_{0}}(\mathcal{B})} := \Sigma(1)_{\mathcal{B}} \setminus \{e_{0}\}$ can only be completed into a fan $\Sigma_{\mathrm{Bd}_{e_{0}}(\mathcal{B})}$ in a unique way, giving the toric variety $\mathrm{Bd}_{e_{0}}(\mathcal{B})$, in which the $B^{0}$ component has been blown-down. This is not true for the 1-skeleton $\Sigma(1)_{\mathcal{B}}$; it can be completed into a fan in 3 different ways, that we now list. They differ in how the edges $s$, $w$, $e_{0}$ and $e_{1}$ are integrated into higher dimensional cones. Let us therefore focus on this aspect of the fan.
\begin{enumerate}[label=(\arabic*)]
	\item The first possibility is to take the 3-cone spanned by $(s,w,e_{0},e_{1})$ and the 2-cones given by its faces in order to define the fan. This leads to a singular toric variety $\mathrm{Sing}(\mathcal{B})$. The other fans that can be constructed from $\Sigma(1)_{\mathcal{B}}$, which are the ones of interest to us, are obtained by subdividing the 3-cone $(s,w,e_{0},e_{1})$ through the addition of a new 2-cone or curve, which can be done in two ways.
	
	\item We can subdivide the 3-cone $(s,w,e_{0},e_{1})$ by adding the curve $(s,e_{0})$. The resulting toric variety is smooth and corresponds to $\mathcal{B}$, whose fan $\Sigma_{\mathcal{B}}$ is represented in \cref{fig:F1xDbverfan} for a model constructed over $\mathbb{F}_{1}$. This toric fan is the one that is obtained from the blow-up process of $\hat{\mathcal{B}}$ described in \cref{sec:geometry-components-vertical}. Computing the orbit-closure of $e_{0}$ and $e_{1}$ we obtain the  $B^{0} = \mathbb{F}_{n}$ and $B^{1} = \mathbb{F}_{0}$ components, respectively, whose fans we represent in \cref{fig:vertical-preflop-orbit-closures}.
	\begin{figure}[t!]
	     \centering
	     \hspace{0.8cm}
	     \begin{subfigure}[b]{0.375\textwidth}
	         \centering
	         \includegraphics[width=\textwidth]{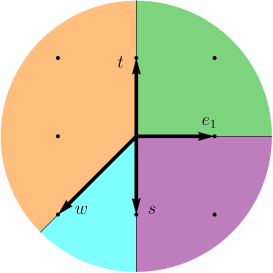}
	         \caption{Orbit closure of $e_{0}$ in $\Sigma_{\mathcal{B}}$.}
	         \label{fig:vertical-preflop-orbit-closures-E0}
	     \end{subfigure}
	     \hfill
	     \begin{subfigure}[b]{0.375\textwidth}
	         \centering
	         \includegraphics[width=\textwidth]{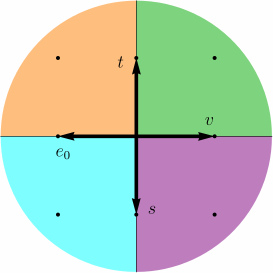}
	         \caption{Orbit closure of $e_{1}$ in $\Sigma_{\mathcal{B}}$.}
	         \label{fig:vertical-preflop-orbit-closures-E1}
	     \end{subfigure}
	     \hspace{0.8cm}
	     \caption{Geometry of the components of $\mathcal{B}$, resulting from the blow-up process.}
	     \label{fig:vertical-preflop-orbit-closures}
	\end{figure}
	
	\item Alternatively, we can subdivide the 3-cone $(s,w,e_{0},e_{1})$ by adding the curve $(w,e_{1})$. This yields a toric variety that we denote by $\mathrm{Flop}_{(s,e_{0})}(\mathcal{B})$, and whose fan $\Sigma_{\mathrm{Flop}_{(s,e_{0})}(\mathcal{B})}$ is repre\-sented in \cref{fig:F1xDbverflopfan} for a model constructed over $\mathbb{F}_{1}$. $\mathrm{Flop}_{(s,e_{0})}(\mathcal{B})$ is smooth when $n = 1$, and singular when $n \geq 2$. This toric fan naturally results from taking the fan $\Sigma_{\mathrm{Bd}_{e_{0}}(\mathcal{B})}$ and refining it by subdividing its $(t,w,e_{1})$ 3-cone through the addition of the $e_{0}$ edge and the appropriate 2- and 3-cones. Computing the orbit closure of $e_{0}$ and $e_{1}$ leads to $\mathbb{P}^{2}_{11n}$ and $\mathrm{Bl}^{1}(\mathbb{F}_{0})$ components, respectively, with the fans given in \cref{fig:vertical-flop-orbit-closures}.
	\begin{figure}[t!]
	     \centering
	     \hspace{0.8cm}
	     \begin{subfigure}[b]{0.375\textwidth}
	         \centering
	         \includegraphics[width=\textwidth]{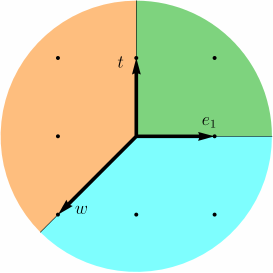}
	         \caption{Orbit closure of $e_{0}$ in $\Sigma_{\mathrm{Flop}_{(s,e_{0})}(\mathcal{B})}$.}
	         \label{fig:vertical-flop-orbit-closures-E0}
	     \end{subfigure}
	     \hfill
	     \begin{subfigure}[b]{0.375\textwidth}
	         \centering
	         \includegraphics[width=\textwidth]{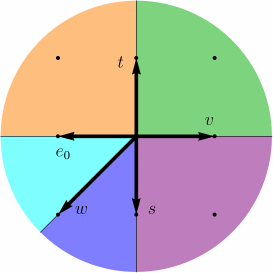}
	         \caption{Orbit closure of $e_{1}$ in $\Sigma_{\mathrm{Flop}_{(s,e_{0})}(\mathcal{B})}$.}
	         \label{fig:vertical-flop-orbit-closures-E1}
	     \end{subfigure}
	     \hspace{0.8cm}
	     \caption{Geometry of the components of $\mathrm{Flop}_{(s,e_{0})}(\mathcal{B})$, obtained from $\mathcal{B}$ by floppting the $(s,e_{0})$ curve into the $(w,e_{1})$ curve.}
	     \label{fig:vertical-flop-orbit-closures}
	\end{figure}
\end{enumerate}
The fan $\Sigma_{\mathrm{Flop}_{(s,e_{0})}(\mathcal{B})}$ is a refinement of the fan $\Sigma_{\mathrm{Bd}_{E_{0}}(\mathcal{B})}$, corresponding to restoring the $B^{0}$ component. The fan $\Sigma_{\mathcal{B}}$ is the result of blowing up $\hat{\mathcal{B}}$, but is not a refinement of $\Sigma_{\mathrm{Bd}_{E_{0}}(\mathcal{B})}$. $\Sigma_{\mathcal{B}}$ and $\Sigma_{\mathrm{Flop}_{(s,e_{0})}(\mathcal{B})}$ are connected by flopping the curve $(s,e_{0})$ into the curve $(w,e_{1})$. This shows that blowing down the $B^{0}$ component entails performing such a flop first. The effect of the flop is also apparent in the components. The component $B^{0}$ is originally $\mathbb{F}_{n}$, but we contract the $(-n)$-curve turning it into $\mathbb{P}^{2}_{11n}$. Meanwhile, the component $B^{1}$, with $\mathrm{F}_{0}$ geometry, acquires a new curve, which corresponds to a (weighted, for $n \geq 2$) blow-up by the addition of the edge $(-1,-n)$. Blowing then the $B^{0}$ component down removes the $(-1,0)$ edge out of the fan of the $B^{1}$ component, leaving us with an $\mathbb{F}_{n}$ surface.

As mentioned earlier, although the discussion has focused on a vertical two-component model, the same results apply for any vertical model in which we try to blow down the strict transform of $\mathcal{U}$, i.e.\ the $B^{0}$ component with $\mathbb{F}_{n}$ geometry, since the blow-up and blow-down operations are local. The components $B^{p} = \mathbb{F}_{0}$, where $p = 1, \dotsc, P$, can always be blown down without the need to perform a flop first.
%auto-ignore

\section{Polynomial factorization in rings with zero divisors}
\label{sec:polynomial-factorization-zero-divisors}

The physical defining polynomials $\fphys$, $\gphys$ and $\Dphys$ studied in \cref{sec:physical-discriminant} are elements of the ring with zero divisors $S_{\mathcal{B}}/I_{\tilde{\mathcal{U}}}$. This strays away from the context of integral domains, in which the factorization of polynomials is most commonly studied. The more general question of the factorization of polynomials in commutative rings with unity and zero divisors has been studied in the mathematical literature, see, e.g.,\ the non-comprehensive list of references \cite{Galovich1978,Anderson1996,Anderson2011,Anderson2013}. Here we only review some of the differences that arise with respect to the factorization theory in integral domains, referring to the literature for an in-depth analysis.

Following \cite{Anderson1996,Anderson2011,Anderson2013}, in a unital commutative ring there are three different notions of associate elements, that can be used to define four different notions of irreducible element.
\begin{definition}
    Let $R$ be a commutative ring with unity and let $a, b \in R$. Then:
    \begin{enumerate}
        \item $a$ and $b$ are associates, written $a \sim b$, if $\langle a \rangle = \langle b \rangle$;
        \item $a$ and $b$ are strong associates, written $a \approx b$, if $a =ub$ for some unit $u$;
        \item $a$ and $b$ are very strong associates, written $a \cong b$, if $a \sim b$ and either $a=b=0$ or $a \neq 0$ and $a=rb$ implies that $r$ is a unit.
    \end{enumerate}
\end{definition}
\begin{definition}
    Let $R$ be a commutative ring and $a \in R$ a non-unit. Then:
    \begin{enumerate}
        \item $a$ is irreducible if $a=bc \Rightarrow a \sim b$ or $a \sim c$;
        \item $a$ is strongly irreducible if $a=bc \Rightarrow a \approx b$ or $a \approx c$;
        \item $a$ is $m$-irreducible if $(a)$ is maximal among the proper principal ideals;
        \item $a$ is very strongly irreducible if $a=bc \Rightarrow a \cong b$ or $a \cong c$.
    \end{enumerate}
\end{definition}
With these definitions, the implications
\begin{equation}
    \begin{tikzcd}
        & & & \textrm{prime} \arrow[d, Rightarrow]\\
    \textrm{very strongly irreducible} \arrow[r, Rightarrow] & m\mathrm{-irreducible} \arrow[r, Rightarrow] & \textrm{strongly irreducible} \arrow[r, Rightarrow] & \textrm{irreducible}                
    \end{tikzcd}
\end{equation}
are satisfied for non-zero non-unit elements of the ring. In an integral domain, the implications in the bottom line reverse, and the four notions of irreducibility coincide. In a GCD domain, the vertical implication reverses as well.

Each of the notions of irreducibility listed above leads to a different notion of atomicity, the property of being able to express each non-zero non-unit element of $R$ as a finite product of irreducible elements. Namely, a unital commutative ring $R$ can be atomic, strongly atomic, \mbox{$m$-atomic} or very strongly atomic, see \cite{Anderson1996} for a detailed analysis. For example, unital commu\-tative rings satisfying the ascending chain condition on principal ideals are atomic, which means that Noetherian rings like $S_{\mathcal{B}}/I_{\tilde{\mathcal{U}}}$ are, in particular, atomic.

\bibliography{references}
\bibliographystyle{JHEP}

\end{document}